\documentclass[nofootinbib,superscriptaddress,floatfix,pra,table,aps, onecolumn, 10pt]{revtex4-2}

\usepackage{scrextend}
\usepackage{soul}
\usepackage[pdftex]{graphicx}
\usepackage{dcolumn}
\usepackage{bm}
\usepackage{epsfig}
\usepackage{amsmath}
\usepackage{amssymb} 
\usepackage{amsfonts}
\usepackage{pifont} 
\usepackage{mathtools}
\usepackage{multirow}
\usepackage{float}
\usepackage[utf8]{inputenc} 
\usepackage{epstopdf}
\usepackage{dsfont}
\usepackage{MnSymbol}
\usepackage{xspace}
\usepackage[most]{tcolorbox}
\tcbuselibrary{breakable}
\usepackage{multirow}
\usepackage{enumitem}
\usepackage{moreenum}
\usepackage{csquotes}
\usepackage{tabularray}
\usepackage{booktabs}
\usepackage[linesnumbered,lined,boxed,commentsnumbered,ruled,vlined]{algorithm2e}
\usepackage{dashbox}
\usepackage[dvipsnames]{xcolor}

\usepackage{amsthm}
\newtheorem{theorem}{Theorem}
\newtheorem*{theorem*}{Theorem}
\newtheorem{lemma}[theorem]{Lemma}
\newtheorem*{lemma*}{Lemma}
\newtheorem{conjecture}[theorem]{Conjecture}
\newtheorem{corollary}[theorem]{Corollary}
\newtheorem*{corollary*}{Corollary}
\newtheorem{proposition}[theorem]{Proposition}
\newtheorem*{proposition*}{Proposition}
\newtheorem{definition}[theorem]{Definition}
\newtheorem*{definition*}{Definition}
\newtheorem{example}[theorem]{Example}
\newtheorem*{example*}{Example}

\newtheorem*{remark*}{Remark}

\newtcbtheorem[]{tcexample}{Example}{breakable, colback=orange!3!white,colframe=orange!85!black, fonttitle=\bfseries}{Exa}

\DeclareMathOperator{\spn}{span}

\def\g{\mathfrak{g}}

\def\gl{\mathfrak{l}}

\def\su{\mathfrak{su}}
\def\Z{\mathbb{Z}}
\def\N{\mathbb{N}}
\def\C{\mathbb{C}}
\def\R{\mathbb{R}}
\def\Q{\mathbb{Q}}

\def\dd{\mathrm{d}}

\renewcommand{\sl}[2]{\mathfrak{sl}({#1},{#2})}
\newcommand{\Np}[1]{\N^{#1}_{\geq 0}}
\newcommand{\der}[1]{\frac{\dd}{\dd{#1}}}
\newcommand{\x}[2]{\chi_{#1}(#2)}

\newcommand{\abs}[1]{|#1|}

\newcommand{\ket}[1]{\ensuremath{| #1 \rangle}{}}

\newcommand{\vspan}[1]{\mathrm{span}\{ #1 \}{}}
\newcommand{\lie}[1]{\langle #1 \rangle_{\mathrm{Lie}}{}}
\newcommand{\kap}{\mathcal{K}}

\newcommand{\im}[1]{\mathrm{Im}(#1)}
\newcommand{\re}[1]{\mathrm{Re}(#1)}

\let\oldcite\cite
\renewcommand{\cite}[1]{\textup{\oldcite{#1}}}

\usepackage{bbm}

\newcommand{\graphlegend}[1]
{\begin{tikzpicture}[scale=0.75, every node/.style={transform shape}]
    \def\axy{0.3} 
    \def\bxy{0.325}   
    \coordinate (x) at (-2,0);
    \coordinate (z) at (2,0);	
    \node[ellipse, minimum width=0.6cm, minimum height=0.65cm, very thick,  draw = #1] (E2) at (x) {\large$x$};
    \node[ellipse, minimum width=0.6cm, minimum height=0.65cm, very thick, draw = #1] (E1) at (z) {\large$z$};
    \path (x) ++({\axy*cos(0)}, {\bxy*sin(0)}) coordinate (P1);
    \path (z) ++({\axy*cos(-180)}, {\bxy*sin(-180)}) coordinate (P2);
    \draw[->, very thick, draw=#1] (0+0.1,0)-- (0+0.3,0);
    \draw[very thick, draw = #1] (P1) -- (P2);
    \node at (0,0.35) {\large$y$};
    \node at (4,0) {\large$\hat{=}$};
    \node at (8.5,0) {\large$[x,y]=\kappa z$\quad with $\kappa\in\mathbb{F}^*$.};
\end{tikzpicture}
}

\usepackage[bookmarks=false,pdfstartview={FitH}]{hyperref}
 
\usepackage{tikz,xcolor}
\usetikzlibrary{patterns,patterns.meta,decorations.pathreplacing,angles,quotes, plotmarks, arrows.meta, calc, hobby}
\usetikzlibrary{shapes.geometric}
\usetikzlibrary{cd, babel, math} 
\definecolor{lime}{HTML}{A6CE39}
\DeclareRobustCommand{\orcidicon}{%
	\begin{tikzpicture}
	\draw[lime, fill=lime] (0,0) 
	circle [radius=0.16] 
	node[white] {{\fontfamily{qag}\selectfont \tiny ID}};
	\draw[white, fill=white] (-0.0625,0.095) 
	circle [radius=0.007];
	\end{tikzpicture}
	\hspace{-2mm}
}
\foreach \x in {A, ..., Z}{%
	\expandafter\xdef\csname orcid\x\endcsname{\noexpand\href{https://orcid.org/\csname orcidauthor\x\endcsname}{\noexpand\orcidicon}}
}

\begin{document}
\title{On the structural properties of Lie algebras via associated labeled directed graphs}
\date{\today}

\author{Tim Heib\,\orcidH{}}
\email{t.heib@fz-juelich.de}	
\affiliation{Institute for Quantum Computing Analytics (PGI-12), Forschungszentrum J\"ulich, 52425 J\"ulich, Germany}
\affiliation{Theoretical Physics, Universit\"at des Saarlandes, 66123 Saarbr\"ucken, Germany}
\author{David Edward Bruschi\,\orcidB{}}
\email{david.edward.bruschi@posteo.net}	
\email{d.e.bruschi@fz-juelich.de}	
\affiliation{Institute for Quantum Computing Analytics (PGI-12), Forschungszentrum J\"ulich, 52425 J\"ulich, Germany}
\affiliation{Theoretical Physics, Universit\"at des Saarlandes, 66123 Saarbr\"ucken, Germany}

\begin{abstract}
    We present a method for associating labeled directed graphs to finite-dimensional Lie algebras that are \emph{graph-admissible}. The proposed framework enables visual identification of key structural algebraic features. We first formalize this approach and establish several classes of Lie algebras that admit the categorization as graph-admissible Lie algebras, and then investigate suitable valid graphs for this approach by translating the antisymmetry of the Lie bracket and the Jacobi identity into rules that such graphs must satisfy. Furthermore, based on these foundations, we provide graph-theoretic criteria for solvability, nilpotency, presence of ideals, simplicity, semisimplicity, and reductiveness of an algebra. We also present practical algorithms for constructing such graphs and those associated with the lower central series and derived series via iterative pruning procedures. We further cover specific examples, such as the Schrödinger and Lorentz algebras, which illustrate the applicability of the developed tools to physically relevant cases. Moreover, we explore applications in physics, where the method facilitates computation of similarity transformations essential for determining quantum mechanical time evolution via the Lie-algebraic factorization method. We also discuss a possible generalization of the framework based on magma-gradings of Lie algebras and formulate related open questions. Our approach bridges algebraic and combinatorial perspectives, offering theoretical tools for this area of mathematical physics.
\end{abstract}
\maketitle


\section*{Introduction}

Lie algebras are fundamental objects in mathematics and theoretical physics that arise naturally in the study of continuous symmetries \cite{Lie:1891,Hall:Lie:groups:16}, differential equations \cite{Olver:1993}, and quantum mechanical systems, where they are ubiquitous \cite{Woit:2017,Isham:1999,Srednicki:2007,Bruschi:Xuereb:2024,A1:project,Gargiulo:2026,Heib:2026:universality}. Their structural properties, such as simplicity, solvability, or nilpotency, play a crucial role in classification theory and representation theory \cite{Knapp:1996}, and find applications across diverse fields, including control theory \cite{Brockett:1973}, the theory of differential equations \cite{Wei:Norman:1963,Wei:Norman:1964}, continuous variable quantum information \cite{Adesso:Ragy:2014}, or particle physics \cite{Georgi:2000}. Thus, understanding and visualizing these properties is not only essential for gaining theoretical insight and easing computational algorithms, but also helpful for conceptual purposes, since visual representations can facilitate understanding of the structures of interest. Here, we aim to construct a visual representation by associating appropriate labeled directed graphs to Lie algebras.

The use of graphs to represent mathematical structures is a well-established and powerful technique. Graphs have long served as tools for encoding algebraic and combinatorial relationships, with notable examples including Feynman diagrams in quantum field theory \cite{Feynman:1949}, Cayley graphs for groups \cite{Cayley:1878}, and weight space diagrams in representation theory \cite{Hall:Lie:groups:16}. More recently, graph-based methods have been employed to classify Pauli Lie algebras \cite{Aguilar:2024}, describe quantum electrical circuits \cite{Ciani:2025}, or to determine the controllability of arrays of coupled qubits \cite{GagoEncinas:2023}. Graphs also serve as effective visualization tools for complex concepts, such as the depiction of unitary irreducible representations for certain Lie algebras related to the symmetries of the Schrödinger equation \cite{Bekaert:2012}, or the Fano plane and its connection to the multiplication table of specific Lie algebras \cite{Dahm:2010}.

This work introduces a framework for associating certain finite-dimensional Lie algebras with labeled directed graphs. The goal is not to produce a unique or basis-independent representation, but rather to construct a graph that reflects aspects of the internal structure of a given Lie algebra in a way that renders its properties visually accessible. Notably, the construction presented here is neither unique nor basis-independent: the same graph may, in certain cases, correspond to two distinct Lie algebras and, conversely, a single Lie algebra may be associated with multiple distinct graphs. Nevertheless, the proposed method enables the identification of solvable and nilpotent Lie algebras, the detection of ideals, and the formulation of graph-theoretic criteria for simplicity, semisimplicity, and reductiveness, thus highlighting the fact that Lie algebras that share the same graph also share similar structural properties.

The connection between graphs and Lie algebras has been explored in various contexts and through diverse implementations. Complex semisimple Lie algebras are, for example, traditionally classified using graphical tools such as Dynkin diagrams, which correspond to specific root systems \cite{Humphreys:1972}. This approach has been generalized to real semisimple Lie algebras using Satake diagrams \cite{helgason2024differential}.

An association between Lie algebras and graphs, more closely related to the idea presented here, has been introduced before \cite{Dani:Mainkar:2005}. In this case, one can construct unlabeled and undirected graphs and associate them with two-step nilpotent Lie algebras. This approach has been adopted \cite{Pouseele:2009,Aldi:2025,Barrionuevo:2025} and extended  \cite{Chakrabarti:2019,Molina:2023} in further work to incorporate labeled and directed edges. In this extended framework, vertices of the graph are identified with basis elements of a Lie algebra, and edges between two vertices are labeled with the result of the corresponding Lie bracket; the direction of the edge encodes the sign of the bracket. One advantage of this construction over the method presented here is that the labels of the edges do not have to coincide with the vertex elements. This approach has also been generalized to certain three-step nilpotent Lie algebras using Schreier graphs \cite{Ray:2015}, and extended to solvable Lie algebras \cite{Grantcharov:2017}. It is, however, important to note that this is not the only effort to associate nilpotent Lie algebras with graphs;  alternative constructions have been also pursued elsewhere \cite{Diaz:2003,Towers:2025}.

Another approach that explores the connection between graphs and Lie algebras begins with a simple graph and attempts to associate a Lie algebra to it \cite{Gintz:2018,Khovanova:1982}. This is done in order to classify these Lie algebras, which are denoted with the ambiguous term \enquote{graph Lie algebras}. However, they should not be confused with \enquote{graph algebras} in the sense of \cite{Raeburn:2005}, where one isn't concerned with Lie algebras but with $C^*$-algebras, or graph algebras as discussed in \cite{Kim:1980}. Moreover, the term \enquote{graph Lie algebra} has been used in the literature in various other contexts \cite{Duchamp:2002,Blumer:2025}, further contributing to the ambiguity in terminology. In contrast, the present work starts with a graph-admissible Lie algebra given a priori and constructs a graph in a consistent manner that encodes its bracket structure. This approach is therefore more closely related to the visualization technique utilized in \cite{Bock:2022}, where basis elements are represented as vertices, and the Lie bracket operations are visualized using lines that start at two distinct vertices and join at the element that results from performing the corresponding Lie bracket. However, that method is restricted to Engel-type algebras only, which are a subclass of nilpotent Lie algebras.

To address and overcome the limitations of the aforementioned approaches, we propose a construction based on labeled directed graphs. In this framework, the vertices of the graph represent a spanning set of the Lie algebra. The starting vertex of an edge corresponds to the element acting via the adjoint action, the edge-label indicates the input element of the adjoint action (also taken from the vertex set), and the target vertex represents the result of the bracket, up to a non-zero proportionality constant. This construction is conceptually similar to the one considered in \cite{Boza:2014}, although that approach employs unlabeled edges, is restricted to Lie algebras over finite fields, and determines all actions via the same adjoint actions of a single external element not represented as a vertex. Our idea is also inspired by commutativity graphs, which connect elements of a Lie algebra whenever their bracket is non-zero \cite{Wang:2020}, and the commutator graphs discussed in \cite{West:2025}, which - similar to \cite{Boza:2014} - link elements based on the adjoint action of a specific element from the Lie algebra. It is important to note, however, that the term \enquote{commutator graph} is not used uniformly in the literature. In some contexts, it refers to graphs that connect elements which commute \cite{Pisanski:1989,Kurzynski:2012}, i.e. whose Lie bracket would vanish - effectively the opposite of the commutativity graph discussed in \cite{Wang:2020}. These constructions are among the closest relatives to the framework presented here. Our approach, therefore, can be seen as a direct adaptation of the idea of commutator graphs developed in \cite{Kazi:2025}, although generalized from strings of Pauli operators to arbitrary finite-dimensional graph-admissible Lie algebras.

Other intriguing approaches exist that investigate the interplay between graphs and Lie algebras, such as the idea presented in  \cite{Marcolli:2015}, where one incorporates the concept of Graph Grammars, originally developed in the field of theoretical computer science, into the construction of Lie algebraic structures, or the one studied in \cite{Caceres:2013}, where the combinatorial structures are the main focus.

This work is organized as follows: In Section~\ref{sec:background}, we recall the essential definitions and conventions pertaining to Lie algebra and graph theory. In Section~\ref{sec:graph:admissible}, we introduce the central concept of \emph{graph-admissibility}, which serves to characterize whether a given Lie algebra can be faithfully represented by a finite labeled directed graph. Section~\ref{sec:valid:graphs} establishes criteria for determining the validity of such graph representations. In Section~\ref{sec:structural:properties}, we examine key structural properties of Lie algebras through their associated graphs, including the identification of central elements, the computation of the derived and lower central series, the detection of ideals, and the formulation of criteria for simplicity and semisimplicity. Section~\ref{sec:prominent:examples} presents applications of this approach to prominent examples, such as the Schrödinger algebra and the Lorentz algebra, and discusses applications to the classification of finite-dimensional Lie algebras of physical relevance. Section~\ref{sec:discussion} discusses limitations and possible extensions, and outlines a series of open questions and conjectures that emerge from the graph-theoretic framework. Finally, Section~\ref{sec:outlook} offers some concluding remarks.

\begin{figure}[htpb]
    \centering
    \includegraphics[width=0.875\linewidth]{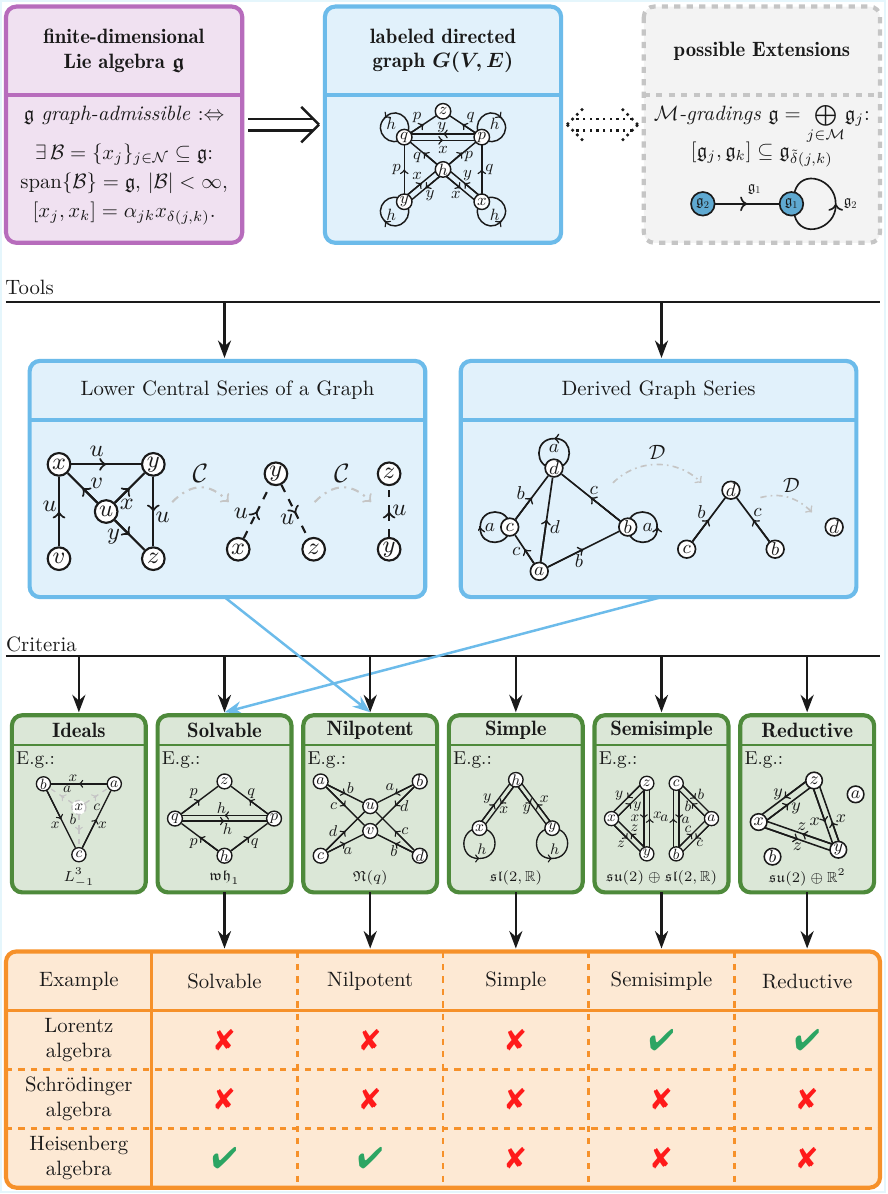}
    \caption{Schematic overview of the workflow presented in this paper. We consider finite-dimensional Lie algebras and begin by introducing the notion of graph-admissibility. Based on this concept, we construct labeled directed graphs associated with these algebras. Subsequently, we develop graph-theoretical tools, such as the lower central series and the derived series for graphs, in order to analyze structural properties. These tools enable the formulation of classification criteria for Lie algebras using graph-theoretical methods. Finally, we consider some specific examples and discuss possible extensions of these methods.}
    \label{fig:schamtics}
\end{figure}

\section{Background concepts}\label{sec:background}
We begin by recalling the most essential definitions relevant to this work, as well as introducing the conventions used. This encompasses several definitions for classifying Lie algebras and graphs. In addition, we adopt the following conventions: We denote any field with the symbol $\mathbb{F}$, and the set $\mathbb{F}$ without the zero element as $\mathbb{F}^*:=\mathbb{F}\setminus\{0\}$. Unless explicitly stated otherwise, all fields considered in this work are assumed to have characteristic zero. Moreover, we assume that $\operatorname{char}(\mathbb{F})\neq 2$ holds for all fields discussed. We also use the symbol \enquote{$\propto$} to express a proportionality relation of two vectors. However, we emphasize that we write $x\propto y$ if and only if there exists a constant $\kappa\in\mathbb{F}^*$ such that $x=\kappa y$ for two vectors $x,y$ in a vector space $V$ over the field $\mathbb{F}$. Consequently, the only vector proportional to the zero vector is the zero vector itself.

\subsection{Fundamental concepts of Lie algebras}
In this subsection we introduce the foundational definitions and algebraic structures that underpin the study of Lie algebras. These include the Lie bracket, the center, and key series used to classify Lie algebras based on their structural properties. These definitions will serve as the backbone for the associated graphs in the later sections. 
\begin{definition}[Lie algebra]
    A \emph{Lie algebra} is a vector space $\g$ over a field $\mathbb{F}$ equipped with an antisymmetric bilinear binary operation $[\;\cdot\;,\;\cdot\;]:\g\times\g\to\g$ that satisfies the \emph{Jacobi identity}
    \begin{align}\label{eqn:Jacobi:identity}
        [x,[y,z]]+[y,[z,x]]+[z,[x,y]]=0
    \end{align}
    for all $x,y,z\in\g$. This operation is called the \emph{Lie bracket} \cite{Knapp:1996,Isham:1999}.
\end{definition} 
Such algebras can also be regarded as inner product spaces when equipped with an appropriate scalar product. For instance, one can use the scalar product defined on the basis elements by $\langle x_j,x_k\rangle=\delta_{jk}$, in case of a finite basis, where $\delta_{jk}$ denotes the Kronecker delta.

\begin{tcolorbox}[breakable, colback=Cerulean!3!white,colframe=Cerulean!85!black,title=\textbf{Remark:} The commutator as Lie bracket]
    In physics, the term Lie bracket is often used interchangeably with the commutator \cite{Arnold:1989}. For example, for square matrices, the commutator is defined as $$[\,\cdot\,,\,\cdot\,]:\C^{n\times n}\times \C^{n\times n}\to \C^{n\times n},\,(\boldsymbol{M}_1,\boldsymbol{M}_2)\mapsto[\boldsymbol{M}_1,\boldsymbol{M}_2]:=\boldsymbol{M}_1\boldsymbol{M}_2-\boldsymbol{M}_2\boldsymbol{M}_1.$$ 
    However, the Lie bracket is not restricted to this operation. Other bilinear operations also satisfy the axioms of a Lie bracket. A notable example is the \emph{vector cross product} \cite{Marsden:1999}, which, according to \cite{Arfken:2013}, is defined by $$\times: \C^3\times \C^3\to \C^3,(v=v^je_j,w=w^je_j)\mapsto v\times w ={\varepsilon_{jk}}^\ell v^j w^ke_\ell,$$ 
    where we employ Einstein's summation convention \cite{Einstein:1916}. Here, $e_j$ denotes the canonical basis vectors of $\C^3$ and $\varepsilon_{jk\ell}$ is the completely antisymmetric Levi-Civita symbol \cite{Arfken:2013}. This example illustrates that the Lie bracket concept extends beyond commutators to other algebraic structures.
\end{tcolorbox}

We are interested in exploring the structural properties of a given Lie algebra. A basic example of such a structural property is the center of the Lie algebra.
\begin{definition}[Center]
    Let $\g$ be a Lie algebra. The \emph{center} of $\g$ is defined as the set $\mathcal{Z}(\g)$ of all elements in $\g$ whose Lie bracket with all other elements in the algebra vanishes identically. That is
    \begin{align}
        \mathcal{Z}(\g):=C_\g(\g):=\{x\in\g\,\mid\,[x,y]=0\text{ for all }y\in\g\},
    \end{align}
    where $C_\g(\g)$ is the \emph{centralizer} of $\g$ in $\g$ \cite{Knapp:1996}. If $\mathcal{Z}(\g)=\g$, we say that $\g$ is \emph{abelian} \cite{Knapp:1996}.
\end{definition}
Other structural properties of Lie algebras can be studied by considering, for example, the derived series or the lower central series. These are generally defined by the operation $[\mathfrak{h}_1,\mathfrak{h}_2]:=\spn\{[h_1,h_2]|\;h_1\in\mathfrak{h}_1,\,h_2\in\mathfrak{h}_2\}$, where $\mathfrak{h}_1,\mathfrak{h}_2\subseteq\g$ are two subalgebras of the Lie algebra $\g$ \cite{Knapp:1996}. It is clear that $[\mathfrak{h}_1,\mathfrak{h}_2]$ is itself a vector space but not necessarily a Lie algebra. 
\begin{definition}[Derived series]
    The \emph{derived series} $(\mathcal{D}^k\g)_{k\in\mathbb{N}_{\geq0}}$ of a Lie algebra $\g$,  also known as the \emph{commutator series} \cite{Knapp:1996}, consists of the algebras constructed by the following recursive procedure 
    \begin{align}
        \mathcal{D}^{k+1}\g:=[\mathcal{D}^k\g,\mathcal{D}^k\g]\quad\text{for all }k\in\N_{\geq0},\quad\text{where}\quad\mathcal{D}^0\g:=\g.
    \end{align}
    The algebras $\mathcal{D}^k\g$ are furthermore denoted as the \emph{$k$-th derived algebra} \cite{Serre:1992}.
\end{definition}

In addition to the derived series, we are also interested in the lower central series, which is defined through a similar recursive construction. 

\begin{definition}[Lower central series]
    The \emph{lower central series} $(\mathcal{C}^k\g)_{k\in\mathbb{N}_{\geq0}}$ of a Lie algebra $\g$  \cite{Humphreys:1972}, also known as the \emph{descending central series} \cite{Serre:1992}, is recursively defined by:
    \begin{align}
        \mathcal{C}^{k+1}\g:=[\g,\mathcal{C}^k\g]\quad\text{for all }k\in\N_{\geq0},\quad\text{where}\quad\mathcal{C}^0\g:=\g.
    \end{align}
\end{definition}

These series --- the derived and lower central series --- are fundamental tools to classify Lie algebras according to their structural properties, namely solvability and nilpotency of a Lie algebra.

\begin{definition}[Solvable and Nilpotent algebra]\label{def:solvable:nilpotent}
    A Lie algebra is \emph{solvable} if its derived series terminates after a finite number of steps; i.e., there exists an integer $k_*\in\N_{\geq0}$ such that $\mathcal{D}^{k_*}\g=\{0\}$, where $\{0\}$ denotes the trivial Lie algebra.
    Similarly, a Lie algebra is \emph{nilpotent} if the lower central series terminates after a finite number of steps; i.e., there exists an integer $k_*\in\N_{\geq0}$, such that $\mathcal{C}^{k_*}\g=\{0\}$. In the case of nilpotent algebras, the smallest such integer is referred to as the \emph{index} of the Lie algebra.
\end{definition}

Note that the notion of nilpotency is analogous to the nilpotency of linear maps, as formalized by Engel's theorem \cite{Humphreys:1972}. It is evident that $\mathcal{C}^1\g=\mathcal{D}^1\g$ and $\mathcal{D}^k\g\subseteq\mathcal{C}^k\g$ for all $k\in \N_{\geq0}$, which shows that every nilpotent Lie algebra is also solvable, although the converse does not hold: there exist solvable Lie algebras that are not nilpotent such as the two-dimensional affine Lie algebra $\mathfrak{aff}(\mathbb{F})$ \cite{Andrada:2005}.

Another important concept in the study and classification of Lie algebras is that of an \emph{ideal}. The ideal plays a central role in defining and understanding structural notions such as \enquote{simplicity}, \enquote{semisimplicity}, and the \enquote{radical} of a Lie algebra.
\begin{definition}[Ideal]
    Let $\g$ be a Lie algebra. An ideal $\mathfrak{h}$ in $\g$ is a subspace $\mathfrak{h}\subseteq \g$ satisfying $[\mathfrak{g},\mathfrak{h}]\subseteq \mathfrak{h}$ \cite{Knapp:1996}.
\end{definition}

Typical examples of ideals include the center $\mathcal{Z}(\g)$ of a Lie algebra, the first derived algebra $\mathcal{D}^1\g$, or the kernel of a Lie-algebra homomorphism \cite{Knapp:1996}.

A particularly important ideal in the structure theory of Lie algebras is the \emph{radical}, which captures the largest solvable part of a Lie algebra.
\begin{definition}[Radical]
    Let $\g$ be a Lie algebra. The \emph{radical} of $\g$, denoted $\operatorname{rad}(\g)$, is defined as the largest solvable ideal of $\g$ \cite{Knapp:1996}.
\end{definition}
The radical of a given Lie algebra exists and is unique, as guaranteed by Proposition 1.12 from \cite{Knapp:1996}. We now define the notions of \emph{simple} and \emph{semisimple} Lie algebras, which play a pivotal role in the structural classification of finite-dimensional Lie algebras. A complete classification of these Lie algebras over the fields of the real and complex numbers has been established \cite{Cartan:1894,Killing:erster:1888,Killing:zweiter:1888,Killing:dritter:1889,Killing:vierter:1890,Magyar:1989}. Moreover, it is known that any finite-dimensional Lie algebra over a field of characteristic zero can be expressed as a semidirect sum of a semisimple Lie algebra and its radical. This is known as the Levi-Mal'tsev decomposition \cite{Kuzmin:1977}. These results form the foundation for understanding the decomposition and internal structure of general Lie algebras.
\begin{definition}[Simple algebra, Semisimple algebra]\label{def:simple:and:semisimple}
    Let $\g$ be a finite-dimensional Lie algebra. This algebra is \emph{simple} if it is non-abelian and contains no proper non-zero ideals. If $\g$ has no non-zero solvable ideals, i.e., $\operatorname{rad}(\g)=\{0\}$, we say that the Lie algebra $\g$ is \emph{semisimple}.
\end{definition}
It follows immediately that every simple Lie algebra is also semisimple. The final structural classification of Lie algebras we want to discuss is that of \emph{reductive} versus \emph{non-reductive} Lie algebras. This distinction generalizes the concept of semisimplicity by allowing for a non-trivial center, while still retaining a semisimple component.
\begin{definition}[Reductive algebra]
    Let $\g$ be a finite-dimensional Lie algebra. The algebra $\g$ is a \emph{reductive algebra} if to each ideal $\mathfrak{i}_1$ in $\g$ there exists a corresponding ideal $\mathfrak{i}_2$ in $\g$ such that $\g=\mathfrak{i}_1\oplus\mathfrak{i}_2$.
\end{definition}
Note that if $\g$ is reductive, then $\g=\mathcal{D}^1\g\oplus\mathcal{Z}(\g)$, where $\mathcal{D}^1\g$ is semisimple \cite{Knapp:1996}.

\subsection{Graphs}
A \emph{labeled directed graph} $G(V,E)$ is defined by an ordered tuple $(V,E)$, where $V$ is a set of vertices and $E$ is a collection of labeled directed edges. Each edge is a three-tuple $e=(v_\mathrm{s},v_\mathrm{l},v_\mathrm{e})$, where the first entry $v_\mathrm{s}$ denotes the starting vertex (i.e., the vertex from which the edge originates), the second entry $v_\mathrm{l}$ is the vertex that labels the edge, and the third entry $v_\mathrm{e}$ is the vertex at which the edge terminates (i.e., the vertex to which the edge points to). Accordingly, we associate with every graph $G(V,E)$ the three functions $\varpi_\mathrm{s},\varpi_\mathrm{l},\varpi_\mathrm{e}:E\to V$, which map each edge $e\in E$ to a vertex from $V$, corresponding to its starting vertex, labeling vertex, and ending vertex respectively. Concretely, we have
\begin{align}
   \varpi_\mathrm{s}:e\equiv(v_\mathrm{s},v_\mathrm{l},v_\mathrm{e})&\mapsto \varpi_\mathrm{s}(e):=v_\mathrm{s},\;&\;
   \varpi_\mathrm{l}:e\equiv(v_\mathrm{s},v_\mathrm{l},v_\mathrm{e})&\mapsto \varpi_\mathrm{l}(e):=v_\mathrm{l},\;&\;
   \varpi_\mathrm{e}:e\equiv(v_\mathrm{s},v_\mathrm{l},v_\mathrm{e})&\mapsto \varpi_\mathrm{e}(e):=v_\mathrm{e}.
\end{align}
It is worth noting that, in the literature, a labeled directed graph $G(V,E)$ is also defined as the triple $G(V,E,\varphi)$, where $V$ denotes the set of vertices, $E$ the set of edges (typically consisting of tuples containing only the starting and ending vertices), and the label of each edge is determined by a corresponding labeling function $\varphi$ \cite{Molina:2023}. In some contexts, the function $\varphi$ is also used to describe the direction of a given edge \cite{Williamson:2010}. Our approach differs from the ones just mentioned.

We now introduce the notions of \enquote{walk}, \enquote{path}, \enquote{trail}, \enquote{circuit}, and \enquote{cycle}. Since our focus is exclusively on directed graphs, we restrict ourselves to their directed counterparts.
\begin{definition}\label{def:walk:trail:path:cycle}
    Let $G(V,E)$ be a labeled directed graph. Let $S:=(v_1,e_1,v_2,e_2,\ldots,e_{s},v_{s+1})$ be a finite, non-trivial sequence of vertices $\{v_j\}_{j=1}^{s+1}\subseteq V$ and edges $\{e_j\}_{j=1}^s\subseteq E$. In line with the literature  \cite{Williamson:2010}, we define the following:
    \begin{enumerate}[label = (\roman*)]
        \item If $\varpi_\mathrm{s}(e_j)=v_j$ and $\varpi_{\mathrm{e}}(e_j)=v_{j+1}$ for all $j\in\{1,\ldots,s\}$, we call $S$ a \emph{directed walk}.
        \item If $S$ is a directed walk and $|\{v_j\}_{j=1}^{s+1}|=s+1$, then $S$ is a \emph{directed path}.
        \item If $S$ is a directed walk and $|\{e_j\}_{j=1}^s|=s$, then $S$ is a \emph{directed trail}.
        \item If $S$ is a directed walk and $v_1=v_{s+1}$, then $S$ is a \emph{closed directed walk}.
        \item If $S$ is a directed trail that is also a closed directed walk, then $S$ is a \emph{directed circuit}.
        \item If $S$ is a directed circuit and $S'=(v_1,e_1,\ldots,e_{s-1},v_s)$ is a directed path, then $S$ is a \emph{cycle}.
    \end{enumerate}
    The \emph{length} of $S$ is defined as $\operatorname{len}(S):=s$.
\end{definition}

In certain contexts, it is more convenient to consider the subgraphs induced by sequences of vertices and edges, rather than the sequences themselves.

\begin{definition}[Induced graphs]\label{def:induced:graph}
    Let $G(V,E)$ be a labeled directed graph, and let $S:=(v_1,e_1,v_2,\ldots,e_s,v_{s+1})$ be a non-trivial sequence of vertices $\{v_j\}_{j=1}^{s+1}\subseteq V$ and edges $\{e_j\}_{j=1}^s\subseteq E$ that satisfy $\varpi_\mathrm{s}(e_j),\varpi_\mathrm{e}(e_j)\in \{v_j\}_{j=1}^{s+1}$. The \emph{graph induced by} $S$ is defined as $G(\{v_j\}_{j=1}^{s+1},\{e_j\}_{j=1}^s)=:G_S$. 
\end{definition}

Since we are exclusively concerned with labeled directed graphs, it is furthermore useful to introduce the concept of \enquote{self-contained subgraphs}, which are subgraphs, where the labels of the edges are only given by the vertices of the subgraph itself.

\begin{definition}\label{def:self:contained}[Self-contained subgraph]
    Let $G(V,E)$ be a labeled directed graph, and let $G_S\equiv G(\Tilde{V},\Tilde{E})\subseteq G(V,E)$ be a subgraph induced by a non-trivial sequence $S$ of vertices and edges. The subgraph $G_S$ is \emph{self-contained} if and only if $\varpi_\mathrm{l}(e)\in \Tilde{V}$ for all edges $e\in\Tilde{E}$.
\end{definition}

Note that a self-contained subgraph induced by a directed walk $W$ does not necessarily imply that $W$ is a closed directed walk. For instance, consider the directed walk $W=(v_1,e_1=(v_1,v_2,v_2),v_2)$: it induces a self-contained subgraph but is not closed.

Finally, we introduce the concept of \emph{unconnected subgraphs} of a directed graph, which, as shown in Proposition~\ref{prop:semisimple:lie:algebra:unconnected:subgraphs} and Lemma~\ref{lem:direct:sum:if:unconnected}, is closely related to the direct sums of the associated Lie algebras.
\begin{definition}\label{def:unconnected:subgraphs}
    Let $G(V,E)$ be a labeled directed graph, and let $G(V_1,E_1), G(V_2,E_2)$ be two subgraphs of $G(V,E)$. These subgraphs are \emph{unconnected} if $V_1\cap V_2=\emptyset$ and, for all $v_1\in V_1$ and $v_2\in V_2$ there exists no edge $e\in E$ such that either $\varpi_\mathrm{s}(e)=v_1$ and $\varpi_\mathrm{e}(e)=v_2$, or vice versa.
\end{definition}
Note that unconnected subgraphs are closely related to the concept of \emph{disconnected components}. These are unconnected subgraphs of a given graph $G(V,E)$ that do not contain further unconnected subgraphs themselves and are not subgraphs of any larger unconnected subgraph \cite{Bollobas:1998,Harary:1955}.

\begin{definition}[Equivalent graphs]\label{def:equivalent:graphs}
    Let $G(V_1,E_1)$ and $G(V_2,E_2)$ be two labeled directed graphs. These two graphs are isomorphic if and only if there exists a bijection $\varphi:V_1\to V_2$, such that $E_2=\{(\varphi(\varpi_\mathrm{s}(e)),\varphi(\varpi_\mathrm{l}(e)),\varphi(\varpi_\mathrm{e}(e)))\;\mid\;e\in E_1\}$. We call the mapping $\varphi$ an isomorphism and write, in the case of isomorphic graphs, $G(V_1,E_1)\cong G(V_2,E_2)$ \cite{Bang:Jensen:2009}. We say two graphs are equivalent if they are isomorphic.
\end{definition}

\section{Formulation of the Central idea}\label{sec:graph:admissible}

This section lays the groundwork for the graph-theoretic framework developed in this paper. It introduces the core concept of \emph{graph-admissibility}, which classifies finite-dimensional Lie algebras according to whether they can be associated with a finite labeled directed graph. The first part of this section formalizes this notion, while the second examines the concrete construction of graphs that encode the structural properties of such Lie algebras.

\subsection{Graph-admissible Lie algebras}\label{sec:graph:admissible:subsec:notion}
Let us begin by introducing a specific classification of $n$-dimensional Lie algebras. This classification is central to the graph-based approach developed in this work, as it determines whether a Lie algebra can be associated with a finite labeled directed graph in a way that preserves its structural properties.

\begin{definition}[Graph-admissible Lie algebra]\label{def:graph:admissible}
    Let $\g$ be an $n$-dimensional Lie algebra. We say that $\g$ is \emph{minimal-graph-admissible} if it admits a basis $\{x_j\}_{j=1}^n$ such that the Lie bracket satisfies:
    \begin{align}
        [x_j,x_k]&=\alpha_{jk} x_{\delta(j,k)}\quad\quad\text{ for all }j,k\in\mathcal{N}:=\{1,\ldots,n\},\label{eqn:desired:basis}
    \end{align}
    where $\boldsymbol{\alpha}\in\mathbb{F}^{n\times n}$ is an antisymmetric matrix and $\delta:\mathcal{N}\times\mathcal{N}\to\mathcal{N}$ is a symmetric function with concrete action $\delta:(j,k)\mapsto\delta(j,k)=j'\in\mathcal{N}$ that maps pairs of indices to another index.

    If $\g$ does not necessarily admit such a basis satisfying the relations \eqref{eqn:desired:basis}, but instead an overcomplete basis $\{\Tilde{x}_j\}_{j=1}^m$ of non-zero elements with $m\geq n$ satisfying
    \begin{align}
        [\Tilde{x}_j,\Tilde{x}_k]=\Tilde{\alpha}_{jk}\Tilde{x}_{\delta(j,k)}\quad\quad\text{ for all }j,k\in\mathcal{M}:=\{1,\ldots,m\},\label{eqn:desired:basis:overcomplete}\quad\quad\text{ and }\quad\quad\Tilde{x}_j\propto \Tilde{x}_k \quad\quad\text{if and only if }j=k,
    \end{align}
    where $\Tilde{\boldsymbol{\alpha}}\in\mathbb{F}^{m\times m}$ is an antisymmetric matrix and $\delta:\mathcal{M}\times\mathcal{M}\to\mathcal{M}$ is a symmetric function on the space of indices as defined above, then we call $\g$ \emph{redundant-graph-admissible}.

    Any finite-dimensional Lie algebra that is either minimal-graph-admissible or redundant-graph-admissible is called \emph{graph-admissible}. Finally, we say that an algebra is \emph{non-graph-admissible} if it is not graph-admissible. 
    
    To ensure consistency and for later convenience, we define $\delta(j,k):=0$ whenever $\alpha_{jk}=0$ (or $\Tilde{\alpha}_{jk}=0$ respectively), and set $x_0:=0=:\Tilde{x}_0$. We also set $\alpha_{j0}:=0=:\alpha_{0j}$ for $j\in\mathcal{N}\cup\{0\}$, and analogously $\tilde{\alpha}_{j0}:=0=:\tilde{\alpha}_{0j}$ for $j\in\mathcal{M}\cup \{0\}$. Moreover, we set $\delta(j,0)=0=\delta(0,j)$ for all $j\in\mathcal{N}$ or $j\in \mathcal{M}$, respectively. This allows us to extend the domain of $\delta$ to include zero: $\delta:\mathcal{N}\cup\{0\}\times\mathcal{N}\cup\{0\}\to\mathcal{N}\cup\{0\}$ or $\delta:\mathcal{M}\cup\{0\}\times\mathcal{M}\cup\{0\}\to\mathcal{M}\cup\{0\}$ if $\g$ is minimal-graph-admissible or respectively redundant-graph-admissible.
\end{definition}

Note that in the more general setting of (anti)commutative algebras, a basis that satisfies the conditions of Definition~\ref{def:graph:admissible} for minimal-graph-admissibility is called \emph{multiplicative}~\cite{Martin:2015,Martin:2016}. Multiplicative bases are a well-established concept occurring for example in Lie algebra set gradings~\cite{Martin:2014} and in the study of Lie algebras with nice bases~\cite{Dere:2026}.

\begin{remark*}
    Minimal-graph-admissibility, redundant-graph-admissibility, graph-admissibility, and non-graph-admiss\-ibility are intrinsic properties of a Lie algebra. Indeed, these properties are defined by the existence, or non-existence, of a suitable basis or finite overcomplete spanning set satisfying the defining relation~\eqref{eqn:desired:basis} or~\eqref{eqn:desired:basis:overcomplete} of the definition above. Their validity therefore does not depend on an initially chosen basis. This should not be confused with basis independence of the associated graph: not every basis or overcomplete spanning set is suitable, and different suitable choices for the same Lie algebra may, in fact, produce non-isomorphic labeled directed graphs. In this respect, graph-admissibility is analogous to diagonalizability: the existence or non-existence of a suitable basis is intrinsic to a given vector-space endomorphism, whereas the concrete representation obtained after choosing such a basis need not be unique~\cite{Lukas:2022}.
\end{remark*}

The condition of graph-admissibility is restrictive; however, it has a powerful consequence: it ensures that the Lie bracket of any two basis elements is proportional to a third basis element, thus allowing us to encode the algebra as a graph where the vertices represent basis elements and the edges represent the bracket relations of these basis elements. The classification of graph-admissible and non-graph-admissible Lie algebras is therefore crucial because it determines whether a Lie algebra can be faithfully represented as a graph. Naturally, the first question that arises is:
\begin{quote}
    \textbf{Q}: \emph{Do there exist Lie algebras that can be uniquely classified as minimal-graph-admissible, redundant-graph-admissible but not minimal-graph-admissible, or non-graph-admissible}?
\end{quote}

As we will demonstrate later, the answer to this question is \enquote{yes}. To begin addressing it, we first show that there exist both minimal-graph-admissible Lie algebras and Lie algebras that are not minimal-graph-admissible. This assertion is formalized in the following theorem, which is proven by providing concrete examples. 

\begin{theorem}\label{thm:existncae:of:non:minimal:graph:admissible:}
    There exist finite-dimensional Lie algebras that are minimal-graph-admissible and ones that are not minimal-graph-admissible.
\end{theorem}

A detailed proof of this result is provided in Appendix~\ref{app:supplementary:proofs}. The one-dimensional abelian Lie algebra provides a trivial example of a minimal-graph-admissible Lie algebra. To obtain Lie algebras that are not minimal-graph-admissible, we consider the family of real Lie algebras $L_\alpha^3=\spn\{x_1,x_2,x_3\}$ determined by the Lie brackets 
\begin{align}
    [x_1,x_2]&=0,\;&\;[x_1,x_3]&=-x_2,\;&\;[x_2,x_3]&=-\alpha x_1-x_2,
\end{align}
which occurs in the classification of low-dimensional Lie algebras~\cite{DeGraaf:2004}. As shown in Appendix~\ref{app:supplementary:proofs}, the real Lie algebra $L_\alpha^3$ is not minimal-graph-admissible whenever $\alpha<-1/4$. The obstruction is specific to the real base field and ultimately arises from the negative discriminant of the relevant quadratic equation and the fact that $\R$ is not algebraically closed.

After establishing that not every finite-dimensional Lie algebra is minimal-graph-admissible, it is natural to ask whether there are Lie algebras that are graph-admissible in a more general sense, i.e., if there are redundant-graph-admissible Lie algebras that are not minimal-graph-admissible, which would justify the distinction of graph-admissible Lie algebras into these two classes. The following corollary addresses this question directly. It is, furthermore, noteworthy that by Definition~\ref{def:graph:admissible} every Lie algebra that is minimal-graph-admissible is also redundant-graph-admissible.

\begin{corollary}\label{cor:existance:of:redundant:but:not:mimimak:graph:admissibel:algebras}
    There exist Lie algebras that are graph-admissible but not minimal-graph-admissible.
\end{corollary}

This corollary can be shown by a direct demonstration; a simple example is the real Lie algebra $L_{-1}^3$, a member of the family of Lie algebras considered in the proof of Theorem~\ref{thm:existncae:of:non:minimal:graph:admissible:}. A detailed derivation can be found in Appendix~\ref{app:supplementary:proofs}. Nonetheless, it is worth recalling that the proof of Theorem~\ref{thm:existncae:of:non:minimal:graph:admissible:} relies on the fact that the field of real numbers is not algebraically closed and therefore highlights that whether a Lie algebra is minimal-graph-admissible can depend on the field over which the Lie algebra is defined. This subtlety underscores the importance of distinguishing between minimal- and redundant-graph-admissibility.

After having seen that not all finite-dimensional Lie algebras are minimal-graph-admissible, it is important to recognize that this property is not exceptional. In fact, there exist important classes of Lie algebras that are minimal-graph-admissible, such as abelian Lie algebras or finite-dimensional subalgebras of the complex single-mode Weyl algebra. 

\begin{proposition}
    Every finite-dimensional abelian Lie algebra is minimal-graph-admissible. Moreover, every basis of such a Lie algebra is suitable in the sense of Definition~\ref{def:graph:admissible}, and all corresponding minimal graphs are isomorphic.
\end{proposition}

\begin{proof}
    This follows immediately from Definition~\ref{def:graph:admissible}, since in an abelian Lie algebra, the Lie bracket of any two elements vanishes. Indeed, this argument applies to every basis. Hence, every corresponding minimal graph consists of $n$ isolated vertices, and any two such graphs are isomorphic.
\end{proof}

\begin{proposition}
    Every finite-dimensional Lie subalgebra of the complex one-mode Weyl algebra $A_1$ is minimal-graph-admissible. 
\end{proposition}

\begin{proof}
    The classification of every non-abelian finite-dimensional Lie subalgebra of the one-mode complex Weyl algebra $A_1$ is known \cite{TST:2006}, and any such Lie algebra is minimal-graph-admissible, which can be verified by delving into the examples given in \cite{TST:2006}. The previous observation that any finite-dimensional abelian Lie algebra is minimal-graph-admissible concludes the proof.
\end{proof}

The same conclusion holds for every finite-dimensional complex Lie subalgebra of the universal enveloping algebra $U(\sl{2}{\C})$, as established by Corollary~34 in the cited work~\cite{Heib:2026:Usl2C}.

The previous proposition raises the question of whether every finite-dimensional Lie subalgebra of the complex $n$-mode Weyl algebra $A_n$ is minimal-graph-admissible. This is a natural question, given the central role of the Weyl algebra in mathematical physics and the fact that its structure becomes increasingly rich as the number of modes increases \cite{Bruschi:Xuereb:2024}. That having been said, the case at hand here is subtle. While the one-mode Weyl algebra $A_1$ is highly restrictive, allowing only certain types of Lie subalgebras \cite{TST:2006}, higher-mode Weyl algebras with $n>1$ admit a much broader class of finite-dimensional Lie subalgebras. This leads us to the following conjecture:
\begin{conjecture}
    There exist non-graph-admissible Lie algebras in the complex $n$-mode Weyl algebra $A_n$ for all $n\geq2$.
\end{conjecture}
This conjecture is motivated by several structural results from the literature: Theorem 4.2 in \cite{Joseph:1970} implies that only rank-one semisimple Lie algebras can be faithfully realized in the complex Weyl algebra $A_1$ of a single mode. It is furthermore known that the only semisimple Lie algebra that can be faithfully realized in $A_1$ is the three-dimensional special linear Lie algebra $\sl{2}{\C}$ \cite{Klep:2010}. This occurs because $\C$ is an algebraically closed field with characteristic zero, and also because the centralizer $C_{A_1}(x)$ of every non-central element $x\in A_1\setminus \mathcal{Z}(A_1)$ is an abelian subalgebra of $A_1$ \cite{Dixmier:1968}, a property known as the \emph{abelian centralizer condition} \cite{A1:project}. Note that $\sl{2}{\C}$ is already the only simple Lie algebra of rank one over $\C$ \cite{Benkart:1984}. 

For $n>1$, the situation changes dramatically. The abelian centralizer condition mentioned above no longer holds for $n\geq2$, and the structure of the Weyl algebra becomes much richer. In particular, the $n$-mode Weyl algebra may admit semisimple Lie algebras of rank greater than one as finite-dimensional subalgebras. This opens the possibility for more intricate finite-dimensional subalgebras, including those that are not minimal-graph-admissible or even graph-admissible.

Beyond the special cases already discussed, there are entire classes of Lie algebras for which graph-admissibility is always guaranteed. In particular, nilpotent Lie algebras, which play an important role in the structure theory of Lie algebras, control theory, representation theory, and geometric analysis \cite{Knapp:1996,Molina:2023}, always admit an appropriate basis or spanning set, as demanded by Definition~\ref{def:graph:admissible}, although they are not necessarily minimal-graph-admissible. The following proposition makes this precise:

\begin{proposition}\label{prop:nilpotent:graph:admissible}
    Every finite-dimensional nilpotent Lie algebra $\g$ is graph-admissible.
\end{proposition}

\begin{proof}
    Let $\g$ be a finite-dimensional nilpotent Lie algebra with basis $\mathcal{B}=\{x_j\}_{j=1}^n$. We construct now an (overcomplete) basis of $\g$ that satisfies the requirement of graph-admissibility, as follows: Start with the original basis $\mathcal{B}$. Then add for every pair $(x_j,x_k)\in\mathcal{B}\times\mathcal{B}$ the Lie bracket $[x_j,x_k]$ to $\mathcal{B}$ if it is non-zero and not proportional to an element that is already in $\mathcal{B}$. These new elements necessarily belong to the first term of the lower central series $\mathcal{C}^1\g=[\g,\g]$. One iterates this now with higher-order Lie brackets. That is, for all $x_j$ in the original basis, one considers the brackets of the form $[x_j,[x_k,x_\ell]]$ and $[[x_j,x_k],[x_p,x_q]]$. If any such bracket is non-zero and not proportional to an element already in $\mathcal{B}$, one adds it to $\mathcal{B}$. One continuously repeats this process and adds at each step new non-zero brackets that are not already proportional to existing elements in $\mathcal{B}$. The crucial observation is that, at each stage, the newly added elements belong to the $k$-th algebra $\mathcal{C}^k\g$ of the lower central series of $\g$. Since $\g$ is nilpotent, there exists a positive integer $k_*$ such that $\mathcal{C}^k\g=\{0\}$ for all $k\geq k_*$. This means that after finitely many steps, no non-zero elements can be generated by further Lie brackets, and the process terminates. The resulting (possibly overcomplete) basis $\mathcal{B}$ is finite and satisfies the bracket relations \eqref{eqn:desired:basis:overcomplete} required for graph-admissibility (i.e., every non-zero Lie bracket of elements from the new $\mathcal{B}$ is proportional to a basis element of $\mathcal{B}$ now that it has been appropriately extended). Therefore, every finite-dimensional nilpotent Lie algebra is graph-admissible. 
\end{proof}

The property of graph-admissibility is not limited to Lie algebras over infinite fields or those with particular algebraic structures. In fact, it holds in a very general setting, including all finite-dimensional Lie algebras defined over a finite field $\mathbb{F}$. This is summarized in the following straightforward proposition:

\begin{proposition}\label{prop:graph:admissibility:finite:fields}
    Every finite-dimensional Lie algebra over a finite field $\mathbb{F}$ is graph-admissible.
\end{proposition}

\begin{proof}
    Let $\g$ be a finite-dimensional Lie algebra over a finite field $\mathbb{F}$. Since $\dim(\g):=n<\infty$ and $|\mathbb{F}|<\infty$, the algebra consists of finitely many elements $\{x_j\}_{j\in\mathcal{J}}$, where $|\mathcal{J}|=|\mathbb{F}|^n$, as can be shown by a simple induction on the dimension $n$. Now, consider an overcomplete spanning set formed by taking one representative from each non-zero element modulo scalar multiples. Now let us construct a subset $\mathcal{B}\subseteq \g$ using the following iterative procedure. Consider the empty set $\mathcal{B}=\emptyset$. Then, fix some ordering of the elements in $\mathcal{J}$. We can assume, without loss of generality that $\mathcal{J}=\{1,\ldots,|\mathbb{F}|^n\}$. Now we iterate over the elements in $\mathcal{J}$. That is, for each $j\in\mathcal{J}$ in order we add $x_j$ to $\mathcal{B}$ if $x_j\neq 0$ and there is no $x_k\in\mathcal{B}$ such that $x_j\propto x_k$. The resulting set $\mathcal{B}$ is finite and spans $\g$, since every non-zero element of $\g$ is proportional to an element of $\mathcal{B}$. Moreover, and crucially, for any two elements in this set, their Lie bracket is either zero or proportional to another element in $\mathcal{B}$. Therefore, the proportionality condition from Definition~\ref{def:graph:admissible} is satisfied, making $\g$ at least redundant-graph-admissible and hence graph-admissible.
\end{proof}

\begin{remark*}
    The spanning set utilized in the proof above contains $(|\mathbb{F}|^n-1)/(|\mathbb{F}|-1)$ elements, one for each one-dimensional subspace of $\g$. For fixed $\mathbb{F}$, its size therefore grows exponentially with $n$, making this construction generally impractical for visualization or computation. Nevertheless, it shows that any such Lie algebra $\g$ is graph-admissible, and, in practice, a suitable spanning set with fewer terms should be constructed, if possible. 
\end{remark*}

While the preceding results establish that several distinct classes of Lie algebras are graph-admissible, thereby demonstrating the broad applicability of the graph-admissibility framework, it is crucial to recognize that this property does not universally extend to all Lie algebras and its limitations have to be acknowledged: Not every finite-dimensional Lie algebra admits a basis or overcomplete generating set that satisfies the structural conditions required for a graph representation to be viable as defined in this work. Demonstrating the existence of such exceptions is essential as it delineates the boundaries of the graph-admissibility approach. The following theorem formalizes this observation and conclusively answers, along with Theorem~\ref{thm:existncae:of:non:minimal:graph:admissible:} and Corollary~\ref{cor:existance:of:redundant:but:not:mimimak:graph:admissibel:algebras}, the introductory question of whether the classification into minimal-graph-admissible, redundant-graph-admissible, and non-graph-admissible Lie algebras is justified.

\begin{theorem}\label{thm:existence:of:non:graph:admissible:lie:algebras}
    There exist finite-dimensional Lie algebras that are non-graph-admissible.
\end{theorem}

The theorem follows directly from Lemma~\ref{lem:L:alpha:three:graph:admissible:if:and:only:if}, delegated to Appendix~\ref{app:supplementary:proofs}. That lemma establishes that the real Lie algebra $L_\alpha^3$, already used to establish Theorem~\ref{thm:existncae:of:non:minimal:graph:admissible:} and Corollary~\ref{cor:existance:of:redundant:but:not:mimimak:graph:admissibel:algebras}, is graph-admissible if and only if
\begin{align*}
    \alpha>-\frac{1}{4}\qquad\text{or}\qquad\frac{1}{\pi}\arctan\left(\frac{\sqrt{|1+4\alpha|}}{1+2\sqrt{|\alpha|}}\right)\in\Q^*.
\end{align*}
For $\alpha=-1/4$, neither condition is satisfied, because the displayed quotient equals zero and $0\notin\Q^*$. Consequently, $L_{-1/4}^3$ is a concrete finite-dimensional non-graph-admissible Lie algebra. The graph-admissibility obstruction at $\alpha=-1/4$ is a non-trivial Jordan block, which persists over every field of characteristic zero, even if it is closed, as the argument in Appendix~\ref{app:supplementary:proofs} shows.

Despite the implications of Theorem~\ref{thm:existence:of:non:graph:admissible:lie:algebras}, it is crucial to emphasize that the graph-admissibility framework remains applicable to a broad and structurally significant group of Lie algebras. In particular, the following result demonstrates that complex semisimple Lie algebras, which play a central role in the classification theory of Lie algebras, fall within the scope of graph-admissibility.

\begin{theorem}\label{thm:semisimple:lie:algebra:graph:admissible}
    Every finite-dimensional semisimple Lie algebra over an algebraically closed field $\mathbb{F}$ is graph-admissible.
\end{theorem}

\begin{proof}
    This is an immediate consequence of the classification of semisimple Lie algebras over algebraically closed fields \cite{Hall:Lie:groups:16,Fuchs:1995,Humphreys:1972}. Any such semisimple Lie algebra $\g$, associated to a set of $r$ simple roots, is uniquely determined by $3r$ elements $\mathcal{S}:=\{E_\pm^{(j)},H^{(j)}\}_{j=1}^{r}$ which satisfy the \emph{Chevalley-Serre} relations:
    \begin{align*}
        [H^{(j)},H^{(k)}]&=0,\;&\;[H^{(j)},E_\pm^{(k)}]&=\pm A_{kj}E_\pm^{(k)},\;&\;[E_+^{(j)},E_-^{(k)}]&=\delta_{jk}H^{(j)},
    \end{align*}
    as well as the Jacobi identity and the \emph{Serre} conditions
    \begin{align*}
        \left(\operatorname{ad}_{E_\pm^{(j)}}\right)^{1-A_{kj}}E_\pm^{(k)}=0\qquad\text{for } j,k\in\{1,\ldots,r\}\text{ with }j\neq k.
    \end{align*}
    Here $\operatorname{ad}_{x}(y):=[x,y]$ denotes the adjoint action, $\operatorname{ad}_x^n(y)$ is a short-hand-notation for $n$-chained brackets, i.e., $\operatorname{ad}_x^n(y)=[x,[x,\ldots,[x,y]]]$, and $\boldsymbol{A}\in\Z^{r\times r}$ is the corresponding Cartan matrix \cite{Fuchs:1995}. The set $\mathcal{S}$ generates the Lie algebra, i.e., $\lie{\mathcal{S}}=\g$. Furthermore, the set $\mathcal{S}$ can be extended to the set $\mathcal{S}'$ which includes all associated root vectors $E_\alpha$, not only those associated with simple roots. The additional root vectors can be obtained as non-zero iterated Lie brackets of the Chevalley generators, while the Serre conditions ensure that the number of these elements remain finite. Moreover, for every non-simple positive root $\alpha\in\Phi^+$, we add the Cartan element $H_\alpha:=[E_\alpha,E_{-\alpha}]$, which need not be proportional to any of the simple Cartan generators $H^{(j)}$. Setting $H_{\alpha_j}:=H^{(j)}$ for every simple root $\alpha_j$, the resulting set can be written as $\mathcal{S}':=\{E_\alpha\mid\alpha\in\Phi\}\cup\{H_\alpha\mid\alpha\in\Phi^+\}$ for the associated root system $\Phi$ with positive roots belonging to $\Phi^+$. Because $\Phi$ is finite~\cite{Hall:Lie:groups:16}, the extended set $\mathcal{S}'$ remains finite. It also spans $\g$, since it contains a non-zero vector from every root space and the elements corresponding to the simple roots contain a basis of the Cartan subalgebra~\cite{Humphreys:1972}. Furthermore, no two distinct elements of $\mathcal{S}'$ are proportional: root vectors associated with distinct roots belong to distinct root spaces, no root vector is proportional to an element of the Cartan subalgebra, and no two distinct positive coroots are proportional because the root system is reduced.

    It remains to verify relation~\eqref{eqn:desired:basis:overcomplete}. For all $\alpha,\beta\in\Phi^+$ and $\gamma\in\Phi$, the standard root-space relations imply
    \begin{align*}
        [H_\alpha,H_\beta]&=0,\;&\;[H_\alpha,E_\gamma]&=\gamma(H_\alpha)E_\gamma.
    \end{align*}
    Moreover, for all $\gamma,\delta\in\Phi$,
    \begin{align*}
        [E_\gamma,E_\delta]\in\left\{\begin{matrix}
            \spn\{E_{\gamma+\delta}\}, &\gamma+\delta\in\Phi,\\
            \spn\{H_\gamma\}, &\delta=-\gamma\text{ and }\gamma\in\Phi^+,\\
            \spn\{H_{-\gamma}\}, &\delta=-\gamma\text{ and }-\gamma\in\Phi^+,\\
            \{0\}, &\gamma+\delta\notin\Phi\cup\{0\}.
        \end{matrix}\right.
    \end{align*}
    Hence, the Lie bracket of any two elements of $\mathcal{S}'$ either vanishes or is proportional to an element of $\mathcal{S}'$. Therefore, the finite overcomplete spanning set $\mathcal{S}'$ satisfies relation~\eqref{eqn:desired:basis:overcomplete}, making $\g$ graph-admissible according to Definition~\ref{def:graph:admissible}.
\end{proof}

\begin{remark*}
    Note that the preceding theorem is solely intended to establish the graph-admissibility of finite-dimensional semisimple Lie algebras. However, a more interesting question is whether every such Lie algebra is minimal-graph-admissible and, if not, which ones are and what distinguishes those from the ones that are not. This question has been partially addressed in the grading literature, where substantially stronger results have been presented. In particular, $\sl{n}{\C}$, $\mathfrak{so}(n,\C)$, and $\mathfrak{sp}(2^k,\C)$ for $k\geq 1$ admit gradings with one-dimensional homogeneous components~\cite{Havlicek:2002,Han:2016}, and hence possess multiplicative bases, making them minimal-graph-admissible. The existence of such bases for all finite-dimensional semisimple Lie algebras is also stated in~\cite{Martin:2015}; we do not address this stronger assertion here.
\end{remark*}

This result generalizes naturally to reductive Lie algebras:
\begin{corollary}\label{cor:rductive:Lie:algebras:graph:admissible}
    Every finite-dimensional reductive Lie algebra over an algebraically closed field $\mathbb{F}$ is graph-admissible.
\end{corollary}

\begin{proof}
    Let $\g$ be a reductive Lie algebra. Then $\g=[\g,\g]\oplus\mathcal{Z}(\g)$, where $[\g,\g]$ is semisimple \cite{Knapp:1996}. Due to Theorem~\ref{thm:semisimple:lie:algebra:graph:admissible}, the semisimple part $[\g,\g]$ admits a (possibly overcomplete) basis $\mathcal{B}$ that satisfies \eqref{eqn:desired:basis:overcomplete}. This (possibly overcomplete) basis can be extended to a (possibly overcomplete) basis of $\g$, by adding elements from the center $\mathcal{Z}(\g)$, which are abelian. This implies that the extended basis satisfies \eqref{eqn:desired:basis:overcomplete}, making $\g$ graph-admissible.
\end{proof}

\subsection{Graphs associated with finite-dimensional Lie algebras}
Having established that a wide range of finite-dimensional Lie algebras are graph-admissible, whilst also highlighting the restrictions to the definition given, it is natural to ask how one can visualize such algebras in a way that makes their structure transparent. The method described here provides a systematic procedure for associating a labeled directed graph to any graph-admissible Lie algebra, making it possible to \enquote{see} the Lie algebra's properties at a glance.

Let $\g$ be a finite-dimensional graph-admissible Lie algebra.
We provide the following method for visualization: First, choose a (possibly overcomplete) basis $\mathcal{B}=\{x_j\}_{j\in\mathcal{N}}$ (or $\mathcal{B}=\{x_j\}_{j\in\mathcal{M}}$) that satisfies relations \eqref{eqn:desired:basis} (respectively, \eqref{eqn:desired:basis:overcomplete}). Then, draw a vertex $v_j$ for every basis element $x_j\in \mathcal{B}$ in a two-dimensional plane and label it using the corresponding element, i.e., $v_j=x_j$ \footnote{To improve the precision of the notation, it can be helpful to introduce the function $\varrho_{\mathrm{v}}:V\to \g$ that maps a vertex $v\in V$ to the corresponding labeling basis element $x_j\in\g$. Similarly, for $K\in\{\mathrm{s},\mathrm{l},\mathrm{e}\}$ define the functions $\varrho_{K}:E\to\g$, $e\mapsto\varrho_\mathrm{v}(\varpi_{K}(e))$ that map the starting vertex, labeling vertex, or end vertex of an edge $e$ to the corresponding element in $\g$. This definition would remove ambiguity in the interpretation of graph elements. However, for convenience, we choose to identify a vertex with its labeling element directly.}. In the next step, one considers every pair of basis elements $x_j,x_k\in\mathcal{B}$. If there exists an element $x_\ell\in \mathcal{B}$ such that $[x_j,x_k]\propto x_\ell$, one draws a directed edge from $v_j$ to $v_\ell$, labeled by $v_k$. In other words, whenever $[x,y]=\kappa z$ with $\kappa\in\mathbb{F}^*$, draw an edge from $x$ to $z$ labeled by $y$:
\begin{quote}
    \centering\graphlegend{Black}
\end{quote}
Thus, for any finite-dimensional graph-admissible Lie algebra $\g$, we can associate the graph $G(V,E)$ with $\g$ by identifying $V=\mathcal{B}$ and $E:=\{e=(v_\mathrm{s},v_\mathrm{l},v_\mathrm{e})\,\mid\,e\in V^{\times 3}\,\wedge\,[v_\mathrm{s},v_\mathrm{l}]\propto v_\mathrm{e}\}$. This construction can be concisely formalized as an algorithm for any graph-admissible Lie algebra, as shown in Algorithm~\ref{alg:creating:graph}. However, we want to note that here and in the subsequent pruning algorithms (Algorithms~\ref{alg:generating:generating:the:graph:derived:altered} and~\ref{alg:generating:lower:central:series}) an admissible (possibly overcomplete) basis $\mathcal{B}$ and its bracket data are supplied. Finding such a set or deciding graph-admissibility for a general input Lie algebra is beyond the scope of this work and left for future work.

\begin{algorithm}[htpb]
        \DontPrintSemicolon
        \KwData{A (possibly overcomplete) basis $\mathcal{B}=\{\Tilde{x}_j\}_{j\in\mathcal{M}}$ of the finite-dimensional graph-admissible Lie algebra $\g$, satisfying the Lie bracket relations \eqref{eqn:desired:basis:overcomplete}.}
        \KwResult{A labeled directed graph associated with the Lie algebra $\g$ constructed with respect to the given basis $\mathcal{B}$.}
        \SetKwData{Left}{left}\SetKwData{This}{this}\SetKwData{Up}{up}
        \SetKwFunction{Union}{Union}\SetKwFunction{FindCompress}{FindCompress}
        \SetKwInOut{Input}{input}\SetKwInOut{Output}{output}
    
        $V$ $\leftarrow$ $\mathcal{B}$
        \tcc*[h]{Initialize the vertex set with the basis $\mathcal{B}$}\;
        $E$ $\leftarrow$ $\emptyset$ 
        \tcc*[h]{Initialize an empty edge set}\;
        \BlankLine
        \ForEach(\tcc*[h]{Add all relevant edges to edge set $E$}){$(j,k)\in \mathcal{M}\times\mathcal{M}$ with $j<k$}{
            \If{$[\tilde{x}_j,\tilde{x}_k]=\tilde{\alpha}_{jk} \tilde{x}_{\delta(j,k)}$ with $\tilde{\alpha}_{jk}\neq 0$}{
                $E$ $\leftarrow$ $E\cup\{(\tilde{x}_j,\tilde{x}_k,\tilde{x}_{\delta(j,k)})\}$\;
                $E$ $\leftarrow$ $E\cup\{(\tilde{x}_k,\tilde{x}_j,\tilde{x}_{\delta(j,k)})\}$
                \tcc*[h]{Add edge to edge set $E$; each edge is an ordered triple: (start vertex, edge-label, end vertex)}\;
            }
        }
        \Return $G(V,E)$\tcc*[h]{Return labeled directed graph $G(V,E)$ associated with $\g$}\;
    \caption{Algorithm for generating a labeled directed graph associated with a finite-dimensional graph-admissible Lie algebra $\g$}\label{alg:creating:graph}
\end{algorithm}

\begin{remark*}
    As required by the input data of Algorithm~\ref{alg:creating:graph}, it is essential to consider \emph{only}  (possibly overcomplete) bases $\mathcal{B}$ that satisfy the relations~\eqref{eqn:desired:basis} (resp. relations~\eqref{eqn:desired:basis:overcomplete}). Thus, for every pair of $x,y\in \mathcal{B}$, one must have either $[x,y]=0$ or $[x,y]\propto z$ for some $z\in \mathcal{B}$. Otherwise the algorithm does not produce an associated graph in the sense of Definitions~\ref{def:graph:admissible} and~\ref{def:assocaited:graph}, because the Lie bracket of two vertex elements is not proportional to at most one other vertex. Consequently, the structural criteria, for example for solvability or nilpotency, developed later will generally not apply. Allowing finite overcomplete bases is important in this context, since additional bracket-output directions can be represented by separate vertices even when no suitable basis exists.
\end{remark*}

It is essential to distinguish between the associated labeled directed graph $G(V,E)$ together with its vertex and edge labels and its underlying combinatorial structure. Regarding the vertex labels as elements of $\g$, one recovers $\g$ as the vector space spanned by $V$, i.e., $\g=\vspan{V}=\lie{V}$. However, regarding the vertex labels as arbitrary elements, the graph $G(V,E)$ does generally not determine the original Lie algebra, because the non-zero proportionality constants appearing in the Lie brackets are not encoded.

The following definition formalizes the correspondence between Lie algebras and their associated graphs, ensuring that the graph encodes the Lie bracket structure to the desired abstraction.

\begin{definition}\label{def:assocaited:graph}
    Let $\g$ be a finite-dimensional graph-admissible Lie algebra. A labeled directed graph $G(V,E)$ is said to be \emph{associated} with $\g$ if and only if it can be obtained by applying Algorithm~\ref{alg:creating:graph} to a---possibly overcomplete---basis $\mathcal{B}$ of $\g$ satisfying the conditions from Definition~\ref{def:graph:admissible}. We say a graph $G(V,E)$, associated with a Lie algebra $\g$, is a \emph{minimal graph} if $|V|=\dim(\g)$, otherwise we say that $G(V,E)$ is a \emph{redundant graph}.
\end{definition}

\begin{tcolorbox}[breakable, colback=Cerulean!3!white,colframe=Cerulean!85!black,title=\textbf{Example}: Graphs associated with graph-admissible Lie algebras]
    We here consider how the construction given above works in practice:
    \begin{example}\label{exa:first:example:graphs:assocaited:with:algebras}
        Consider the simple real Lie algebra $\mathfrak{su}(2)$, defined by the following Lie bracket relations \cite{Pfeifer:2003}:
        \begin{align}
            [e_1,e_2]&=e_3,\;&\;[e_2,e_3]&=e_1,\;&\;[e_3,e_1]&=e_2.\label{eqn:basis:su:2:canonical}
        \end{align}
        Applying Algorithm~\ref{alg:creating:graph}, we construct the associated directed labeled graph, which is depicted in Figure~\ref{fig:su2:sl2} (a).
    \end{example}
    It is important to emphasize that the graph associated with a given Lie algebra is not unique, as illustrated by Example~\ref{exa:second:example:graphs:assocaited:with:algebras}. 
    \begin{example}\label{exa:second:example:graphs:assocaited:with:algebras}
        Consider the real Lie algebra $\sl{2}{\R}$, which is not isomorphic to $\mathfrak{su}(2)$ \cite{helgason2024differential}. A standard basis of $\sl{2}{\R}$ is given by the elements $h$, $x$, $y$ with the Lie bracket relations \cite{helgason2024differential}:
        \begin{align}
            [h,x]&=2x,\;&\;[h,y]&=-2y,\;&\;[x,y]&=h.
        \end{align}
        This basis yields the graph depicted in Figure~\ref{fig:su2:sl2} (b). However, we can also choose an alternative basis, defined by the elements $e_1:=h/2$, $e_2:=(x-y)/2$, and $e_3:=(x+y)/2$. It is straightforward to verify that this new basis satisfies the Lie bracket relations:
        \begin{align}
            [e_1,e_2]&=e_3,\;&\;[e_2,e_3]&=e_1,\;&\;[e_3,e_1]&=-e_2.\label{eqn:basis:sl2R:second:version}
        \end{align}
        These relations closely resemble those of $\mathfrak{su}(2)$ from Example~\ref{exa:first:example:graphs:assocaited:with:algebras}, differing only by a sign in the final bracket. Consequently, the graph associated with this basis is identical to the graph for $\mathfrak{su}(2)$ depicted in Figure~\ref{fig:su2:sl2} (a), despite the fact that $\sl{2}{\R}$ and $\mathfrak{su}(2)$ are not isomorphic.
    \end{example}
    \begin{figure}[H]
        \centering
        \includegraphics[width=0.75\linewidth]{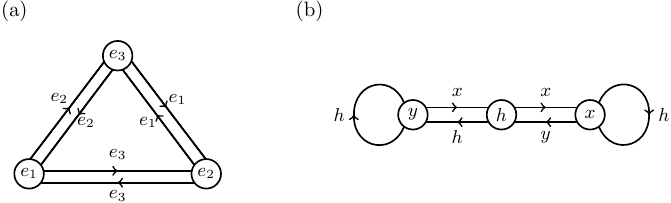}
        \caption{Depiction of the graphs associated with the Lie algebras $\mathfrak{su}(2)$ and $\sl{2}{\R}$ from Examples~\ref{exa:first:example:graphs:assocaited:with:algebras} and~\ref{exa:second:example:graphs:assocaited:with:algebras} respectively.}
        \label{fig:su2:sl2}
    \end{figure}
\end{tcolorbox}

\begin{remark*}
    Examples~\ref{exa:first:example:graphs:assocaited:with:algebras} and~\ref{exa:second:example:graphs:assocaited:with:algebras} demonstrate two important limitations of the construction discussed here. First,\emph{ the graph associated with a given Lie algebra is not unique} and generally depends on the chosen (possibly overcomplete) basis. In particular, the two bases of $\mathfrak{sl}(2,\R)$ considered in Example~\ref{exa:second:example:graphs:assocaited:with:algebras} produce different associated graphs. Second, \emph{non-isomorphic Lie algebras may produce isomorphic labeled directed graphs}: the graph obtained from the alternative basis of $\mathfrak{sl}(2,\R)$ from above coincides with the graph obtained for $\su(2)$ in Example~\ref{exa:first:example:graphs:assocaited:with:algebras}.  Consequently, the associated graph does not generally determine the Lie algebra up to isomorphism and need not distinguish compact Lie algebras from non-compact ones, as demonstrated with the examples~\cite{Hall:Lie:groups:16}. This loss of information results from suppressing the nonzero proportionality constants appearing in the Lie brackets. To address such limitations, one may introduce additional edge weights or colors encoding the sign and magnitude of these structure constants. For instance, edges corresponding to brackets with positive $\alpha_{jk}$ could be colored red, while those with negative coefficients blue. Such a refinement would allow the graph to distinguish between otherwise indistinguishable bracket structures and potentially retain more of the algebra's intrinsic structure. We refrain from this approach, as the focus of the present framework is not on encoding all numerical data of the structure constants, but rather on utilizing graph-theoretic features to highlight structural properties, such as solvability, nilpotency, and the presence of ideals. In this context, coloring is primarily employed to identify and emphasize subgraphs that correspond to these properties.
\end{remark*}

The family of real Lie algebras $L_\alpha^3$ was essential to prove Theorem~\ref{thm:existncae:of:non:minimal:graph:admissible:}, Corollary~\ref{cor:existance:of:redundant:but:not:mimimak:graph:admissibel:algebras}, and Theorem~\ref{thm:existence:of:non:graph:admissible:lie:algebras}, as it serves as an example for demonstrating the existence of Lie algebras that are only redundant-graph-admissible but not minimal-graph-admissible, and Lie algebras that are non-graph-admissible altogether. This motivates Example~\ref{exa:different:L:alpha:3:graphs}, discussed in Appendix~\ref{app:supplementary:proofs}, which provides explicit graphs for $L_\alpha^3$ at specific values of $\alpha$.

\begin{lemma}\label{lem:isomorphism:unchanged}
    Let $\g$ be a finite-dimensional graph-admissible Lie algebra, and let $\varphi:\g\to\varphi(\g)$ be a Lie-algebra isomorphism. Let $\mathcal{B}$ be a (possibly overcomplete) basis satisfying \eqref{eqn:desired:basis} (respectively \eqref{eqn:desired:basis:overcomplete}). Then, $\varphi$ induces an isomorphism between the labeled directed graph generated by $\mathcal{B}$ and the labeled directed graph generated by $\varphi(\mathcal{B})$.
\end{lemma}

\begin{proof}
    Let $\mathcal{B}=\{x_j\}_{j=1}^m$ be a (possibly overcomplete) basis of $\g$ that satisfies the relations \eqref{eqn:desired:basis} (respectively \eqref{eqn:desired:basis:overcomplete}) and $\varphi:\g\to \varphi(\g)$ a Lie-algebra isomorphism. Define, $y_j:=\varphi(x_j)\in\varphi(\g)$. Then, since $\varphi$ is an invertible linear map that preserves the Lie bracket, $\varphi(\mathcal{B})=\{y_j\}_{j=1}^m$ is again a (possibly overcomplete) basis, and one has for all $y_j,y_k\in\varphi(\mathcal{B})$:
    \begin{align*}
        [y_j,y_k]_{\varphi(\g)}=[\varphi(x_j),\varphi(x_k)]_{\varphi(\g)}=\varphi([x_j,x_k]_\g)=\varphi(\alpha_{jk}x_{\delta(j,k)})=\alpha_{jk}\varphi(x_{\delta(j,k)})=\alpha_{jk}y_{\delta(j,k)}.
    \end{align*}
    Hence, by application of Algorithm~\ref{alg:creating:graph}, the sets $\mathcal{B}$ and $\varphi(\mathcal{B})$ generate isomorphic graphs.
\end{proof}

\begin{remark*}
    The preceding Lemma~\ref{lem:isomorphism:unchanged} demonstrates two important properties. First, graph-admissibility, including the distinction between minimal- and redundant-graph-admissibility, is preserved under Lie-algebra isomorphisms, an essential sanity check. We can again compare this with the property of diagonalizability for vector-space endomorphisms which is preserved under conjugation by invertible matrices~\cite{Lukas:2022}. Second, transporting a fixed suitable (possibly overcomplete) basis along a Lie-algebra isomorphism preserves the corresponding graph structure. Different graphs associated with the same Lie algebra therefore arise from choosing different suitable bases, rather than from replacing the Lie algebra and the chosen spanning set by an isomorphic copy.
\end{remark*}

The observations in Theorem~\ref{thm:semisimple:lie:algebra:graph:admissible} and Corollary~\ref{cor:rductive:Lie:algebras:graph:admissible} are crucial, as they show that the graph-based approach developed here is widely applicable and not restricted to a narrow subclass of Lie algebras. In particular, it applies to all semisimple and reductive Lie algebras over algebraically closed fields. Recall that a Lie algebra is semisimple if and only if it admits a unique decomposition of the form $\g=\g_1\oplus\ldots\oplus\g_s$, where each $\g_j$ is a simple ideal of $\g$ \cite{Knapp:1996}. 
This algebraic decomposition has a direct analog in the graph-theoretic framework, as shown in the next proposition.

\begin{proposition}\label{prop:semisimple:lie:algebra:unconnected:subgraphs}
    Let $\g$ be a finite-dimensional semisimple Lie algebra over an algebraically closed field $\mathbb{F}$. Then $\g$ can be associated with a labeled directed graph $G(V,E)$ that splits into multiple mutually unconnected subgraphs $G(V_j,E_j)$, where each subgraph corresponds to a simple ideal $\g_j\subseteq\g$. In particular, if $\g$ is simple, then the associated graph consists of a single connected component. 
\end{proposition}

\begin{proof}
    By the classification of semisimple Lie algebras over algebraically closed fields, any such Lie algebra admits a decomposition $\g=\mathfrak{g}_1\oplus\mathfrak{g}_2\oplus\ldots\oplus \mathfrak{g}_k$, where each $\g_j\subseteq\g$ is a simple ideal of $\g$, and $k\geq 1$. These ideals are mutually commutative, i.e., $[\mathfrak{g}_j,\mathfrak{g}_\ell]=\{0\}$ if $j\neq \ell$. By Theorem~\ref{thm:semisimple:lie:algebra:graph:admissible}, each $\g_j$ is furthermore graph-admissible. Thus, for every $\g_j$, there exists a labeled directed graph $G_j\equiv G(V_j,E_j)$. These graphs must be mutually unconnected, otherwise there would exist elements $x_j\in \g_j$ and $x_k\in\g_k$ that do not vanish under the Lie bracket, even if $j\neq k$. Each graph $G_j$ can moreover be chosen to be connected. Indeed, if an associated graph of the simple ideal $\g_j$ had two distinct connected components, the spans of their respective vertex sets would be non-zero ideals of $\g_j$ and would commute with one another. Simplicity would force both spans to equal $\g_j$, implying $[\g_j,\g_j]=\{0\}$, which contradicts the fact that $\g_j$ is non-abelian and simple. If $\g$ is simple, then $k=1$, and the graph $G(V,E)$ consists of a single connected component.
\end{proof}

\section{Criteria for valid graphs}\label{sec:valid:graphs}

In this section, we establish formal criteria for determining whether a labeled directed graph can be associated with a finite-dimensional Lie algebra. These criteria are essential for distinguishing \emph{proper} and \emph{improper} graphs and for understanding the structural constraints imposed by the Lie bracket. We focus primarily on minimal graphs associated with minimal-graph-admissible Lie algebras, although some results generalize to redundant graphs as well. Note that we are here tackling the ``reverse problem'' to that tackled above.

\begin{definition}
    Let $G(V,E)$ be a labeled directed graph. We call $G(V,E)$ \emph{proper} if it can be associated with a finite-dimensional Lie algebra; we call it furthermore \emph{proper-minimal} if, in addition, it can be associated with a finite-dimensional Lie algebra as a minimal graph. We call $G(V,E)$ \emph{improper} if it is not a proper graph. 
    
    A proper graph $G(V,E)$ is said to be \emph{choice-independent} if every edge $e\in E$ can be translated into a valid Lie bracket $[\varpi_\mathrm{s}(e),\varpi_\mathrm{l}(e)]=\kappa \varpi_\mathrm{e}(e)$, for every $\kappa\in\mathbb{F}^*$. Otherwise, the graph is called \emph{choice-dependent}, meaning that the validity of the bracket depends on specific values for $\kappa$.
\end{definition}

Simple criteria stemming from the antisymmetric property of the Lie bracket include the following:
\begin{proposition}\label{prop:conditions:for:edges}
    Let $G(V,E)$ be a labeled directed graph. Then $G(V,E)$ is a proper graph only if all the following conditions are satisfied:
    \begin{enumerate}[label =(\alph*)]
        \item All edges $e\in E$ must satisfy $\varpi_\mathrm{s}(e)\neq\varpi_\mathrm{l}(e)$.
        \item The edge $e=(\varpi_\mathrm{s}(e),\varpi_\mathrm{l}(e),\varpi_\mathrm{e}(e))$ belongs to $E$ if and only if the edge $e'=(\varpi_\mathrm{l}(e),\varpi_\mathrm{s}(e),\varpi_\mathrm{e}(e))$ also belongs to $E$.
        \item Each edge $e\in E$ must be uniquely defined by $\varpi_\mathrm{s}(e)$ and $\varpi_\mathrm{l}(e)$. That is, if two edges $e,e'\in E$ satisfy $\varpi_\mathrm{s}(e)=\varpi_\mathrm{s}(e')$ and $\varpi_\mathrm{l}(e)=\varpi_\mathrm{l}(e')$, then it must follow that $e=e'$.
    \end{enumerate}
\end{proposition}

\begin{proof}
    Let $G(V,E)$ be a labeled directed graph and suppose it can be associated with a finite-dimensional Lie algebra~$\g$. Consider an edge $e\in E$. Then:
    \begin{itemize}
        \item \textbf{Condition (a):} If $\varpi_\mathrm{s}(e)=\varpi_\mathrm{l}(e)$, then the corresponding Lie bracket is $[\varpi_\mathrm{s}(e),\varpi_\mathrm{l}(e)]=[\varpi_\mathrm{s}(e),\varpi_\mathrm{s}(e)]\propto \varpi_\mathrm{e}(e)$. This is not possible, since $[x,x]$ must vanish for all $x\in \g$, due to the antisymmetry of the Lie bracket. Therefore, such an edge cannot exist in the proper graph $G(V,E)$.
        \item \textbf{Condition (b):} Suppose $e=(\varpi_\mathrm{s}(e),\varpi_\mathrm{l}(e),\varpi_\mathrm{e}(e))\in E$ and $e'=(\varpi_\mathrm{l}(e),\varpi_\mathrm{s}(e),\varpi_\mathrm{e}(e))\notin E$, then, by construction of an associated graph, $[\varpi_\mathrm{s}(e),\varpi_\mathrm{l}(e)]\propto \varpi_\mathrm{e}(e)$ and $[\varpi_\mathrm{l}(e),\varpi_\mathrm{s}(e)]=0$. This violates the antisymmetry of the Lie bracket, since $0=[\varpi_\mathrm{l}(e),\varpi_\mathrm{s}(e)]=-[\varpi_\mathrm{s}(e),\varpi_\mathrm{l}(e)]\not\propto \varpi_\mathrm{e}(e)$ is impossible.
        \item \textbf{Condition (c):} Suppose there exist two distinct edges $e,e'\in E$ with $\varpi_\mathrm{s}(e')=\varpi_\mathrm{s}(e)$ and $\varpi_\mathrm{l}(e')=\varpi_\mathrm{l}(e)$ but $e\neq e'$, i.e., $\varpi_\mathrm{e}(e')\neq\varpi_\mathrm{e}(e)$. This implies that the Lie bracket $[\varpi_\mathrm{s}(e),\varpi_\mathrm{l}(e)]$ would be simultaneously proportional to two distinct elements, namely $\varpi_\mathrm{e}(e)$ and $\varpi_\mathrm{e}(e')$. Since $\propto$ is an equivalence relation, it follows that $\varpi_\mathrm{e}(e)\propto \varpi_\mathrm{e}(e')$, which is prohibited by Definition~\ref{def:assocaited:graph}, as it would imply that the two distinct elements $\varpi_\mathrm{e}(e),\varpi_\mathrm{e}(e')$ lie in $V$ and are proportional to each other. Hence, such ambiguity is not allowed in a proper graph.
    \end{itemize}
\end{proof}

We can also attribute some properties to the entire associated graph:

\begin{proposition}\label{prop:criteria:number:of:vertices:and:edges}
    Let $\g$ be an $n$-dimensional minimal-graph-admissible Lie algebra. Then any minimal graph $G(V,E)$ associated with $\g$ must satisfy the following properties: (a) The number of vertices equals the dimension of the Lie algebra: $|V|=n$; (b) The number of edges satisfies the bound: $0\leq |E|\leq n(n-1)$; (c) For any vertex $v_\mathrm{s}\in V$, the number of outgoing edges from $v_\mathrm{s}$ is bounded by:  $0\leq |\{e\in E\,\mid\, \varpi_\mathrm{s}(e)=v_\mathrm{s}\}|\leq n-1$ ; (d) For any vertex  $v_\mathrm{e}\in V$, the number of incoming edges to $v_\mathrm{e}$ is bounded by: $0\leq |\{e\in E\,\mid\,\varpi_\mathrm{e}(e)=v_\mathrm{e}\}|\leq n(n-1)$.
\end{proposition}

\begin{proof}
    This claim follows directly from Definitions~\ref{def:graph:admissible} and~\ref{def:assocaited:graph}; in particular claim (d) follows immediately from claim~(b).
\end{proof}

\begin{proposition}\label{prop:tight:bounds}
    The bounds stated in Proposition~\ref{prop:criteria:number:of:vertices:and:edges} are tight. That is, there exist minimal-graph-admissible Lie algebras that can be associated with minimal graphs that saturate these bounds.
\end{proposition}

\begin{proof}
    We start by constructing an explicit example of a minimal-graph-admissible Lie algebra $\g$ whose associated graph achieves the upper bounds from Proposition~\ref{prop:criteria:number:of:vertices:and:edges}. Let $\{x_1,\ldots, x_n\}$ be a set of linearly independent elements and set $x_0:=0$ by convention. Define the symmetric function:
    \begin{align*}
        \delta:\{0,1,\ldots,n\}\times\{0,1\ldots,n\}\to \{0,1,\ldots,n\},\,(j,k)\mapsto\delta(j,k):=\left\{\begin{matrix}
            0&,\text{ if }j=k\text{ or }j=0\text{ or }k=0\\
            1&,\text{otherwise}
        \end{matrix}\right.,
    \end{align*}
    and the antisymmetric matrix
    \begin{align*}
        \boldsymbol{\alpha}\in \mathbb{F}^{n\times n}\quad\text{with}\quad\alpha_{jk}:=\left\{\begin{matrix}
            0&\text{, if }j=k\\
            1&\text{, if }j=1\text{ and }k>1\\
            -1&\text{, if }k=1\text{ and }j>1\\
            j-k&\text{ otherwise}
        \end{matrix}\right..
    \end{align*}
    We now prove that these define the Lie algebra $\g=\lie{\{x_1,\ldots,x_n\}}$ via the relations: $[x_j,x_k]:=\alpha_{jk}x_{\delta(j,k)}$. Thus, we need to confirm that the Lie bracket is antisymmetric and satisfies the Jacobi identity. One has $[x_j,x_k]=\alpha_{jk} x_{\delta(j,k)}=-\alpha_{kj} x_{\delta(k,j)}=-[x_k,x_j]$ for all $j,k\in\{1,\ldots,n\}$, since $\boldsymbol{\alpha}$ is by construction antisymmetric and $\delta$ symmetric. Now consider three indices $j,k,\ell$, all unequal to one. We compute:
    \begin{align*}
        [x_j,[x_k,x_\ell]]+[x_k,[x_\ell,x_j]]+[x_\ell,[x_j,x_k]]&=-\alpha_{k\ell}\alpha_{1j}x_1-\alpha_{\ell j}\alpha_{1k}x_1-\alpha_{jk}\alpha_{1\ell}x_1=-(k-\ell+\ell-j+j-k)x_1=0.
    \end{align*}
    Now, let at least one of these indices be one. Without loss of generality, let $j=1$. Then, we have:
    \begin{align*}
        [x_1,[x_k,x_\ell]]+[x_k,[x_\ell,x_1]]+[x_\ell,[x_1,x_k]]&=\alpha_{1\ell}\alpha_{1k}x_1-\alpha_{1k}\alpha_{1\ell}x_1=0.
    \end{align*}
    Hence, by linear combinations, this constitutes a valid Lie algebra. We can now apply Algorithm~\ref{alg:creating:graph} to obtain the graph $G(V,E)$ which is associated with $\g$. This graph has by construction $n$ vertices. Next, we observe that the graph has $n(n-1)$ distinct edges, as every distinct pair of indices $(j,k)\in\mathcal{N}\times\mathcal{N}$ for which $\alpha_{jk}\neq 0$ induces a distinct edge. By construction, $\alpha_{jk}=0$ if and only if $j=k$. Thus $|E|=n(n-1)$. We also see that exactly $n-1$ edges start at $x_1$, and $n(n-1)$ edges end at $x_1$.

    To show that the lower bounds are also tight, consider the abelian Lie algebra of dimension $n$. It is immediate to notice that any basis of this Lie algebra satisfies \eqref{eqn:desired:basis}, and the graphs constructed by Algorithm~\ref{alg:creating:graph} do not possess any edges. 
\end{proof}

We can draw, as an example, the graph associated with the non-abelian Lie algebra constructed in the proof of Proposition~\ref{prop:tight:bounds} for the case $n=4$, as shown in Figure~\ref{fig:example:to:many:edges} (a). In this visualization, every edge in the set $E$ is explicitly drawn, resulting in a visually cluttered representation, showing that this approach is not practical and loses its visual simplicity. This is especially true for large algebras, as the number of edges grows quadratically with the number of vertices. To address this issue, we introduce the convention: For any pair of vertices $v_\mathrm{s},v_\mathrm{e}\in V$, if there exists at least one edge $e\in E$ with $\varpi_\mathrm{s}(e)=v_\mathrm{s}$ and $\varpi_\mathrm{e}(e)=v_\mathrm{e}$ we draw only a single edge from $v_\mathrm{s}$ to $v_\mathrm{e}$, regardless of how many such edges exist. We then annotate this edge with all corresponding labeling vertices $v_\mathrm{l}$ that appear in the original edge set. This convention is illustrated in Figure~\ref{fig:example:to:many:edges} (b).
\begin{figure}[H]
    \centering
    \includegraphics[width=0.75\linewidth]{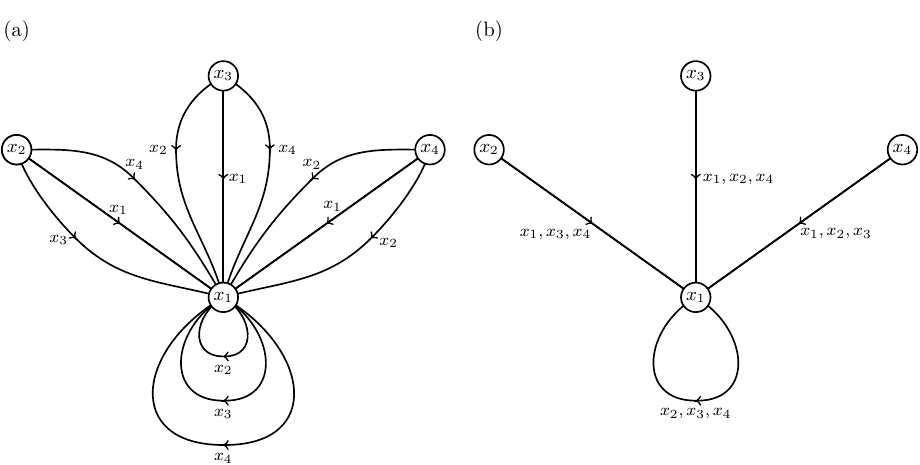}
    \caption{Depiction of the graph $G(V,E)$ associated with the non-abelian Lie algebra from the proof of Proposition~\ref{prop:tight:bounds}, for the case $n=4$. In (a), the graph is visualized by drawing every edge in $E$ explicitly; in (b) the graph is drawn using the convention of collapsing multiple edges with the same starting and end vertices into a single edge and annotating them with all relevant labeling vertices.}
    \label{fig:example:to:many:edges}
\end{figure}

We now proceed with our analysis. For the sake of simplicity we assume that $\operatorname{char}(\mathbb{F})=0$ for the remainder of this section.

\begin{theorem}\label{thm:neccersary:conditions:proper:graph}
    Let $G(V,E)$ be a labeled directed graph with three or more distinct vertices. Moreover let $\mathcal{V}:=\{\Tilde{V}\subseteq V\,\mid\,|\tilde{V}|=3\}$ be the set containing all subsets $\Tilde{V}\subseteq V$ of three distinct vertices and define, for a given set $\Tilde{V}\in \mathcal{V}$, the edge set $\Tilde{E}(\Tilde{V})\equiv \Tilde{E}:=\{e\in E\,\mid\,\varpi_\mathrm{s}(e),\varpi_\mathrm{l}(e),\varpi_\mathrm{e}(e)\in \tilde{V}\}$ containing all edges in $E$, where starting vertex, labeling vertex, and end vertex belong to $\Tilde{V}$. Then $G(V,E)$ is a proper-minimal graph only if for every $\Tilde{V}\in\mathcal{V}$, the subgraph $G(\Tilde{V},\Tilde{E})\subseteq G(V,E)$ is equivalent to one of the graphs depicted in Figure~\ref{fig:all:3:vertex:subgraphs}.
\end{theorem}

\begin{remark*}
    Before proceeding with the proof of this claim, we observe that Theorem~\ref{thm:neccersary:conditions:proper:graph} provides a necessary condition for verifying whether a given labeled directed graph $G(V,E)$ with at least three vertices can be proper-minimal. The corresponding test requires considering every subset $\tilde{V}\subseteq V$ with exactly three distinct vertices, and the corresponding subgraph must contain all edges from $E$ whose starting, labeling, and end vertices belong to $\Tilde{V}$. Hence, for a graph with $m$ vertices, one must examine 
    \begin{align*}
        \binom{m}{3}=\frac{m(m-1)(m-2)}{6}\in\mathcal{O}(m^3)
    \end{align*}
    three-vertex subgraphs. Since each such subgraph contains only three vertices and is compared with the fixed list of eleven graphs depicted in Figure~\ref{fig:all:3:vertex:subgraphs}, the number of required checks belongs to $\mathcal{O}(m^3)$. However, passing all these checks is not sufficient, as illustrated by Examples~\ref{exa:thm.Necceserary:not:sufficient:on1} and~\ref{exa:thm.Necceserary:not:sufficient:two} found below.
\end{remark*}

\begin{tcolorbox}[breakable, colback=Cerulean!3!white,colframe=Cerulean!85!black,title=\textbf{Remark}: Theorem~\ref{thm:neccersary:conditions:proper:graph} does not provide sufficient conditions]
\begin{example}\label{exa:thm.Necceserary:not:sufficient:on1}
    Graphs of Type VIII (see Figure~\ref{fig:all:3:vertex:subgraphs}) satisfy all conditions stated in Theorem~\ref{thm:neccersary:conditions:proper:graph}. However, as demonstrated in the upcoming Corollary~\ref{cor:proper:choice:dependent}, such graphs are improper graphs, meaning they cannot be associated with a finite-dimensional Lie algebra as a minimal graph.
\end{example}

\begin{example}\label{exa:thm.Necceserary:not:sufficient:two}
    The graph depicted below clearly satisfies all conditions from Theorem~\ref{thm:neccersary:conditions:proper:graph}, as the subgraphs $G(\Tilde{V},\Tilde{E})$ containing three distinct vertices and all edges $e$ in $E$ with $\varpi_\mathrm{s}(e),\varpi_\mathrm{l}(e),\varpi_\mathrm{e}(e)\in\Tilde{V}$ are all equivalent to graphs of Type I or II. Nevertheless, the Lie bracket relations yield: $[a,[b,c]]=0$, $[b,[c,a]]\propto e$, and $[c,[a,b]]=0$, which implies $[a,[b,c]]+[b,[c,a]]+[c,[a,b]]\neq0$, constituting a violation of the Jacobi identity.
    \begin{figure}[H]
        \centering
        \includegraphics[width=0.225\linewidth]{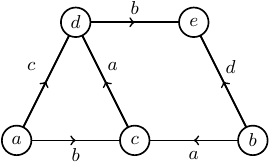}
        \caption{Illustration of an improper labeled directed graph that satisfies all conditions from Theorem~\ref{thm:neccersary:conditions:proper:graph}.}
    \end{figure}
\end{example}
\end{tcolorbox}

\begin{figure}[htpb]
    \centering
    \includegraphics[width=0.80\linewidth]{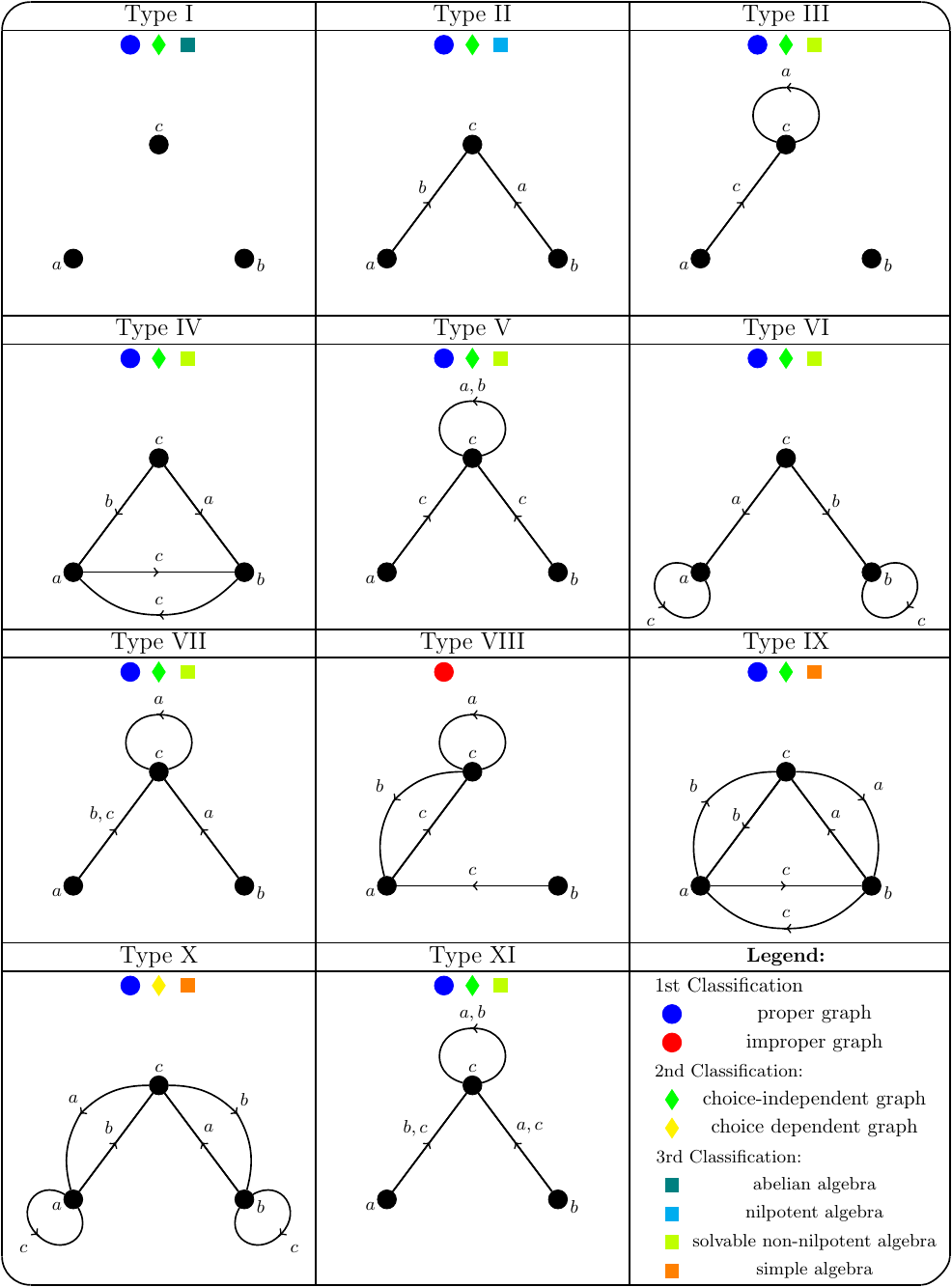}
    \caption{Visualization of all possible candidate three-vertex subgraph configurations $G(\Tilde{V},\Tilde{E})$ of a proper graph $G(V,E)$, where the edges $e\in \Tilde{E}\subseteq E$ are those that satisfy $\varpi_\mathrm{s}(e),\varpi_\mathrm{l}(e),\varpi_\mathrm{e}(e)\in \Tilde{V}$, considered up to topological equivalence. The eleven subgraphs are furthermore classified according to the following criteria: (a) proper versus improper graphs; (b) choice-independent versus choice-dependent, in case the subgraphs are proper; and (c) an algebraic classification of whether these graphs can be associated with abelian, nilpotent, solvable and non-nilpotent, and simple Lie algebras. A legend is provided to indicate these classifications using colored classification markers where applicable\footnote{Note that the convention for drawing graphs has changed compared to the previous figures. This adjustment was made to distinguish general graphs from those graphs associated with specific Lie algebras, where the vertices represent particular elements. In the second convention, the label of a vertex is placed beside the vertex rather than inside.}.}
    \label{fig:all:3:vertex:subgraphs}
\end{figure}

\begin{proof}
    Let $G(V,E)$ be a labeled directed graph with three or more distinct vertices that is proper-minimal. 
    Since any proper-minimal graph is, by definition, proper, one can invoke Proposition~\ref{prop:conditions:for:edges}, which implies that only two types of edges can occur:
    \begin{itemize}
        \item \emph{Wedge-type} edges: These are edges $e\in E$ for which all the three vertices $\varpi_\mathrm{s}(e),\varpi_\mathrm{l}(e),\varpi_\mathrm{e}(e)$ are pair-wise distinct. Such edges are called wedge-type edges because,  as shown in Proposition~\ref{prop:conditions:for:edges}, an edge $e$ belongs to $E$ if and only if the edge $e'=(\varpi_\mathrm{l}(e),\varpi_\mathrm{s}(e),\varpi_\mathrm{e}(e))$ also belongs to $E$. The graph $G(\{\varpi_\mathrm{s}(e),\varpi_\mathrm{l}(e),\varpi_\mathrm{e}(e)\},\{e,e'\})$, depicted in Figure~\ref{fig:wedge:and:loop:type:edges} (a), resembles a wedge when the three vertices are placed at the corners of an isosceles triangle.
        \item \emph{Loop-type} edges. These are edges $e\in E$, where either $\varpi_\mathrm{e}(e)=\varpi_\mathrm{l}(e)$ or $\varpi_\mathrm{e}(e)=\varpi_\mathrm{s}(e)$, with $\varpi_\mathrm{s}(e)\neq \varpi_\mathrm{l}(e)$. These are called loop-type edges because: (i) an edge $e$ belongs to $E$ if and only if the edge $e'=(\varpi_\mathrm{l}(e),\varpi_\mathrm{s}(e),\varpi_\mathrm{e}(e))$ also belongs to $E$, and (ii) one of those edges is of the form $(a,b,a)$. In other words, one of those edges is a loop, connecting a vertex to itself. The graph $G(\{\varpi_\mathrm{s}(e),\varpi_\mathrm{l}(e)\},\{e,e'\})$ is depicted in Figure~\ref{fig:wedge:and:loop:type:edges} (b).
    \end{itemize}
    It is immediately evident that wedge-type and loop-type edges are distinct, and no graph isomorphism exists that maps one to the other.

    \begin{figure}[H]
        \centering
        \includegraphics[width=0.5\linewidth]{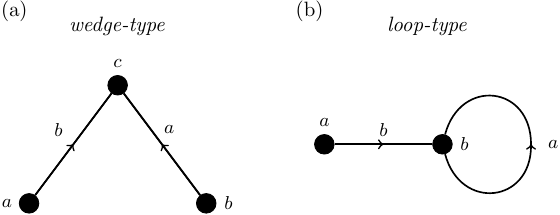}
        \caption{Depiction of wedge-type edges in (a) and loop-type edges in (b). Note that we use a different convention for visualizing vertices, if the figure is intended to convey a general feature and not a particular Lie algebra, as before.}
        \label{fig:wedge:and:loop:type:edges}
    \end{figure}
   
    We proceed here with the proof of the theorem. Consider every subgraph $G(\Tilde{V},\Tilde{E})\subseteq G(V,E)$, where $\Tilde{V}=\{a,b,c\}\subseteq V$ is a set of three distinct vertices, and $\Tilde{E}=\{e\in E\,\mid\, \varpi_\mathrm{s}(e),\varpi_\mathrm{l}(e),\varpi_\mathrm{e}(e)\in \Tilde{V}\}$ is the set of edges whose source, label, and target all lie in $\Tilde{V}$. Due to the edge-rule discussed in Propositions~\ref{prop:conditions:for:edges} and~\ref{prop:criteria:number:of:vertices:and:edges} , the number of edges is even and bounded from above, such that one has $|\Tilde{E}|\in\{0,2,4,6\}$. We now analyze these cases separately, and use upper-case roman indices to refer to the corresponding graphs in Figure~\ref{fig:all:3:vertex:subgraphs}:
    
    \begin{enumerate}[label = (\roman*)]
        \item $\boldsymbol{\abs{\Tilde{E}}=0}$. Here, there is only one possible graph: $G_{\mathrm{I}}\equiv G(\Tilde{V},\Tilde{E}_\mathrm{I})$, where $\Tilde{E}_\mathrm{I}=\emptyset$.
        \item $\boldsymbol{\abs{\Tilde{E}}=2}$. Here, there are nine possible graphs, divided into two groups, one group of three graphs with only wedge-type edges, and another group of six graphs with only loop-type edges. The three graphs with only wedge-type edges are the ones with $\Tilde{E}=\{(a,b,c),(b,a,c)\}$, $\Tilde{E}=\{(b,c,a),(c,b,a)\}$, and $\Tilde{E}=\{(c,a,b),(a,c,b)\}$. These are clearly all equivalent, as any one of those can be transformed into the other two by a simple relabeling of the vertices. The six graphs with only loop-type edges are those with $\Tilde{E}=\{(a,b,a),(b,a,a)\}$, $\Tilde{E}=\{(a,b,b),(b,a,b)\}$, $\Tilde{E}=\{(a,c,a),(c,a,a)\}$, $\Tilde{E}=\{(a,c,c),(c,a,c)\}$, $\Tilde{E}=\{(b,c,b),(c,b,b)\}$, and $\Tilde{E}=\{(b,c,c),(c,b,c)\}$. These are also all clearly equivalent among themselves, but not equivalent to the previous wedge-type graphs. Thus, we identify two inequivalent graph types: $G_{\mathrm{II}}\equiv G(\Tilde{V},\Tilde{E}_\mathrm{II})$ and $G_{\mathrm{III}}\equiv G(\Tilde{V},\Tilde{E}_\mathrm{III})$, with $\Tilde{E}_\mathrm{II}\equiv\{(a,b,c),(b,a,c)\}$ and $\Tilde{E}_\mathrm{III}\equiv\{(a,c,c),(c,a,c)\}$ respectively.
        \item $\boldsymbol{\abs{\Tilde{E}}=4}$. This case is subdivided into three distinct subcases, depending on the types of edges involved.
        \begin{enumerate}[label = (iii.\roman*)]
            \item All four edges are wedge-type. Here, there is only one class of inequivalent subgraphs in this category, namely subgraphs of the form $G_{\mathrm{IV}}\equiv G(\Tilde{V},\Tilde{E}_\mathrm{IV})$, with $\Tilde{E}_\mathrm{IV}\equiv\{(c,b,a),(b,c,a),(c,a,b),(a,c,b)\}$.
            \item All four edges are loop-type. Here, one finds three inequivalent subgraphs: 
            \begin{itemize}
                \item $G_{\mathrm{XII}}\equiv G(\Tilde{V},\Tilde{E}_\mathrm{XII})$, with $\Tilde{E}_\mathrm{XII}\equiv\{(a,c,c),(c,a,c),(c,b,b),(b,c,b)\}$,
                \item $G_{\mathrm{V}}\equiv G(\Tilde{V},\Tilde{E}_\mathrm{V})$, with $\Tilde{E}_\mathrm{V}\equiv\{(a,c,c),(c,a,c),(b,c,c),(c,b,c)\}$,
                \item $G_{\mathrm{VI}}\equiv G(\Tilde{V},\Tilde{E}_\mathrm{VI})$ with $\Tilde{E}_\mathrm{VI}\equiv\{(c,a,a),(a,c,a),(c,b,b),(b,c,b)\}$.
            \end{itemize}
            Notably, subgraphs of the type $G_{\mathrm{XII}}$ are incompatible with the Jacobi identity. Specifically:  $[b,[c,a]]\propto b$, while $[a,[b,c]]\propto [a,b]\in \spn\{V\setminus\Tilde{V}\}$. If $[a,b]=0$, then $[a,[b,c]]+[b,[c,a]]+[c,[a,b]]\propto b$, which violates the Jacobi identity. If instead $[a,b]\propto x$, where $x\in V\setminus \Tilde{ V}$, then $[c,x]$ must be proportional to an appropriate linear combination of $x$ and $b$, which is prohibited for a minimal graph since all non-zero bracket results must be proportional to elements of $V$, which must be a basis. Hence, graphs containing subgraphs of type $G_{\mathrm{XII}}$ cannot be associated with finite-dimensional Lie algebras as minimal graphs.
            \item Two wedge-type and two are loop-type edges. Here, one finds the following two inequivalent subgraphs: 
            \begin{itemize}
                \item $G_{\mathrm{VII}}\equiv G(\Tilde{V},\Tilde{E}_\mathrm{VII})$, with $\Tilde{E}_\mathrm{VII}\equiv\{(a,c,c),(c,a,c),(a,b,c),(b,a,c)\}$,
                \item $G_{\mathrm{VIII}}\equiv G(\Tilde{V},\Tilde{E}_\mathrm{VIII})$, with $\Tilde{E}_\mathrm{VIII}\equiv\{(a,c,c),(c,a,c),(b,c,a),(c,b,a)\}$.
            \end{itemize}
        \end{enumerate}
        \item $\boldsymbol{\abs{\Tilde{E}}=6}$. This last case is divided into four subcases based on the composition of wedge-type and loop-type edges.
        \begin{enumerate}[label = (iv.\roman*)]
            \item All six edges are wedge-type. There is only one inequivalent subgraph in this category, namely $G_{\mathrm{IX}}\equiv G(\Tilde{V},\Tilde{E}_\mathrm{IX})$ with $\Tilde{E}_\mathrm{IX}\equiv\{(a,b,c),(b,a,c),(b,c,a),(c,b,a),(c,a,b),(a,c,b)\}$.
            \item All six edges are loop-type. Here, there exist two inequivalent subgraphs:
            \begin{itemize}
                \item $G_{\mathrm{XIII}}\equiv G(\Tilde{V},\Tilde{E}_\mathrm{XIII})$, with $\Tilde{E}_\mathrm{XIII}\equiv\{(a,c,c),(c,a,c),(c,b,b),(b,c,b), (b,a,a),(a,b,a)\}$,
                \item $G_{\mathrm{XIV}}\equiv G(\Tilde{V},\Tilde{E}_\mathrm{XIV})$, with $\Tilde{E}_\mathrm{XIV}\equiv\{(c,a,a),(a,c,a),(c,b,b),(b,c,b),(a,b,b),(b,a,b)\}$.
            \end{itemize}
            We will now show that in both cases the graph $G(V,E)$ cannot be associated with a finite-dimensional Lie algebra as a minimal subgraph. For $G_{\mathrm{XIII}}$, one would have that $[a,[b,c]]\propto a$, $[b,[c,a]]\propto b$, and $[c,[a,b]]\propto c$. Thus, the Jacobi identity cannot be satisfied as one assumes that the three elements $a,b,c$ are linearly independent and can therefore not be non-trivially linearly combined such that their weighted sum vanishes. Also, graphs with subgraphs of type $G_{\mathrm{XIV}}$ are prohibited, since this would imply $[a,b]=\kappa_1 b$, $[b,c]=\kappa_2 b$, and $[c,a]=\kappa_3 a$, where $\kappa_1,\kappa_2,\kappa_3\in \mathbb{F}^*$. The Jacobi identity in this case would read: $[a,[b,c]]+[b,[c,a]]+[c,[a,b]]=\kappa_1\kappa_2b-\kappa_1\kappa_3b-\kappa_1\kappa_2b=-\kappa_1\kappa_3b\neq 0$, which contradicts the very definition of the Jacobi identity.
            \item Four wedge-type edges and two are loop-type edges. Here, there is only one inequivalent subgraph, namely $G_{\mathrm{XV}}\equiv G(\Tilde{V},\Tilde{E}_\mathrm{XV})$, with $\Tilde{E}_\mathrm{XV}\equiv\{(a,c,c),(c,a,c),(a,b,c),(b,a,c), (b,c,a),(c,b,a)\}$. This configuration implies: $[a,[b,c]]+[b,[c,a]]+[c,[a,b]]\propto a$, which is again a violation of the Jacobi identity.
            \item Two wedge-type edges and four are loop-type edges. Here, there are three inequivalent subgraphs:
            \begin{itemize}
                \item $G_{\mathrm{XVI}}\equiv G(\Tilde{V},\Tilde{E}_\mathrm{XVI})$, with $\Tilde{E}_\mathrm{XVI}\equiv\{(a,c,c),(c,a,c),(a,b,c),(b,a,c),(c,b,b),(b,c,b)\}$, 
                \item $G_{\mathrm{X}}\equiv G(\Tilde{V},\Tilde{E}_\mathrm{X})$, with $\Tilde{E}_\mathrm{X}\equiv\{(c,a,a),(a,c,a),(c,b,b),(b,c,b),(a,b,c),(b,a,c)\}$,
                \item $G_{\mathrm{XI}}\equiv G(\Tilde{V},\Tilde{E}_\mathrm{XI})$, with $\Tilde{E}_\mathrm{XI}\equiv\{(a,c,c),(c,a,c),(b,c,c),(c,b,c),(a,b,c),(b,a,c)\}$.
            \end{itemize}
            Subgraphs of the type $G_{\mathrm{XVI}}$ would imply $[a,[b,c]]\propto c$, $[b,[c,a]]\propto b$, and $[c,[a,b]]=0$. Thus, $a,b,c$ cannot satisfy the Jacobi identity under the assumption that they are linearly independent.
        \end{enumerate}
    \end{enumerate}
    This concludes the case-by-case analysis of all possible three-vertex subgraphs. Each valid configuration corresponds to one of the eleven graphs depicted in Figure~\ref{fig:all:3:vertex:subgraphs}, and any deviation from these possible configurations leads to a violation of  either the Jacobi identity, bilinearity, antisymmetry of the Lie bracket, or the minimality condition.
\end{proof}

\begin{corollary}\label{cor:proper:choice:dependent}
    All graphs depicted in Figure~\ref{fig:all:3:vertex:subgraphs} are proper subgraphs, except those of Type VIII. Furthermore, all proper graphs depicted in Figure ~\ref{fig:all:3:vertex:subgraphs} are choice-independent, except those of Type X.
\end{corollary}

\begin{proof}
    The edges of the graphs depicted in Figure~\ref{fig:all:3:vertex:subgraphs} define the Lie bracket relations among the three vertex elements, up to irrelevant nonzero proportionality constants. In this context, and without loss of generality, we here treat the graphs in Figure~\ref{fig:all:3:vertex:subgraphs} as standalone graphs themselves, not as subgraphs of a larger graph. Consequently, missing edges imply that the corresponding Lie bracket vanishes since the bracket cannot connect to a vertex outside the graph. 
    
    To determine whether a graph is proper, it suffices to verify whether the Jacobi identity is satisfied, since the vertex set $V$ spans a vector space and the Lie bracket is antisymmetric by construction. To determine whether a given graph is choice-independent, one must similarly verify if the Jacobi identity holds irrespective of the choice for proportionality constants associated with the Lie bracket defined by the edges. We now analyze each graph type individually.
    
    \begin{itemize}
        \item \textbf{Type I}: The Lie bracket of all basis elements and therefore of all elements vanishes. Thus, Type I graphs can be associated with the three-dimensional abelian Lie algebra, making them proper and choice-independent.
        \item \textbf{Type II}: The only non-trivial Lie bracket is $[a,b]\propto c$. Thus $[a,[b,c]]=0=[b,[c,a]]=[c,[a,b]]$, making graphs of Type II proper and choice-independent.
        \item \textbf{Type III}: The only non-trivial Lie bracket is $[a,c]\propto c$. Thus $[a,[b,c]]=0=[b,[c,a]]=[c,[a,b]]$, making graphs of Type III proper and choice-independent.
        \item \textbf{Type IV}: The only non-trivial Lie brackets are $[b,c]\propto a$ and $[a,c]\propto b$. Thus $[a,[b,c]]=0=[b,[c,a]]=[c,[a,b]]$, making graphs of Type IV proper and choice-independent.
        \item \textbf{Type V}: The only non-trivial Lie brackets are $[a,c]=\kappa_1 c$ and $[b,c]=\kappa_2 c$, where $\kappa_1,\kappa_2\in\mathbb{F}^*$. Thus $[a,[b,c]]=\kappa_1\kappa_2c$, $[b,[c,a]]=-\kappa_1\kappa_2c$, and $[c,[a,b]]=0$. Hence $[a,[b,c]]+[b,[c,a]]+[c,[a,b]]=0$, making graphs of Type V proper and choice-independent.
        \item \textbf{Type VI}: The only non-trivial Lie brackets are $[c,a]\propto a$ and $[b,c]\propto b$. Thus $[a,[b,c]]=0=[b,[c,a]]=[c,[a,b]]$, making graphs of Type VI proper and choice-independent.
        \item \textbf{Type VII}: The only non-trivial Lie brackets are $[a,c]\propto c$ and $[a,b]\propto c$. Thus $[a,[b,c]]=0=[b,[c,a]]=[c,[a,b]]$, making graphs of Type VII proper and choice-independent.
        \item \textbf{Type VIII}: The only non-trivial Lie brackets are $[a,c]=\kappa_2 c$ and $[b,c]=\kappa_3 a$ with $\kappa_2,\kappa_3\in\mathbb{F}^*$. Thus $[a,[b,c]]=0$, $[b,[c,a]]=-\kappa_2\kappa_3 a$, and $[c,[a,b]]=0$. Hence $[a,[b,c]]+[b,[c,a]]+[c,[a,b]]\propto a$, making graphs of Type VIII improper. 
        \item \textbf{Type IX}: The non-trivial Lie brackets are $[a,b]\propto c$, $[b,c]\propto a$, and $[a,c]\propto b$. Thus $[a,[b,c]]=0=[b,[c,a]]=[c,[a,b]]$, making graphs of Type IX proper and choice-independent.
        \item \textbf{Type X}: The non-trivial Lie brackets are $[a,b]=\kappa_1 c$, $[a,c]=\kappa_2 a$, and $[b,c]=\kappa_3 b$, where $\kappa_1,\kappa_2,\kappa_3\in\mathbb{F}^*$. Thus $[a,[b,c]]=\kappa_1\kappa_3 c$, $[b,[c,a]]=\kappa_1\kappa_2 c$, and $[c,[a,b]]=0$. Hence $[a,[b,c]]+[b,[c,a]]+[c,[a,b]]=\kappa_1(\kappa_2+\kappa_3) c$. To satisfy the Jacobi identity, one needs, therefore, to require $\kappa_2=-\kappa_3$, making graphs of Type X proper and choice-dependent.
        \item \textbf{Type XI}: The non-trivial Lie brackets are $[a,b]=\kappa_1 c$, $[a,c]=\kappa_2 c$, and $[b,c]=\kappa_3 c$, where $\kappa_1,\kappa_2,\kappa_3\in\mathbb{F}^*$. Thus $[a,[b,c]]=\kappa_2\kappa_3 c$, $[b,[c,a]]=-\kappa_2\kappa_3 c$, and $[c,[a,b]]=0$. Hence $[a,[b,c]]+[b,[c,a]]+[c,[a,b]]=0$. Hence, the Jacobi identity is satisfied for all $\kappa_1,\kappa_2,\kappa_3\in\mathbb{F}^*$, making graphs of Type XI proper and choice-independent.
    \end{itemize}
\end{proof}

\begin{corollary}\label{cor:criteria:three:vertices:vanish:under:twO.brackets}
    Let $\g$ be a finite-dimensional graph-admissible Lie algebra associated with a labeled directed graph $G(V,E)$, where the vertices satisfy $[v_j,[v_k,v_\ell]]=0$ for all $j,k,\ell$ with $j\neq k\neq \ell \neq j$. Then, for any choice of three distinct vertices $\{v_1,v_2,v_3\}=\Tilde{V}\subseteq V$, and considering only edges $e\in \Tilde{E}=\{e=(v_\mathrm{s},v_\mathrm{l},v_\mathrm{e})\in E\,\mid\,v_\mathrm{s},v_\mathrm{l},v_\mathrm{e}\in \Tilde{V}\}$, the resulting subgraph $G(\Tilde{V},\Tilde{E})$ must be of Type I, II, III, IV, VI, VII, or IX, depicted in Figure~\ref{fig:all:3:vertex:subgraphs}.
\end{corollary}

\begin{proof}
    This statement follows from the same computations performed in the proof of Corollary~\ref{cor:proper:choice:dependent}.
\end{proof}

\begin{corollary}\label{cor:3:graphs:Bianchi:identification}
    The proper graphs in Figure~\ref{fig:all:3:vertex:subgraphs} can each be associated with at least one real three-dimensional Lie algebra. The naming conventions used below for these Lie algebras follow the classification given in the literature \cite{Mubarakzyanov:1963,Popovych:2003}. We have:
    \begin{itemize}
        \item \textbf{Type I}: The abelian Lie algebra $3\gl_1=\gl_1\oplus\gl_1\oplus\gl_1$, also known as Bianchi I. 
        \item \textbf{Type II}: The Heisenberg algebra $\gl_{3,1}\cong\mathfrak{h}_1$ \cite{Kac:1990,Gosson:2006}, also known as Bianchi II, determined by the nontrivial Lie bracket: $[e_2,e_3]=e_1$.
        \item \textbf{Type III}: The Lie algebra $\gl_{2,1}\oplus \gl_1$, also known as Bianchi III, determined by the nontrivial Lie bracket: $[e_1,e_2]=e_2$.
        \item \textbf{Type IV}: The solvable Lie algebra $\gl_{3,5}^{\beta=0}$, also known as Bianchi VII, determined by the nontrivial Lie brackets $[e_1,e_3]=\beta e_1-e_2$ and $[e_2,e_3]=e_1+\beta e_2$ with $\beta=0$.
        \item \textbf{Type V}: The Lie algebra $\gl_{2,1}\oplus \gl_1$, also known as Bianchi III, determined by the nontrivial Lie bracket: $[e_1,e_2]=e_2$.
        \item \textbf{Type VI}: The solvable Lie algebra $\gl_{3,3}$, also known as Bianchi V, determined by the nontrivial Lie brackets  $[e_1,e_3]=e_1$ and $[e_2,e_3]=e_2$.
        \item \textbf{Type VII}: The Lie algebra $\gl_{2,1}\oplus \gl_1$, also known as Bianchi III, determined by the nontrivial Lie bracket: $[e_1,e_2]=e_2$.
        \item \textbf{Type IX}: The simple Lie algebra $\mathfrak{gl}_{3,7}$, also known as Bianchi IX, or $\mathfrak{so}_\R(3)\cong\mathfrak{su}(2)$ \cite{fuchs2003symmetries}, determined by the nontrivial Lie brackets $[e_2,e_3]=e_1$, $[e_3,e_1]=e_2$, and $[e_1,e_2]=e_3$.
        \item \textbf{Type X}: The simple Lie algebra $\mathfrak{gl}_{3,6}$, also known as Bianchi VIII, or $\sl{2}{\R}$ \cite{A1:project}, determined by the non-trivial Lie brackets $[e_1,e_2]=e_1$, $[e_2,e_3]=e_3$, and $[e_1,e_3]=2e_2$.
        \item \textbf{Type XI}: The Lie algebra $\gl_{2,1}\oplus \gl_1$, also known as Bianchi III, determined by the nontrivial Lie bracket: $[e_1,e_2]=e_2$.
    \end{itemize}
\end{corollary}

This statement does not claim completeness of identification or classification. For instance, the Lie algebra $\gl_{3,4}$, also known as Bianchi VI, determined by the nontrivial Lie brackets $[e_1,e_3]=e_1$ and $[e_2,e_3]=\alpha e_2$, with $-1\leq \alpha <1 $ and $\alpha\neq 0$, can be associated with a Type VI graph.

\begin{proof}
    First note that by Corollary~\ref{cor:proper:choice:dependent}, the graphs listed above are all proper graphs depicted in Figure~\ref{fig:all:3:vertex:subgraphs}. Most of these graphs can be furthermore directly obtained by applying Algorithm~\ref{alg:creating:graph} to the respective bases of the Lie algebras stated above. However, this is not the case for the graphs of Type V, VII, and XI. To demonstrate that these graphs can be associated with the Lie algebra $\gl_{2,1}\oplus\gl_1$, which is defined by the nontrivial Lie bracket $[e_1,e_2]=e_2$, while $e_3\in\mathcal{Z}(\gl_{2,1}\oplus\gl_1)$, we consider the following alternative bases: For Type V, one uses the basis $\{e_1,e_1+e_3,e_2\}$.  Applying Algorithm~\ref{alg:creating:graph} to this basis yields a graph of Type V. For Type VII, one uses the basis $\{e_1,e_2+e_3,e_2\}$, and for Type XI, the basis $\{e_2,e_1-e_2,e_1+e_2+e_3\}$, which both yield the respective graphs.
\end{proof}

\begin{lemma}\label{lem:subgraph:subalgebra:relation}
    Let $G(V,E)$ be a proper graph associated with a finite-dimensional Lie algebra $\g$. Let $G(\Tilde{V},\Tilde{E})$ be a subgraph of $G(V,E)$ with $|\Tilde{V}|\geq 3$. Then $G(\Tilde{V},\Tilde{E})$ is a proper subgraph associated with a Lie subalgebra $\mathfrak{h}\subseteq \g$ if $\Tilde{E}:=\{e\in E\,\mid\,\varpi_\mathrm{s}(e),\varpi_\mathrm{l}(e),\varpi_\mathrm{e}(e)\in\Tilde{V}\}=\{e\in E\,\mid\, \varpi_{\rm s}(e),\varpi_{\rm l}(e)\in\widetilde V\}$. 
\end{lemma}

\begin{proof}
    This result is evident from the construction of the associated graphs and requires no further argument.
\end{proof}

\begin{corollary}\label{cor:simple:subalgebra}
    Let $G(V,E)$ be a proper-minimal graph associated with a real finite-dimensional Lie algebra $\g$. If $G(V,E)$ contains a subgraph of Type IX or Type X, then $\g$ contains a simple three-dimensional Lie subalgebra, i.e., a Lie subalgebra isomorphic to $\sl{2}{\R}$ or $\mathfrak{su}(2)$, rendering $\g$ not solvable.
\end{corollary}

\begin{proof}
    By Lemma~\ref{lem:subgraph:subalgebra:relation}, the Lie algebra associated with the graph $G(V,E)$ contains a subalgebra $\mathfrak{h}\subseteq\g$ that is associated with a subgraph of Type IX or X. Let this subgraph be denoted by $G(\Tilde{V},\Tilde{E})$. Since $G(V,E)$ is proper-minimal, the three vertices in $\Tilde{V}$ are linearly independent. Furthermore, one observes that $\Tilde{V}\subseteq [\Tilde{V},\Tilde{V}]$, which implies that the subalgebra $\mathfrak{h}:=\spn\{\Tilde{V}\}$ is not solvable. According to the Bianchi-classification, the only non-solvable real Lie algebras of dimension three are the simple Lie algebras $\sl{2}{\R}$, and $\mathfrak{su}(2)$ \cite{Bianchi:1903}. Therefore, $\g$ contains a simple three-dimensional Lie subalgebra, which makes it not solvable.
\end{proof}

\begin{proposition}\label{prop:allowed:trails}
    Let $G(V,E)$ be a labeled directed graph. For $G(V,E)$ to be a proper-minimal graph, it must satisfy the following condition: For every subset $\Tilde{V}=\{a,b,c\}\subseteq V$ of three distinct vertices, one of the following must hold:
    \begin{enumerate}[label = (\alph*)]
        \item There exist three directed trails $C_1=(a,e_{11},x_{1},e_{12},y)$, $C_2=(b,e_{21},x_2,e_{22},y)$, and $C_3=(c,e_{31},x_3,e_{32},y)$, where $x_1,x_2,x_3,y\in V$ and the labeling conditions are $\varpi_\mathrm{l}(e_{22})=\varpi_\mathrm{l}(e_{31})=a$, $\varpi_\mathrm{l}(e_{11})=\varpi_\mathrm{l}(e_{32})=b$, $\varpi_{\mathrm{l}}(e_{12})=\varpi_\mathrm{l}(e_{21})=c$.
        \item After a possible relabeling of $a,b,c$, there exist two directed trails $C_1=(a,e_{11},x_1,e_{12},y)$ and $C_2=(c, e_{21},x_2,e_{22},y)$ with $x_1,x_2,y\in V$ and $\varpi_\mathrm{l}(e_{22})=a$, $\varpi_\mathrm{l}(e_{12})=\varpi_\mathrm{l}(e_{21})=b$, and $\varpi_{\mathrm{l}}(e_{11})=c$, but no directed trail $C_3=(b,e_{31},x_3,e_{32},y)$ with $x_3\in V$, $\varpi_\mathrm{l}(e_{31})=a$, and $\varpi_\mathrm{l}(e_{32})=c$.
        
        \item For every $v\in \Tilde{V}$ there exist no edges $e,e'\in E$ such that $\varpi_\mathrm{s}(e)=v$, $\varpi_\mathrm{l}(e)\in \Tilde{V}\setminus\{v\}$, $\varpi_\mathrm{s}(e')=\varpi_\mathrm{e}(e)$ and $\varpi_\mathrm{l}(e')\in \Tilde{V}\setminus\{v,\varpi_\mathrm{l}(e)\}$.
    \end{enumerate}
\end{proposition}

\begin{figure}[H]
    \centering
    \includegraphics[width=0.85\linewidth]{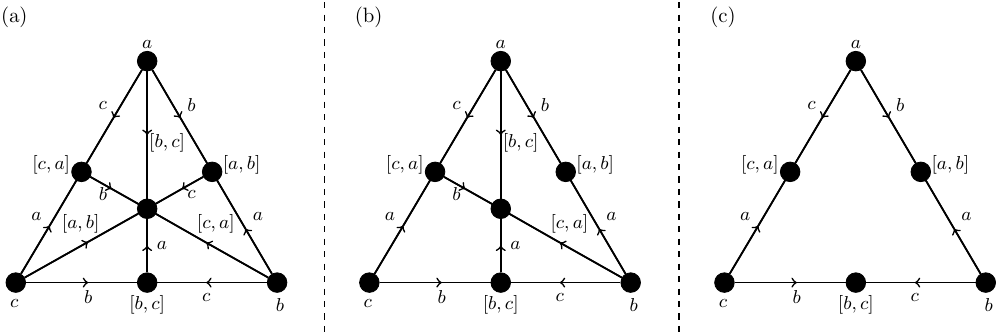}
    \caption{Visualization of the three structural conditions described in Proposition~\ref{prop:allowed:trails}. Note that this figure presents an abstraction and does not necessarily depict the actual graphs, since the Lie brackets involving the vertices $\tilde{V}=\{a,b,c\}$ may be proportional to the vertices themselves, e.g., $[c,a]\propto c$.}
    \label{fig:all:two:step:walks:sketch}
\end{figure}

\begin{proof}
    Every set $\{a,b,c\}\subseteq V$ of distinct vertices must satisfy the Jacobi identity:
    \begin{align*}
        [a,[b,c]]+[b,[c,a]]+[c,[a,b]]=0.
    \end{align*}
    Since all vertices satisfy relations \eqref{eqn:desired:basis} and are linearly independent elements, there are three distinct possibilities for how the Jacobi identity can be satisfied:
    \begin{enumerate}[label =(\roman*)]
        \item  All three nested brackets are proportional to the same nonzero element: $[a,[b,c]]\propto[b,[c,a]]\propto[c,[a,b]]$.
        \item Two nested brackets are proportional to the same element, while the third one vanishes. After a possible relabeling of $a,b,c$: $[a,[b,c]]\propto[b,[c,a]]$, while $[c,[a,b]]=0$.
        \item All nested brackets vanish: $[a,[b,c]]=0=[b,[c,a]]=[c,[a,b]]$.
    \end{enumerate}
    These cases are the only algebraic configurations that satisfy the Jacobi identity under the given assumption. It is immediate to observe that
    \begin{align*}
        \text{Case }(i)&\Leftrightarrow \text{Condition }(a),\;&\;\text{Case }(ii)&\Leftrightarrow \text{Condition }(b),\;&\;\text{Case }(iii)&\Leftrightarrow \text{Condition }(c),
    \end{align*}
    as verified by the application of Algorithm~\ref{alg:creating:graph}.
\end{proof}

We can also simplify the procedure to check every triple of vertices by the following simple observation:
\begin{proposition}\label{prop:conditions:proper:choice:independent:proper:minimal}
    Let $G(V,E)$ be a labeled directed graph with vertices $V=\{v_j\}_{j\in\mathcal{M}}$ and edges $E=\{e_p\}_{p\in\mathcal{P}}$ that satisfies the properties (a) - (c) from Proposition~\ref{prop:conditions:for:edges}. 
    Define the map $\delta:(\mathcal{M}\cup\{0\})\times(\mathcal{M}\cup\{0\})\to\mathcal{M}\cup\{0\}$ with concrete action
    \begin{align}
        \delta:(j,k)\mapsto\delta(j,k):=\left\{\begin{matrix}
            \ell&\text{ if } \exists \text{ edge $e\in E$ with }e=(v_j,v_k,v_\ell)\\
            0&\text{ otherwise}
        \end{matrix}
        \right..
    \end{align}
    Then: (i) $G(V,E)$ is a proper choice-independent graph if  $\delta(j,\delta(k,\ell))=0$ for all $j,k,\ell\in\mathcal{M}$; (ii) $G(V,E)$ is a proper-minimal graph only if $\left|\left\{\delta(j,\delta(k,\ell)),\delta(k,\delta(\ell,j)),\delta(\ell,\delta(j,k))\right\}\setminus\{0\}\right|\leq 1$ for every fixed triple $j,k,\ell\in\mathcal{M}$ with $j\neq k\neq\ell\neq j$. 
\end{proposition}

\begin{proof}
    Suppose $G(V,E)$ is a labeled directed graph satisfying the conditions stated above, i.e., properties (a) - (c) from Proposition~\ref{prop:conditions:for:edges}. We begin by observing the following: The map $\delta$ is well-defined by property (c); it is symmetric by property (b); and satisfies $\delta(j,j)=0$ by property (a) from Proposition~\ref{prop:conditions:for:edges}.
    Now consider the vector space $\spn\{V\}$ equipped with the bilinear operation $[\,\cdot\,,\,\cdot\,]$ defined on the basis elements via $[v_j,v_k]=\alpha_{jk}v_{\delta(j,k)}$, where $v_0:=0$ and $\boldsymbol{\alpha}$ is chosen to be antisymmetric. One can for example define $\alpha_{jk}:=(j-k)(1-\delta_{\delta(j,k),0})$, ensuring that $\boldsymbol{\alpha}^\mathrm{Tp}=-\boldsymbol{\alpha}$. By construction, the bracket $[\,\cdot\,,\,\cdot\,]$ is antisymmetric. Thus, one is left verifying that $[\,\cdot\,,\,\cdot\,]$ satisfies the Jacobi identity. Observe that if $\delta(j,\delta(k,\ell))=0$ for all $j,k,\ell\in\mathcal{M}$, then
    \begin{align*}
        [v_j,[v_k,v_\ell]]+[v_k,[v_\ell,v_j]]+[v_\ell,[v_j,v_k]]&=\alpha_{j,\delta(k,\ell)}\alpha_{k\ell}v_{\delta(j,\delta(k,\ell))}+\alpha_{k,\delta(\ell,j)}\alpha_{\ell j}v_{\delta(k,\delta(\ell,j))}+\alpha_{\ell,\delta(j,k)}\alpha_{jk}v_{\delta(\ell,\delta(j,k))}=0,
    \end{align*}
    since each term is proportional to a vertex of the form $v_{\delta(j,\delta(k,\ell))}\equiv v_0=0$. This proves claim (i). Note this conclusion does not require the assumption that $V$ is a linearly independent set.

    Let us proceed to prove claim (ii). Suppose the contrapositive, i.e., that there exist fixed indices $j,k,\ell\in\mathcal{M}$ with $j\neq k\neq\ell\neq j$ such that
    \begin{align*}
        \left|\left\{\delta(j,\delta(k,\ell)),\delta(k,\delta(\ell,j)),\delta(\ell,\delta(j,k))\right\}\setminus\{0\}\right|\geq 2.
    \end{align*}
    Without loss of generality, assume that $\delta(j,\delta(k,\ell))\neq \delta(k,\delta(\ell,j))$, where $\delta(j,\delta(k,\ell)), \delta(k,\delta(\ell,j))\in\mathcal{M}$, and $j\neq k\neq\ell\neq j$. Then, the terms $\alpha_{j,\delta(k,\ell)}\alpha_{k\ell}v_{\delta(j,\delta(k,\ell))}$ and $\alpha_{k,\delta(\ell,j)}\alpha_{\ell j}v_{\delta(k,\delta(\ell,j))}$ are linearly independent elements in $\spn\{V\}$, since $\operatorname{char}(\mathbb{F})=0$. If $V$ is taken to be a basis, then the third term $\alpha_{\ell,\delta(j,k)}\alpha_{jk}v_{\delta(\ell,\delta(j,k))}$ cannot cancel both of the previous terms. This shows that violating the condition $\left|\left\{\delta(j,\delta(k,\ell)),\delta(k,\delta(\ell,j)),\delta(\ell,\delta(j,k))\right\}\setminus\{0\}\right|\leq 1$ implies that $G(V,E)$ cannot be associated with a minimal-graph-admissible Lie algebra as a minimal graph.
\end{proof}

Here, one needs to mention that Proposition~\ref{prop:allowed:trails} and part~(ii) of~\ref{prop:conditions:proper:choice:independent:proper:minimal} only provide weak criteria for checking the validity of graphs that are based on the Jacobi identity, since they are only applicable for minimal graphs and only provide necessary conditions. However, they are not sufficient to distinguish between choice-dependent and choice-independent graphs, which is essential in some cases, especially when this classification is also affected by the characteristic of the field over which the algebra is defined. For further remarks on this topic see Appendix~\ref{app:field:dependence:choice:dependent:graphs}. More fundamentally, the principal obstruction is that the prescribed graph construction suppresses the non-zero proportionality constants appearing in the Lie brackets. Consequently, particularly for choice-dependent graphs, the Jacobi identity cannot generally be verified from the graph alone. Indeed, the nested commutators entering the Jacobi identity may be proportional to the same vertex element without their coefficients cancelling. Although the graph records their common output direction, it does not retain the proportionality constants required to determine whether their sum vanishes. 
The situation is further complicated for redundant graphs because their vertex elements may satisfy non-trivial linear relations. Consequently, nested commutators terminating at distinct, and even pairwise non-proportional, vertex elements may nevertheless cancel in a suitable linear combination. These linear relations are likewise not encoded by the graph itself. Strong purely graph-theoretic sufficient criteria for properness may therefore be expected primarily in the minimal choice-independent setting. In the general choice-dependent or redundant setting, one must either solve additional algebraic equations for the relevant proportionality constants and linear relations or modify the graph construction by introducing suitable weights, colors, or additional labeling rules as discussed earlier.

To summarize the current status of the obtained conditions, Appendix~\ref{app:subsec:tables:summariting:proper:graph:conditions} contains two complementary tables. Table~\ref{tab:summary:criteria:proper:graphs} summarizes the available criteria, indicating whether they are necessary or sufficient, the settings in which they apply, and their remaining limitations. Table~\ref{tab:status:proper:graphs:minimal:redundant} organizes the results according to the Lie-algebra axioms and distinguishes between minimal and redundant graphs, and identifies aspects that remain unresolved.

\section{Structural properties of finite-dimensional graph-admissible Lie algebras via associated graphs}\label{sec:structural:properties}
We are now equipped with all the necessary tools to identify structural properties of a finite-dimensional graph-admissible Lie algebra $\g$ using only the graphs $G(V,E)$ associated with $\g$. Among the most common structural classifications of Lie algebras are those based on the notions of abelian, nilpotent, solvable, simple, semisimple, or reductive Lie algebras. These classifications rely typically on algebraic constructions such as the lower central series, the derived series, or the identification of ideals or central elements. In what follows, we explore how these structural features can be inferred directly from the graph representation, without requiring explicit computation of the Lie bracket beyond its graphical encoding.

For simplicity, we proceed by assuming that the characteristic of the field $\mathbb{F}$ is $\operatorname{char}(\mathbb{F})=0$, and continue thereby with the same assumption introduced in the previous section.

\subsection{Identifying central elements}

\begin{lemma}\label{lem:subalgebra:of:center}
    Let $\g$ be a finite-dimensional graph-admissible Lie algebra associated with the labeled directed graph $G(V,E)$. Let $W\subseteq V$ be the set of all vertices with no outgoing edges. Then $\spn\{W\}$ is a subalgebra of the center $\mathcal{Z}(\g)$ of $\g$.
\end{lemma}

\begin{proof}
    Let $w\in W\subseteq V$ be a vertex with no outgoing edges. Then, by the construction of Algorithm~\ref{alg:creating:graph}, there exists no element $v\in V$ such that $[v,w]\neq 0$. Since the set $V$ spans the Lie algebra $\g$, any element $g\in\g$ can be written as a linear combination: $g=\sum_{v\in V}c_v v$ with some coefficients $c_v\in\mathbb{F}$. Therefore, for any such $g$, one has $[w,g]=0$ for all $g\in \g$, implying that $w\in\mathcal{Z}(\g)$. This argument holds for every vertex $w\in W$ that has no outgoing edges, and all their linear combinations. Thus $\spn\{W\}\subseteq \mathcal{Z}(\g)$.
\end{proof}

It is important to note that Lemma~\ref{lem:subalgebra:of:center} does not guarantee that the set of all vertices with no outgoing edges spans the whole center $\mathcal{Z}(\g)$ of the Lie algebra $\g$. That is, while such vertices always lie in the center, the center may contain additional elements not represented directly by single vertices without outgoing edges. To illustrate this, consider the real Lie algebra visualized in Figure~\ref{fig:no:center:visible}.

\begin{figure}[H]
    \centering
    \includegraphics[width=0.6\linewidth]{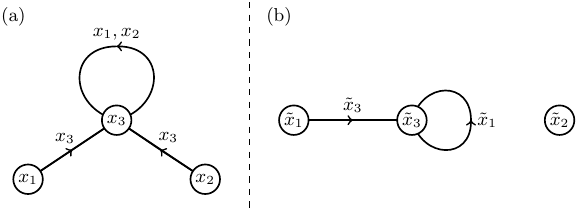}
    \caption{Graph associated to the real Lie algebra $\g$, spanned by the basis elements $x_1$, $x_2$, $x_3$ satisfying $[x_1,x_2]=0$, $[x_2,x_3]=x_3$, and $[x_3,x_1]= x_3$, which is known as the Lie algebra $\mathfrak{aff}(\R)\oplus \R$ \cite{Popovych:2003,Andrada:2005}. In panel (a), the graph $G(V,E)$ constructed from this basis has no vertices without outgoing edges, i.e., the set $W$ from Lemma~\ref{lem:subalgebra:of:center} is empty. Nonetheless, the center of the Lie algebra is clearly nontrivial, as  $\mathcal{Z}(\g)=\lie{\{x_1+x_2\}}$. In panel (b), we constructed the graph using the modified basis $\mathcal{B}=\{\tilde{x}_1:=x_1,\tilde{x}_3:=x_3,\tilde{x}_2:=x_1+x_2\}$. Here, one has clearly $\spn\{W\}=\mathcal{Z}(\g)$.}
    \label{fig:no:center:visible}
\end{figure}
We can, however, remedy this problem with the following result:

\begin{lemma}\label{lem:central:elements:extended:graphs}
    Let $\g$ be a finite-dimensional graph-admissible Lie algebra with non-trivial center. Then $\g$ admits a labeled directed graph $G(V,E)$ such that the set of all vertices with no outgoing edges spans the center $\mathcal{Z}(\g)$. 
\end{lemma}

\begin{proof}
    Let $G(V,E)$ be a graph associated with a finite-dimensional graph-admissible Lie algebra $\g$, and suppose $\mathcal{Z}(\g)\neq\{0\}$. Let $V=\{v_j\}_{j=1}^m$ denote the set of vertices spanning $\g$. Let $C=\{c_j\}_{j=1}^k$ be a basis for the center $\mathcal{Z}(\g)$. If any $c_j\in C$ is not proportional to an existing vertex $v\in V$, one can simply add $c_j$ to $V$. Since each $c_j\in\mathcal{Z}(\g)$, it follows that $[v,c_j]=0$ for all $v\in V$, and hence the modified set $V$ is a valid set of vertices, according to Definition~\ref{def:graph:admissible}, and hence the subset of vertices in $V$ with no outgoing edges spans the center $\mathcal{Z}(\g)$.
\end{proof}

\begin{lemma}\label{lem:abelian:criterion}
    Let $\g$ be a finite-dimensional graph-admissible Lie algebra associated with a labeled directed graph $G(V,E)$. Then $\g$ is abelian if and only if the set of edges $E$ is empty.
\end{lemma}

\begin{proof}
    Let $G(V,E)$ be a graph associated with a finite-dimensional Lie algebra $\g$. Suppose $E$ is empty. Then there exists no edge $e\in E$, and hence no triple $(v_\mathrm{s},v_\mathrm{l},v_\mathrm{e})\in V^{\times3}$ such that $[v_\mathrm{s},v_\mathrm{l}]\propto v_\mathrm{e}$. Thus $[v_1,v_2]=0$ for all $v_1,v_2\in V$. Since $\spn\{V\}=\g$, it follows by linearity $[\spn\{V\},\spn\{V\}]=[\g,\g]\subseteq\{0\}$, implying that $\g$ is abelian. Conversely, suppose $\g$ is abelian. Then any (possibly overcomplete) basis $\mathcal{B}=\{\Tilde{x}_j\}_{j\in\mathcal{M}}$ satisfies $[\Tilde{x}_j,\Tilde{x}_k]=0$ for all $j,k\in\mathcal{M}$. By the construction of  Algorithm~\ref{alg:creating:graph}, the graph associated with $\mathcal{B}$ contains no edges, i.e., $E=\emptyset$.
\end{proof}

\begin{lemma}\label{lem:direct:sum:if:unconnected}
    Let $\g$ be a finite-dimensional minimal-graph-admissible Lie algebra associated with a minimal labeled directed graph $G(V,E)$. If $G(V,E)$ splits into two non-trivial unconnected subgraphs, then $\g$ is the direct sum of two mutually commuting ideals.
\end{lemma}

\begin{proof}
    Suppose $G(V,E)$ splits into two non-trivial unconnected subgraphs $G(V_1,E_1)$ and $G(V_2,E_2)$. Then $E=E_1\cup E_2$ and $V=V_1\cup V_2$, where $V_1\cap V_2=\emptyset=E_1\cap E_2$. An edge $e\in E_j$ cannot have its labeling vertex $\varpi_\mathrm{l}(e)$ outside of $V_j$, since by Proposition~\ref{prop:conditions:for:edges} its antisymmetric counterpart would belong to $E$ and connect vertices from $V_1$ and $V_2$, contradicting the assumption that the two graphs are unconnected. Thus, by construction, $\spn\{V_j\}$ are two commuting ideals. Moreover, because $V$ is a basis partitioned by $V_1$ and $V_2$, $\g=\spn\{V_1\}\oplus\spn\{V_2\}$. 
\end{proof}

We conclude with the following observation. In Lemma~\ref{lem:central:elements:extended:graphs} we established that for any finite-dimensional graph-admissible Lie algebra, one can construct a graph that faithfully represents the center $\mathcal{Z}(\g)$. That is the set of vertices with no outgoing edges spans the center. However, the provided construction does generally not yield a minimal graph since it involves adding new vertices to represent central elements not already present in the original basis. This raises a natural question: \emph{Is it possible to construct a minimal graph that faithfully represents the center of the Lie algebra if it is minimal-graph-admissible?} The following conjecture addresses this question.

\begin{conjecture}
    Let $\g$ be a finite-dimensional minimal-graph-admissible Lie algebra with non-trivial center. Then $\g$ admits a minimal graph $G(V,E)$, where the set of all vertices with no outgoing edges spans the center $\mathcal{Z}(\g)$. 
\end{conjecture}

\subsection{Computing the derived series}\label{sec:structural:properties:derived:series}
Let $\g$ be a finite-dimensional minimal-graph-admissible Lie algebra with a basis $\mathcal{B}=\{x_j\}_{j=1}^n$ that satisfies the Lie bracket relations \eqref{eqn:desired:basis}. Let $G(V,E)$ denote the labeled directed graph associated with $\g$, constructed via Algorithm~\ref{alg:creating:graph} using the basis $\mathcal{B}$. The first derived algebra $\mathcal{D}^1\g=[\g,\g]$ is spanned by the set $\{\alpha_{jk} x_{\delta(j,k)}\}_{j,k=1}^n$, which is clearly linearly dependent, since it contains the zero element $0=\alpha_{jj}x_{\delta(j,j)}$ for all $j\in\mathcal{N}$. To refine this set, we proceed as follows: We remove all elements $\alpha_{jk} x_{\delta(j,k)}$ for which $\alpha_{jk}=0$ or respectively $\delta(j,k)=0$. We remove, furthermore, duplicate elements. That is, we remove all but one element for which $\delta(j,k)=\delta(p,q)$ and $\alpha_{jk}\neq 0 \neq \alpha_{pq}$. In particular, due to the antisymmetry of $\boldsymbol{\alpha}$ and the symmetry of $\delta$, the elements $\alpha_{jk} x_{\delta(j,k)}$ and $\alpha_{kj} x_{\delta(k,j)}$ are proportional to each other. To formalize the recursive computation of the derived series, we define the following sequence of index sets:
\begin{align}
    \mathcal{N}_{\mathrm{D}}^{(\ell+1)}:=\delta\left(\left\{(j,k)\in\mathcal{N}_{\mathrm{D}}^{(\ell)}\times\mathcal{N}_{\mathrm{D}}^{(\ell)}\mid \;\alpha_{jk}\neq0\right\}\right)\quad\text{for all }\ell\in\N_{\geq0},\quad\text{where}\quad\mathcal{N}_{\mathrm{D}}^{(0)}:=\mathcal{N}.\label{eqn:recusrice:derived:index:set}
\end{align}
This recursive definition can be analogously applied not only to minimal-graph-admissible Lie algebras but also for redundant-graph-admissible ones and bases of redundant graphs. The key observation is that the following result implies that $\mathcal{D}^1\g=\spn\{x_j\mid \,j\in \mathcal{N}_{\mathrm{D}}^{(1)}\}$, which allows one to compute the entire derived series by this recursive procedure via $\mathcal{D}^\ell\g=\spn\{x_j\mid \,j\in \mathcal{N}_{\mathrm{D}}^{(\ell)}\}$.

\begin{lemma}\label{lem:if:g:graph:admissible:so:derived:algebras}
    Let $\g$ be a finite-dimensional graph-admissible Lie algebra. Then, the derived algebra $\mathcal{D}^\ell\g$ is also graph-admissible for all $\ell\in\N_{\geq0}$. Furthermore, if $\g$ is minimal-graph-admissible, then $\mathcal{D}^\ell \g$ is also minimal-graph-admissible for all $\ell\in\N_{\geq0}$.
\end{lemma}

\begin{proof}
    Let $\{\Tilde{x}_j\}_{j\in\mathcal{M}}$ be an overcomplete basis of the finite-dimensional graph-admissible Lie algebra $\g$ that satisfies the Lie bracket relations \eqref{eqn:desired:basis:overcomplete}. We define recursively the sets:
    \begin{align}
        \mathcal{M}_{\mathrm{D}}^{(\ell+1)}:=\delta\left(\left\{(j,k)\in\mathcal{M}_{\mathrm{D}}^{(\ell)}\times\mathcal{M}_{\mathrm{D}}^{(\ell)}\mid \;\Tilde{\alpha}_{jk}\neq0\right\}\right)\quad\text{for all }\ell\in\N_{\geq0},\quad\text{where}\quad\mathcal{M}_{\mathrm{D}}^{(0)}:=\mathcal{M}.\label{eqn:recusrice:derived:index:set:overcomplete}
    \end{align}
    By virtue of Definition~\ref{def:graph:admissible} of $\delta$ it is immediate that $\mathcal{M}_\mathrm{D}^{(\ell+1)}\subseteq \mathcal{M}_\mathrm{D}^{(\ell)}$ for all $\ell\in\N_{\geq0}$. Furthermore, one notes the following:
    \begin{itemize}
        \item Suppose $z\in\mathcal{D}^1\g$. Then there exist pairs of elements $x_j,y_j\in\g$, indexed by a finite non-empty index set $\mathcal{J}\subseteq\N_{\geq1}$, such that $z=\sum_{j\in\mathcal{J}} [x_j,y_j]$ holds\footnote{To realize that, recall that $z\in [\g,\g]$ if and only if $z\in\spn\{[x,y]\,\mid\,x,y\in\g\}$. Thus, any $z\in\mathcal{D}\g$ can be written as $z=\sum_{j,k}c_{jk}[x_j,x_k]$, where $j,k$ denote indices from suitable index sets and $c_{jk}$ are appropriate coefficients. However, due to the bilinearity of the Lie bracket, one can write $z=\sum_j[x_j,y_j]$ with $y_j=\sum_k c_{jk}x_k$, establishing the validity of the initial claim.}. Since $\{\Tilde{x}_j\}_{j\in\mathcal{M}}$ spans $\g$, there exist coefficients $\mu_{jp}\in\mathbb{F}$ and $\nu_{jq}\in\mathbb{F}$, with $p,q\in\mathcal{M}$, such that $x_j=\sum_{p\in\mathcal{M}}\mu_{jp}\Tilde{x}_p$ and $y_j=\sum_{q\in\mathcal{M}}\nu_{jq}\Tilde{x}_q$. Then: 
        \begin{align*}
            z=\sum_{p,q\in\mathcal{M}}\sum_{j\in\mathcal{J}} \mu_{jp}\nu_{jq}[\Tilde{x}_p,\Tilde{x}_q]=\sum_{p,q\in\mathcal{M}}\sum_{j\in\mathcal{J}} \mu_{jp}\nu_{jq}\Tilde{\alpha}_{pq}\Tilde{x}_{\delta(p,q)}
        \end{align*}
        Thus $z\in \spn\{\Tilde{x}_j\,\mid\, j\in\mathcal{M}_\mathrm{D}^{(1)}\}$, by the definition of $\mathcal{M}_\mathrm{D}^{(1)}$, and the convention that $\delta(p,q)=0$ whenever $\Tilde{\alpha}_{pq}=0$. 
        \item Suppose $z\in \spn\{\Tilde{x}_j\,\mid\, j\in\mathcal{M}_\mathrm{D}^{(1)}\}$. Then there exist coefficients $\lambda_j\in\mathbb{F}$ with $j\in\mathcal{M}_\mathrm{D}^{(1)}$ such that $z=\sum_{j\in\mathcal{M}_\mathrm{D}^{(1)}}\lambda_j\Tilde{x}_j$. By the definition of $\mathcal{M}_\mathrm{D}^{(1)}$, there exist indices $p_j,q_j\in\mathcal{M}$ such that $[\Tilde{x}_{p_j},\Tilde{x}_{q_j}]=\Tilde{\alpha}_{p_jq_j}\Tilde{x}_j$ with $\Tilde{\alpha}_{p_jq_j}\neq 0$. We now choose a particular pair of $(\hat{p}_j,\hat{q}_j)\in\mathcal{M}^2$ that satisfies $[\Tilde{x}_{\hat{p}_j},\Tilde{x}_{\hat{q}_j}]=\Tilde{\alpha}_{\hat{p}_j\hat{q}_j}\Tilde{x}_j$ with $\Tilde{\alpha}_{\hat{p}_j\hat{q}_j}\neq 0$. Then:
        \begin{align*}
            z=\sum_{j\in\mathcal{M}_\mathrm{D}^{(1)}}\lambda_j\frac{1}{\Tilde{\alpha}_{\hat{p}_j\hat{q}_j}}[\Tilde{x}_{\hat{p}_j},\Tilde{x}_{\hat{q}_j}]=\sum_{j\in\mathcal{M}_\mathrm{D}^{(1)}}\left[\frac{\lambda_j}{\Tilde{\alpha}_{\hat{p}_j\hat{q}_j}}\Tilde{x}_{\hat{p}_j},\Tilde{x}_{\hat{q}_j}\right],
        \end{align*}
        which implies that $z\in \mathcal{D}^1\g$. 
    \end{itemize}
    Hence, we conclude that $\mathcal{D}^1\g=\spn\{\Tilde{x}_j\,\mid\, j\in\mathcal{M}_\mathrm{D}^{(1)}\}$. Given the recursive definition $\mathcal{D}^{\ell+1}\g=[\mathcal{D}^\ell\g,\mathcal{D}^\ell\g]$ of the elements of the derived series, it follows immediately by means of induction that $\mathcal{D}^\ell\g=\spn\{\Tilde{x}_j\,\mid\, j\in\mathcal{M}_\mathrm{D}^{(\ell)}\}$. Consequently, one observes that $\{\delta(j,k)\,\mid\, j,k\in\mathcal{M}_\mathrm{D}^{(\ell)}\}\setminus\{0\}=\mathcal{M}_\mathrm{D}^{(\ell+1)}$, since $\g$ is assumed to be graph-admissible, and $\{\Tilde{x}_j\}_{j\in\mathcal{M}}$ is an appropriate overcomplete basis satisfying the bracket relations:
    \begin{align*}
        [\Tilde{x}_j,\Tilde{x}_k]&=\Tilde{\alpha}_{jk}\Tilde{x}_{\delta(j,k)}\quad\text{ for all }j,k\in\mathcal{M}_\mathrm{D}^{(\ell)},
    \end{align*}
    where $\delta(j,k)=0$ if $\Tilde{\alpha}_{jk}=0$ and otherwise $\delta(j,k)\in \mathcal{M}_\mathrm{D}^{(\ell)}$ for all $j,k\in\mathcal{M}_\mathrm{D}^{(\ell)}$. This confirms that every derived algebra is graph-admissible. 

    The second claim follows analogously by replacing $\mathcal{M}_\mathrm{D}^{(\ell)}$ with $\mathcal{N}_\mathrm{D}^{(\ell)}$ and noticing that the sets $\{x_j\,\mid\, j\in\mathcal{N}_\mathrm{D}^{(\ell)}\}$ are linearly independent as they are subsets of the basis $\{x_j\}_{j\in\mathcal{N}}$.
\end{proof}

Let us now denote the graph associated with the $\ell$-th derived algebra $\mathcal{D}^\ell\g$ by $\mathcal{D}^\ell G(V,E)\equiv G(\mathcal{D}^\ell V,\mathcal{D}^\ell E)$, and investigate the relationship between $\mathcal{D}^{\ell-1} G(V,E)$ and $\mathcal{D}^{\ell} G(V,E)$. Constructing a graph associated with the $\ell$-th derived algebra is achieved by applying Algorithm~\ref{alg:creating:graph} to the set $\{x_j\,\mid\,j\in \mathcal{N}_\mathrm{D}^{(\ell)}\}$, where $\{x_j\}_{j\in\mathcal{N}}$ is a basis of $\g$ satisfying the bracket relations \eqref{eqn:desired:basis}. A detailed implementation of this procedure is provided in Appendix~\ref{app:remaining:algorithms}, specifically in Algorithm~\ref{alg:generating:generating:the:graph:derived:unaltered}. 

In analyzing the graph $\mathcal{D}^\ell G(V,E)$ associated with the derived algebra $\mathcal{D}^\ell\g$, it becomes necessary to understand the conditions for an index $j$ to belong to $\mathcal{N}^{(\ell-1)}_{\mathrm{D}}$ but not to $\mathcal{N}^{(\ell)}_{\mathrm{D}}$. Consider therefore the definition of the sets $\mathcal{N}^{(\ell)}_{\mathrm{D}}$, provided in \eqref{eqn:recusrice:derived:index:set}. If an index $j$ belongs to $\mathcal{N}_{\mathrm{D}}^{(\ell-1)}$ but not to $\mathcal{N}_{\mathrm{D}}^{(\ell)}$, this implies that there exists no pair $(p,q)\in\mathcal{N}_{\mathrm{D}}^{(\ell-1)}\times\mathcal{N}_{\mathrm{D}}^{(\ell-1)}$ such that $[x_p,x_q]=\kappa x_j$ with $\kappa\in\mathbb{F}^*$ and $j\in\mathcal{N}$. In other words, the element $x_j$ does not appear as a nontrivial Lie bracket of any two basis elements from the previous derived set. Translated into graph-theoretic terms, $x_j$ corresponds to a vertex that is not the endpoint of any edge in $\mathcal{D}^{\ell-1}G(V,E)$. Consequently all such vertices with no incoming edges must be removed when constructing the graph $\mathcal{D}^{\ell}G(V,E)$ from $\mathcal{D}^{\ell-1} G(V,E)$. Due to their importance, and to differentiate them from abelian elements, we introduce the following notation:
\begin{definition}
    A vertex $v$ with no incoming edges but at least one outgoing edge is called a \emph{loose end}.
\end{definition}
To maintain the validity of the graph structure, it is mandatory to remove every edge in $\mathcal{D}^{\ell-1} G(V,E)$ whose starting, labeling, or end vertex belongs to the set of vertices eliminated in the current pruning round. All vertices and edges satisfying these deletion conditions are first marked and removed simultaneously only after all vertices have been inspected. All other vertices and edges persist since their corresponding elements are evidently part of $\mathcal{D}^{\ell}\g$. This ``pruning'' procedure is formalized into Algorithm~\ref{alg:generating:generating:the:graph:derived:altered}.

\begin{algorithm}
        \DontPrintSemicolon
        \KwData{A (possibly overcomplete) basis $\mathcal{B}=\{\Tilde{x}_j\}_{j\in\mathcal{M}}$ of the finite-dimensional graph-admissible Lie algebra $\g$, satisfying the Lie bracket relations \eqref{eqn:desired:basis:overcomplete}.}
        \KwResult{A sequence of labeled directed graphs $\mathcal{D}^\ell G(V,E)$ associated with the derived algebras $\mathcal{D}^\ell\g$ of the Lie algebra $\g$ for all $\ell\geq0$, constructed with respect to the given basis $\mathcal{B}$.}
        \SetKwData{Left}{left}\SetKwData{This}{this}\SetKwData{Up}{up}
        \SetKwFunction{Union}{Union}\SetKwFunction{FindCompress}{FindCompress}
        \SetKwInOut{Input}{input}\SetKwInOut{Output}{output}
    
        \BlankLine
        $\mathcal{D}$ $\leftarrow$ $\emptyset$
        \tcc*[h]{Initialize the set containing all graphs $\mathcal{D}^\ell G(V,E)$ associated with each derived algebra $\mathcal{D}^\ell\g$}\;
        \BlankLine
        $V$ $\leftarrow$ $\mathcal{B}$
        \tcc*[h]{Initialize the vertex set for $G(V,E)$ with the basis $\mathcal{B}$}\;
        $E$ $\leftarrow$ $\emptyset$
        \tcc*[h]{Initialize the edge set for $G(V,E)$}\;
        \BlankLine
        \ForEach(\tcc*[h]{Add all relevant edges to edge set $E$}){$(j,k)\in \mathcal{M}\times\mathcal{M}$ with $j<k$}{
            \If{$[\Tilde{x}_j,\Tilde{x}_k]=\Tilde{\alpha}_{jk} \tilde{x}_{\delta(j,k)}$ with $\Tilde{\alpha}_{jk}\neq 0$}{
                $E$ $\leftarrow$ $E\cup \{(\Tilde{x}_j,\Tilde{x}_k,\Tilde{x}_{\delta(j,k)})\}\cup \{(\Tilde{x}_k,\Tilde{x}_j,\Tilde{x}_{\delta(j,k)})\}$
                \tcc*[h]{Add edge to edge set $E$; each edge is an ordered triple: (start vertex, edge-label, end vertex)}\;
            }
        }
        $\mathcal{D}$ $\leftarrow$ $\{G(V,E)\}$\tcc*[h]{Add initial graph associated with $\g=\mathcal{D}^0\g$ to $\mathcal{D}$}\;
        \BlankLine
        \ForEach{$\ell\in\N_{\geq1}$}{
            $R_V$ $\leftarrow$ $\emptyset$\tcc*[h]{Initialize set of vertices that must be deleted in each pruning step}\;
            $R_E$ $\leftarrow$ $\emptyset$\tcc*[h]{Initialize set of edges that must be deleted in each pruning step}\;
            \ForEach{vertex $v\in V $}{
                \If{there exists no edge $e\in E$ such that $\varpi_\mathrm{e}(e)=v$}{
                    $R_V$ $\leftarrow$ $R_V\cup\{v\}$\tcc*[h]{Mark $v$ for deletion from $V$}\;
                    \ForEach{edge $e\in E$}{
                        \If(\tcc*[h]{i.e., edge is of the form $(\cdot,\cdot,v)$ or $(\cdot,v,\cdot)$ or $(v,\cdot,\cdot)$}){$v\in e$}{
                            $R_E$ $\leftarrow$ $R_E\cup\{e\}$\tcc*[h]{Mark $e$ for deletion from $E$}\;
                        }
                    }
                }
            }
            $V$ $\leftarrow$ $V\setminus R_V$ \tcc*[h]{Simultaneously remove untargeted vertices from $V$}\;
            $E$ $\leftarrow$ $E\setminus R_E$ \tcc*[h]{Simultaneously remove marked edges from $E$}\;
            $V^{(\ell)}\equiv\mathcal{D}^\ell V\leftarrow V$\;
            $E^{(\ell)}\equiv\mathcal{D}^\ell E\leftarrow E$\;
            $\mathcal{D}$ $\leftarrow$ $\mathcal{D}\cup\{\mathcal{D}^\ell G(V,E)=G(V^{(\ell)},E^{(\ell)})\}$\tcc*[h]{Add updated graph $\mathcal{D}^\ell G(V,E)$}\;
        }
        \Return $\mathcal{D}$ \tcc*[h]{Return set of all graphs $\mathcal{D}^\ell G(V,E)$ associated with the derived algebras $\mathcal{D}^\ell \g$ for all $\ell\geq0$}\;
\caption{Algorithm for generating a labeled directed graph associated with the $\ell$-th derived algebra $\mathcal{D}^\ell\g$ for a finite-dimensional graph-admissible Lie algebra $\g$}\label{alg:generating:generating:the:graph:derived:altered}
\end{algorithm}

\begin{lemma}\label{lem:derived:series:graph:alg:valid}
    Let $\g$ be a finite-dimensional graph-admissible Lie algebra associated with the labeled directed graph $G(V,E)$. Then, the sequence of graphs obtained via Algorithm~\ref{alg:generating:generating:the:graph:derived:altered} can be associated with the Lie algebras of the derived series of $\g$.
\end{lemma}

\begin{proof}
    This claim follows directly from the preceding discussion, since the logic does not change when replacing the basis $\{x_j\}_{j\in\mathcal{N}}$---satisfying the bracket relations \eqref{eqn:desired:basis}---with an overcomplete basis $\{\Tilde{x}_j\}_{j\in\mathcal{M}}$---satisfying the bracket relations \eqref{eqn:desired:basis:overcomplete}. However, note that in each pruning round, the sets $R_V$ and $R_E$ are determined with respect to the unchanged graph from the preceding round. Hence, the membership of a vertex in $R_V$ is not affected by the processing order. After all vertices have been inspected, the simultaneous removal of the marked vertices and edges produces precisely the graph associated with the next Lie algebra in the derived series.
\end{proof}

In order to highlight the working mechanisms of Algorithm~\ref{alg:generating:generating:the:graph:derived:altered}, it is instructive to apply it to a specific example of a solvable minimal-graph-admissible Lie algebra, as demonstrated in Example~\ref{exa:demonstration:algorithm:derived:series:minimal:graph}.

\begin{remark*}
    It is important to note that this method is not restricted to minimal-graph-admissible Lie algebras and minimal graphs. It also applies to redundant-graph-admissible Lie algebras and redundant graphs. In the case of redundant graphs, however, the set of vertices $V$ is a set of linearly dependent elements spanning the associated Lie algebra $\g$. Consequently, when performing Algorithm~\ref{alg:generating:generating:the:graph:derived:altered} to compute $\mathcal{D}^{\ell+1}G(V,E)$ starting from $\mathcal{D}^\ell G(V,E)$, it is possible that certain vertices $v_\ell$ with indices $\ell\in\mathcal{M}_\mathrm{D}^{(\ell)}$ are removed, even though the corresponding elements $\Tilde{x}_\ell$ belong to $\mathcal{D}^{\ell+1}\g$. This occurs whenever $\ell\notin\delta(\mathcal{M}_\mathrm{D}^{(\ell)},\mathcal{M}_\mathrm{D}^{(\ell)})$ but $\Tilde{x}_\ell$ is linearly dependent on other elements $\Tilde{x}_{\delta(j,k)}\in \mathcal{D}^{\ell+1}V$. If it is desirable to retain such elements in the graph, one can modify the Algorithm~\ref{alg:generating:generating:the:graph:derived:altered} to reintroduce these vertices back into the set $\mathcal{D}^{\ell+1} V$, along with all corresponding edges into $\mathcal{D}^{\ell+1} E$. This modification does not affect the correctness of the procedure, as the algorithm remains valid regardless of whether the initial basis is overcomplete or not. The modified version is presented in Appendix~\ref{app:remaining:algorithms} as Algorithm~\ref{alg:generating:generating:the:graph:derived:altered:redundant:included}.
\end{remark*}

\begin{tcolorbox}[breakable, colback=Cerulean!3!white,colframe=Cerulean!85!black,title=\textbf{Example}: Application of Algorithm~\ref{alg:generating:generating:the:graph:derived:altered}]
    
    \begin{example}\label{exa:demonstration:algorithm:derived:series:minimal:graph}
        To illustrate Algorithm~\ref{alg:generating:generating:the:graph:derived:altered} we apply it to the real Lie algebra $\g$, defined on the seven-dimensional vector space spanned by the basis $\mathcal{B}=\{x_0,\ldots, x_6\}$ whose elements satisfy the following Lie brackets:
        \begin{align*}
            [x_0,x_1]&=x_1,\;&\;[x_0,x_j]&=-(7-j)x_j\quad\text{for }j\in\{2,3,4,5,6\},\;&\;[x_1,x_j]&=x_{j+1}\quad\text{for }j\in\{2,3,4,5\},
        \end{align*}
        which are the only  non-vanishing ones.
        This Lie algebra is part of a family of solvable Lie algebras that can be faithfully realized within the skew-hermitian Weyl algebra $\hat{A}_1$, where the first derived algebra is non-abelian~\cite{A1:project}.
        
        Applying Algorithm~\ref{alg:generating:generating:the:graph:derived:altered} to this algebra yields a sequence of directed graphs, each corresponding to a derived algebra in the series $\mathcal{D}^\ell\g$. The graphs $\mathcal{D}^\ell G(V,E)$ are constructed with respect to the basis $\mathcal{B}$ and according to the pruning procedure of Algorithm~\ref{alg:generating:generating:the:graph:derived:altered}. The respective graphs are depicted in Figure~\ref{fig:demonstration:algorithm:derived:series:minimal:graph}.
        \begin{figure}[H]
            \centering
            \includegraphics[width=0.95\linewidth]{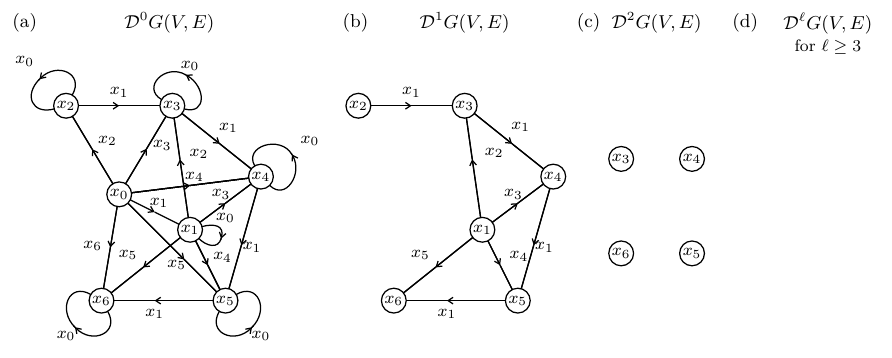}
            \caption{Graphs associated with the derived series of the Lie algebra from Example~\ref{exa:demonstration:algorithm:derived:series:minimal:graph}. In panel (a): The initial graph $\mathcal{D}^0G(V,E)\equiv G(V,E)$ associated with the full Lie algebra $\g$. In panel (b): The graph after removing all loose ends and isolated vertices, along with the corresponding edges, which can be associated with the first derived algebra $\mathcal{D}^1\g$. In panel (c): The graph $\mathcal{D}^2 G(V,E)$, and in panel (d):  The graphs $\mathcal{D}^\ell G(V,E)$ for $\ell\geq 3$.}
            \label{fig:demonstration:algorithm:derived:series:minimal:graph}
        \end{figure}
    \end{example}
\end{tcolorbox}

We are now able to state the following fundamental results:
\begin{theorem}\label{thm:solvable:termination:derived:series:graph}
    Let $\g$ be a finite-dimensional graph-admissible Lie algebra, and let $G(V,E)$ be any labeled directed graph associated with $\g$. Then, $\g$ is solvable if and only if the series of derived graphs $\mathcal{D}^\ell G(V,E)$ constructed by Algorithm~\ref{alg:generating:generating:the:graph:derived:altered} terminates after a finite number of steps.
\end{theorem}

Before proceeding with the proof, we need to mention that the series of derived graphs $\mathcal{D}^\ell G(V,E)$ is said to terminate, if for some $\ell_*$, the corresponding derived graph $\mathcal{D}^{\ell_*} G(V,E)$ is trivial, i.e., both the vertex and edge sets are empty: $\mathcal{D}^{\ell_*}V=\emptyset=\mathcal{D}^{\ell_*}E$. We adopt furthermore the convention that the trivial graph induces the trivial Lie algebra: $\lie{\emptyset}=\spn\{0\}$.

\begin{proof}
    This follows immediately from Lemma~\ref{lem:derived:series:graph:alg:valid} and the observation that each derived algebra satisfies $\mathcal{D}^\ell\g=\lie{\mathcal{D}^\ell V}$, where $\mathcal{D}^\ell G(V,E)=G(\mathcal{D}^\ell V,\mathcal{D}^\ell E)$ is the graph associated with the $\ell$-th derived algebra. The termination of the derived graph sequence corresponds precisely to the vanishing of the derived series, which characterizes solvability.
\end{proof}

We can further establish some graph-theoretic criteria for solvability that solely rely on the structure of the graph associated with the original Lie algebra, without needing to explicitly compute the derived series.

\begin{proposition}\label{prop:non:solvability:condition:weak}
    Let $\g$ be a finite-dimensional graph-admissible Lie algebra associated with the labeled directed graph $G(V,E)$. Then, $\g$ is non-solvable if $G(V,E)$ contains a non-empty self-contained subgraph $G_C\equiv G(\Tilde{V},\Tilde{E})$ that is induced by a cycle $C$\footnote{Note that every cycle is, by Definition~\ref{def:walk:trail:path:cycle}, a directed walk. Thus a graph induced by a cycle is analogously defined as a graph induced by a directed walk (cf. Definition~\ref{def:induced:graph}).}.
\end{proposition}

\begin{proof}
    Let $G(V,E)$ be a labeled directed graph associated with the finite-dimensional graph-admissible Lie algebra $\g$, and suppose it contains a non-empty self-contained subgraph $G_C\equiv G(\Tilde{V},\Tilde{E})$ induced by a cycle $C$. By definition, this means that $G(\Tilde{V},\Tilde{E})$ is a labeled directed graph, where every edge $e\in\Tilde{E}$ is labeled by a vertex $v\in\Tilde{V}$. The set of vertices $V$ spanning $\g$ is a possibly overcomplete, albeit finite, basis of $\g$. We can therefore define $V=\{v_j\}_{j\in\mathcal{M}}$. Without loss of generality, we may relabel the vertices such that $\Tilde{V}=\{v_j\}_{j=1}^{s}$ and $C=(v_1,e_1,v_2,\ldots,v_{s-1},e_{s-1},v_s,e_s,v_1)$, where $s\leq m=|\mathcal{M}|$. Consequently, the Lie brackets of the subgraph elements satisfy:
    \begin{align*}
        [v_{j},v_{k_j}]&=\alpha_{jk_j}v_{j+1}\quad\text{for all }j\in \{1,\ldots,s-1\}\quad\text{and}\quad[v_{s},v_{k_{s}}]=\alpha_{sk_{s}}v_1,
    \end{align*}
    where $k_j\in\{1,\ldots,s\}$ and $\alpha_{j k_j}\in\mathbb{F}^*$ for all $j\in\{1,\ldots,s\}$. It is therefore clear that $v_j,v_{k_j}\in [\spn\{\Tilde{V}\},\spn\{\Tilde{V}\}]$ for all $j\in\{1,\ldots,s\}$, and hence $\spn\{\Tilde{V}\}\subseteq [\spn\{\Tilde{V}\},\spn\{\Tilde{V}\}]$. It follows consequently that $\spn\{\Tilde{V}\}$ is contained in every derived algebra $\mathcal{D}^k\g$ for all $k\in\N_{\geq0}$. Since $\Tilde{V}\neq \emptyset$ by assumption, the derived series of $\g$ does not terminate, making $\g$ non-solvable.
    \begin{figure}[H]
        \centering
        \includegraphics[width=0.3\linewidth]{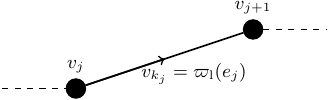}
        \caption{Visualization of the Lie bracket relation $[v_j,v_{k_j}]\propto v_{j+1}$, where $v_j,v_{j+1}$ are two consecutive vertices of a cycle $C$ connected by the directed edge $e_j$. The vertex $v_{k_j}$ is the labeling vertex of the edge $e_j$, i.e., $\varpi_\mathrm{l}(e_j)=v_{k_j}$.}
        \label{fig:small:cycle:proof:explainer}
    \end{figure}
\end{proof}

We can further generalize Proposition~\ref{prop:non:solvability:condition:weak} by extending the criterion from subgraphs induced by a directed cycle to those induced by closed directed walks. 

We begin with an observation concerning graphs that do not contain self-contained subgraphs induced by closed directed walks, under the assumption that they are associated with a finite-dimensional graph-admissible Lie algebra:
\begin{lemma}\label{lem:self:contained:subgraphs:induced:by:closed:walks}
    Let $\g$ be a finite-dimensional graph-admissible Lie algebra associated with the labeled directed graph $G(V,E)$. Then, $G(V,E)$ contains no self-contained subgraphs $G_W\equiv G(\Tilde{V},\Tilde{E})$ induced by non-trivial closed directed walks $W$ if and only if, in every subgraph $G_{W}\equiv G(\Tilde{V},\Tilde{E})$ that is induced by a non-trivial closed directed walk $W$, there exists a vertex $w\in\Tilde{V}$ such that for all edges $e\in E$ with $\varpi_\mathrm{e}(e)=w$, either the source $\varpi_\mathrm{s}(e)\notin \Tilde{V}$ or the label $\varpi_\mathrm{l}(e)\notin\Tilde{V}$.
\end{lemma}

While this lemma might appear self-evident at first glance, it is not trivial, as demonstrated by the example below.

\begin{tcolorbox}[breakable, colback=Cerulean!3!white,colframe=Cerulean!85!black,title=\textbf{Remark}: Discussion of  Lemma~\ref{lem:self:contained:subgraphs:induced:by:closed:walks}]
    \begin{example}
    Let $G(V,E)$ be a labeled directed graph associated with a finite-dimensional graph-admissible Lie algebra $\g$ that does not contain any subgraphs that are self-contained and induced by a closed directed walk. Envision a situation in which the subgraph $G(\Tilde{V},\Tilde{E})$, depicted in Figure~\ref{fig:example:condition:self:contained:a}, is a subgraph of $G(V,E)$ that is induced by a closed directed walk $W$. In this case, at least one edge must be labeled by a vertex from $V\setminus\Tilde{V}$. For illustration purposes, we assume here that two such edges exist and have colored them orange in the figure. Additionally, we can now suppose that for every vertex $v$ in $\Tilde{V}$, there exists an edge $e\in E$ such that $\varpi_\mathrm{e}(e)=v$ and both $\varpi_\mathrm{s}(e)\in\tilde{V}$ and $\varpi_\mathrm{l}(e)\in\Tilde{V}$. These edges, which do not belong to the walk $W$, are colored blue in Figure~\ref{fig:example:condition:self:contained:a}. 
    \begin{figure}[H]
        \centering
        \includegraphics[width=0.5\linewidth]{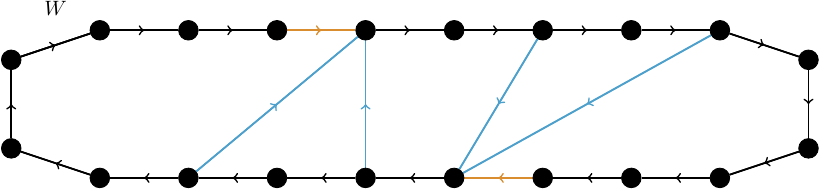}
        \caption{Depiction of a subgraph $G(\Tilde{V},\Tilde{E})$ that is induced by the closed directed walk $W$. The edges of $W$ are colored black if the labels are vertices from $\Tilde{V}$ and orange if this is not the case. Edges that do not belong to the walk $W$, but to the larger graph $G(V,E)$, are colored blue.}
        \label{fig:example:condition:self:contained:a}
    \end{figure}
    At first glance, it may appear that $G(V,E)$ contains a different closed directed walk $W'$ that induces a self-contained subgraph $G(\Tilde{V}',\Tilde{E}')$, as illustrated in Figure~\ref{fig:example:condition:self:contained:b}. However, upon closer inspection, some edges in this walk may be labeled by vertices that do not belong to $\Tilde{V}'$. This suggests that, even if no subgraph of $G(V,E)$ is induced by a closed directed walk that is self-contained, in principle it might still be possible to encounter subgraphs where, for every vertex, there exists an edge $e\in E$ such that the starting, labeling, and ending vertex of $e$ belong to the subgraph. In the proof that follows, we show that such a configuration is not possible.
    \begin{figure}[H]
        \centering
        \includegraphics[width=0.5\linewidth]{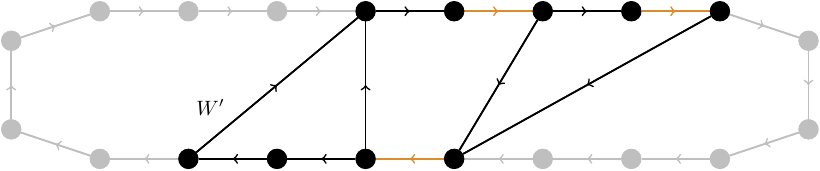}
        \caption{Overlay of the subgraph $G(\Tilde{V}',\Tilde{E}')$, induced by the closed directed walk $W'$, on top of the subgraph $G(\Tilde{V},\Tilde{E})\subseteq G({V},{E})$ from Figure~\ref{fig:example:condition:self:contained:a}. The edges of $G(\Tilde{V}',\Tilde{E}')$ are colored black if the labels are vertices from $\Tilde{V}'$, and orange if this is not the case. Vertices and edges that do not belong to the walk $W'$ but to the walk $W$ are colored gray.}
        \label{fig:example:condition:self:contained:b}
    \end{figure}
    \end{example}
\end{tcolorbox}

\begin{proof}
    Suppose that in every subgraph $G_{W}\equiv G(\Tilde{V},\Tilde{E})$ induced by a non-trivial closed directed walk $W$ there exists a vertex $w\in\Tilde{V}$ such that all edges $e\in E$ with $\varpi_\mathrm{e}(e)=w$ satisfy either $\varpi_\mathrm{l}(e)\notin \Tilde{V}$ or $\varpi_\mathrm{s}(e)\notin\Tilde{V}$. Then, by Definition~\ref{def:self:contained}, no such subgraph $G_W$ can be self-contained, since at least one edge of the induced subgraph $G_W$ violates the condition that all components of the edge belong to $\Tilde{V}$.
    
    For the converse direction, we argue by contraposition. That is, we want to show that if $G(V,E)$ possesses a subgraph $G_W\equiv G(\Tilde{V},\Tilde{E})$, induced by a non-trivial closed directed walk $W$ in $G(V,E)$, such that for all vertices $v\in\tilde{V}$ there exists an edge $e\in E$ such that $\varpi_\mathrm{e}(e)=v$ and $\varpi_\mathrm{s}(e),\varpi_\mathrm{l}(e)\in\Tilde{V}$, then there exists a closed directed walk $W'$ that induces a self-contained subgraph $G_{W'}\equiv G(\tilde{V}',\tilde{E}')$. Thus, let now $G_W\equiv G(\tilde{V},\tilde{E})$ be such a subgraph induced by $W$. Then, by assumption, for each $v\in\tilde{V}$ there exists at least one edge $e_v\in E$ with $\varpi_\mathrm{e}(e_v)=v$ and $\varpi_\mathrm{s}(e_v),\varpi_\mathrm{l}(e_v)\in \tilde{V}$. For each $v\in\Tilde{V}$ choose exactly one such edge $e_v\in E$ and keep these choices fixed.  For later convenience, we say that $W$ satisfies the \emph{incoming-edge condition}.
    
    All closed directed walks considered here and below have, without loss of generality, finite length. For a directed walk $W''$ with induced graph $G_{W''}\equiv G(\Tilde{V}'',\Tilde{E}'')$, we introduce the set
    \begin{align*}
        \mathcal{M}(W''):=\left\{e\in\Tilde{E}''\,\middle\mid\,\varpi_\mathrm{l}(e)\notin \Tilde{V}''\right\}
    \end{align*}
    of all edges of the walk $W''$ whose labeling vertices are not visited by $W''$ itself and set $m:=|\mathcal{M}(W)|$. For each later walk, $\mathcal{M}$ is defined analogously using that walk's own visited vertices and traversed edges.

    We will now show by complete induction that, for every $m\geq 0$, every non-trivial closed directed walk $W''$ that satisfies the incoming-edge condition and $|\mathcal{M}(W'')|=m$ implies the existence of a closed directed walk $W'$ that induces a self-contained subgraph $G_{W'}\equiv G(\tilde{V}',\Tilde{E}')$.

    The base case $m=0$ is immediate, since in this case $W''$ visits every labeling vertex of its edges, implying that $G_{W''}$ is self-contained.

    Let us continue with the induction step. Thus, let $m>0$ and assume that the claim has been established for all non-trivial closed directed walks satisfying the incoming-edge condition and having fewer than $m$ distinct edges labeled by vertices not visited by the walk itself. We may now choose some $e\in\mathcal{M}(W)$ and set $v:=\varpi_\mathrm{e}(e)$. For each $t\in\Tilde{V}$ with $t\neq v$, we can find a non-trivial directed segment $S_{vt}$ of $W$ that starts at $v$, ends at $t$, and does not traverse $e$. Indeed, start immediately after an occurrence of $e$ in the walk $W$, i.e., at $v$, and follow the walk, cyclically if necessary, until $t$ is reached. Whenever $v$ is visited again before reaching $t$, discard the part of the segment preceding this visit. The remaining segment therefore starts at $v$ and does not return to $v$ before reaching $t$. It therefore does not traverse $e$, as traversing any occurrence of $e$ would return to $v$, and the part through that occurrence would have been discarded. For $t=v$ set $S_{vv}$ to be the empty segment. By assumption, the previously chosen edge $e_v\in E$ satisfies $\varpi_\mathrm{e}(e_v)=v$ and $\varpi_\mathrm{s}(e_v),\varpi_\mathrm{l}(e_v)\in \tilde{V}$. Set $u:=\varpi_\mathrm{s}(e_v)$. We may now construct the closed directed walk $W_1=(v,\ldots,u,e_v,v)$ by following the previously constructed segment $S_{vu}$ from $v$ to $u$ and then traversing $e_v$ back to $v$. By construction, $W_1$ is a non-trivial finite closed directed walk that does not traverse $e$. Indeed, $S_{vu}$ does not traverse $e$ by construction and $e_v\neq e$, since $\varpi_\mathrm{l}(e_v)\in \Tilde{V}$, whereas $\varpi_\mathrm{l}(e)\notin \Tilde{V}$. Hence $W_1$ avoids every occurrence of $e$ and only visits vertices from $\Tilde{V}$.

    Observe that $W_1$ need not satisfy the incoming-edge condition or induce a self-contained subgraph $G_{W_1}\equiv G(\tilde{V}_1,\Tilde{E}_1)$, as, for example, it may omit vertices that label some edges $\tilde{e}\in\Tilde{E}_1$, including the label of $e_v$. In order to apply the induction hypothesis, we must restore the incoming-edge condition and obtain a walk with fewer than $m$ edges whose labeling vertices are not visited.

    With the construction presented below, we will include the missing vertices by an iterative procedure. Write $G_{W_j}\equiv G(\tilde{V}_j,\Tilde{E}_j)$ for the graph induced by the current closed directed walk $W_j$, where we start with the walk $W_1$. This walk has the following properties: (i) $W_1$ is a non-trivial finite closed directed walk; (ii) $v\in\Tilde{V}_1\subseteq\Tilde{V}$; (iii) $e\notin \Tilde{E}_1$; and (iv) each edge $\tilde{e}\in\tilde{E}_1$ satisfies either $\tilde{e}\in\tilde{E}\setminus\{e\}$ or $\varpi_\mathrm{s}(\tilde{e}),\varpi_\mathrm{l}(\tilde{e}),\varpi_\mathrm{e}(\tilde{e})\in\Tilde{V}$. We will preserve these properties with index $j$ at each stage in the following procedure. The index $j$ labels the successively constructed walks. The outer induction parameter $m=|\mathcal{M}(W)|$ remains fixed during this construction. Suppose $W_j$ has the same properties. We now describe how to obtain $W_{j+1}$ if (a) $|\mathcal{M}(W_j)|\geq m$ or (b) there exists a vertex $t\in\Tilde{V}_j$ whose chosen incoming edge $e_t$ satisfies either $\varpi_\mathrm{s}(e_t)\notin\Tilde{V}_j$ or $\varpi_\mathrm{l}(e_t)\notin \Tilde{V}_j$. If either condition holds, proceed as follows.
    \begin{itemize}
        \item \textbf{Case (a).} Here, there must exist an edge $\tilde{e}\in \Tilde{E}_j$ such that $\varpi_\mathrm{l}(\tilde{e})\in\Tilde{V}\setminus\Tilde{V}_j$, since $|\mathcal{M}(W)\setminus\{e\}|=m-1$ and, by property (iv), every edge $\tilde{e}\in \Tilde{E}_j\setminus\mathcal{M}(W)$ satisfies $\varpi_\mathrm{l}(\tilde{e})\in\Tilde{V}$. Choose such an edge $\tilde{e}$ and set $a:=\varpi_\mathrm{l}(\tilde{e})$, $t:=\varpi_\mathrm{e}(\tilde{e})$, and $u:=\varpi_\mathrm{s}(\tilde{e})$. By construction, $t,u\in\Tilde{V}_j$ as they are both visited by $W_j$. According to Proposition~\ref{prop:conditions:for:edges}, the antisymmetric counterpart $\tilde{e}'=(a,u,t)$ belongs to $E$.

        Now start at an occurrence of $t$ in $W_j$ and follow $W_j$, cyclically if necessary, until reaching $v$. Then follow the segment $S_{va}$ of $W$, constructed as before so that it avoids traversing $e$. Finally, traverse $\tilde{e}'$ from $a$ back to $t$. This produces a closed directed walk that can be inserted into $W_j$ at $t$ to obtain $W_{j+1}$. This procedure adds the previously missing vertex $a$ to the set of vertices visited by $W_{j+1}$. Its closing edge $\tilde{e}'$ is labeled by $u$, which was already visited by $W_j$, and its starting and ending vertices are visited by $W_{j+1}$.

        \item \textbf{Case (b).} Let $t\in\Tilde{V}_j$ be a vertex whose chosen incoming edge $e_t$ satisfies either $\varpi_\mathrm{s}(e_t)\notin \Tilde{V}_j$ or $\varpi_\mathrm{l}(e_t)\notin \Tilde{V}_j$. Denote one such missing vertex by $a$. By assumption, $a\in\Tilde{V}$. If $\varpi_\mathrm{s}(e_t)=a$, set $\tilde{e}':= e_t$. Otherwise, $\varpi_\mathrm{l}(e_t)=a$, and we set $\tilde{e}':=(\varpi_\mathrm{l}(e_t),\varpi_\mathrm{s}(e_t),t)$, i.e., the antisymmetric counterpart of $e_t$, which is contained in $E$ by Proposition~\ref{prop:conditions:for:edges}. In either case, $\tilde{e}'$ starts at $a$ and ends at $t$. 

        Now, similarly to case (a), start at an occurrence of $t$ in $W_j$ and follow $W_j$, cyclically if necessary, until reaching $v$. Then follow the segment $S_{va}$ of $W$ that avoids $e$, and finally traverse $\tilde{e}'$ from $a$ back to $t$. Insert this closed directed walk at $t$ in $W_j$ to obtain $W_{j+1}$. This newly constructed walk visits the previously missing vertex $a$. If the label of its closing edge $\tilde{e}'$ is still missing, it is the other starting or labeling vertex of the chosen edge $e_t$. Case (b) can then be applied again at the same vertex $t$ in the next iteration of this procedure, and the next insertion includes this remaining missing input. 
    \end{itemize}

    In both cases, the inserted walk avoids $e$: its initial segment along $W_j$ avoids $e$ by property (iii), the segment $S_{va}$ avoids it by construction, and the closing edge has its labeling vertex in $\Tilde{V}$, whereas $\varpi_\mathrm{l}(e)\notin \Tilde{V}$. Also, every edge of $W_{j+1}$ is either an edge of $W_j$, an edge of the original walk $W$ other than $e$, or a new closing edge whose three vertices all belong to $\tilde{V}$. Thus, $W_{j+1}$ satisfies conditions (i)-(iv) from before. The inserted excursion visits the previously missing vertex $a$ and retains all vertices and edges of $W_j$. Hence, $\tilde{V}_j\subsetneq \tilde{V}_{j+1}\subseteq\tilde{V}$ and $\tilde{E}_j\subseteq\Tilde{E}_{j+1}$. 

    One now repeats the construction in cases (a) and (b) as long as either condition applies. This procedure must terminate after at most $|\tilde{V}|-|\tilde{V}_1|$ insertions, since every insertion strictly enlarges the visited set $\tilde{V}_j$, which is always contained in the finite set $\tilde{V}$. If the visited set reaches $\Tilde{V}$, case (b) cannot apply because, for each $v\in\Tilde{V}$, the starting and labeling vertices of every chosen edge $e_v$ lie in $\Tilde{V}$. Case (a) cannot apply either: by property (iv), the edges labeled outside of $\Tilde{V}$ must belong to $\mathcal{M}(W)\setminus \{e\}$, of which there are $m-1$. Thus, after finitely many steps we obtain a walk $W''$ with induced graph $G_{W''}\equiv G(\Tilde{V}'',\Tilde{E}'')$ for which neither case applies. The absence of case (a) implies that $|\mathcal{M}(W'')|<m$, while the absence of (b) shows that $W''$ satisfies the incoming-edge condition. The walk $W''$ is thus a non-trivial finite closed directed walk that satisfies the incoming-edge condition, as well as $|\mathcal{M}(W'')|<m$. By the complete induction hypothesis, there exists therefore a closed directed walk $W'$ that induces a self-contained subgraph $G_{W'}\equiv G(\tilde{V}',\Tilde{E}')$. This proves the converse by contraposition. That is, if $G(V,E)$ contains no self-contained subgraphs $G_W\equiv G(\Tilde{V},\Tilde{E})$ induced by a non-trivial closed directed walk $W$, then, in every subgraph $G_{W}\equiv G(\tilde{V},\Tilde{E})$ that is induced by a non-trivial closed directed walk $W$, there exists a vertex $w\in\Tilde{V}$ such that for all edges $e\in E$ with $\varpi_\mathrm{e}(e)=w$, either $\varpi_\mathrm{s}(e)\notin \Tilde{V}$ or $\varpi_\mathrm{l}(e)\notin\Tilde{V}$.
\end{proof}

\begin{theorem}\label{thm:non:solvability:condition:strong}
    Let $\g$ be a finite-dimensional graph-admissible Lie algebra and let $G(V,E)$ be any labeled directed graph associated with $\g$. Then $\g$ is non-solvable if and only if $G(V,E)$ contains a self-contained subgraph $G_W\equiv G(\Tilde{V},\Tilde{E})$ that is induced by a non-trivial closed directed walk $W$.
\end{theorem}

\begin{remark*}
    Note that the classification of a Lie algebra as solvable or non-solvable is intrinsic to the Lie algebra under consideration. Thus, the result above holds for every proper labeled directed graph that can be associated with the graph-admissible Lie algebra $\g$.
\end{remark*}

\begin{proof}
    Let $\g$ be a finite-dimensional graph-admissible Lie algebra associated with the labeled directed graph $G(V,E)$. For each $\ell\in\Np{}$, we write $G_\ell:=G(V_\ell,E_\ell):=\mathcal{D}^\ell G(V,E)$, where $\mathcal{D}^\ell G(V,E)$ is the derived graph of the initial graph $G(V,E)=G(V_0,E_0)=\mathcal{D}^0G(V,E)=G_0$ obtained by $\ell$ applications of Algorithm~\ref{alg:generating:generating:the:graph:derived:altered}. Since the graphs $G_\ell$ are obtained by Algorithm~\ref{alg:generating:generating:the:graph:derived:altered}, we may further observe that
    \begin{align*}
        V_{\ell+1}:=\left\{v\in V_\ell\,\middle\mid\,\exists\, e\in E_\ell\text{ such that }v=\varpi_\mathrm{e}(e)\right\}=\{\varpi_\mathrm{e}(e)\,\mid\, e\in E_\ell\},\qquad E_{\ell+1}:=\left\{e\in E_\ell\,\middle\mid\, \varpi_\mathrm{s}(e),\varpi_\mathrm{l}(e),\varpi_\mathrm{e}(e)\in V_{\ell+1}\right\}.
    \end{align*}
    This implies in particular that $V_{\ell+1}\subseteq V_\ell$ and $E_{\ell+1}\subseteq E_\ell$ for all $\ell\in\Np{}$.

    First, suppose that $G(V,E)$ contains a self-contained subgraph $G_W\equiv G(\Tilde{V},\Tilde{E})$ induced by a non-trivial closed directed walk $W=(w_0,e_0,w_1,\ldots,w_s,e_s,w_0)$. We now show by induction over $\ell$ that $\Tilde{V}\subseteq V_\ell$ and $\Tilde{E}\subseteq E_\ell$ for every $\ell\in\Np{}$. The base case $\ell=0$ is trivially true. Thus, we may assume that the induction hypothesis, i.e., the two inclusions $\Tilde{V}\subseteq V_\ell$ and $\Tilde{E}\subseteq E_\ell$, holds for some $\ell\in\Np{}$. Let now $w\in \Tilde{V}$ be a vertex visited by $W$. Then, there exists an edge $e_w\in \Tilde{E}\subseteq E_\ell$ such that $\varpi_\mathrm{e}(e_w)=w$. This implies that $w\in V_{\ell+1}$ for each $w\in\tilde{V}$, and therefore $\Tilde{V}\subseteq V_{\ell+1}$. Moreover, since $G_W$ is self-contained, it holds that $\varpi_\mathrm{s}(e),\varpi_\mathrm{l}(e),\varpi_\mathrm{e}(e)\in \Tilde{V}\subseteq V_{\ell+1}$ for every $e\in\Tilde{E}$. Thus, by construction of the sets $E_\ell$, it holds that $\Tilde{E}\subseteq E_{\ell+1}$, which completes the induction. Since $W$ is non-trivial, this shows that $V_\ell\neq\emptyset\neq E_\ell$ for all $\ell\in \Np{}$. Thus, by Theorem~\ref{thm:solvable:termination:derived:series:graph}, $\g$ is non-solvable, since the series $\mathcal{D}^\ell G(V,E)$ never terminates at the trivial graph.

    Let us now show the converse direction. For this, we argue by contraposition. That is, we assume that $G(V,E)$ contains no self-contained subgraph $G_W\equiv G(\Tilde{V},\Tilde{E})$ that is induced by a non-trivial closed directed walk $W$ and aim to show that $\g$ is in that case solvable. For this, according to Theorem~\ref{thm:solvable:termination:derived:series:graph}, we have to show that there exists a non-negative integer $\ell_*\in\Np{}$ such that $G_{\ell_*}=G(\emptyset,\emptyset)$. To show that there exists such a non-negative integer $\ell_*\in\Np{}$ such that $V_{\ell_*}=\emptyset$, it suffices to show that $V_{\ell+1}\subsetneq V_\ell$ for all $\ell\in\Np{}$ such that $V_\ell\neq\emptyset$, since $V$ is finite. This would then immediately also imply that $E_{\ell_*}=\emptyset$, and therefore $G_{\ell_*}=G(\emptyset,\emptyset)$.

    Suppose, for a proof by contradiction, that there exists some non-negative integer $f\in\Np{}$ such that $V_{f+1}=V_f\neq\emptyset$. This would then, by construction, immediately imply that $E_{f+1}=E_f\neq\emptyset$. Indeed, each edge $e\in E_f$ has all three of its vertices in $V_{f}=V_{f+1}$ and is therefore also retained, which implies that $E_{f}\subseteq E_{f+1}$. The converse inclusion follows from the earlier observation that $E_{f+1}\subseteq E_{f}$. Non-emptiness follows from $V_{f+1}\neq\emptyset$. Next, let us introduce, for each vertex $v\in V_f$, the set
    \begin{align*}
        R_v:=\left\{w\in V_f\,\middle\mid\,\exists\,\text{closed directed walk $W$ in $G_f$ such that $W$ visits $w$ and $v$}\right\}\cup\{v\}.
    \end{align*}
    Here, we observe that, for $u,v\in V_f$, it clearly holds that $R_v=R_u$ exactly when $u\in R_v$. More generally, for two vertices $u,v\in V_f$, it holds that either $R_v=R_u$ or $R_v\cap R_u=\emptyset$. Furthermore, for at least one $v\in V_f$, the set $R_v$ is not trivial, i.e., there exists a second vertex $w\in R_v\setminus\{v\}$. To show this, we first observe that, for every $\ell\in \Np{}$, if $(v_\mathrm{s},v_\mathrm{l},v_\mathrm{e})\in E_\ell$, then the anti-symmetric counterpart $(v_\mathrm{l},v_\mathrm{s},v_\mathrm{e})$ also belongs to $E_\ell$. The base case $\ell=0$ follows immediately from Proposition~\ref{prop:conditions:for:edges}. Hence, assume the claim holds for some $\ell\in\Np{}$, and let $e=(v_\mathrm{s},v_\mathrm{l},v_\mathrm{e})\in E_{\ell+1}$. Since $E_{\ell+1}\subseteq E_\ell$, one has that $e\in E_\ell$. By the induction hypothesis, one has $e':=(v_\mathrm{l},v_\mathrm{s},v_\mathrm{e})\in E_\ell$. However, since $e\in E_{\ell+1}$, it also holds that $\varpi_\mathrm{s}(e)=\varpi_\mathrm{l}(e')=v_\mathrm{s}$, $\varpi_\mathrm{l}(e)=\varpi_\mathrm{s}(e')=v_\mathrm{l}$, and $\varpi_\mathrm{e}(e)=\varpi_\mathrm{e}(e')=v_\mathrm{e}$ all belong to $V_{\ell+1}$. Thus, $e'\in E_{\ell+1}$, by construction of the sets $E_\ell$, which completes the proof by induction.

    Now, pick some $v\in V_f=V_{f+1}$. Then, there exists an edge $e_v\in E_f$ such that $\varpi_\mathrm{e}(e_v)=v$ and $\varpi_\mathrm{s}(e_v),\varpi_\mathrm{l}(e_v)\in V_f$. By the observation above, the antisymmetric counterpart $\tilde{e}_v:=(\varpi_\mathrm{l}(e_v),\varpi_\mathrm{s}(e_v),v)$ also belongs to $E_f$. According to Proposition~\ref{prop:conditions:for:edges}, the starting and labeling vertices of $e_v$ and therefore also of $\tilde{e}_v$ are distinct, so at least one of these two edges has distinct starting and ending vertices. Pick this edge and move backwards to its starting vertex. Then repeat the same process for the vertex reached by this process. Since $V$ and therefore $V_f$ are finite, after at most $|V_f|$ iterations a vertex must be visited twice. Retain the part of the walk constructed by this process that starts at the second occurrence of the first repeated vertex and ends at its first occurrence, where the walk shall be followed in the direction of the edges, i.e., the reverse direction compared to its construction. The resulting walk is a non-trivial closed directed walk that visits at least two distinct vertices. Hence, at least one set $R_v$ contains at least two distinct vertices.

    We may now decompose $V_f$ as follows: pick some $v_1\in V_f$ and compute $R_{v_1}$. Then choose each successive $v_k\in V_f$ from $V_f\setminus\bigcup_{j=1}^{k-1}R_{v_j}$ and continue until every vertex $v\in V_f$ is included in exactly one of these sets $R_{v_k}$. Since $V_f$ is finite and non-empty, one has
    \begin{align*}
        V_f=\bigcup_{k=1}^d R_{v_k}\quad\text{for some }d\in\N_{\geq1}.
    \end{align*}
    Note that by an earlier observation, the sets $R_{v_k}$ are pairwise disjoint.

    In a next step, we want to show that for at least one $k\in\{1,\ldots,d\}$ there exists no edge $e\in E_f$ such that $\varpi_\mathrm{e}(e)\in R_{v_k}$ while $\varpi_\mathrm{s}(e)\notin R_{v_k}$. To establish this, we assume the contrary. In this case, we may pick some $k_1\in\{1,\ldots,d\}$ such that there exists an edge $e_1\in E_f$ with $\varpi_\mathrm{e}(e_1)\in R_{v_{k_1}}$ and $\varpi_\mathrm{s}(e_1)\notin R_{v_{k_1}}$, where $\varpi_\mathrm{s}(e_1)\in R_{v_{k_2}}$ for some $k_2\in\{1,\ldots,d\}\setminus\{k_1\}$. We may now repeat this process iteratively. That is, we choose an edge $e_j\in E_f$ such that $\varpi_\mathrm{e}(e_j)\in R_{v_{k_j}}$ while $\varpi_\mathrm{s}(e_j)\in R_{v_{k_{j+1}}}$, where $k_{j+1}\neq k_j$. Since $d$ is finite, there must exist a smallest integer $j_*\geq 2$ such that $k_{j_*+1}\in\{k_1,\ldots, k_{j_*}\}\subseteq\{1,\ldots,d\}$. Since $k_{j_*+1}\neq k_{j_*}$, we may write $k_{j_*+1}=k_q$ for some $q\in\{1,\ldots,j_*-1\}$. By the choice of $j_*$, the indices $k_q,k_{q+1},\ldots, k_{j_*}$ are pairwise distinct. For each $p\in\{q+1,\ldots,j_*\}$, the two vertices $\varpi_\mathrm{e}(e_p)$ and $\varpi_\mathrm{s}(e_{p-1})$ both belong to $R_{v_{k_p}}$. Moreover, $\varpi_\mathrm{e}(e_q)$ and $\varpi_\mathrm{s}(e_{j_*})$ both belong to $R_{v_{k_q}}$, since $k_{j_*+1}=k_q$. By construction of the sets $R_v$, for each $p\in\{q+1,\ldots,j_*\}$ there exists a directed walk $W_p$ from $\varpi_\mathrm{e}(e_p)$ to $\varpi_\mathrm{s}(e_{p-1})$ only visiting vertices in $R_{v_{k_p}}$. Moreover, there exists a directed walk $W_q$ from $\varpi_\mathrm{e}(e_q)$ to $\varpi_\mathrm{s}(e_{j_*})$ only visiting vertices in $R_{v_{k_q}}$. If the two vertices are identical take the corresponding walk to be the empty walk. Now follow the edges $e_{j_*},e_{j_*-1},\ldots, e_{q}$ in this order and join the ending and starting vertices $\varpi_\mathrm{e}(e_p)$ and $\varpi_\mathrm{s}(e_{p-1})$ of two consecutive edges with the corresponding walks $W_p$ for $p\in\{q+1,\ldots,j_*\}$. Finally follow $W_q$ from $\varpi_\mathrm{e}(e_q)$ to $\varpi_\mathrm{s}(e_{j_*})$ to close the resulting walk. This construction yields a non-trivial closed directed walk that visits the starting and ending vertices of every $e_p$ with $p\in\{q,q+1,\ldots,j_*\}$.  By the definition of the sets $R_v$ and the equality of the overlapping sets, all these vertices belong to the same set. In particular, this would imply that
    \begin{align*}
        \varpi_\mathrm{s}(e_{j_*})\in\bigcap_{p=q}^{j_*}R_{v_{k_p}}=\emptyset,
    \end{align*}
    which is impossible. Thus, for at least one $k\in\{1,\ldots,d\}$ there exists no edge $e\in E_f$ with $\varpi_\mathrm{e}(e)\in R_{v_k}$ and $\varpi_\mathrm{s}(e)\notin R_{v_k}$.

    Let now $R_{v_{\hat{k}}}$ be such that every edge $e\in E_f$ with $\varpi_\mathrm{e}(e)\in R_{v_{\hat{k}}}$ satisfies $\varpi_\mathrm{s}(e)\in R_{v_{\hat{k}}}$. We will now construct a closed directed walk $W$ in $R_{v_{\hat{k}}}$ where the labeling vertex of each traversed edge is also visited. Set $\{w_1,\ldots,w_s\}=R_{v_{\hat{k}}}$ and start with $w_1$. By construction, there exists an edge $e_{w_1}\in E_f$ such that $\varpi_\mathrm{e}(e_{w_1})=w_1$ and $a:=\varpi_\mathrm{s}(e_{w_1})$ and $b:=\varpi_\mathrm{l}(e_{w_1})$ both belong to $V_f$. Moreover, by an earlier observation, the antisymmetric counterpart $\tilde{e}_{w_1}=(b,a,w_1)$ also belongs to $E_f$. Since both edges end in $R_{v_{\hat{k}}}$, they must also both start in $R_{v_{\hat{k}}}$. Thus, both $a$ and $b$ belong to $R_{v_{\hat{k}}}$. The same argument applies to every edge of $E_f$ ending in $R_{v_{\hat{k}}}$. The set $R_{v_{\hat{k}}}$ contains at least two vertices, since Proposition~\ref{prop:conditions:for:edges} does not allow edges of the form $(w_1,w_1,w_1)$. By construction, for each $w_\ell\in R_{v_{\hat{k}}}\setminus\{w_1\}$ there exists a closed directed walk $W_\ell$ that visits both $w_1$ and $w_\ell$. For each $\ell\in\{2,\ldots,s\}$, choose $W_\ell$ such that it begins at $w_1$ and follow these walks one after another. Each such walk stays within $R_{v_{\hat{k}}}$, because every vertex visited lies on a closed directed walk with $w_1$. The resulting walk therefore visits precisely all vertices from $R_{v_{\hat{k}}}$ and is hence non-trivial. Let now $G_W\equiv G(\Tilde{V},\Tilde{E})$ be the subgraph induced by this walk, where now $\Tilde{V}=R_{v_{\hat{k}}}$. By construction and the argument before, for each vertex $v\in \Tilde{V}$, there exists an edge $e_v\in E_f$ with $\varpi_\mathrm{e}(e_v)=v$ and $\varpi_\mathrm{s}(e_v),\varpi_\mathrm{l}(e_v)\in \Tilde{V}$. Lemma~\ref{lem:self:contained:subgraphs:induced:by:closed:walks} now implies that $G(V,E)$ contains a self-contained subgraph induced by a closed directed walk. This contradicts the initial assumption.

    Consequently $V_{\ell+1}\subsetneq V_\ell$, whenever $V_\ell\neq \emptyset$. Since $V$ is finite, after at most $|V|$ steps we obtain $V_{\ell_*}=\emptyset$ for some $\ell_*\in\Np{}$. By the discussion at the beginning of the proof of this direction, it follows that $\g$ is solvable. Hence, the contraposition holds: if $\g$ is non-solvable, then $G(V,E)$ contains a non-trivial self-contained subgraph induced by a non-trivial closed directed walk. This proves the converse direction and therefore the complete claim.
\end{proof}

\begin{corollary}
    Let $\g$ be a finite-dimensional redundant-graph-admissible Lie algebra associated with the redundant graph $G(V,E)$. Suppose $G(V,E)$ contains a closed directed walk $W$ that induces the subgraph $G_W\equiv G(\Tilde{V},\Tilde{E})$, where every edge in $\Tilde{E}$ is labeled either by a vertex $v\in \Tilde{V}$ or a vertex $w\in (V \cap \spn\{\Tilde{V}\})\setminus \Tilde{V}$---that is, a linear combination of the vertices from $\Tilde{V}$. Then, $G(V,E)$ contains a closed directed walk $W'$ that induces a self-contained subgraph $G_{W'}\equiv G(\Tilde{V}',\Tilde{E}')$.
\end{corollary}

\begin{proof}
    Let $W=(v_1,e_1,v_2,e_2,\ldots, e_{s},v_1)$ be a closed directed walk of length $s$ that satisfies the conditions stated in the claim above. This walk induces the subgraph $G_W\equiv G(\Tilde{V},\Tilde{E})$, where $\Tilde{V}=\{v_j\,\mid\,j\in\{1,\ldots,s\}\}$ and $\Tilde{E}=\{e_j\,\mid \,j\in\{1,\ldots,s\}\}$. This implies the following sequence of Lie bracket relations:
    \begin{align*}
        [v_1,\varpi_\mathrm{l}(e_1)]&\propto v_2,\;&\;[v_2,\varpi_\mathrm{l}(e_2)]&\propto v_3,\;&\;\;&\;\ldots\;&\;[v_{s},\varpi_\mathrm{l}(e_{s})]&\propto v_1.
    \end{align*}
    That is, each bracket operation involving a vertex in the walk and the labeling vertex of the consecutive edge yields the next vertex in the walk, up to a scalar multiple. Recall that for two subspaces $\mathfrak{v},\mathfrak{w}$ of a Lie algebra $\g$, one defines $[\mathfrak{v},\mathfrak{w}]:=\spn\{[v,w]\,\mid\,v\in\mathfrak{v},\,w\in\mathfrak{w}\}$ \cite{Knapp:1996}. We have, therefore, $\spn\{\Tilde{V}\}\subseteq[\spn\{\Tilde{V}\},\spn\{\Tilde{V}\}]$, since $\varpi_\mathrm{l}(e)\in\spn\{\Tilde{V}\}$ for all $e\in \Tilde{E}$. Hence $\spn\{\Tilde{V}\}\subseteq\mathcal{D}^1\g$, because $\Tilde{V}\subseteq V$. Consequently: $\spn\{\Tilde{V}\}\subseteq \mathcal{D}^\ell \g$ for all $\ell\in\N_{\geq0}$. It follows that the derived series of $\g$ does not terminate, making $\g$ non-solvable. By Theorem~\ref{thm:non:solvability:condition:strong}, this implies that $G(V,E)$ must contain a closed directed walk that induces a self-contained subgraph of $G(V,E)$. 
\end{proof}

\subsection{Computing the lower central series}

The computation of the lower central series using a labeled directed graph associated with a finite-dimensional graph-admissible Lie algebra $\g$ follows a procedure similar to that used for the derived series. In fact, the first step must yield the same result as the one obtained from the first step for the procedure utilized in case of the derived series, since $\mathcal{D}^1\g=\mathcal{C}^1\g$ holds for all Lie algebras $\g$. The key difference lies in the recursive construction of the consecutive algebras: while the derived series is formed by taking Lie brackets among elements from the preceding algebra, the lower central series is obtained by taking Lie brackets between elements of the original algebra and those of the previous algebra in the series. Let $\mathcal{B}=\{{x}_j\}_{j=1}^n$ be a basis of the finite-dimensional minimal-graph-admissible Lie algebra $\g$ satisfying the relations \eqref{eqn:desired:basis}. We now define the index sets
\begin{align}
    \mathcal{N}_{\mathrm{C}}^{(\ell+1)}:=\delta\left(\left\{(j,k)\in\mathcal{N}\times\mathcal{N}_{\mathrm{C}}^{(\ell)}\mid \;\alpha_{jk}\neq0\right\}\right)\quad\text{for all }\ell\in\N_{\geq0},\quad\text{where}\quad\mathcal{N}_{\mathrm{C}}^{(0)}:=\mathcal{N}.\label{eqn:recusrice:lower:index:set}
\end{align}
An analogous construction applies to finite-dimensional redundant-graph-admissible Lie algebras and their associated redundant graphs. 
The essential observation is now that the elements of the lower central series are generally given by $\mathcal{C}^{\ell}\g=\spn\{x_j\,\mid j\in \mathcal{N}_{\mathrm{C}}^{(\ell)}\}$, as will be formalized in the following lemma:

\begin{lemma}\label{lem:if:g:graph:admissible:so:central:series:algebras}
    Let $\g$ be a finite-dimensional graph-admissible Lie algebra. Then, the algebras $\mathcal{C}^\ell\g$ in the lower central series are also graph-admissible. Moreover, if $\g$ is minimal-graph-admissible, then each algebra $\mathcal{C}^\ell\g$ is also minimal-graph-admissible for all $\ell\in\N_{\geq0}$.
\end{lemma}

\begin{proof}
    Let $\g$ be a finite-dimensional graph-admissible Lie algebra. Then, there exists a (possibly overcomplete) basis $\mathcal{B}=\{\Tilde{x}_j\}_{j\in\mathcal{M}}$ that satisfies the relations \eqref{eqn:desired:basis:overcomplete} and spans $\g$. We now introduce the recursively constructed sets
    \begin{align}
        \mathcal{M}_{\mathrm{C}}^{(\ell+1)}:=\delta\left(\left\{(j,k)\in\mathcal{M}\times\mathcal{M}_{\mathrm{C}}^{(\ell)}\mid \;\Tilde{\alpha}_{jk}\neq0\right\}\right)\quad\text{for all }\ell\in\Np{},\quad\text{where}\quad\mathcal{M}_{\mathrm{C}}^{(0)}:=\mathcal{M}.\label{eqn:recusrice:lower:index:set:redundant}
    \end{align}
    It is straightforward to observe that $\mathcal{M}_\mathrm{C}^{(\ell+1)}\subseteq \mathcal{M}_\mathrm{C}^{(\ell)}$ for all $\ell\in\N_{\geq0}$, which can be shown by induction: Clearly $\mathcal{M}_\mathrm{C}^{(1)}\subseteq \mathcal{M}_\mathrm{C}^{(0)}=\mathcal{M}$, since $\delta(\mathcal{M},\mathcal{M})\subseteq\mathcal{M}\cup \{0\}$ and $\Tilde{\alpha}_{jk}\neq0$ if and only if $\delta(j,k)\neq 0$. Suppose now $j\in \mathcal{M}_\mathrm{C}^{(\ell+1)}$. Then, there exists an index $p\in\mathcal{M}$ and $k\in\mathcal{M}_\mathrm{C}^{(\ell)}$ such that $\delta(p,k)=j$. By the induction hypothesis, $\mathcal{M}^{(\ell)}_\mathrm{C}\subseteq\mathcal{M}_\mathrm{C}^{(\ell-1)}$, it follows that $k\in \mathcal{M}_\mathrm{C}^{(\ell-1)}$, and consequently $j\in\mathcal{M}_\mathrm{C}^{(\ell)}$, proving $\mathcal{M}_\mathrm{C}^{(\ell+1)}\subseteq \mathcal{M}_\mathrm{C}^{(\ell)}$ for all $\ell\in\N_{\geq0}$.
    
    Furthermore, $\spn\{\Tilde{x}_j\,\mid\,j\in\mathcal{M}_\mathrm{C}^{(1)}\}=\mathcal{C}^1\g$, which follows analogously to the proof of Lemma~\ref{lem:if:g:graph:admissible:so:derived:algebras}, since $\mathcal{M}_\mathrm{C}^{(1)}=\mathcal{M}_\mathrm{D}^{(1)}$ and $\mathcal{C}^1\g=\mathcal{D}^1\g$. We now proceed with the induction step:
    \begin{itemize}
        \item Suppose $z\in \mathcal{C}^{\ell+1}\g$. Then, there exist elements $x_j\in \g$ and $y_j\in \mathcal{C}^\ell\g$, for $j\in \mathcal{J}\subseteq\N$, a non-empty finite index set, such that $z=\sum_{j\in\mathcal{J}}[x_j,y_j]$. Since $\mathcal{B}$ spans $\g$, and, by induction hypothesis, $\spn\{\Tilde{x}_j\,\mid\,j\in\mathcal{M}_\mathrm{C}^{(\ell)}\}=\mathcal{C}^\ell\g$, there exist coefficients $\mu_{jp}$ and $\nu_{jq}$ with $p\in\mathcal{M}$ and $q\in \mathcal{M}_\mathrm{C}^{(\ell)}$, such that $x_j=\sum_{p\in\mathcal{M}}\mu_{jp}\Tilde{x}_p$ and $y_j=\sum_{q\in\mathcal{M}_\mathrm{C}^{(\ell)}}\nu_{jq}\Tilde{x}_q$. Thus:
        \begin{align*}
            z=\sum_{p\in\mathcal{M}}\sum_{q\in\mathcal{M}_\mathrm{C}^{(\ell)}}\sum_{j\in\mathcal{J}}\mu_{jp}\nu_{jq}[\Tilde{x}_p,\Tilde{x}_q]=\sum_{p\in\mathcal{M}}\sum_{q\in\mathcal{M}_\mathrm{C}^{(\ell)}}\sum_{j\in\mathcal{J}}\mu_{jp}\nu_{jq}\Tilde{\alpha}_{pq} \Tilde{x}_{\delta(p,q)}.
        \end{align*}
        Hence, $z\in\spn\{\Tilde{x}_j\,\mid\,j\in\mathcal{M}_\mathrm{C}^{(\ell+1)}\}$, by the definition of $\mathcal{M}_\mathrm{C}^{(\ell+1)}$ and the convention that $\delta(p,q)=0$ if $\Tilde{\alpha}_{pq}$ vanishes.
        \item Conversely, suppose $z\in \spn\{\Tilde{x}_j\,\mid\, j\in\mathcal{M}_\mathrm{C}^{(\ell+1)}\}$. Then, there exist coefficients $\lambda_j\in\mathbb{F}$ for $j\in \mathcal{M}_\mathrm{C}^{(\ell+1)}$, such that $z=\sum_{j\in\mathcal{M}_\mathrm{C}^{(\ell+1)}}\lambda_j\Tilde{x}_j$. An index $j$ belonging to $\mathcal{M}_\mathrm{C}^{(\ell+1)}$ implies the existence of indices $p_j\in\mathcal{M}$ and  $q_j\in\mathcal{M}_\mathrm{C}^{(\ell)}$ such that $\delta(p_j,q_j)=j\neq0$, or equivalently $[\Tilde{x}_{p_j},\Tilde{x}_{q_j}]=\Tilde{\alpha}_{p_jq_j}\Tilde{x}_j$ with $\Tilde{\alpha}_{p_jq_j}\neq 0$. Thus, for every $j\in\mathcal{M}_\mathrm{C}^{(\ell+1)}$, we can select exactly one pair $(\hat{p}_j,\hat{q}_j)\in\mathcal{M}\times\mathcal{M}_\mathrm{C}^{(\ell)}$ such that $\delta(\hat{p}_j,\hat{q}_j)=j$, and write:
        \begin{align*}
            z=\sum_{j\in\mathcal{M}_\mathrm{C}^{(\ell+1)}}\lambda_j\frac{1}{\Tilde{\alpha}_{\hat{p}_j\hat{q}_j}}[\Tilde{x}_{\hat{p}_j},\Tilde{x}_{\hat{q}_j}]=\sum_{j\in\mathcal{M}_\mathrm{C}^{(\ell+1)}}\left[\Tilde{x}_{\hat{p}_j},\frac{\lambda_j}{\Tilde{\alpha}_{\hat{p}_j\hat{q}_j}}\Tilde{x}_{\hat{q}_j}\right].
        \end{align*}
        Hence, by the induction hypothesis $\mathcal{C}^\ell\g=\spn\{\Tilde{x}_j\,\mid\,j\in\mathcal{M}_\mathrm{C}^{(\ell)}\}$, it follows that $z\in \mathcal{C}^{\ell+1}\g$.
    \end{itemize}
    Thus, one has inductively proven that $\mathcal{C}^\ell\g=\spn\{\Tilde{x}_j\,\mid\,j\in\mathcal{M}_\mathrm{C}^{(\ell)}\}$ for all $\ell\in\N_{\geq0}$. The second claim can be shown analogously, by recognizing that $\{{x}_j\,\mid\,j\in\mathcal{N}_\mathrm{C}^{(\ell)}\}$ forms a linearly independent set, as it is a subset of the basis $\mathcal{B}$.
\end{proof}

\begin{corollary}\label{cor:if:g:graph:admissible:so:central:series:algebras:subgraph:basis}
    Let $\g$ be a finite-dimensional graph-admissible Lie algebra associated with the labeled directed graph $G(V,E)$. Then, for every $\ell\in\N_{\geq0}$, there exists a subgraph $G(\Tilde{V}^{(\ell)},\Tilde{E}^{(\ell)})\subseteq G(V,E)$ that can be associated with the $\ell$-th term of the lower central series $\mathcal{C}^\ell\g$, and $\Tilde{V}^{(\ell)}=\{v_j\,\mid\, j\in\mathcal{M}_\mathrm{C}^{(\ell)}\}$.
\end{corollary}

\begin{proof}
    This follows directly by adapting the argument presented in the proof of Lemma~\ref{lem:if:g:graph:admissible:so:central:series:algebras}.
\end{proof}

Lemma~\ref{lem:if:g:graph:admissible:so:central:series:algebras} implies that one can now associate a graph $\mathcal{C}^{\ell} G(V,E)$ with the $\ell$-th element of the lower central series by applying Algorithm~\ref{alg:creating:graph} to the (possibly overcomplete) basis $\mathcal{C}^\ell V:=\{\Tilde{x}_j\,\mid\,j\in\mathcal{M}_\mathrm{C}^{(\ell)}\}$ of $\mathcal{C}^\ell\g$. However, to facilitate this, we propose an alternative approach via Algorithm~\ref{alg:generating:lower:central:series}. There, we extend the graph $G(V,E)$ to a tripartite structure $G(V,E,E_\mathrm{r}^{(0)})$, where $E_\mathrm{r}^{(0)}$ is an initially empty set of labeled directed edges of a distinct other type\footnote{These edges will be illustrated using dashed lines throughout this work. However, to visualize the semidirect structure of a Lie algebra, we will also represent the corresponding edges with dashed lines. The convention being used in each case will be clear from the context.}. We now denote by $\mathcal{C}^\ell G(V,E,E_\mathrm{r}^{(0)})\equiv G(\mathcal{C}^\ell V,\mathcal{C}^\ell E,E_\mathrm{r}^{(\ell)})$ the subgraphs of $G(V,E,\emptyset)$ that can be employed to construct graphs that can be in turn associated with $\mathcal{C}^\ell\g$ and whose vertices are $\{\Tilde{x}_j\,\mid\,j\in\mathcal{M}_\mathrm{C}^{(\ell)}\}$. The edge set $E_\mathrm{r}^{(\ell)}$ is defined as:
\begin{align*}
    E_\mathrm{r}^{(\ell)}:=\{e\in E\,\mid\,\varpi_\mathrm{s}(e),\varpi_\mathrm{e}(e)\in \mathcal{C}^\ell V\text{ and }\varpi_\mathrm{l}(e)\in V\setminus\mathcal{C}^\ell V\}\qquad\text{for all }\ell\in\N_{\geq0}.
\end{align*}
The actual labeled directed graphs $\mathcal{C}^\ell G(V,E)$ associated with the algebras $\mathcal{C}^\ell\g$ of the lower central series are defined as the subgraphs of the graphs $\mathcal{C}^\ell G(V,E,E_\mathrm{r}^{(0)})$, where the edge set $E_\mathrm{r}^{(\ell)}$ is removed. We then have the following result:

\begin{lemma}\label{lem:lower:central:series:of:graphs}
    Let $\g$ be a finite-dimensional graph-admissible Lie algebra associated with the labeled directed graph $G(V,E)$. Then, for each $\ell\in\N_{\geq0}$, the graphs $\mathcal{C}^\ell G(V,E)=G(\mathcal{C}^\ell V,\mathcal{C}^\ell E)$, obtained via Algorithm~\ref{alg:generating:lower:central:series}, can be associated with the $\ell$-th Lie algebra in the lower central series of $\g$.
\end{lemma}

The core idea behind Algorithm~\ref{alg:generating:lower:central:series} is to construct a graph associated with $\mathcal{C}^\ell\g$ not directly using the set $\{v_j\,\mid\,j\in\mathcal{M}_\mathrm{C}^{(\ell)}\}$, but iteratively pruning the original graph $G(V,E)$. Starting with $G(V,E)$, one classifies the vertices from $V$ into the two sets $\mathcal{C}^1V:=\{v_j\,\mid\,j\in\mathcal{M}_\mathrm{C}^{(1)}\}$, which corresponds to the elements spanning $\mathcal{C}^1\g$, and the set $V\setminus \mathcal{C}^1V$, which do not. The vertices $v\in\mathcal{C}^1V$ are precisely those for which there exists an edge $e\in E$ such that $\varpi_\mathrm{e}(e)=v$. Thus, to obtain the graph associated with $\mathcal{C}^1\g$, one must simply remove all vertices from $V$ that are not endpoints of any edge $e\in E$. The corresponding edge set $\mathcal{C}^1E$ is then acquired by  removing all edges $e\in E$ for which any of $\varpi_\mathrm{s}(e),\varpi_\mathrm{l}(e),\varpi_\mathrm{e}(e)\notin\mathcal{C}^1V$. The resulting graph $G(\mathcal{C}^1V,\mathcal{C}^1 E):=\mathcal{C}^1 G(V,E)$ is clearly associated with $\mathcal{C}^1\g$. 

To construct $\mathcal{C}^2G(V,E)$, we identify $\mathcal{C}^2 V:=\{v_j\,\mid\,j\in\mathcal{M}_\mathrm{C}^{(2)}\}$ as the subset of $\mathcal{C}^1V$ that consists of all vertices $v\in\mathcal{C}^1 V$ for which there exists an edge $e\in E$ such that $\varpi_\mathrm{e}(e)=v$ and either $\varpi_\mathrm{l}(e)\in\mathcal{C}^1 V$ or $\varpi_\mathrm{s}(e)\in\mathcal{C}^1V$. Crucially, this step considers all edges $e\in E$ to find the relevant ones, not just those in $e\in\mathcal{C}^1 E$. To account for this, we introduce the aforementioned set $E_\mathrm{r}^{(1)}$, consisting of edges $e\in E\setminus\mathcal{C}^1 E$ such that $\varpi_\mathrm{e}(e)\in\mathcal{C}^1 V$ and either $\varpi_\mathrm{l}(e)\in\mathcal{C}^1 V$ or $\varpi_\mathrm{s}(e)\in\mathcal{C}^1 V$. By Proposition~\ref{prop:conditions:for:edges}, these edges can be reduced to those originating from vertices $w\in\mathcal{C}^1 V$, since any edge $e\in E_\mathrm{r}^{(1)}$ with $\varpi_\mathrm{s}(e)\in V\setminus\mathcal{C}^1V $ must satisfy $\varpi_\mathrm{l}(e)\in\mathcal{C}^1 V$. Such an edge $e$ belongs to $E$ if and only if the edge $e'=(\varpi_\mathrm{l}(e),\varpi_\mathrm{s}(e),\varpi_\mathrm{e}(e))$ belongs to $E$. Thus removing these edges $e\in E\setminus\mathcal{C}^1E$ with $\varpi_\mathrm{l}(e)\in \mathcal{C}^1V $ and $\varpi_\mathrm{s}(e)\in V\setminus\mathcal{C}^1 V$ from $E_\mathrm{r}^{(1)}$ does not lose any information since $e'=(\varpi_\mathrm{l}(e),\varpi_\mathrm{s}(e),\varpi_\mathrm{e}(e))\in E_\mathrm{r}^{(1)}$ and $\varpi_\mathrm{e}(e')=\varpi_\mathrm{e}(e)$. Therefore, removing edges $e$ from $E_\mathrm{r}^{(1)}$ that satisfy $\varpi_\mathrm{s}(e)\notin \mathcal{C}^1V$ does not result in loss of information and we can simply define $E_\mathrm{r}^{(1)}$, as the set of edges $e\in E\setminus \mathcal{C}^1E$ with $\varpi_\mathrm{s}(e)\in\mathcal{C}^1V$. This procedure can be iteratively applied to construct all graphs $\mathcal{C}^\ell G(V,E)=G(\mathcal{C}^\ell V,\mathcal{C}^\ell E)$ associated with the algebras $\mathcal{C}^\ell \g$ of the lower central series. To visually distinguish edges from $E_\mathrm{r}^{(\ell)}$, we represent them using dashed lines.

\begin{proof}
    Let $\g$ be a finite-dimensional graph-admissible Lie algebra. Then $\g$ admits a (possibly overcomplete) basis $\mathcal{B}=\{\Tilde{x}_j\}_{j\in\mathcal{M}}\equiv V$ that satisfies the relations~\eqref{eqn:desired:basis:overcomplete}. We aim to show that the graph $G(\mathcal{C}^\ell V,\mathcal{C}^\ell E)$, constructed via Algorithm~\ref{alg:generating:lower:central:series}, can be associated with the Lie algebra $\mathcal{C}^\ell \g$. This can, due to Corollary~\ref{cor:if:g:graph:admissible:so:central:series:algebras:subgraph:basis}, be shown by inductively proving that $\mathcal{C}^\ell V=\{\Tilde{x}_j\,\mid\,j\in\mathcal{M}_\mathrm{C}^{(\ell)}\}$ and $\mathcal{C}^\ell E=\{e\in E\,\mid\,\varpi_\mathrm{s}(e),\varpi_\mathrm{l}(e),\varpi_\mathrm{e}(e)\in\mathcal{C}^\ell V\}$. The base case $\ell=0$ is immediate. For the inductive step, consider the case $\ell+1$. By Algorithm~\ref{alg:generating:lower:central:series} and the induction hypothesis, the set $\mathcal{C}^{\ell+1}V$ includes every vertex $v\in\mathcal{C}^\ell V=\{\Tilde{x}_j\,\mid\,j\in\mathcal{M}_\mathrm{C}^{(\ell)}\}$ for which there exists an edge $e$ in $\mathcal{C}^\ell E$ or $E_\mathrm{r}^{(\ell)}$ such that $\varpi_\mathrm{e}(e)=v$ (cf. lines 14-17,29,32). Consider both cases separately:
    \begin{itemize}
        \item If $e\in \mathcal{C}^\ell E$ and $\varpi_\mathrm{e}(e)=v$, then there exist two vertices $v_\mathrm{s},v_\mathrm{l}\in\mathcal{C}^\ell V$ such that $[v_\mathrm{s},v_\mathrm{l}]\propto v$. This implies that $v\in \{\Tilde{x}_j\,\mid\,j\in\mathcal{M}_\mathrm{C}^{(\ell+1)}\}$, since $\mathcal{C}^\ell V=\{\Tilde{x}_j\,\mid\,j\in\mathcal{M}_\mathrm{C}^{(\ell)}\}\subseteq V$.
        \item If $e\in E_\mathrm{r}^{(\ell)}$ and $\varpi_\mathrm{e}(e)=v$, then there exists a vertex $v_\mathrm{s}\in \mathcal{C}^\ell V$ and a vertex $v_\mathrm{l}\in V\setminus\mathcal{C}^\ell V$ such that $[v_\mathrm{s},v_\mathrm{l}]\propto v$. This implies again that $v\in \{\Tilde{x}_j\,\mid\,j\in\mathcal{M}_\mathrm{C}^{(\ell+1)}\}$, since $V\setminus\mathcal{C}^\ell V=\{\Tilde{x}_j\,\mid\, j\in \mathcal{M}\setminus\mathcal{M}_\mathrm{C}^{(\ell)}\}\subseteq V$.
    \end{itemize}
    Thus one has $\mathcal{C}^{\ell+1}V\subseteq\{\Tilde{x}_j\,\mid\,j\in\mathcal{M}_\mathrm{C}^{(\ell+1)}\}$. Conversely, suppose $v\in \{\Tilde{x}_j\,\mid\,j\in\mathcal{M}_\mathrm{C}^{(\ell+1)}\}$. Then, by the induction hypothesis, there exists an element $v_\mathrm{s}\in \{\Tilde{x}_j\,\mid\,j\in\mathcal{M}_\mathrm{C}^{(\ell)}\}=\mathcal{C}^\ell V$ and an element  $v_\mathrm{l}\in V$ such that $[v_\mathrm{s},v_\mathrm{l}]\propto v$. Therefore, there exists an edge $e\in E$ such that $\varpi_\mathrm{s}(e),\varpi_\mathrm{e}(e)\in \mathcal{C}^\ell V$ and $\varpi_\mathrm{l}(e)\in V$, implying $e\in\mathcal{C}^\ell E\cup E_\mathrm{r}^{(\ell)}$, and hence $v\in\mathcal{C}^{\ell+1} V$, since $\varpi_\mathrm{e}(e)=v$. This shows $\{\Tilde{x}_j\,\mid\,j\in\mathcal{M}_\mathrm{C}^{(\ell+1)}\}\subseteq \mathcal{C}^{\ell+1}V$ and consequently $\mathcal{C}^{\ell+1}V=\{\Tilde{x}_j\,\mid\,j\in\mathcal{M}_\mathrm{C}^{(\ell+1)}\}$. Finally, by Algorithm~\ref{alg:generating:lower:central:series} (cf. lines 24-26,30,33), one has $\mathcal{C}^{\ell+1}E=\{e\in E\,\mid\,\varpi_\mathrm{s}(e),\varpi_\mathrm{l}(e),\varpi_\mathrm{e}(e)\in\mathcal{C}^{\ell+1} V\}$, which completes the proof.
\end{proof}

\begin{theorem}\label{thm:nilpotebt:if:graph:series.Terminates}
    Let $\g$ be a finite-dimensional graph-admissible Lie algebra, and let $G(V,E)$ be any labeled directed graph associated with $\g$. Then, $\g$ is nilpotent if and only if the series of graphs $\mathcal{C}^\ell G(V,E)$ constructed by Algorithm~\ref{alg:generating:lower:central:series} terminates after a finite number of steps.
\end{theorem}

Here, similar to Theorem~\ref{thm:solvable:termination:derived:series:graph}, we need to clarify what is meant by termination in this context. The sequence of graphs $\mathcal{C}^\ell G(V,E)$ is said to terminate, if for some $\ell_*\in\N_{\geq0}$, the corresponding graph $\mathcal{C}^{\ell_*}G(V,E)$ is trivial, i.e., both its vertex and edge sets are empty: $\mathcal{C}^{\ell_*}V=\emptyset$ and $\mathcal{C}^{\ell_*}E=\emptyset$. We also recall that we adopted the convention that the trivial graph induces the trivial Lie algebra: $\lie{\emptyset}=\spn\{0\}$. This aligns with the standard interpretation of the termination of the lower central series, where $\mathcal{C}^{\ell_*}\g=\{0\}$ for some $\ell_*\in\N_{\geq0}$.

\begin{proof}
    This result is a direct consequence of the observation that the vertex set $\mathcal{C}^\ell V$ spans the Lie algebra $\mathcal{C}^\ell\g$, i.e., $\spn\{\mathcal{C}^\ell V\}=\mathcal{C}^\ell \g$. Therefore,  $\mathcal{C}^\ell V=\emptyset$ if and only if $\mathcal{C}^\ell\g=\{0\}$. Note that $\mathcal{C}^\ell V=\emptyset$ implies trivially that $\mathcal{C}^\ell E=\emptyset$. This establishes the equivalence between the termination of the lower central  series of the Lie algebra $\g$ and the corresponding graph sequence generated by Algorithm~\ref{alg:generating:lower:central:series}. The claim follows then directly from Definition~\ref{def:solvable:nilpotent}.
\end{proof}

We can demonstrate Algorithm~\ref{alg:generating:lower:central:series} for a simple example:

\begin{tcolorbox}[breakable, colback=Cerulean!3!white,colframe=Cerulean!85!black,title=\textbf{Example}: Application of Algorithm~\ref{alg:generating:lower:central:series}]
    \begin{example}\label{exa:example:lower:central:series}
    Consider the real seven-dimensional Lie algebra $\g$, spanned by the basis $\{v_j\}_{j=1}^7$, and defined by the non-trivial Lie bracket relations $[v_1,v_j]=v_{j+1}$ for all $j\in\{2,\ldots,6\}$. The lower central series of the Lie algebra is given by:
    \begin{align*}
        \mathcal{C}^0\g&=\spn\{\{v_j\}_{j=1}^7\},\;&\;\mathcal{C}^1\g&=\spn\{\{v_j\}_{j=3}^7\},\;&\;\mathcal{C}^2\g&=\spn\{\{v_j\}_{j=4}^7\},\\
        \mathcal{C}^3\g&=\spn\{\{v_j\}_{j=5}^7\},\;&\;
        \mathcal{C}^4\g&=\spn\{\{v_j\}_{j=6}^7\},\;&\;\mathcal{C}^5\g&=\spn\{\{v_j\}_{j=7}^7\},
    \end{align*}
    while $\mathcal{C}^6\g=\{0\}=\mathcal{C}^\ell \g$ for all $\ell\geq 7$.
    This Lie algebra, similar to the Lie algebra considered in Example~\ref{exa:demonstration:algorithm:derived:series:minimal:graph}, is part of a family of nilpotent Lie algebras that can be faithfully realized within the skew-hermitian Weyl algebra $\hat{A}_1$ \cite{A1:project}. In fact, this algebra can be seen as the first derived algebra of the solvable eight-dimensional Lie algebra from the same family of algebras as the seven-dimensional Lie algebra from Example~\ref{exa:demonstration:algorithm:derived:series:minimal:graph}.
    
    Algorithm~\ref{alg:generating:lower:central:series} can be utilized to compute the sequence of labeled directed graphs $\mathcal{C}^\ell G(V,E,E_\mathrm{r})$ that are associated with each term $\mathcal{C}^\ell \g$ in the lower central series. The graphs are depicted in Figure~\ref{fig:example:nilpotent:minimal}.
        \begin{figure}[H]
            \centering
            \includegraphics[width=0.95\linewidth]{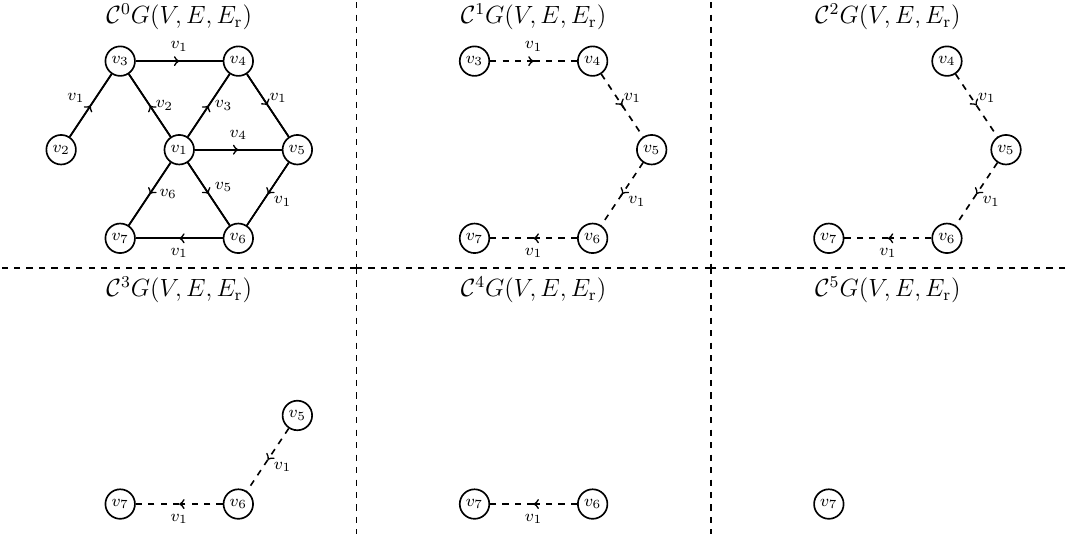}
            \caption{Application of Algorithm~\ref{alg:generating:lower:central:series} to generate the graphs $\mathcal{C}^\ell G(V,E)$ associated with the Lie algebras $\mathcal{C}^\ell\g$ of the lower central series from Example~\ref{exa:example:lower:central:series}. Each graph corresponds to a step in the lower central series, from $\mathcal{C}^0\g$ down to the last non-trivial Lie algebra $\mathcal{C}^5\g\neq\{0\}.$ Note that the dashed edges in the figure represent elements of the supplementary edge sets $E_\mathrm{r}^{(\ell)}=\mathcal{C}^\ell E_\mathrm{r}$, which are not part of the graph $\mathcal{C}^\ell G(V,E)$ itself but are used in the construction process to identify relevant vertices and edges for the next iteration, and highlight structural aspects of the Lie algebra.}
            \label{fig:example:nilpotent:minimal}
        \end{figure}
    \end{example}
\end{tcolorbox}

\begin{algorithm}[htpb]
        \DontPrintSemicolon
        \KwData{A (possibly overcomplete) basis $\mathcal{B}=\{\Tilde{x}_j\}_{j\in\mathcal{M}}$ of the finite-dimensional graph-admissible Lie algebra $\g$, satisfying the Lie bracket relations \eqref{eqn:desired:basis:overcomplete}.}
        \KwResult{A sequence of labeled directed graphs $\mathcal{C}^\ell G(V,E)$ associated with Lie algebras $\mathcal{C}^\ell\g$ of the lower central series for the Lie algebra $\g$ for all $\ell\geq 0$, constructed with respect to the given basis $\mathcal{B}$.}
        \SetKwData{Left}{left}\SetKwData{This}{this}\SetKwData{Up}{up}
        \SetKwFunction{Union}{Union}\SetKwFunction{FindCompress}{FindCompress}
        \SetKwInOut{Input}{input}\SetKwInOut{Output}{output}
    
        \BlankLine
        $\mathcal{C}$ $\leftarrow$ $\emptyset$
        \tcc*[h]{Initialize the set that contains all graphs $\mathcal{C}^\ell G(V,E,E_{\mathrm{r}})$ associated with each Lie algebra $\mathcal{C}^\ell \g$ of the lower central series}\;
        \BlankLine
        $V$ $\leftarrow$ $\mathcal{B}$\tcc*[h]{Initialize the vertex set for $G(V,E)$ with the basis $\mathcal{B}$}\;
        $E$ $\leftarrow$ $\emptyset$\tcc*[h]{Initialize the edge set for $G(V,E)$}\;
        $E_{\mathrm{r}}$ $\leftarrow$ $\emptyset$\;
        \BlankLine
        \ForEach(\tcc*[h]{add all relevant edges to edge set $E$}){$(j,k)\in \mathcal{M}\times\mathcal{M}$ with $j<k$}{
            \If{$[\Tilde{x}_j,\Tilde{x}_k]=\Tilde{\alpha}_{jk} \Tilde{x}_{\delta(j,k)}$ with $\Tilde{\alpha}_{jk}\neq 0$}{
                $E$ $\leftarrow$ $E\cup \{(\Tilde{x}_j,\Tilde{x}_k,\Tilde{x}_{\delta(j,k)})\}\cup \{(\Tilde{x}_k,\Tilde{x}_j,\Tilde{x}_{\delta(j,k)})\}$
                \tcc*[h]{Add edge to edge set $E$; each edge is an ordered triple: (start vertex, edge-label, end vertex)}\;
            }
        }
        $\mathcal{C}$ $\leftarrow$ $\{G(V,E,\emptyset)\}$\tcc*[h]{Add initial graph associated with $\g=\mathcal{C}^0\g$ to $\mathcal{C}$}\;
        \BlankLine
        \ForEach{$\ell\in\N_{\geq1}$}{
            $V_\mathrm{old}$ $\leftarrow$ $V$\;
            $E_\mathrm{old}$ $\leftarrow$ $E$\;
            $E_{\mathrm{r,old}}$ $\leftarrow$ $E_\mathrm{r}$\;
            \BlankLine
            $R_V$ $\leftarrow$ $\emptyset$\tcc*[h]{Initialize the set of vertices marked for removal}\;
            \ForEach{vertex $v\in V_{\mathrm{old}}$}{
                \If{there exists no edge $e\in E_{\mathrm{old}}\cup E_{\mathrm{r,old}}$ such that $\varpi_\mathrm{e}(e)=v$}{
                    $R_V$ $\leftarrow$ $R_V\cup\{v\}$\tcc*[h]{Mark $v$ for removal from $V$}\;
                }
            }
            $V_\mathrm{new}$ $\leftarrow$ $V_\mathrm{old}\setminus R_V$\;
            \BlankLine
            $R_E$ $\leftarrow$ $\emptyset$\tcc*[h]{Initialize the set of edges marked for removal from $E$}\;
            $R_{E_\mathrm{r}}$ $\leftarrow$ $\emptyset$\tcc*[h]{Initialize the set of edges marked for removal from $E_\mathrm{r}$}\;
            $A_{E_\mathrm{r}}$ $\leftarrow$ $\emptyset$\tcc*[h]{Initialize the set of edges marked for transfer from $E$ to $E_\mathrm{r}$}\;
            \BlankLine
            \ForEach{edge $e\in E_{\mathrm{r,old}}$}{
                \If{there exist no vertices $v_\mathrm{s},v_\mathrm{e}\in V_{\mathrm{new}}$ such that $\varpi_\mathrm{s}(e)=v_{\mathrm{s}}$ and $\varpi_\mathrm{e}(e)=v_{\mathrm{e}}$}{
                    $R_{E_\mathrm{r}}$ $\leftarrow$ $R_{E_\mathrm{r}}\cup\{e\}$
                    \tcc*[h]{Mark $e$ for removal from $E_\mathrm{r}$}\;
                }
            }
            \ForEach{edge $e\in E_{\mathrm{old}}$}{
                \uIf{there exist no vertices $ v_\mathrm{s},v_\mathrm{l},v_\mathrm{e}\in V_{\mathrm{new}}$ such that $e=(v_\mathrm{s},v_\mathrm{l},v_\mathrm{e})$}{
                    $R_E$ $\leftarrow$ $R_E\cup\{e\}$\tcc*[h]{Mark $e$ for removal from $E$}\; 
                    \uIf{there exist vertices $ v_\mathrm{s},v_\mathrm{e}\in V_\mathrm{new}$ such that $\varpi_\mathrm{s}(e)=v_\mathrm{s}$ and $\varpi_\mathrm{e}(e)=v_\mathrm{e}$}{
                        $A_{E_\mathrm{r}}$ $\leftarrow$ $A_{E_\mathrm{r}}\cup\{e\}$\tcc*[h]{Mark $e$ for transfer from $E$ to $E_\mathrm{r}$}\;
                    }
                }
            }
            $V$ $\leftarrow$ $V_\mathrm{new}$\;
            $E$ $\leftarrow$ $E_\mathrm{old}\setminus R_E$\;
            $E_\mathrm{r}$ $\leftarrow$
            $(E_{\mathrm{r,old}}\setminus R_{E_\mathrm{r}})\cup A_{E_\mathrm{r}}$\;
            \BlankLine
            $V^{(\ell)}\equiv\mathcal{C}^\ell V\leftarrow V$\;
            $E^{(\ell)}\equiv \mathcal{C}^\ell E\leftarrow E$\;
            $E_\mathrm{r}^{(\ell)}$ $\leftarrow$ $ E_\mathrm{r}$\;
            \BlankLine
            $\mathcal{C}$ $\leftarrow$ $\mathcal{C}\cup\{\mathcal{C}^\ell G(V,E,E_{\mathrm{r}})=G(\mathcal{C}^{\ell}V,\mathcal{C}^{\ell}E, E_{\mathrm{r}}^{(\ell)})\}$\tcc*[h]{Add updated graph $\mathcal{C}^\ell G(V,E)$}\;
        }
        \Return $\mathcal{C}$ \tcc*[h]{Return the set of all graphs $\mathcal{C}^\ell G(V,E)$ associated with the Lie algebras $\mathcal{C}^\ell\g$ of the lower central series for all $\ell\geq 0$}\;
\caption{Algorithm for generating a labeled directed graph associated with the $\ell$-th Lie algebra of the lower central series $\mathcal{C}^\ell\g$ for a finite-dimensional graph-admissible Lie algebra $\g$}\label{alg:generating:lower:central:series}
\end{algorithm}

\begin{lemma}\label{lem:relation:walks:lower:central:series:of:vertices}
    Let $\g$ be a finite-dimensional graph-admissible Lie algebra and let $G(V,E)$ be any labeled directed graph associated with $\g$. Then $\mathcal{C}^\ell V=\mathcal{P}_\ell$, where the sets $\mathcal{P}_\ell$ are recursively defined by
    \begin{align*}
        \mathcal{P}_{\ell+1}:=\left\{\varpi_\mathrm{e}(e)\,\middle\mid\,e\in E,\,\varpi_\mathrm{s}(e)\in \mathcal{P}_\ell\right\},\qquad\text{with}\qquad\mathcal{P}_0:=V.
    \end{align*}
    For each $\ell\in\N_{\geq1}$, the set $\mathcal{P}_\ell$ denotes the set of ending vertices of directed walks in $G(V,E)$ of length $\ell$.
\end{lemma}

\begin{proof}
    This can be shown by induction. The base case $\ell=0$ is clear, because $\mathcal{C}^0V=V=\mathcal{P}_0$ by definition. Thus, we may assume that $\mathcal{C}^\ell V=\mathcal{P}_\ell$ for some $\ell\in \Np{}$ and aim to show that $\mathcal{C}^{\ell+1}V=\{v_j\,\mid\, j\in \mathcal{M}_\mathrm{C}^{(\ell+1)}\}=\mathcal{P}_{\ell+1}$, where $V=\{v_j\}_{j\in \mathcal{M}}$. Let $j\in \mathcal{M}_\mathrm{C}^{(\ell+1)}$. Then there exists an index $p\in \mathcal{M}$ and an index $q\in \mathcal{M}_\mathrm{C}^{(\ell)}$ such that $\delta(p,q)=j$ and $\tilde{\alpha}_{pq}\neq 0$. This implies that there exist some $v_p\in V$ and $v_q\in \mathcal{C}^\ell V$ such that $[v_p,v_q]=-[v_q,v_p]=\kappa v_j$ with $\kappa\neq 0$. By construction, it follows that there exists an edge $e=(v_q,v_p,v_j)\in E$ such that $\varpi_\mathrm{e}(e)=v_j$ and $\varpi_\mathrm{s}(e)=v_q\in \mathcal{C}^\ell V=\mathcal{P}_\ell$. Hence, $v_j\in \mathcal{P}_{\ell+1}$, and therefore $\mathcal{C}^{\ell+1} V\subseteq \mathcal{P}_{\ell+1}$. The reverse inclusion follows analogously, since every step in the argument above can be read in both directions.

    To see that, for $\ell\in \N_{\geq1}$, the set $\mathcal{P}_\ell$ is precisely the set  of ending vertices of directed walks in $G(V,E)$ of length $\ell$, we again argue by induction. Notice first that $\mathcal{P}_1$ is the set of vertices that are the ending points of at least one edge in $E$ and therefore precisely the ending vertices of walks of length one. Assume now that the assertion holds for some $\ell\in\N_{\geq1}$. If $v\in \mathcal{P}_{\ell+1}$, then there exists an edge $e\in E$ that starts at some vertex $w\in \mathcal{P}_\ell$ and ends at $v$. By the induction hypothesis, there exists a directed walk of length $\ell$ that ends at $w$. Extending this walk by a traversal of $e$ thus yields a walk of length $\ell+1$ that ends at $v$. Conversely, let a walk of length $\ell+1$ end at $v$, and let $e$ be its last edge. Removing the final occurrence of $v$ and this edge results in a walk of length $\ell$ ending at $w:=\varpi_\mathrm{s}(e)$. The induction hypothesis now implies that $w\in\mathcal{P}_\ell$ and therefore $v=\varpi_\mathrm{e}(e)\in \mathcal{P}_{\ell+1}$, which completes the proof by induction.
\end{proof}

\begin{theorem}\label{thm:nilpotency:criteria:strong}
    Let $\g$ be a finite-dimensional graph-admissible Lie algebra and let $G(V,E)$ be any labeled directed graph associated with $\g$. Then, $\g$ is nilpotent if and only if $G(V,E)$ contains no non-trivial closed directed walk.
\end{theorem}

\begin{remark*}
    Note that, similarly to Theorem~\ref{thm:non:solvability:condition:strong}, nilpotency and non-nilpotency are intrinsic properties of the Lie algebra under consideration. Thus, the result above holds for every proper labeled directed graph that can be associated with a graph-admissible Lie algebra $\g$.
\end{remark*}

\begin{proof}
    Recall the definition of the sets $\mathcal{P}_\ell$ in Lemma~\ref{lem:relation:walks:lower:central:series:of:vertices} and suppose $G(V,E)$ contains a non-trivial closed directed walk $W=(v_0,e_0,v_1,\ldots,v_{s-1},e_{s-1},v_0)$ of length $s$. Then $v_j\in \mathcal{P}_{j+ks}$ for all $j\in\{0,1,\ldots,s-1\}$ and $k\in \Np{}$ because the walk $W$ can be repeated arbitrarily often to construct directed walks of length $j+ks$ that end at the vertices $v_j$. Thus, $\mathcal{P}_\ell\neq \emptyset$ for all $\ell\in \Np{}$. Since $\mathcal{P}_\ell= \mathcal{C}^\ell V$ by Lemma~\ref{lem:relation:walks:lower:central:series:of:vertices} and, by construction, $\mathcal{C}^{\ell+1}V\subseteq \mathcal{C}^{\ell} V$, for all $\ell\in\Np{}$, the Lie algebra $\g$ is, according to Theorem~\ref{thm:nilpotebt:if:graph:series.Terminates}, not nilpotent.
    
    Conversely, suppose $G(V,E)$ contains no non-trivial closed directed walk. Since $V$ is finite, every directed walk in $G(V,E)$ has length at most $m-1=|V|-1$. Indeed, if there existed a walk $W$ of length $s\geq m$, then it would have to visit a vertex $v\in V$ twice. However, the directed segment of $W$ starting at the first occurrence of $v$ and ending at the second occurrence of $v$ is a non-trivial closed directed walk, contradicting the initial premise. This now implies that $\mathcal{P}_m=\emptyset$ and therefore, by Lemma~\ref{lem:relation:walks:lower:central:series:of:vertices}, $\mathcal{C}^mV=\emptyset$. Hence, it follows from Theorem~\ref{thm:nilpotebt:if:graph:series.Terminates} that $\g$ is nilpotent. The contraposition of this implication is thus also true, showing that if $\g$ is not nilpotent then $G(V,E)$ contains a non-trivial closed directed walk.
\end{proof}

\begin{corollary}\label{cor:index:of:nilpotent:lie:algebra}
    Let $\g$ be a finite-dimensional, non-abelian, nilpotent Lie algebra associated with the labeled directed graph $G(V,E)$, and let $W$ be a directed walk within $G(V,E)$. Then, $\operatorname{len}(W)\leq \ell_*$, where $\ell_*+1$ is the index of $\g$. This inequality is sharp, i.e., there exists at least one directed walk $W'$ within $G(V,E)$ such that $\operatorname{len}(W')=\ell_*$. Furthermore, every non-trivial directed walk within $G(V,E)$ is a directed trail and a directed path.
\end{corollary}

\begin{proof}
    Let $\ell_*+1$ be the index of $\g$. Then $\mathcal{C}^{k}\g=\{0\}$ for all $k>\ell_*$, while $\mathcal{C}^{\ell_*}\g\neq\{0\}$. Since $\g$ is assumed to be non-abelian, $\ell_*\geq 1$. By Lemmata~\ref{lem:lower:central:series:of:graphs} and~\ref{lem:relation:walks:lower:central:series:of:vertices} it holds that $\mathcal{P}_k=\mathcal{C}^kV=\emptyset$ for all $k>\ell_*$, while $\mathcal{P}_{\ell_*}=\mathcal{C}^{\ell_*}V\neq\emptyset$. Since, by construction, $\mathcal{C}^{\ell+1}V\subseteq\mathcal{C}^\ell V$ for all $\ell\in \Np{}$, it also holds that $\mathcal{P}_\ell=\mathcal{C}^\ell V\neq \emptyset$ for all $\ell\in\{0,1,\ldots,\ell_*\}$. According to Lemma~\ref{lem:relation:walks:lower:central:series:of:vertices}, for each $\ell\in\{0,1,\ldots,\ell_*\}$, there exists therefore a directed walk in $G(V,E)$ of length $s=\ell$ but no directed walk of length $s>\ell_*$, which proves the first part of the claim.
    
    Finally, we want to show that every non-trivial directed walk within $G(V,E)$ is a directed trail and a directed path. For this, we recall that Theorem~\ref{thm:nilpotency:criteria:strong} prohibits non-trivial closed directed walks within $G(V,E)$. If a directed walk visited a vertex more than once, the directed segment of the walk connecting the two occurrences of that vertex would be a non-trivial closed directed walk, which is prohibited. Thus, any directed walk within $G(V,E)$ cannot visit a vertex more than once and therefore also not traverse the same edge more than once. Thus, by Definition~\ref{def:walk:trail:path:cycle}, every non-trivial directed walk within $G(V,E)$ is a directed trail and a directed path.
\end{proof}

\begin{lemma}
    Let $\g$ be a finite-dimensional graph-admissible Lie algebra associated with the labeled directed graph $G(V,E)$. If $G(V,E)$ contains a closed directed walk $C=(v_1,e_1,v_2,e_2,\ldots,e_{s},v_1)$, then $\dim(\operatorname{span}\{v_j,\varpi_\mathrm{l}(e_j)\}_{j=1}^{s})\geq 2$.
\end{lemma}

\begin{proof}
    This is an immediate consequence of the anti-symmetric property of the Lie bracket.
\end{proof}

\subsection{Identifying ideals}
Another important structural feature of Lie algebras is the presence of ideals, which play a central role in classification, decomposition, and representation theory \cite{Knapp:1996,Kuzmin:1977,Humphreys:1972}. The graph-theoretic approach introduced here also provides insight into the identification and characterization of ideals. 
Since every ideal is, by definition, a subalgebra, we decide to begin with a lemma that is concerned with graphs that can be associated with subalgebras.
\begin{lemma}\label{lem:subalgebra:graph:weak}
    Let $\g$ be a finite-dimensional graph-admissible Lie algebra associated with the labeled directed graph $G(V,E)$, and let $\Tilde{V}$ be a non-empty subset of $V$. Then, the subgraph $G(\Tilde{V},\Tilde{E})$ can be associated with the Lie algebra $\mathfrak{v}:=\vspan{\Tilde{V}}$ if and only if $\Tilde{E}=\{e\in E\,\mid\,\varpi_\mathrm{s}(e),\varpi_\mathrm{l}(e),\varpi_{\mathrm{e}}(e)\in \Tilde{V}\}=\{e\in E\,\mid\,\varpi_\mathrm{s}(e),\varpi_\mathrm{l}(e)\in \Tilde{V}\}$.
\end{lemma}

\begin{proof}
We omit the proof, as this claim follows straightforwardly from the definition of a Lie subalgebra and the construction rules for graphs associated with graph-admissible Lie algebras. The interested reader may verify this by checking that the closure under the Lie bracket corresponds precisely to the inclusion of all edges whose source, label and endpoint lie within $\Tilde{V}$.
\end{proof}

A first result connecting ideals of a Lie algebra with the associated graph is presented in the following claim:

\begin{lemma}\label{lem:ideal:span}
    Let $\g$ be a finite-dimensional graph-admissible Lie algebra associated with the labeled directed graph $G(V,E)$. If  $W\subseteq V$ is a subset such that no edge $e\in E$ points from a vertex $w\in W$ to a vertex $v\in V\setminus W$, then $W$ spans an ideal of $\g$.
\end{lemma}

\begin{proof}
    Let $G(V,E)$ be a graph associated with a finite-dimensional Lie algebra $\g$, and let $W\subseteq V$. Suppose that every edge $e\in E$ with $\varpi_\mathrm{s}(e)\in W$ also satisfies $\varpi_\mathrm{e}(e)\in W$. This implies that for every $w\in W$ and $v\in V$, the Lie bracket $[w,v]=\kappa_{wv} \Tilde{w}$ holds for some scalar $\kappa_{wv}\in\mathbb{F}$ and $\Tilde{w}\in W$. Therefore, for arbitrary linear combinations
    \begin{align*}
        \left[\sum_{w\in W}\lambda_ww,\sum_{v\in V}\lambda_vv\right]=\sum_{w\in W}\sum_{v\in V}\lambda_w\lambda_v\kappa_{wv}\Tilde{w}\in \spn\{W\}\qquad\text{ for all }\qquad\lambda_v,\lambda_w\in\mathbb{F}.
    \end{align*}
    This shows that $[\spn\{W\},\spn\{W\}]\subseteq [\spn\{V\},\spn\{W\}]\subseteq\spn\{W\}$ and hence $\spn\{W\}=\lie{W}$ as well as $[\g,\lie{W}]\subseteq\lie{W}$, demonstrating that $\lie{W}$ is indeed an ideal of $\g$.
\end{proof}

This lemma inspires us to formally define the following useful graphical concept in order to facilitate obtaining subsequent results:

\begin{definition}\label{def:ideal:graph:property}
    Let $G(V,E)$ be a labeled directed graph associated with a finite-dimensional graph-admissible Lie algebra $\g$. A subset $W\subseteq V$ is said to satisfy the \emph{ideal-graph-property} if and only if no edge $e\in E$ points from a vertex $w\in W$ to a vertex $v\in V\setminus W$.
\end{definition}

\begin{lemma}\label{lem:ideal:basis:subset:minimal}
    Let $\g$ be a finite-dimensional minimal-graph-admissible Lie algebra associated with the minimal graph $G(V,E)$. If $\mathfrak{i}\subseteq\g$ is an ideal and $\mathfrak{i}=\spn\{W\}$ for some $W\subseteq V$, then $W$ satisfies the ideal-graph-property.
\end{lemma}

\begin{proof}
    Let $W=\{v_k\}_{k\in\mathcal{K}}\subseteq V=\{v_j\}_{j\in\mathcal{N}}$ be a basis of an ideal $\mathfrak{i}\subseteq \g$, where $G(V,E)$ is a minimal graph associated with the finite-dimensional Lie algebra $\g$, and $\mathcal{K}\subseteq \mathcal{N}$. Then, $[v_j,v_k]=\alpha_{jk} v_{\delta(j,k)}$ for all $j\in\mathcal{N}$ and $k\in\mathcal{K}$. Since $\mathfrak{i}$ is an ideal of $\g$, it must be closed under the Lie bracket with all elements of $\g$, which implies $\delta(\mathcal{N},\mathcal{K})\subseteq \mathcal{K}\cup\{0\}$. Thus, every edge $e\in E$ with $\varpi_\mathrm{s}(e)\in W$ also satisfies $\varpi_\mathrm{e}(e)\in W$. Therefore, $W$ satisfies the ideal-graph-property as defined in Definition~\ref{def:ideal:graph:property}.
\end{proof}

\begin{lemma}\label{lem:no:ideal:subset:walk:connecting:all:vertices}
    Let $\g$ be a finite-dimensional graph-admissible Lie algebra associated with the labeled directed graph $G(V,E)$. If $V$ contains no proper non-empty subset $V'\subsetneq V$ that satisfies the ideal-graph-property, and $\dim(\g)\geq 2$, then, for every pair of vertices $v,v'\in V$, there exists a directed walk $W$ that starts at $v$ and ends at $v'$.
\end{lemma}

\begin{proof}
    Let $v\in V$ be an arbitrary but fixed vertex, and define the set
    \begin{align*}
        R_v:=\left\{w\in V\,\middle\mid\,\exists\,\text{directed walk in $G(V,E)$ starting at $v$ and ending at $w$}\right\}\cup\{v\}.
    \end{align*}
    Let $e\in E$ be such that $w:=\varpi_\mathrm{s}(e)\in R_v$. By definition there exists a directed walk $W$ in $G(V,E)$ that starts at $v$ and ends at $w$. Appending $e$ to this walk shows that $\varpi_\mathrm{e}(e)\in R_v$. Thus $R_v$ satisfies the ideal-graph-property. Since $v\in R_v$, it must hold that $R_v=V$, as we assume that $V$ contains no proper non-empty subset that satisfies the ideal-graph-property. But this also implies that each $u\in V$ can be reached from $v$ by following a directed walk $W$ in $G(V,E)$. The claim thus follows immediately since $v$ was arbitrary.
\end{proof}

\begin{corollary}\label{cor:no:proper:igp:closed:walk:induces:V}
    Let $\g$ be a finite-dimensional graph-admissible Lie algebra associated with the labeled directed graph $G(V,E)$, and suppose $\dim(\g)\geq 2$. Then, the vertex set $V$ contains no proper non-empty subset $W\subsetneq V$ that satisfies the ideal-graph-property if and only if there exists a closed directed walk $C$ that induces the subgraph $G_C\equiv G(\Tilde{V},\Tilde{E})$, where $\Tilde{V}=V$.
\end{corollary}

\begin{proof}
    Suppose there exists no proper non-empty subset $W\subsetneq V=\{v_j\}_{j\in\mathcal{N}}$ that satisfies the ideal-graph-property. Then, by Lemma~\ref{lem:no:ideal:subset:walk:connecting:all:vertices}, for every pair $(j,k)\in\mathcal{N}\times\mathcal{N}$, there exists a directed walk $W_{jk}=(v_j,e_1^{(jk)},v_1^{(jk)},\ldots,e_{s_{jk}}^{(jk)},v_k)$ that starts at $v_j$ and ends at $v_k$. Using these walks, we can construct the closed directed walk
    \begin{align*}
        C=(v_1,e_1^{(12)},v_1^{(12)},\ldots,e_{s_{12}}^{(12)},v_2,e_{1}^{(23)},\ldots,e_{s_{n-1,n}}^{(n-1,n)},v_n,e_1^{(n,1)},v_1^{(n,1)},\ldots,e_{s_{n,1}}^{(n,1)},v_1),
    \end{align*}
    which induces the graph $G_C\equiv G(\Tilde{V},\Tilde{E})$, where $\Tilde{V}=V$.

    Now suppose there exists a closed directed walk $C=(v_{j_1}^{(1)},e_1,v_{j_2}^{(2)},\ldots,e_s,v_{j_1}^{(1)})$ that induces the subgraph $G_C\equiv G(\Tilde{V},\Tilde{E})$ with $\Tilde{V}=V$. Let $W\subseteq V$ be a non-empty subset that satisfies the ideal-graph-property. Then there exists at least one vertex $v_1\in W\cap\Tilde{V}=W$. Without loss of generality, let us focus on $v_1=v_{j_1}^{(1)}$. Since $C$ is a closed directed walk that visits every vertex $w\in V$ at least once, and $\dim(\g)\geq 2$, there must exist an edge $e\in \Tilde{E}\subseteq E$, such that $\varpi_\mathrm{s}(e)=v_1$ and $\varpi_\mathrm{e}(e)\neq v_1$. Iterating this process by replacing $v$ with $\varpi_\mathrm{e}(e)$ yields eventually the set $S=\{v_1,\varpi_\mathrm{e}(e),\ldots\}$. If $S\neq V$, then there must exist an edge $e'\in\Tilde{E}\subseteq E$, such that $\varpi_\mathrm{s}(e')\in S$ and $\varpi_\mathrm{e}(e')\notin S$. This allows us to add the vertex $\varpi_\mathrm{e}(e')$ to $S$ and repeat the process. Importantly, since $W$ satisfies the ideal-graph-property, one always has $S\subseteq W$. Due to the finiteness of $V$, this process must eventually terminate with $S=W=V$.     
    This shows that no proper non-empty subsets of $V$ satisfy the ideal-graph-property, completing this proof.
\end{proof}

Instead of considering only closed directed walks that visit every vertex at least once, one can immediately conclude the following:

\begin{corollary}
    Let $\g$ be a finite-dimensional graph-admissible Lie algebra with $\dim(\g)\geq 2$, associated with the labeled directed graph $G(V,E)$. Then, the vertex set $V$ contains no proper non-empty subset $W\subsetneq V$ that satisfies the ideal-graph-property if and only if there exists a closed directed walk $C$ that induces the graph $G_C\equiv G(V,E)$ itself.
\end{corollary}

\begin{proof}
    This result follows directly from the observation that the existence of a closed directed walk $C$ that induces a subgraph $G_C\equiv G(\Tilde{V},\Tilde{E})$ with $\Tilde{V}=V$ is equivalent to the existence of a closed directed walk $C'$ that induces the entire graph $G_{C'}\equiv G(V,E)$ itself. The necessary condition of this claim is trivially satisfied: If a closed directed walk $C$ induces the full graph $G(V,E)$, then clearly $V=\Tilde{V}$.
    
    Now, let us continue establishing the sufficient condition. Suppose there exists a closed directed walk $C=(v_1,e_1,v_2,\ldots,e_s,v_1)$ that induces a subgraph $G_C\equiv G(\Tilde{V},\Tilde{E})$ with $V=\Tilde{V}$. If $E\neq \Tilde{E}$, then there exists some $e\in E \setminus \Tilde{E}$. Since $V=\Tilde{V}$, one has $\varpi_\mathrm{s}(e),\varpi_\mathrm{e}(e)\in V$, and both vertices are already visited by the walk $C$. We can now construct the closed directed walk $C'$ as follows: First, one repeats the walk $C$, then follows $C$ again until reaching the vertex $\varpi_\mathrm{s}(e)$. There, one traverses the edge $e$, such that one arrives at the vertex $\varpi_\mathrm{e}(e)$. Then, one continues along the walk $C$, starting from $\varpi_\mathrm{e}(e)$, until returning to the starting vertex $v_1$. This extended walk then induces the subgraph $G_{C'}\equiv G(\Tilde{V}',\Tilde{E}')$, where $V=\Tilde{V}'$ and $\Tilde{E}'=\Tilde{E}\cup\{e\}$. Repeating this procedure for every edge $e'\in E\setminus \Tilde{E}$, will eventually incorporate all edges $e\in E$ into the extended walk. The resulting walk $C''$ then induces the full graph $G(V,E)$ itself, completing the proof.
\end{proof}

We now formalize a process for identifying ideals within finite-dimensional graph-admissible Lie algebras by leveraging the structure of their associated labeled directed graphs. The following Algorithm~\ref{alg:identifying:ideals} provides a systematic method for detecting all subsets of vertices that satisfy the ideal-graph-property, and consequently span ideals. This procedure is particularly useful for minimal-graph-admissible Lie algebras associated with minimal graphs, as each resulting vertex set corresponds to a basis of the corresponding ideal.

\begin{algorithm}[H]
        \DontPrintSemicolon
        \KwData{A labeled directed graph $G(V,E)$ associated with the graph-admissible Lie algebra $\g$.}
        \KwResult{A set $\mathcal{F}$ of ideals of $\g$, represented as subspaces spanned by subsets of the vertex set $V$.}
        \SetKwData{Left}{left}\SetKwData{This}{this}\SetKwData{Up}{up}
        \SetKwFunction{Union}{Union}\SetKwFunction{FindCompress}{FindCompress}
        \SetKwInOut{Input}{input}\SetKwInOut{Output}{output}
    
        \BlankLine
        $\mathcal{I}$ $\leftarrow$ $\{\emptyset,V\}$
        \tcc*[h]{Initialize the set $\mathcal{I}$ with the generator-sets of the trivial ideals: the zero ideal and the full Lie algebra}\;
        $\mathcal{C}$ $\leftarrow$ $V$
        \tcc*[h]{Initialize the set of candidate starting vertices}\;
        \BlankLine
        \ForEach{vertex $c\in \mathcal{C}$}{
            $\mathfrak{i}$ $\leftarrow$ $\{c\}$\tcc*[h]{Initialize the current ideal basis candidate $\mathfrak{i}$ as the singleton $\{c\}$}\;
            $\mathfrak{i}_\mathrm{new}$ $\leftarrow$ $\{c\}$\tcc*[h]{Track newly added vertices in the current iteration}\;
            \While{there exists an edge $e\in E$ such that $\varpi_\mathrm{s}(e)\in\mathfrak{i}$ and $\varpi_\mathrm{e}(e)\notin \mathfrak{i}$}{
                $\mathfrak{i}_\mathrm{new}'\leftarrow\emptyset$\tcc*[h]{Reset the set $\mathfrak{i}'_\mathrm{new}$ collecting new vertices to be added}\;
                \ForEach{vertex $w\in \mathfrak{i}_\mathrm{new}$ and $e\in E$}{
                    \If{$\varpi_\mathrm{s}(e)=w$ and $\varpi_\mathrm{e}(e)\notin \mathfrak{i}$}{
                        $\mathfrak{i}_\mathrm{new}'$ $\leftarrow$ $\mathfrak{i}_\mathrm{new}'\cup\{\varpi_\mathrm{e}(e)\}$\tcc*[h]{Add the target vertex of the edge to the ideal candidates}
                    }
                }
                $\mathfrak{i}$ $\leftarrow$ $\mathfrak{i}\cup\mathfrak{i}_\mathrm{new}'$\tcc*[h]{Update the ideal basis candidate $\mathfrak{i}$ with newly added vertices}\;
                $\mathfrak{i}_\mathrm{new}$ $\leftarrow$ $\mathfrak{i}_\mathrm{new}'$\tcc*[h]{Prepare for the next iteration}\;
            }
            $\mathcal{I}$ $\leftarrow$ $\mathcal{I}\cup\{\mathfrak{i}\}$\tcc*[h]{Add the complete ideal basis candidate to the set of ideals}\;
            \If{there exists a closed directed walk $W=(w_1,e_1,w_2,\ldots,w_p,e_{p},w_{1})$ such that $\bigcup_{j=1}^{p}\{w_j\}= \mathfrak{i}$}{
                $\mathcal{C}$ $\leftarrow$ $\mathcal{C}\setminus\mathfrak{i}$\tcc*[h]{Remove all vertices in $\mathfrak{i}$ from the candidate set if they form a closed directed walk}
            }
            \Else{
                $\mathcal{C}$ $\leftarrow$ $\mathcal{C}\setminus\{c\}$\tcc*[h]{Remove only the starting vertex if no such closed directed walk exists}
            }
        }
        $\mathcal{F}$ $\leftarrow$ $\emptyset$
        \tcc*[h]{Initialize the set that contains all ideals}\;
        \ForEach{subcollection $\mathfrak{i}=\{\mathfrak{i}_1,\mathfrak{i}_2,\ldots,\mathfrak{i}_s\}\in \mathcal{P}(\mathcal{I})$}{
            $\mathcal{F}$ $\leftarrow$ $\mathcal{F}\cup\{\spn\{\mathfrak{i}_1\cup\mathfrak{i}_2\cup\ldots\cup\mathfrak{i}_s\}\}$\tcc*[h]{Add the span of the union of subsets as an ideal}
        }
        \Return $\mathcal{F}$\tcc*[h]{Return the complete set of identified ideals}\;
\caption{Identifying ideals of a finite-dimensional graph-admissible Lie algebra}\label{alg:identifying:ideals}
\end{algorithm}

\begin{lemma}
    Let $\mathfrak{g}$ be a finite-dimensional graph-admissible Lie algebra associated with the labeled directed graph $G(V,E)$. Then Algorithm~\ref{alg:identifying:ideals} identifies precisely the ideals of $\g$ of the form $\mathfrak{h}=\vspan{W}$, where $W\subseteq V$ satisfies the ideal-graph-property.
\end{lemma}

\begin{proof}
    We must show two claims:
    \begin{enumerate}[label = (\arabic*)]
        \item Every subalgebra $\mathfrak{h}\subseteq\g$ computed by Algorithm~\ref{alg:identifying:ideals} is spanned by a subset $W\subseteq V$ that satisfies the ideal-graph-property.
        \item For every subset $W\subseteq V$ satisfying the ideal-graph-property, the ideal $\vspan{W}$ is computed by Algorithm~\ref{alg:identifying:ideals}.
    \end{enumerate}
    Let us begin with (1). Suppose $\mathfrak{h}\subseteq\g$ is obtained via the application of Algorithm~\ref{alg:identifying:ideals}, where $\g$ is associated with the graph $G(V,E)$. Then, by construction, $\mathfrak{h}$ is spanned by a vertex set $W\subseteq V$ that is the union of several vertex sets $W_1,\ldots,W_s$, i.e., $\bigcup_{j=1}^s W_j=W$ (cf. lines 19-20). Henceforth, $W\subseteq V$, since the sets $W_j$ are iteratively constructed from a starting vertex $c\in V$ (cf. line 4), by collecting all vertices reachable via directed edges originating from $c$, i.e., all $\varpi_\mathrm{e}(e)\in V$ for each $e\in E$ with $\varpi_\mathrm{s}(e)=c$ (cf. line 9), and repeating this procedure for the newly added vertices (cf. lines 6-12). 
        
    Now, assume by contradiction that there exists a vertex $w\in W$ and an edge $e\in E$ such that $\varpi_\mathrm{s}(e)=w$ and $\varpi_\mathrm{e}(e)\notin W$. Without loss of generality, let $w\in W_1$. This then contradicts the update rule (cf. lines 9-10) of Algorithm~\ref{alg:identifying:ideals}, which guarantees that any such target vertex $\varpi_\mathrm{e}(e)$ would have been added to $W_1$ during the iteration and consequently also been added to $W$. Therefore any edge $e\in E$ originating from a vertex $w\in W$ does not point outside of $W$, and $W$ satisfies the ideal-graph-property. Since $W\subseteq V$ and the Lie bracket of any element in $W$ with any element in $V$ remains in $\spn\{W\}$, it follows, by Lemma~\ref{lem:ideal:span}, that $[\mathfrak{h},\g]\subseteq\mathfrak{h}$, and hence $[\mathfrak{h},\mathfrak{h}]\subseteq\mathfrak{h}$. Thus, $\mathfrak{h}$ is an ideal and a subalgebra of $\g$.
        
    Let us proceed with (2), and let $\mathfrak{h}$ be a subalgebra of $\g$ spanned by a subset $W\subseteq V$ that satisfies the ideal-graph-property. The trivial case $W=\emptyset$ is immediate. This allows us to write: $W=\{w_j\}_{j=1}^r\subseteq \{w_j\}_{j=1}^m=V$ for $1\leq r\leq m$. For each $w_j\in W$ define the recursively constructed sets
    \begin{align}
        V_j^{(0)}:=\{w_j\}\qquad V_j^{(\ell+1)}:=\left\{v\in V\setminus\bigcup_{k=0}^\ell V_j^{(k)}\,\middle|\,\exists\, e\in E\text{ such that }\varpi_\mathrm{s}(e)\in V_j^{(\ell)}\text{ and }\varpi_\mathrm{e}(e)=v\right\}\quad\text{for all }\ell\in\Np{}.
    \end{align}
    These sets collect all vertices reachable from $w_j$ via a directed walk of length $\ell$, but by no shorter directed walk. The following observations hold:
    \begin{enumerate}[label = (\alph*)]
        \item One has $V_j^{(\ell)}\cap V_{j}^{(k)}=\emptyset$ if $\ell\neq k$.
        \item If $V_j^{(\ell)}=\emptyset$, then $V_j^{(k)}=\emptyset$ for all $k\geq\ell$.
        \item Since $V$ is finite, one must have that $V_j^{(\ell)}=\emptyset$ for all $\ell\geq|V|$.
        \item Every vertex $w\in V$ for which there exists a directed walk starting at $w_j$ and ending at $w$ belongs to the set $W_j^{(\infty)}:=\bigcup_{\ell=0}^\infty V_j^{(\ell)}$: To see this, let $W'=(w_j,e_1,v_1,e_2,\ldots,e_{s},w)$ be such a directed walk connecting $w_j$ and $w$. Then, clearly, $v_k\in W_j^{(k)}:= \bigcup_{\ell=0}^{k} V_j^{(\ell)}$ for all $k\in\{1,\ldots,s\}$, and consequently $w\in W_j^{(s)}\subseteq W_j^{(\infty)}$.
        \item For every vertex $w\in V_j^{(\ell)}$ with $\ell\geq 1$, there exists a directed trail that starts at $w_j$ and ends at $w$, which is also a directed path. This follows from the recursive construction of the sets $V_j^{(\ell)}$, and can be shown by induction on $\ell$. The base case $\ell=1$ is immediate from the definition, and the induction step follows by chaining edges and vertices from earlier levels.
    \end{enumerate}
        
    Suppose now that for some $j\in\{1,\ldots,r\}$, one has $w\in W_j^{(\infty)}$. Then there exists an integer $\ell\in\{0,\ldots,|V|-1\}$ such that $w\in V_j^{(\ell)}$. If $\ell=0$, then $w_j=w\in W$. If $\ell\geq 1$, then by the previous observations (a)-(e), there exists a directed walk $W''=(w_j,e_{j_0},w_{j_1},\ldots,w_{j_{\ell-1}},e_{j_{\ell-1}},w)$ such that $\varpi_\mathrm{s}(e_{j_k})=w_{j_k}\in V_j^{(k)}$, $\varpi_\mathrm{e}(e_{j_k})=w_{j_{k+1}}\in V_j^{(k+1)}$, $\varpi_\mathrm{s}(e_{j_0})=w_j$ and $\varpi_\mathrm{e}(e_{j_{\ell-1}})=w$. 
    Since $w_j\in W$ and $W$ satisfies the ideal-graph-property, all vertices along this walk $w_{j_1},\ldots,w_{j_{\ell-1}},w$, in particular $w$, must belong to $W$. Hence, $W_j^{(\infty)}\subseteq W$ for all $j\in\{1,\ldots,r\}$.
        
    Following Algorithm~\ref{alg:identifying:ideals}, the set $\mathcal{I}$ generally contains the sets $W_j^{(\infty)}$. That is, $\{W_j^{(\infty)}\}_{j\in\Tilde{\mathcal{M}}}\subseteq \mathcal{I}$ (cf. line 13), where $\Tilde{\mathcal{M}}\subseteq \mathcal{M}$ indexes the starting vertices used in the algorithm (cf. line 4). All $j\in\mathcal{M}$ are included in $\Tilde{\mathcal{M}}$, except for those $j'\in\mathcal{M}$ for which there exists a closed directed walk $C'=(w_{p_0},e_{p_0},w_{p_1},\ldots w_{p_{s}},e_{p_{s}},w_{p_0})$ such that $\bigcup_{k=0}^{s} \{w_{p_{k}}\}=W_j^{(\infty)}$ for some $j=p_0$ with $W_j^{(\infty)}\in \mathcal{I}$, and $j'\in\{p_k\}_{k=1}^s$ (cf. lines 14-15). In such cases, the algorithm avoids recomputing the same ideal multiple times by removing redundant starting vertices. 
        
    Suppose now that for some $k\in\{1,\ldots,r\}$, the set $W_k^{(\infty)}\notin \mathcal{I}$. Then, by the logic of Algorithm~\ref{alg:identifying:ideals}, there exists some $j\in\mathcal{M}$ and a closed directed walk $C''=(w_{k_0},e_{k_0},w_{k_1},\ldots, e_{k_{s}},w_{k_0})$ such that $k_0=j$, and $\bigcup_{q=0}^s\{w_{k_q}\}= W_j^{(\infty)}$ with $k=k_q$ for some $q\in \{1,\ldots,s\}$. We now show that in this case, $W_j^{(\infty)}= W_k^{(\infty)}$. Let $w\in W_j^{(\infty)}$. If $w=w_j$, then the segment of the closed directed walk $C''$ that starts at $w_k$ and ends at $w_j$ provides a directed walk from $w_k$ to $w_j$, showing that $w_j\in W_k^{(\infty)}$. If $w\neq w_j$, then there exists a directed walk $W'''$ that starts at $w_j$ and ends at $w$. This walk can be adapted such that it starts at $w_k$ instead, passes through $w_j$, and ends at $w$, by prepending the segment of $C''$ that starts at $w_k$ and ends at $w_j$. Hence, $w\in W_k^{(\infty)}$. By the same argument, any $w\in W_k^{(\infty)}$ also belongs to $W_j^{(\infty)}$, since $C''$ is a closed directed walk. We conclude: $W_k^{(\infty)}=W_j^{(\infty)}$. Therefore, even if for some $k\in\{1,\ldots,r\}$ the set  $W_k^{(\infty)}$ does not belong to $\mathcal{I}$ via a construction starting at the vertex $w_k$, it is represented via $W_j^{(\infty)}$ for some $j\in\mathcal{M}$; thus, all such sets are accounted for.
        
    We now conclude by showing that the full set $W$ is recovered as the union of all reachable sets: $W=\bigcup_{k=1}^rW_k^{(\infty)}$. The inclusion $W\subseteq \bigcup_{k=1}^r W_k^{(\infty)}$ is trivial, since $W=\{w_k\}_{k=1}^r$ and $\{w_k\}=W_k^{(0)}\subseteq W_k^{(\infty)}$. Conversely, suppose $w\in \bigcup_{k=1}^r W_k^{(\infty)}$. Then, either $w=w_j$ for some $j\in\{1,\ldots,r\}$, or there exists a directed path $P=(w_j,\ldots, w)$ that is also a directed trail, starting at some $w_j\in W$ for some $j\in\{1,\ldots,r\}$ and ending at $w$. Since $W$ satisfies the ideal-graph-property, all vertices along the directed path $P$ must also belong to $W$. Hence, $w\in W$, and we conclude $W=\bigcup_{k=1}^rW_k^{(\infty)}$. This shows that $W$ is the union of sets already included in $\mathcal{I}$, and therefore the ideal $\mathfrak{h}=\spn\{W\}$ is obtained by Algorithm~\ref{alg:identifying:ideals} (cf. lines 19-21), completing the proof.
\end{proof}

\begin{tcolorbox}[breakable, colback=Cerulean!3!white,colframe=Cerulean!85!black,title=\textbf{Example}: Application of Algorithm~\ref{alg:identifying:ideals} to identify ideals]
    \begin{example}
        We now demonstrate the application of Algorithm~\ref{alg:identifying:ideals} to an exemplary finite-dimensional graph-admissible Lie algebra $\g$, utilizing the graph depicted in Figure~\ref{fig:different:ideals:alg:application}. In this example, we select two distinct vertices to initiate the algorithm, namely $y$ and $p$. In both panels of the figure, the respective starting vertex is visually distinguished by coloring it red. During the first iteration of the algorithm, the vertices added to the candidate ideal basis $\mathfrak{i}$, apart from the initial vertex, are colored blue, along with the edges considered in this step. In the subsequent iteration, additional vertices and edges are added to $\mathfrak{i}$; these are colored green to indicate their inclusion in the second step of the expansion process. 
    \begin{figure}[H]
        \centering
        \includegraphics[width=0.75\linewidth]{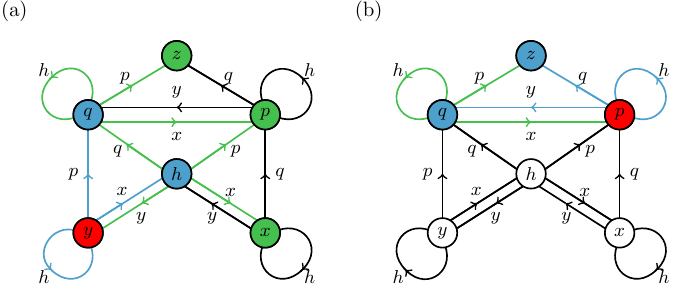}
        \caption{Application of Algorithm~\ref{alg:identifying:ideals} to identify ideals of a finite-dimensional graph-admissible Lie algebra using different initial vertices. The ideals are $V=\{x,y,h,q,p,z\}$ in panel (a) and $\{q,p,z\}$ in panel (b). According to \cite{A1:project}, the depicted graph can be associated with the Lie algebra $\sl{2}{\R}\ltimes\mathfrak{h}_1$, where the first ideal would coincide with the complete algebra, i.e., $\spn{V}\cong \sl{2}{\R}\ltimes\mathfrak{h}_1$. The second ideal, on the other hand, represents the radical $\spn\{q,p,z\}\cong\mathfrak{h}_1$, the Heisenberg algebra.}
        \label{fig:different:ideals:alg:application}
    \end{figure}
    We observe that there is potential to optimize Algorithm~\ref{alg:identifying:ideals}: For instance, if the vertex $p$ belongs to an ideal of the Lie algebra, then the vertex $q$ must also belong to the same ideal, and vice versa. This follows from the existence of a closed directed walk connecting $p$ and $q$. Therefore, it is sufficient to include only one of these two vertices in the initial set of candidate vertices $\mathcal{C}$ (cf. line 2). While a related reduction of $\mathcal{C}$ is currently implemented (cf. lines 14-17), the one discussed here is not, but would be a feasible and potentially a beneficial enhancement. 
    \end{example}
\end{tcolorbox}

The result above motivates the following classification of vertices and edges of a labeled directed graph $G(V,E)$, associated with a finite-dimensional graph-admissible Lie algebra $\g$, into three distinct structural categories:
\begin{itemize}
    \item \emph{L-type} (Lake-type): An edge $e\in E$ is classified as L-type if there exists a closed directed walk $W=(v_1,e_1,v_2,\ldots,e_s,v_1)$, such that $e\in \Tilde{E}$, where $G_W\equiv G(\Tilde{V},\Tilde{E})$ is the subgraph induced by the walk $W$. A vertex $v\in V$ is similarly classified L-type if $v\in\Tilde{V}$.
    \item \emph{T-type} (Tributary-type): A vertex $v\in V$ is classified as T-type if it is not L-type and there exists a directed walk $W=(v_1,e_1,v_2,\ldots,e_s,v_{s+1})$, such that $v_1=v$ and $v_{s+1}$ is an L-type vertex. An edge $e\in E$ is classified T-type if it is not L-type and belongs to a directed walk that starts at a T-type vertex and ends at an L-type vertex.
    \item \emph{D-type} (Disappearing-stream-type): Vertices and edges that are neither L-type nor T-type are classified as D-type.
\end{itemize}
The terminology used (L-type, T-type, and D-type) is inspired by the classification of water systems in river networks. The underlying idea behind this naming scheme is conceptually simple yet structurally insightful: One can envision the graph as a dynamic system in which water flows along the direction of the edges. In this analogy, vertices represent specific locations in the landscape, while edges represent the flow of water from one vertex to another, just as the Lie bracket operations connect elements of the algebra. This perspective enables the following interpretation of the graph's structure: L-type elements correspond to \emph{lakes} \cite{Lehner:2024}. These are regions of the graph where flow circulates. Specifically, a closed directed walk exists, allowing flow to traverse all involved vertices and return to the starting point, similar to circular currents found in lakes or oceans. In contrast, T-type elements represent \emph{tributaries} \cite{Matthews:2014}. These are vertices or edges that lead into a lake but are not themselves part of any closed walk; they feed into cyclic structures without directly participating in them. Consequently, it is generally impossible to start at one vertex of a tributary and find a directed walk that reaches every other vertex within the same tributary. This behavior is analogous to a river feeding into a lake: drifting along the river following its current does not lead to upstream regions. Vertices classified as D-type can be further subdivided into two subcategories: Either there exists a stream that originates at a tributary or a lake and terminates at a given vertex, making it part of a \enquote{distributary} \cite{Matthews:2014}, or no such stream exists. These vertices share the common characteristic of belonging to rivers that terminate without ever reaching a lake, thereby forming \enquote{disappearing streams} \cite{Ahmad:2025}. Accordingly, vertices from which no edges originate can be referred to as \emph{sinkholes}. 

This nomenclature also facilitates the introduction of the notions of \emph{upstream} and \emph{downstream}: A vertex $v$ is said to be downstream of a vertex $w$ if there exists a directed walk that starts at $w$ and ends at $v$, but not vice versa. Analogously, the vertex $w$ is said to be \emph{upstream} of the vertex $v$. Clearly, vertices that belong to the same closed directed walk are neither upstream nor downstream of each other. If the need arises, one can further subdivide L-type vertices and edges into distinct lakes, and similarly classify T-type and D-type vertices and edges into distinct rivers. This classification can also aid in easily identifying ideals visually. For instance, a lake together with all its distributaries and downstream vertices typically forms an ideal. Figure~\ref{fig:ideals:sketch:preview} visualizes the conceptual framework introduced here.

\begin{figure}[H]
    \centering
    \includegraphics[width=0.85\linewidth]{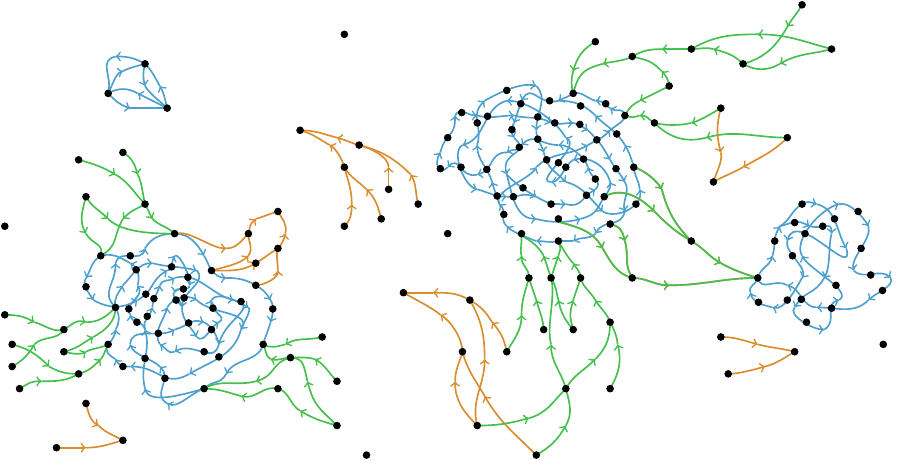}
    \caption{Schematic representation of the flow-based classification of the graph elements in the context of Lie algebraic structures. Edges are color-coded for visual clarity: L-type edges (blue) represent cyclic flows analogous to lakes; T-type edges (green) correspond to tributaries feeding into closed walks; D-type edges (orange) indicate disappearing streams. Note that the lack of labels is imposed to emphasize flow dynamics, and that the figure serves as a conceptual sketch rather than a representation of a valid graph.}
    \label{fig:ideals:sketch:preview}
\end{figure}

To further illustrate the concept of interpreting edges as streams, we consider the derived series of a Lie algebra. Let $\g$ be a finite-dimensional minimal-graph-admissible Lie algebra associated with the minimal labeled directed graph $G(V,E)$. Then, there exists an integer $\ell\in\N_{\geq0}$ such that $\mathcal{D}^\ell V=\mathcal{D}^{\ell+j}V$ for all $j\in\N_{\geq0}$. We define the smallest such integer $\ell$ as $\ell_*$. This allows us to introduce the following decomposition of the vertex set $V$:
\begin{align}
    S^{(1)}:=\left\{\begin{matrix}
        V&\text{, if }\ell_*=0\\
        V\setminus \mathcal{D}^1V&\text{, otherwise}
    \end{matrix}\right.\qquad\text{and}\qquad S^{(j)}:=\left\{\begin{matrix}
        V\setminus\left(\bigcup_{\ell=1}^{j-1}S^{(\ell)}\cup \mathcal{D}^{j}V\right)&,\text{ for all }j\in\{2,\ldots,\ell_*\}\\
        V\setminus \bigcup_{\ell=1}^{j-1}S^{(\ell)}&,\text{ for }j=\ell_*+1
    \end{matrix}\right..
\end{align}
It is evident that the sets $S^{(j)}$ are pairwise disjoint, i.e., $S^{(j)}\cap S^{(k)}=\emptyset$ if $j\neq k$, and their union reconstructs the entire vertex set, i.e., $\bigcup_{\ell=1}^{\ell_*+1} S^{(\ell)}=V$. Moreover, the vertices in $S^{(\ell)}$ are generally downstream of at least one vertex in some $S^{(j)}$ with $j<\ell$, but never downstream of any vertex in $S^{(k)}$ for $k>\ell$. These relationships can be schematically visualized as shown in Figure~\ref{fig:derived:series:csacade}. Specifically, the set $S^{(1)}$ is generally upstream of the sets $S^{(2)},\ldots ,S^{(\ell_*+1)}$, and this pattern continues such that $S^{(\ell)}$ is generally upstream of the subsequent sets $S^{(j)}$ with $j>\ell$, in the sense that there can only exist edges that point from vertices in $S^{(\ell)}$ to $S^{(j)}$ but never in the converse direction. Conversely, the set $S^{(\ell_*+1)}$ is not upstream of any other set, but at most downstream of all preceding sets $S^{(\ell)}$. A similar decomposition can be constructed using other series of ideals, such as the lower central series, by replacing the operation $\mathcal{D}^\ell$ with the operation $\mathcal{C}^\ell$.

We can also interpret the graphs through the lens of the Levi-Mal'tsev decomposition theorem, which states that any Lie algebra over a field of characteristic zero can be expressed as the semidirect sum of a semisimple component and its radical \cite{Kuzmin:1977}. In the graphical representation, the semisimple component is typically depicted as a collection of L-type vertices and edges, since non-solvable Lie algebras must contain a closed directed walk (cf. Theorem~\ref{thm:non:solvability:condition:strong}). The semidirect sum structure manifests then in the observation that the portion of the graph corresponding to the radical is generally positioned downstream from the semisimple part. Furthermore, if the radical is nilpotent, this part will generally be represented by D-type edges and vertices.

\begin{figure}[htpb]
    \centering
    \includegraphics[width=0.5\linewidth]{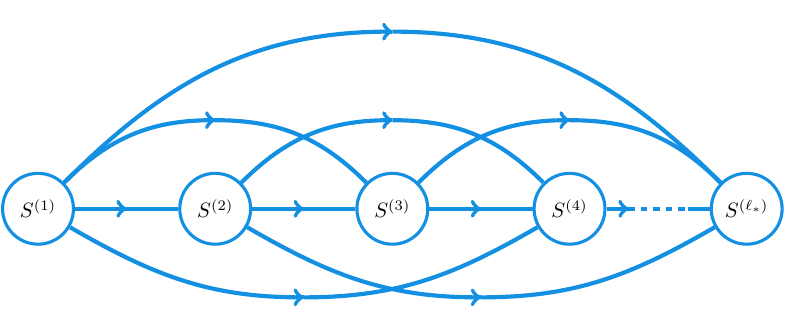}
    \caption{Schematic depiction of the relations among the sets $S^{(\ell)}$, defined via the derived series $\mathcal{D}^\ell V$ of the vertices from a minimal graph $G(V,E)$ associated with a finite-dimensional minimal-graph-admissible Lie algebra. This diagram illustrates that edges which originate at a vertex in the set $S^{(\ell)}$ only point to vertices in the sets $S^{(j)}$ with $j\geq \ell$, but never to vertices in sets $S^{(k)}$ with $k<\ell$.}
    \label{fig:derived:series:csacade}
\end{figure}

\subsection{Criteria for simplicity and semisimplicity}\label{sec:structural:properties:simple:semisimple}

The results of the previous subsection allow us to formulate criteria for simplicity and semisimplicity, both of which fundamentally depend on the structure and existence of ideals within a Lie algebra. Note that, as mentioned earlier, we only consider Lie algebras that are taken over a field of characteristic zero, which provides us with several well-known conditions for semisimplicity \cite{Serre:1992}.

We begin by presenting a necessary criterion, derived from Lemma~\ref{lem:ideal:span}, which applies to minimal-graph-admissible Lie algebras. However, it is important to emphasize that the simplicity of a Lie algebra does not imply that its associated graph is a directed simple graph, nor does the converse hold. A directed simple graph is defined as a directed graph that neither contains multiple edges between the same pair of source and target vertices nor loops \cite{Weisstein:2025}. A counterexample illustrating the non-equivalence between graph simplicity and algebraic simplicity is the graph associated with the simple Lie algebra $\sl{2}{\R}$ depicted in Figure~\ref{fig:su2:sl2}~(b). Conversely, the simplicity of a directed graph does not rule out the simplicity of the corresponding Lie algebra. For instance, the graph associated with the simple Lie algebra $\mathfrak{su}(2)$, shown in Figure~\ref{fig:su2:sl2}~(a), is also a simple directed graph. These examples demonstrate that, within this framework, there is no causal or structural correspondence between the simplicity of a graph and the simplicity of the associated Lie algebra.

\begin{lemma}\label{lem:nec:condition:first:simplicity:minimal}
    Let $\g$ be an $n$-dimensional minimal-graph-admissible Lie algebra associated with the minimal graph $G(V,E)$, and suppose $n\geq 2$. If $\g$ is simple, then the only subsets $W\subseteq V$ that satisfy the ideal-graph property are $W=V$ and $W=\emptyset$.
\end{lemma}

\begin{proof}
    Let $G(V,E)$ be a minimal graph associated with the finite-dimensional Lie algebra $\g$. First, observe that the sets $V$ and $\emptyset$ both trivially satisfy the ideal-graph-property. Suppose, for contradiction, that there exists a proper non-empty subset $W\subsetneq V$ that satisfies the ideal-graph property. By Lemma~\ref{lem:ideal:span}, the subset $W$ spans an ideal of $\g$, which is a proper non-zero ideal, since $V$ is a basis of $\g$ and $W$ is a proper non-empty subset of $V$. This contradicts the simplicity of $\g$, as defined in Definition~\ref{def:simple:and:semisimple}, and thus concludes this proof.
\end{proof}

We now restate Lemma~\ref{lem:nec:condition:first:simplicity:minimal} in a way that highlights the structural properties directly observable from the graph associated to the Lie algebra $\g$:

\begin{proposition}\label{prop:weak:simplicity:cond:graph:based}
    Let $\g$ be a finite-dimensional minimal-graph-admissible Lie algebra associated with the minimal graph $G(V,E)$, and suppose $\dim(\g)\geq 2$. If $\g$ is simple, then there exists a closed directed walk $W$ that induces a subgraph $G_W\equiv G(\Tilde{V},\Tilde{E})$ with $V=\Tilde{V}$.
\end{proposition}

\begin{proof}
    Let $\g$ be a simple finite-dimensional minimal-graph-admissible Lie algebra associated with the minimal graph $G(V,E)$. Then, by Lemma~\ref{lem:nec:condition:first:simplicity:minimal}, there exists no proper non-empty subset $W\subsetneq V=\{v_j\}_{j\in\mathcal{N}}$ that satisfies the ideal-graph-property. By Corollary~\ref{cor:no:proper:igp:closed:walk:induces:V}, it follows that the graph must contain a closed directed walk $W$ that traverses all vertices in $V$, thereby inducing a subgraph $G_W\equiv G(\Tilde{V},\Tilde{E})$ with $V=\Tilde{V}$.
\end{proof}

We now introduce a sufficient but weak condition for non-semisimplicity that applies not only to minimal graphs but also to redundant ones. This condition is based on the presence of sinkholes in the associated graph structure.
\begin{proposition}\label{prop:very:very:weak:semisimplicity:condition}
    Let $\g$ be a finite-dimensional graph-admissible Lie algebra associated with the labeled directed graph $G(V,E)$. If $G(V,E)$ contains sinkholes, i.e., vertices $v\in V$ for which there is no edge $e\in E$ with $\varpi_\mathrm{s}(e)= v$, then $\g$ is not semisimple.
\end{proposition}

\begin{proof}
    Let $v$ be a sinkhole. Then, the singleton set $\{v\}$ satisfies the ideal-graph-property. By Lemma~\ref{lem:ideal:span}, this implies that $\{v\}$ spans a non-zero ideal of $\g$. Moreover, the Lie subalgebra $\spn\{v\}$ is trivially solvable, as it is abelian. Therefore, $\g$ contains a solvable non-zero ideal, and by Definition~\ref{def:simple:and:semisimple}, it follows that $\g$ is not semisimple. 
\end{proof}

We now generalize the criterion provided in Proposition~\ref{prop:very:very:weak:semisimplicity:condition} by considering not just sinkholes, but a broader class of vertices, which leads to a more robust condition for detecting non-semisimplicity in graph-admissible Lie algebras.
\begin{proposition}
    Let $\g$ be a finite-dimensional graph-admissible Lie algebra associated with the labeled directed graph $G(V,E)$. Suppose there exists a proper non-empty subset $W\subsetneq V$ that spans a proper subspace $\mathfrak{w}:=\spn\{W\}\subsetneq\g$ and satisfies the ideal-graph-property. If there exists no closed directed walk $C$ in $G(V,E)$ that induces a self-contained subgraph $G_C\equiv G(\Tilde{V},\Tilde{E})$ with $\Tilde{V}\subseteq W$, then $\g$ is not semisimple. 
\end{proposition}

\begin{proof}
    By Lemma~\ref{lem:ideal:span}, the proper non-empty subset $W\subsetneq V$ spans an ideal of $\g$. This ideal $\mathfrak{w}:=\spn\{W\}$ is, by assumption, a proper non-zero subalgebra of $\g$. Suppose there exists no closed directed walk $C$ in $G(V,E)$ that induces a self-contained subgraph $G_C\equiv G(\Tilde{V},\Tilde{E})$ with $\Tilde{V}\subseteq W$. Then, in particular, there exists also no closed directed walk $C'$ in the subgraph $G(W,{E}')$ that induces a self-contained subgraph $G_{C'}\equiv G(\Tilde{V}',\Tilde{E}')$, where ${E}':=\{e\in E\,\mid\,\varpi_\mathrm{s}(e),\varpi_\mathrm{l}(e),\varpi_\mathrm{e}(e)\in W\}$. By Lemma~\ref{lem:subalgebra:graph:weak}, the graph $G(W,{E}')$ is  associated with the subalgebra $\mathfrak{w}$. Then, by Theorem~\ref{thm:non:solvability:condition:strong}, it follows from the absence of a closed directed walk that induces a self-contained subgraph that $\mathfrak{w}$ is solvable, making $\g$, by Definition~\ref{def:simple:and:semisimple}, not semisimple.
\end{proof}

As a complementary condition to  Proposition~\ref{prop:very:very:weak:semisimplicity:condition} which is based on the presence of sinkholes, we now consider a criterion for non-semisimplicity that relies on the existence of loose ends and isolated vertices, that is, on vertices that are never targets of edges. Note that this includes those vertices that are never sources of any edges.

\begin{proposition}
    Let $\g$ be a finite-dimensional graph-admissible Lie algebra associated with the labeled directed graph $G(V,E)$. If the set $\Tilde{V}:=\{v\in V\,\mid\,\forall\,e\in E:\,\varpi_\mathrm{e}(e)\neq v\}$ of loose ends and isolated vertices is such that $\mathfrak{v}:=\spn\{V\setminus \Tilde{V}\}$ is a proper non-zero subspace of $\g$, then $\g$ is not semisimple.
\end{proposition}

This also implies that the existence of loose ends in a minimal graph implies that the associated Lie algebra is not semisimple.

\begin{proof}
    Lemma~\ref{lem:derived:series:graph:alg:valid} informs us that applying Algorithm~\ref{alg:generating:generating:the:graph:derived:altered} to the graph $G(V,E)$ yields the (possibly overcomplete) basis $\mathcal{D}^1 V$ of the first derived algebra $\mathcal{D}^1\g$. The basis is given by $\mathcal{D}^1 V= V\setminus \Tilde{V}$, and therefore spans a proper subspace of $\g$. Since $\g\neq \mathcal{D}^1\g$, it follows from Corollary 2 in \oldcite[page 45]{Serre:1992} that $\g$ is not semisimple.
\end{proof}

The criteria developed for semisimple Lie algebras can be naturally extended to reductive Lie algebras. Recall that the first derived algebra $\mathcal{D}^1\g$ of a reductive Lie algebra $\g$ is semisimple, and it decomposes as $\g=\mathcal{D}^1\g\oplus\mathcal{Z}(\g)$ \cite{Knapp:1996}. Consequently, one can employ Algorithm~\ref{alg:generating:generating:the:graph:derived:altered} to obtain the first derived graph $\mathcal{D}^1 G(V,E)$ associated with the first derived algebra $\mathcal{D}^1\g$ of a finite-dimensional graph-admissible Lie algebra $\g$ and its associated graph $G(V,E)$. Applying the previously established conditions for semisimplicity to the graph $\mathcal{D}^1G(V,E)$, and finding that it does not correspond to a graph of a semisimple Lie algebra, one can conclude that $\g$ is not reductive.

Note that we can also deduce local criteria for semisimplicity, i.e., criteria that depend on the existence of special subgraphs, which are based on considerations involving the Killing form. For this see Proposition~\ref{prop:semisimple:killing:form:condition} in Appendix~\ref{app:symmetries}.

We want to conclude this section with a conjecture that proposes a strong condition for semisimplicity, based on the associated graphs.

\begin{conjecture}\label{con:semisimple}
    Let $\g$ be a finite-dimensional minimal-graph-admissible Lie algebra associated with a minimal labeled directed graph $G(V,E)$. Then $\g$ is semisimple if and only if every vertex is part of a closed directed walk that induces a self-contained subgraph of $G(V,E)$.
\end{conjecture}

\section{Prominent examples}\label{sec:prominent:examples}

We present here several examples of graph-admissible Lie algebras that play a significant role in many areas of physics, including general relativity and quantum mechanics. These examples are analyzed utilizing the graph-theoretic framework developed in this work. Our objective is to underscore the practical utility of the proposed methods and to demonstrate that their strength lies not only in computational applicability but also in their capability to visually represent the structural properties, thereby facilitating deeper structural insight. This graphical perspective can be particularly advantageous for those encountering these algebraic concepts for the first time.

\subsection{The Schrödinger algebra}\label{sec:prominent:examples:subsec:Schroedinger:algebra}
We begin by discussing the real Lie algebra $\sl{2}{\R}\ltimes\mathfrak{h}_{m}$, where $\mathfrak{h}_{m}$ denotes  the real $(2m+1)$-dimensional Heisenberg algebra and $\sl{2}{\R}$ the real three-dimensional special linear algebra \cite{Gosson:2006,Lang:2012}. For $m=1$, this algebra is physically particularly relevant as it is associated with the symmetries of the Schrödinger equation \cite{Aizawa:2011}. Moreover, it represents the largest non-solvable Lie algebra that can be faithfully realized within the skew-hermitian Weyl algebra $\hat{A}_1$---the algebra of all skew-hermitian polynomials made from the bosonic creation and annihilation operators of one mode \cite{A1:project}. The term \enquote{Schrödinger algebra} used to describe the Lie algebra $\sl{2}{\R}\ltimes\mathfrak{h}_{m}$ is somewhat ambiguous, as it is often used to refer to different and distinct Lie algebras \cite{A1:project,Dobrev:1997,Nikitin:2020,Duval:1994,Aizawa:2011,Tao:2022}. However, for the simplest case $m=1$, these definitions all coincide. There are generally two approaches to extending this basic Schrödinger algebra. Both stem from the observation that the fundamental Lie algebra $\sl{2}{\R}\ltimes\mathfrak{h}_1$  encodes the symmetries of the free Schrödinger equation in a $(1+1)$-dimensional flat spacetime \cite{Liu:2021}. Thus, a natural extension of this algebra is to consider the groups associated with the symmetries of the Schrödinger equation in an $(n+1)$-dimensional flat spacetime \cite{Nikitin:2020,Duval:1994,Aizawa:2011,Liu:2021}. To distinguish these and related algebras, they are sometimes referred to as the \emph{extended} or \emph{super Schrödinger algebra}. Another approach is to extend the dimension of the Heisenberg algebra, i.e., by replacing the algebra $\mathfrak{h}_1$ in the first Schrödinger algebra $\sl{2}{\R}\ltimes\mathfrak{h}_1$ with the higher-dimensional Heisenberg algebra $\mathfrak{h}_m$ \cite{Tao:2022}, or similar extensions \cite{Liu:2021}. 

The Lie algebra $\sl{2}{\R}\ltimes\mathfrak{h}_m$ is defined over $\R$ with basis elements $\mathcal{B}=\{h,x,y,q_j,p_j,z\,\mid\,j\in\{1,\ldots,m\}\}$, and Lie bracket relations given by
\begin{align}
    [h,x]&=2x,\;&\;[h,y]&=-2y,\;&\;[x,y]&=h,\;&\;[h,q_j]&=q_j,\nonumber\\
    [y,q_j]&=p_j,\;&\;[h,p_j]&=-p_j,\;&\;[x,p_j]&=q_j,\;&\;[q_j,p_k]&=\delta_{jk}z,\label{eqn:basis:schroedinger:algebra}
\end{align}
while all other Lie brackets vanish \cite{Tao:2022}. Note that the subalgebra $\sl{2}{\R}$ is spanned by the three elements $h,x,y$, while the Heisenberg algebra $\mathfrak{h}_{m}$ is spanned by the remaining elements $q_1,p_1,\ldots,q_m,p_m,z$. This algebra is physically relevant because it can be faithfully realized within the Weyl algebra $A_m$, which is used to describe all bosonic polynomial Hamiltonians \cite{Woit:2017,Bruschi:Xuereb:2024} (see Appendix~\ref{app:realization:of:schroedinger:algebra} for more details). The basis choice in \eqref{eqn:basis:schroedinger:algebra} clearly implies that the algebra is minimal-graph-admissible, and its associated graph for $m=2$ is depicted in Figure~\ref{fig:schroedinger:algebra:eight:dim}.
\begin{figure}[htpb]
    \centering
    \includegraphics[width=0.6\linewidth]{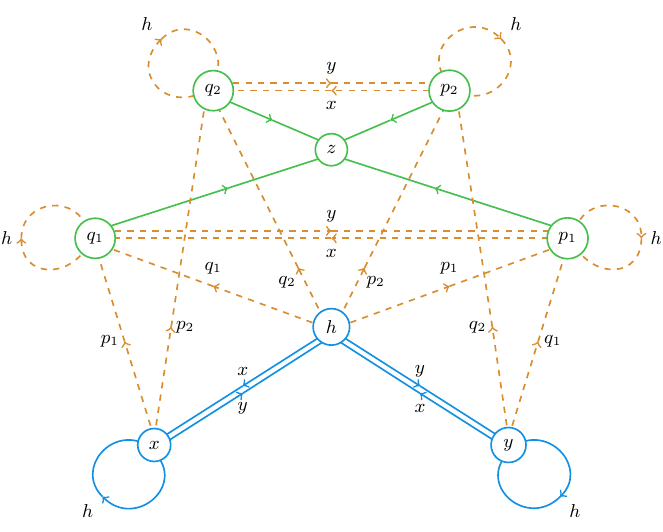}
    \caption{Depiction of the minimal graph associated with the Lie algebra $\sl{2}{\R}\ltimes\mathfrak{h}_{2}$, using the basis given in \eqref{eqn:basis:schroedinger:algebra}.}
    \label{fig:schroedinger:algebra:eight:dim}
\end{figure}

From this graph, we observe that this algebra is not solvable, as established by the criteria in Theorem~\ref{thm:non:solvability:condition:strong}: the subgraph with blue edges and vertices contains a closed directed walk that induces a self-contained subgraph. Furthermore, according to Proposition~\ref{prop:very:very:weak:semisimplicity:condition}, the algebra is also not semisimple and therefore not simple, since the vertex $z$ acts as a sinkhole (and thus it spans a nontrivial ideal of the full algebra). 

This structure adheres to the Levi-Mal'tsev decomposition theorem, which states that any Lie algebra is the semi-direct sum of a semisimple Lie algebra and its radical \cite{Kuzmin:1977}. In this case, the semisimple component is $\sl{2}{\R}$ and the radical is $\mathfrak{h}_2$. The graph illustrates this decomposition: vertices and edges in green correspond to basis elements from $\mathfrak{h}_2$, while those in blue represent elements from $\sl{2}{\R}$. Dashed orange edges indicate non-vanishing brackets between elements of the semisimple part and the radical, thereby encoding the semidirect sum. Notably, the subgraph associated with $\mathfrak{h}_2$ satisfies the ideal-graph-property, and its structure confirms that $\mathfrak{h}_2$ is solvable and, in fact, nilpotent, as there exists no closed directed walk within this subgraph (cf. Theorem~\ref{thm:nilpotency:criteria:strong}). It should therefore be evident that the graph-theoretic framework provides a clear visualization of the Levi-Mal'tsev decomposition.

\subsection{The Lie algebra of the Lorentz group}\label{sec:prominent:examples:subsec:Lorentz:algebra}
A second example we want to discuss is the complexified Lie algebra of the Lorentz group, denoted by $\mathfrak{so}(3;1)_\C$ \cite{Sexl:2000}, and hereafter referred to as the \emph{Lorentz algebra} for convenience. This algebra is of paramount importance to quantum field theory \cite{Srednicki:2007} and general relativity \cite{Wald:1984,Carroll:2019} as it encodes symmetries of Minkowski spacetime (i.e., flat spacetime). This algebra is generated by six elements: three boost operators $K_1,K_2,K_3$, and the three angular momentum operators $J_1,J_2,J_3$. Their Lie bracket relations are given by
\begin{align*}
    [J_j,J_k]&=i\sum_{\ell=1}^3\varepsilon_{jk\ell}J_\ell,\;&\;[J_j,K_k]&=i\sum_{\ell=1}^3\varepsilon_{jk\ell} K_\ell,\;&\;[K_j,K_k]&=-i\sum_{\ell=1}^3\varepsilon_{jk\ell}J_\ell,
\end{align*}
where $\varepsilon_{jk\ell}$ denotes the completely antisymmetric Levi-Civita symbol \cite{Srednicki:2007}. This algebra is therefore minimal-graph-admissible, and its corresponding graph $G_1(V_1,E_1)$ is depicted in Figure~\ref{fig:lorentz:algebra:v01:ab} (a). The graph representation provides an intuitive visualization of the corresponding non-abelian structure: The presence of closed directed walks that induce self-contained subgraphs indicates non-solvability. Nevertheless, this example illustrates that Proposition~\ref{prop:weak:simplicity:cond:graph:based} offers only a necessary condition, not a sufficient one: although the graph $G_1(V_1,E_1)$ contains a closed directed walk that induces the entire graph $G_1(V_1,E_1)$ itself, the algebra $\mathfrak{so}(3;1)_\C$ is only semisimple, not simple. 
\begin{figure}
    \centering
    \includegraphics[width=0.98\linewidth]{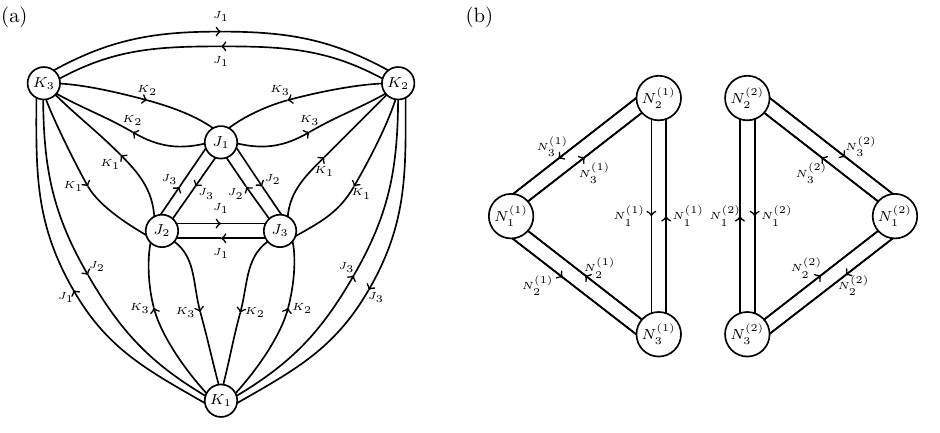}
    \caption{Panel (a): Depiction of the graph $G_1(V_1,E_1)$ associated with the complex Lorentz algebra $\mathfrak{so}(3;1)_\C$, constructed using the basis $\{J_j,K_j\}_{j=1}^3$ corresponding to the angular momentum and boost operators. Panel (b): Depiction of the graph $G_2(V_2,E_2)$ associated with the complex Lorentz algebra $\mathfrak{so}(3;1)_\C$, constructed using the basis $\{N_j^{(1)},N_j^{(2)}\}_{j=1}^3$, where $N_j^{(1)}=J_j-iK_j$ and $N_j^{(2)}=J_j+iK_j$. This representation reveals the decomposition $\mathfrak{so}(3;1)_\C\cong\sl{2}{\C}\oplus\sl{2}{\C}$, as the graph splits into two unconnected components. Although the physical interpretation of these new generators is obscure \cite{Srednicki:2007}, the graph-theoretic perspective makes the algebraic structure explicit.}
    \label{fig:lorentz:algebra:v01:ab}
\end{figure}
To reveal its deeper structure, we can construct a second graph that can be associated with $\mathfrak{so}(3;1)_\C$. To this end, we perform a change of basis by introducing the elements $N_j^{(1)}:=J_j-iK_j$ and $N_j^{(2)}:=J_j+iK_j$ for $j\in\{1,2,3\}$. These elements satisfy the following Lie brackets:
\begin{align*}
    [N_j^{(1)},N_k^{(1)}]&=2i\sum_{\ell=1}^3\varepsilon_{jk\ell} N_\ell^{(1)},\;&\;[N_j^{(2)},N_k^{(2)}]&=2i\sum_{\ell=1}^3\varepsilon_{jk\ell} N_\ell^{(2)},\;&\;[N_j^{(1)},N_k^{(2)}]&=0.
\end{align*}
This decomposition shows that $\mathfrak{so}(3;1)_\C\cong\sl{2}{\C}\oplus\sl{2}{\C}$, as noted in \cite{Hall:Lie:groups:16}, and is a demonstration of Proposition~\ref{prop:semisimple:lie:algebra:unconnected:subgraphs}, since the associated graph $G_2(V_2,E_2)$  now splits into two unconnected components, as depicted in Figure~\ref{fig:lorentz:algebra:v01:ab} (b).

Importantly, regardless of the chosen basis $\{J_j,K_j\}_{j=1}^3$ or $\{N_j^{(1)},N_j^{(2)}\}_{j=1}^3$, every vertex belongs to a closed directed walk that induces a self-contained subgraph, supporting Conjecture~\ref{con:semisimple}, which asserts that every minimal graph associated with a semisimple Lie algebra exhibits this property. Furthermore, by Proposition~\ref{prop:weak:simplicity:cond:graph:based}, we can confirm that the Lorentz algebra is not simple, as the graph $G_2(V_2,E_2)$ in Figure~\ref{fig:lorentz:algebra:v01:ab} (b) does not contain a single closed directed walk encompassing all vertices of the original graph $G_2(V_2,E_2)$. 

The Lorentz algebra can be extended to the Poincaré algebra by adding the energy and momentum operators. For a further discussion on this topic see Appendix~\ref{app:poincare}, where we also consider the Galilei algebra, which can be obtained as the non-relativistic limit of the Poincaré algebra.

\subsection{Graphical representation for the search of finite-dimensional physically relevant Lie algebras}\label{sec:prominent:examples:linear:quantum}
One of the main goals of the approach presented in this work is to apply graph-based representations to Lie algebras that arise in a physically relevant context. One such context is that of the Lie algebras generated by the defining operators of Hamiltonians, which are hermitian operators that describe the dynamics of coupled bosonic systems. This avenue has been the object of study of a seminal work \cite{Bruschi:Xuereb:2024}, further applied to the single mode case in subsequent work \cite{A1:project}. There, the focus is on the skew-hermitian Weyl algebra $\hat{A}_n$ of $n$ modes, which consists of all skew-hermitian polynomials in the complex Weyl algebra $A_n$. The complex Weyl algebra, in turn, plays a central role in the description of bosonic systems \cite{Woit:2017}. It is generated by the creation and annihilation operators $\hat{a}_k^\dagger$ and $\hat{a}_k$ that satisfy the canonical commutation relations $[\hat{a}_k,\hat{a}^\dagger_\ell]=\delta_{k\ell}$, while all others vanish. In this context, the Lie algebra $\g:=\lie{\{\hat{a}_1,\ldots,\hat{a}_n,\hat{a}^\dagger_1,\ldots,\hat{a}_n^\dagger,1\}}$ is the $(2n+1)$-dimensional Heisenberg algebra $\mathfrak{h}_n$. Consequently, the Weyl algebra $A_n$ can be identified with the universal enveloping algebra of $\mathfrak{h}_n$, provided the central basis element is chosen to be the identity, i.e., $A_n=U(\mathfrak{h}_n)/\langle z-1\rangle$, where $z$ denotes the central basis element of the abstract Heisenberg algebra~\cite{Gaddis:2025}.  Since bosonic systems are governed by Hamiltonian operators, which are defined in this context as hermitian operators in $A_n$, the operators of interest belong to a specific subspace of the Weyl algebra: the real hermitian Weyl space, which is isomorphic to the real skew-hermitian Weyl algebra. Hermiticity can, for simplicity in this context, be formalized by introducing the conjugate-linear anti-automorphism $(\cdot)^\dagger:A_n\to A_n$ with action $(\cdot)^\dagger:p\mapsto p^\dagger$ that satisfies the following properties: $1^\dagger=1$, $(\hat{a}_k)^\dagger=\hat{a}_k^\dagger$, $(\hat{a}_k^\dagger)^\dagger=\hat{a}_k$, $i^\dagger=-i$, and the order of multiplication is reversed, i.e., $(\hat{p}_1\hat{p}_2)^\dagger=\hat{p}_2^\dagger\hat{p}_1^\dagger$. An operator $\hat{O}$ is hermitian if it obeys $\hat{O}^\dagger=\hat{O}$, while a skew-hermitian operator $\hat{O}$ satisfies $\hat{O}^\dagger=-\hat{O}$.

The central question posed in the literature \cite{Bruschi:Xuereb:2024}, which we address here from a motivational and outlook perspective, is
\begin{quote}
    \textbf{Q}: \emph{Given a Hamiltonian $\hat{H}(t)=i\sum_{j\in\mathcal{J}}u_j(t) \hat{g}_j$, where $u_j(t)$ are real-valued time-dependent functions and $\hat{g}_j$ are time-independent and linearly independent skew-hermitian operators, is the Lie algebra $\g:=\lie{\mathcal{G}}$ generated by the set $\mathcal{G}:=\{\hat{g}_j\}_{j\in\mathcal{J}}$ finite dimensional? }
\end{quote}
This question is of profound importance because its answer determines whether the dynamics of the quantum system can be effectively characterized or even controlled via a finite number of degrees of freedom \cite{Bruschi:Xuereb:2024}. Such considerations are central to quantum optimal control \cite{Boscain:2021,Huang:1983}, a rapidly growing field driven by the recent advances in quantum information science \cite{Google:2024,Werninghaus:2021,Kim:2023}.

In order to connect this aforementioned avenue with our work on graphical representations of Lie algebras, we now consider the prototypical Hamiltonian
\begin{align*}
    \hat{H}(t)=\omega_1(t)\hat{a}_1^\dagger \hat{a}_1+\omega_2(t) \hat{a}_2^\dagger \hat{a}_2+ g_2(t)\left(\hat{a}_1^2+(\hat{a}_1^\dagger)^2\right)+i(\hat{a}_1-\hat{a}_1^\dagger)\left(g_1(t)+g_3(t)\hat{a}_2^\dagger \hat{a}_2+g_4(t)(\hat{a}_2^\dagger)^2\hat{a}_2^2\right),
\end{align*}
which models an optomechanical system \cite{Larson:2024} that has been extended by the addition of nonlinear interaction terms. Optomechanical systems are paradigmatic models of the interaction between the radiation pressure of photons and one or more vibrating mirrors in a cavity \cite{Aspelmeyer:Kippenberg:2014}, and they are now the subject of great theoretical and experimental efforts due to their potential as core components of future quantum technologies \cite{Barzanjeh:2021,Andrews:2014,Riedinger:2018}. The set of generators $\mathcal{G}$ of $\hat{H}(t)$ is given by
\begin{align*}
    \mathcal{G}:=\left\{i\hat{a}_1^\dagger \hat{a}_1,i\hat{a}_2^\dagger \hat{a}_2,(\hat{a}_1-\hat{a}_1^\dagger),i\left(\hat{a}_1^2+(\hat{a}_1^\dagger)^2\right),(\hat{a}_1-\hat{a}_1^\dagger)\hat{a}_2^\dagger \hat{a}_2,(\hat{a}_1-\hat{a}_1^\dagger)(\hat{a}_2^\dagger)^2\hat{a}_2^2\right\}.
\end{align*}
We can apply Theorem 54 of the cited work \cite{Bruschi:Xuereb:2024} to confirm that the Lie algebra $\g:=\lie{\mathcal{G}}$ is finite dimensional. A direct computation shows that $\dim(\g)=14$. A convenient (overcomplete) basis is:
\begin{align*}
    a&:=i\left(\hat{a}_1^\dagger\hat{a}_1+\frac{1}{2}\right),\;&\;b&:=\hat{a}_1^2-(\hat{a}_1^\dagger)^2,\;&\;c&:=i\left(\hat{a}_1^2+(\hat{a}_1^\dagger)^2\right),\\
    p_1&:=\hat{a}_1-\hat{a}_1^\dagger,\;&\;q_1&:=i\left(\hat{a}_1+\hat{a}_1^\dagger\right),\;&\;z_1&:=i,\\
    p_2&:=\left(\hat{a}_1-\hat{a}_1^\dagger\right)\hat{a}_2^\dagger\hat{a}_2,\;&\;q_2&:=i\left(\hat{a}_1+\hat{a}_1^\dagger\right)\hat{a}_2^\dagger\hat{a}_2,\;&\;z_2&:=i(\hat{a}_2^\dagger)^2\hat{a}_2^2+i\hat{a}_2^\dagger\hat{a}_2,\\
    p_3&:=\left(\hat{a}_1-\hat{a}_1^\dagger\right)(\hat{a}_2^\dagger)^2\hat{a}_2^2,\;&\;q_3&:=i\left(\hat{a}_1+\hat{a}_1^\dagger\right)(\hat{a}_2^\dagger)^2\hat{a}_2^2,\;&\;z_3&:=i(\hat{a}_2^\dagger)^4\hat{a}_2^4+4i(\hat{a}_2^\dagger)^3\hat{a}_2^3+2i(\hat{a}_2^\dagger)^2\hat{a}_2^2,\\
    y_1&:=i(\hat{a}_2^\dagger)^2\hat{a}_2^2,\;&\;y_2&:=i\hat{a}_2^\dagger\hat{a}_2,\;&\;y_3&:=i(\hat{a}_2^\dagger)^3\hat{a}_2^3+2i(\hat{a}_2^\dagger)^2\hat{a}_2^2.
\end{align*}
Computing the Lie brackets confirms that this algebra is graph-admissible. However, because certain elements (e.g.,~$z_2$) can be expressed as a linear combination of others (such as $y_1$ and $y_2$), the associated graph $G(V,E)$ is redundant. Figure~\ref{fig:om:example} gives an intuitive illustration of this graph.

\begin{figure}
    \centering
    \includegraphics[width=0.75\linewidth]{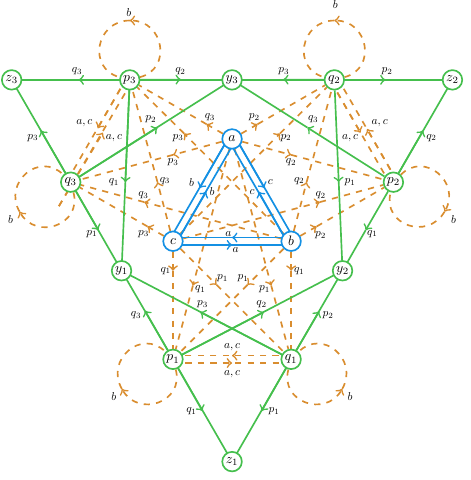}
    \caption{Depiction of the redundant graph $G(V,E)$ associated with the fourteen-dimensional Lie algebra $\g$ and the overcomplete  basis $\{a,b,c,p_j,q_j,z_j,y_j\}_{j=1}^3$. Blue vertices and edges correspond to the semisimple subalgebra spanned by $\{a,b,c\}$, while green vertices and edges indicate the radical $\operatorname{rad}(\g)=\lie{\{p_j,q_j,z_j,y_j\}_{j=1}^3}$, a two-step nilpotent ideal. Dashed orange edges depict interactions between the semisimple part and the radical, illustrating the Levi-Mal'tsev decomposition as a semidirect sum. The subgraph of the radical $\operatorname{rad}(\g)$ exhibits rotational symmetry, suggesting a connection to the alternating/cyclic group $A_3\cong\Z_3$.}
    \label{fig:om:example}
\end{figure}

We now apply the tools developed in this work to analyze the structural properties of the Lie algebra associated with the Hamiltonian above:
\begin{enumerate}[label = (\roman*)]
    \item \emph{Non-trivial center.} According to Lemma~\ref{lem:subalgebra:of:center}, the Lie algebra $\g$ has a non-trivial center, as the vertices labeled by the elements $\{y_j,z_j\}_{j=1}^3$ are sinkholes in the associated graph $G(V,E)$.
    \item \emph{Presence of a simple subalgebra.} The three vertices $\{a,b,c\}$ are linearly independent. Thus, the blue subgraph in Figure~\ref{fig:om:example} is a proper-minimal graph of Type~IX and the corresponding Lie algebra is according to Corollary~\ref{cor:simple:subalgebra} simple. Moreover, the Lie algebra $\spn\{a,b,c\}$ is isomorphic to $\sl{2}{\R}$~\cite{A1:project}.
    \item \emph{Non-solvability.} The Lie algebra $\g$ is non-solvable, as the associated graph $G(V,E)$ contains a closed directed walk that induces a self-contained subgraph. For instance, starting at vertex $a$, one traverses the edge labeled by $b$ to reach $c$, then follows the edge labeled by $a$ to reach $b$, and finally traverses the edge labeled by $c$ to return to $a$. By Theorem~\ref{thm:non:solvability:condition:strong}, this confirms the non-solvability of $\g$.
    \item \emph{Existence of a nilpotent ideal (the radical).} Due to Lemma~\ref{lem:ideal:span}, the Lie algebra $\g$ contains an ideal that is spanned by the vertices $\{p_j,q_j,z_{j},y_j\}_{j=1}^3$, which is highlighted in the graph $G(V,E)$ in Figure~\ref{fig:om:example} by coloring the corresponding edges and vertices green. This ideal is nilpotent, as its associated subgraph contains no closed directed walk (Theorem~\ref{thm:nilpotency:criteria:strong}). A close inspection reveals that this ideal is the radical $\operatorname{rad}(\g)$ of this Lie algebra. Furthermore, employing Algorithm~\ref{alg:generating:lower:central:series} and Lemma~\ref{lem:lower:central:series:of:graphs} reveals that $\operatorname{rad}(\g)$ is a two-step nilpotent Lie algebra.

    This observation enables once more an interpretation of the structure of the algebra via the Levi-Mal'tsev decomposition theorem: the semisimple component (indicated with the blue edges and vertices) interacts with its radical (indicated with the green edges and vertices) exclusively through brackets that result in elements from the radical (illustrated with dashed orange edges). This provided a clear graphical illustration of the semidirect sum decomposition that governs the internal structure and the interplay of the two components.  

    \item \emph{Symmetry and conjecture.}
    The graph of the radical $\operatorname{rad}(\g)$ exhibits an intrinsic symmetry: rotations by $\frac{2}{3}\pi$ combined with the corresponding permutations of the labels that leave the graph invariant. Thus, its symmetry group corresponds to that of an equilateral triangle restricted to rotations, suggesting a possible connection between this two-step nilpotent Lie algebra and the dihedral group $D_3$, the smallest non-abelian group \cite{Seiler:2008}, and its subgroup $\Z_3\cong A_3$ \cite{Fischer:2005}.
    
    \item \emph{Non-semisimplicity and non-reductiveness.} The Lie algebra $\g$ is not semisimple, which is consistent with the presence of a non-trivial radical. 
    Furthermore, by Algorithm~\ref{alg:generating:generating:the:graph:derived:altered}, the first derived graph $\mathcal{D}G(V,E)$ coincides with the original graph $G(V,E)$ (i.e., there are no loose ends), which implies by Lemma~\ref{lem:derived:series:graph:alg:valid} that the first derived algebra $\mathcal{D}\g$ coincides with the original Lie algebra $\g$. Since $\mathcal{Z}(\g)\neq\{0\}$ by item~(i), $\g$ is non-reductive.
\end{enumerate}

\section{Discussion}\label{sec:discussion}

Before concluding we aim to provide a comprehensive discussion of the work presented. Specifically, we wish to emphasize the essential results, highlight the strengths of the proposed approach, and critically examine its limitations and weaknesses. In addition, we outline possible avenues for extending the framework to mitigate these weaknesses while preserving its main advantages. Finally, we identify several interesting open questions that remain unsolved and merit further investigation.

\subsection{Considerations}
We begin by highlighting the principal strengths of the proposed framework for associating Lie algebras with labeled directed graphs that encode important structural properties of the algebra. 
First, the algorithm for generating a labeled directed graph associated with a given Lie algebra is conceptually simple and easy to implement, provided a suitable basis is available. 
Furthermore, the present framework enables the detection of key structural properties of the Lie algebras through a visual method. 
In particular, we presented necessary and sufficient criteria for solvability and nilpotency (cf. Theorems~\ref{thm:non:solvability:condition:strong} and~\ref{thm:nilpotency:criteria:strong}, respectively), which are key criteria in the classification of Lie algebras \cite{Knapp:1996}. 
Moreover, we introduced concrete and straightforward algorithms to construct graphs associated with the Lie algebras of the derived and lower central series of a given algebra (cf. Algorithms~\ref{alg:generating:generating:the:graph:derived:altered} and~\ref{alg:generating:lower:central:series}, respectively). 
We also established the equivalence between the termination of these algebraic series and the corresponding graph sequences (cf. Theorem~\ref{thm:solvable:termination:derived:series:graph} and Theorem~\ref{thm:nilpotebt:if:graph:series.Terminates}, respectively). 
In addition, this approach allows for the identification of central elements (cf. Lemma~\ref{lem:subalgebra:of:center})  and ideals (cf. Lemma~\ref{lem:ideal:span}), where Algorithm~\ref{alg:identifying:ideals} can assist in finding ideals. This framework further allows one to derive multiple conditions to recognize simple, semisimple and reductive Lie algebras (cf. Section~\ref{sec:structural:properties:simple:semisimple}).  

These results provide an intuitive way to understand the internal structure of a given Lie algebra, thereby offering significant concrete benefits by presenting a visual characterization that can be understood by experts and non-experts alike. 
Furthermore, each criterion discussed above can, in principle, be implemented through straightforward algorithms and therefore readily integrated into computer-based classification methods even in the case of Lie algebras of large dimensions.

Finally, as indicated in the case of the Schrödinger algebra in Section~\ref{sec:prominent:examples:subsec:Schroedinger:algebra}, this graph-based approach allows us to visualize and understand a fundamental structural property of Lie algebras established by the Levi-Mal'tsev theorem: any Lie algebra over a field of characteristic zero decomposes as the semidirect sum of a semisimple subalgebra and its radical \cite{Kuzmin:1977}. 
The graph representation not only highlights the structure of the semisimple component and the radical but also reflects the interaction between them induced by the semidirect sum.

\subsection{Weaknesses and limitations}
Now, let us discuss the weaknesses and limitations of the proposed approach. 
An obvious limitation concerns the range of applicability, as established by Theorem~\ref{thm:existence:of:non:graph:admissible:lie:algebras}, which demonstrates the existence of finite-dimensional non-graph-admissible Lie algebras.
Although, as shown in Section~\ref{sec:graph:admissible:subsec:notion} and Section~\ref{sec:prominent:examples}, a broad class of physically relevant Lie algebras falls within the graph-admissible category, the restriction to graph-admissible Lie algebras remains significant.
Another weakness arises from the existence of Lie algebras that are not minimal-graph-admissible but only redundant-graph-admissible, as established by Corollary~\ref{cor:existance:of:redundant:but:not:mimimak:graph:admissibel:algebras}. 
Graphs associated with redundant-graph-admissible but not minimal-graph-admissible Lie algebras, or redundant graphs in general, lack several useful properties that minimal graphs possess, which complicates both the visualization and interpretation of their structure. 
For instance, in a minimal graph, the number of vertices coincides with the dimension of the associated Lie algebra. Redundant graphs, on the other hand, require additional vertices, which makes the representation by the associated graph less transparent. 
These extra vertices correspond to elements that are no longer linearly independent, which introduces difficulties in interpreting the graph and building intuition. 
Furthermore, the rules for verifying the validity of a redundant graph, as discussed in Section~\ref{sec:valid:graphs}, are also more involved than those for minimal graphs, especially those induced by the Jacobi identity. This added complexity makes redundant graphs less practical for quick structural analysis compared to minimal graphs.

The current framework is restricted to finite-dimensional Lie algebras, leaving infinite-dimensional cases outside its scope. This limitation excludes several important classes of Lie algebras that arise in mathematical physics and functional analysis, such as the Weyl algebra \cite{Dixmier:1968}, universal enveloping algebras \cite{Hall:Lie:groups:16}, or the Lie algebra of compact operators on an infinite-dimensional Hilbert space \cite{Harpe:1972}.

In practice, this approach may, in certain cases, be difficult to implement, as it relies on finding a specific basis that satisfies the structural requirements given by equations \eqref{eqn:desired:basis} or, alternatively, constructing an overcomplete basis that satisfies equations \eqref{eqn:desired:basis:overcomplete}.
Such bases are not always easy to obtain, and not every basis is suitable for this method. 
Another significant weakness is the non-uniqueness of the associated graphs: a single Lie algebra can, in principle, admit multiple bases that correspond to distinct graphs. 
The complex Lorentz algebra $\mathfrak{so}(3;1)_\C$, as demonstrated in Section~\ref{sec:prominent:examples:subsec:Lorentz:algebra}, clearly illustrates this phenomenon. 
Conversely, the same graph may be associated with several distinct Lie algebras, as shown in Example~\ref{exa:second:example:graphs:assocaited:with:algebras}. 
While this non-uniqueness complicates the classification of Lie algebras, it can also be viewed as a strength, as it reveals structural similarities between different Lie algebras.
Moreover, this property can also be exploited to differentiate simple from non-simple but semisimple Lie algebras in the complex setting. Specifically, Proposition~\ref{prop:semisimple:lie:algebra:unconnected:subgraphs} guarantees that a semisimple but non-simple Lie algebra can be associated with a labeled directed graph $G(V,E)$ that decomposes into multiple (at least two) mutually unconnected subgraphs $G(V_j,E_j)$, where each subgraph corresponds to a simple ideal. In contrast, Proposition~\ref{prop:weak:simplicity:cond:graph:based} shows that a minimal graph associated with a simple Lie algebra can never be decomposed in such a way.

This approach is effective in highlighting certain structural properties of the Lie algebra, such as solvability or nilpotency, but it does not capture all aspects of the algebraic structure. 
In particular, the criteria derived from the graph representation are not sufficient to reliably distinguish between simple and semisimple Lie algebras when only a single graph is considered. Furthermore, we could not provide any criterion that aids in determining whether a Lie algebra corresponds to a compact or non-compact Lie group using solely the graph structure. 
For instance, the two real Lie algebras $\mathfrak{su}(2)$ and $\sl{2}{\R}$ can both be associated with the same graph (cf. Example~\ref{exa:second:example:graphs:assocaited:with:algebras}), despite the fact that their corresponding Lie groups differ fundamentally: the special unitary group $\operatorname{SU}(2)$ is compact, whereas the special linear group $\operatorname{SL}(2,\R)$ is not \cite{Woit:2017}. One possible refinement to address this limitation might be obtained by introducing an edge-coloring scheme to encode additional information, such as the sign or magnitude of the structure constants. For example, the Lie algebras $\mathfrak{su}(2)$ and $\sl{2}{\R}$ differ, with respect to the bases used here, only by a sign factor in their defining brackets (cf. equations~\eqref{eqn:basis:su:2:canonical} and~\eqref{eqn:basis:sl2R:second:version}), which could be indicated by distinct edge colors.

\subsection{Potential extensions}
The observations presented above naturally lead to the discussion of possible extensions of the proposed framework. 
One promising direction is to generalize labeled directed graphs to a labeled version of directed hypergraphs \cite{Bretto:2013}, as this allows a single edge to point towards multiple vertices. 
Such a generalization is necessary whenever one works with a basis where the adjoint action of one basis element on another yields a linear combination of several basis elements, a situation that occurs, for example, for all finite-dimensional non-graph-admissible Lie algebras and can also arise for arbitrary basis choices. 
In this setting, the corresponding edge would indicate that the Lie bracket of the element labeling the vertex from which the edge originates and the element labeling the edge itself produces a linear combination of the elements labeling all vertices to which the edge points. However, this approach has certain drawbacks. 
For instance, the weights of the individual basis elements in a linear combination are not clearly represented in the hypergraph, even though they carry more structural significance than the scalar coefficient of a single resulting element. In the current approach, where the basis is chosen so that the Lie bracket of two basis elements produces at most one basis element, the exact scalar factor is irrelevant for the construction of the associated graph. This omission is justified by the bilinearity of the Lie bracket, as any non-zero scalar multiple can be absorbed without changing the underlying structure of the graph. In contrast, the relative weights in a linear combination determine which elements appear in the bracket and how strongly they contribute, making them essential for accurately capturing the algebra's structure. Additionally, hypergraphs are generally less visually intuitive, making it harder to develop a clear understanding of the structure of the associated Lie algebra. 
Despite these challenges, this generalization is worth investigating as it could significantly broaden the scope of applicability of the framework. 

Another immediate generalization would be achieved by relaxing the requirement that the (possibly overcomplete) basis satisfying relations \eqref{eqn:desired:basis} (or \eqref{eqn:desired:basis:overcomplete} respectively) is finite and allowing, instead, also suitable countably infinite bases. This extension would enable the construction of graphs with countably infinite vertices and edges, thereby covering also certain finite but non-graph-admissible Lie algebras as well as other countably infinite-dimensional Lie algebras. 
Such generalizations would not only alleviate the difficulty of finding a suitable basis but also allow the visualization of additional structural features, such as chains of Lie brackets, which are particularly relevant in the study of derived and lower central series or the classification of finite-dimensional subalgebras of an infinite-dimensional Lie algebra \cite{Bruschi:Xuereb:2024,A1:project}. For a more detailed explanation, see Example~\ref{exa:commutator:chains}.

\begin{tcolorbox}[breakable, colback=Cerulean!3!white,colframe=Cerulean!85!black,title=\textbf{Example}: Visualizing Commutator Chains]
    \begin{example}\label{exa:commutator:chains}
        The seminal work \cite{Bruschi:Xuereb:2024} investigates whether a specific set of elements in the skew-hermitian Weyl algebra $\hat{A}_n$ generates an infinite-dimensional subalgebra. The key tool in this analysis is provided by Commutator Chains, which were later generalized in \cite{A1:project}. The idea behind these objects is to start with two seed elements and iteratively calculate nested Lie brackets. If the resulting elements are of strictly increasing degree, one can conclude that the two initial elements generate an infinite-dimensional Lie algebra, since non-zero terms of differing degree are linearly independent. Here, in Figure~\ref{fig:commutator:chain:type:II}, we illustrate a Commutator Chain of Type II by representing only edges relevant to the construction of the chain. All elements are linearly independent since they have strictly monotonically increasing degree, and therefore the chain never terminates. Thus, the corresponding graph consists of countably infinitely many vertices and edges. For details, see \cite{Bruschi:Xuereb:2024}.
        \begin{figure}[H]
            \centering
            \includegraphics[width=0.95\linewidth]{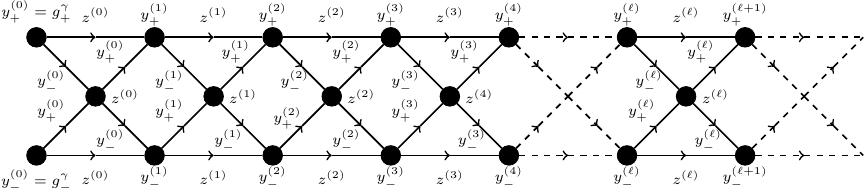}
            \caption{Commutator Chain of Type II generated by the seed elements $y_+^{(0)}:=g_+^{\gamma}$ and $y_-^{(0)}:=g_-^{\gamma}$.}
            \label{fig:commutator:chain:type:II}
        \end{figure}
    \end{example}
\end{tcolorbox}

These extensions generally preserve the strengths of the presented approach and could provide a richer and more flexible framework for representing Lie algebras using labeled directed graphs. Nevertheless, in the next Subsection~\ref{sec:subsec:magma:gradings} we want to focus on the most promising generalization, which suggests a direct connection to the concept of Lie algebra gradations.

\subsection{Generalization via magma-gradings}\label{sec:subsec:magma:gradings}

Another idea that naturally arises when attempting to generalize the graph-theoretic approach developed in this work is inspired by the concept of group gradings, which play a central role in various areas of mathematics and theoretical physics, such as representation theory \cite{Kac:1968}, supersymmetry \cite{Schimmrigk:2004}, and the theory of algebraic groups \cite{Bahturin:2014,Hornhardt:2025}. In particular, gradations provide a powerful tool for organizing the internal structure of Lie algebras \cite{Elduque:2013} and aid in classification efforts \cite{Reeder:2012,Bahturin:2014:II}. We begin by recalling the standard definition:
\begin{definition}[graded Lie algebra]
    Let $\g$ be a Lie algebra over the field $\mathbb{F}$ and let $(G,\circ)$ be an abelian group. A \emph{$G$-gradation} of $\g$ is a direct vector-space decomposition
    \begin{align*}
        \g=\bigoplus_{j\in G}\g_j,\qquad\text{such that}\qquad [\g_j,\g_k]\subseteq \g_{j\circ k}\qquad\text{for all }j,k\in G.
    \end{align*}
    A Lie algebra is called a \emph{$G$-graded Lie algebra} if such a decomposition exists for a given abelian group $G$ \cite{helgason2024differential}.
\end{definition}
It is common practice to refer to a $G$-graded Lie algebra simply as a graded Lie algebra when either the specific grading group $G$ is clear from the context or deemed unimportant. In particular, any Lie algebra can be trivially regarded as a graded Lie algebra by considering the trivial group $G=\{0\}$, where the entire algebra is concentrated in a single component. It is important to note that in the literature, the term graded Lie algebra is sometimes also used to refer to the concept of a graded Lie superalgebra \cite{Kac:1977,CATTANEO:2006,Schimmrigk:2004,Nijenhuis:1966,Nijenhuis:Richardson:1967}. While closely related, the latter is equipped with different constraints on the Lie bracket of individual spaces. One notable example is the modification of the Jacobi identity to the graded Jacobi identity \cite{Nijenhuis:Richardson:1967}.

The concept of graded Lie algebras bears a structural resemblance to that of minimal-graph-admissible Lie algebras, as these admit a finite basis $\mathcal{B}=\{e_j\}_{j\in\mathcal{J}}$, with an index set $\mathcal{J}\subseteq\N_{\geq1}$, such that the basis elements $e_j\in\mathcal{B}$ satisfy the relations \eqref{eqn:desired:basis}. Accordingly, one can decompose a minimal-graph-admissible Lie algebra $\g$ as $\g=\bigoplus_{j\in\mathcal{J}}\g_j$, where each vector-space component is defined by $\g_j:=\spn\{e_j\}$. These spaces satisfy either the bracket relation $[\g_j,\g_k]=\g_{\delta(j,k)}$ or $[\g_j,\g_k]=\{0\}$ for a symmetric function $\delta:\mathcal{J}\times\mathcal{J}\to\mathcal{J}$. For pairs $(j,k)$ satisfying $[\g_j,\g_k]=\{0\}$, the value of $\delta(j,k)$ may be chosen arbitrarily, provided that it preserves the symmetry of $\delta$.
Thus, a minimal-graph-admissible Lie algebra satisfies a generalized notion of a graded Lie algebra, where the abelian group $(G,\circ)$ is replaced by the abelian magma $(\mathcal{J},\delta)$ \cite{Serre:1992}. Such a generalization of a Lie-algebra grading is, in some contexts, also referred to as a graded Lie algebra \cite{Cordova:Martinez:2018}. Nevertheless, since group-gradings remain the most widely used and standard formulation, we adopt the more descriptive term \enquote{magma-graded} to emphasize the generalization and to avoid ambiguity. Abelian magmas are also referred to as abelian binars \cite{Bergman:2012} or abelian groupoids \cite{Brandt:1927}. The use of an abelian magma rather than an abelian group allows $\delta$ to be non-associative and not to possess an identity element or inverses. This motivates the following definition:
\begin{definition}[magma-graded-graph-admissible Lie algebra]\label{def:Mmgga:lie:algebra}
    Let $\g$ be a Lie algebra over the field $\mathbb{F}$. Then $\g$ is called an \emph{$\mathcal{M}$-magma-graded-graph-admissible ($\mathcal{M}$-mgga) Lie algebra} if there exists an abelian magma $(\mathcal{M},\delta)$ such that $\g$ admits a direct vector-space decomposition
    \begin{align*}
        \g=\bigoplus_{j\in \mathcal{M}}\g_j,\qquad\text{such that either}\qquad [\g_j,\g_k]=\g_{\delta(j,k)}\qquad\text{or}\qquad [\g_j,\g_k]=\{0\},\qquad\text{for all }j,k\in\mathcal{M},
    \end{align*}
    where $\g_j\neq\{0\}$ for all $j\in\mathcal{M}$.
    Furthermore, if there exists a magma $(\mathcal{M},\delta)$ such that $|\mathcal{M}|=\dim(\g)$ and $\g$ is an $\mathcal{M}$-mgga Lie algebra, then $\g$ is called a \emph{minimal-magma-graded-graph-admissible} (mmgga) Lie algebra.
\end{definition}

From the definition above, it follows that a finite-dimensional Lie algebra is minimal-graph-admissible if and only if it is an mmgga Lie algebra. Indeed, every minimal-graph-admissible basis $\{e_j\}_{j\in\mathcal{N}}$ induces such a decomposition by setting $\g_j:=\spn\{e_j\}$. Conversely, if $|\mathcal{M}|=\dim(\g)$, then every non-zero component $\g_j$ of the direct decomposition must be one-dimensional, and choosing one non-zero element from each component produces a minimal-graph-admissible basis. Thus, we can employ the same methods to associate a labeled directed graph with an mmgga Lie algebra.

Note also that Definition~\ref{def:Mmgga:lie:algebra} also encompasses Lie algebras that are non-graph-admissible in the sense of the original Definition~\ref{def:graph:admissible}, as illustrated in Example~\ref{exa:graph:associated:to:non:graph:admissible:lie:algebra} below. In this example, formally applying suitably modified versions of the methods developed above reproduces several structural properties of the Lie algebra. At present, however, this remains a heuristic observation and requires further investigation. In particular, the example motivates the following questions: Is every redundant-graph-admissible Lie algebra an $\mathcal{M}$-mgga Lie algebra for some finite abelian magma $(\mathcal{M},\delta)$? Does every $\mathcal{M}$-mgga Lie algebra admit a natural finest magma-gradation? The latter question is nontrivial because every perfect Lie algebra admits the coarse one-component magma-gradation $\g=\g_1$ with $[\g_1,\g_1]=\g_1$. The corresponding modified graph consists of a single vertex with a self-loop labeled by that vertex, but it does not retain the internal bracket structure of $\g$.

\begin{tcolorbox}[breakable, colback=Cerulean!3!white,colframe=Cerulean!85!black,title=\textbf{Example}: $\mathcal{M}$-magma-gradation for non-graph-admissible Lie algebras]
    \begin{example}\label{exa:graph:associated:to:non:graph:admissible:lie:algebra}
        To establish the existence of non-graph-admissible Lie algebras as formulated by Theorem~\ref{thm:existence:of:non:graph:admissible:lie:algebras}, in the corresponding proof we consider the family $L_\alpha^3$ of solvable real three-dimensional Lie algebras \cite{DeGraaf:2004} for specific non-zero choices of the real parameter $\alpha$. The Lie algebra $L_\alpha^3$ is spanned by the three elements $\{e_1,e_2,e_3\}$ that satisfy the Lie brackets:
        \begin{align*}
            [e_1,e_2]&=0,\;&\;[e_1,e_3]&=-e_2,\;&\;[e_2,e_3]&=-\alpha e_1-e_2.
        \end{align*}
        For $\alpha\neq 0$, we can now choose the subspaces $\g_1:=\spn\{e_1,e_2\}$ and $\g_2:=\spn\{e_3\}$. It is straightforward to verify that
        \begin{align*}
            [\g_1,\g_1]&=\{0\},\;&\;[\g_1,\g_2]&=\g_1,\;&\;[\g_2,\g_1]&=\g_1,\;&\;[\g_2,\g_2]&=\{0\}.
        \end{align*}
        Thus, $L_\alpha^3$ admits an $\mathcal{M}$-magma-gradation, and is in particular an $\mathcal{M}$-mgga Lie algebra, where $\mathcal{M}=\{1,2\}$ and the binary operation $\delta:\mathcal{M}\times\mathcal{M}\to\mathcal{M}$ must satisfy $\delta(1,2)=1=\delta(2,1)$, while the values of $\delta(1,1)$ and $\delta(2,2)$ may be chosen arbitrarily. The resulting labeled directed graph $G(V,E)$, obtained by formally applying the construction rules of Algorithm~\ref{alg:creating:graph} with each homogeneous subspace $\g_j$ treated as a single vertex, is depicted in Figure~\ref{fig:L:alpha:3:for:non:graph:admissible}.
        \begin{figure}[H]
            \centering
            \includegraphics[width=0.3\linewidth]{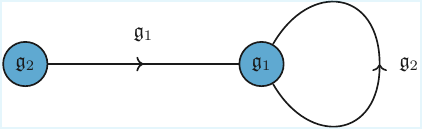}
            \caption{Modified labeled directed graph corresponding to the $\mathcal{M}$-magma-gradation of $L_\alpha^3$, where each blue vertex represents an entire homogeneous subspace.}
            \label{fig:L:alpha:3:for:non:graph:admissible}
        \end{figure}
        Note that we color the vertices blue to highlight that they represent entire homogeneous subspaces rather than individual basis elements.

        We can now analyze this graph using suitably modified versions of the central results of this work, as if it were obtained from Algorithm~\ref{alg:creating:graph} for a valid basis of a Lie algebra $\g$. We find that the analogue of Theorem~\ref{thm:non:solvability:condition:strong} would correctly identify $L_\alpha^3$ as solvable, since the graph $G(V,E)$ contains no self-contained subgraph that is induced by a closed directed walk. Similarly, an appropriately modified analogue of Theorem~\ref{thm:nilpotency:criteria:strong} would correctly imply that $L_\alpha^3$ is not nilpotent because the graph $G(V,E)$ contains a closed directed walk. Moreover, an appropriately modified version of Lemma~\ref{lem:derived:series:graph:alg:valid}, together with Algorithm~\ref{alg:generating:generating:the:graph:derived:altered}, would imply that the derived algebras of $L_\alpha^3$ are $\mathcal{D}^0L_\alpha^3=\g_1\oplus\g_2=L_\alpha^3$, $\mathcal{D}^1L_\alpha^3=\g_1$, and $\mathcal{D}^\ell L_\alpha^3=\{0\}$ for all $\ell\in\N_{\geq2}$, which is correct. Similarly, an appropriately modified version of Lemma~\ref{lem:lower:central:series:of:graphs}, together with Algorithm~\ref{alg:generating:lower:central:series}, correctly constructs the lower central series of $L_\alpha^3$, given by $\mathcal{C}^0L_\alpha^3=\g_1\oplus\g_2=L_\alpha^3$ and $\mathcal{C}^\ell L_\alpha^3=\g_1$ for all $\ell\in\N_{\geq1}$. Since this series does not terminate, $L_\alpha^3$ is correctly identified as non-nilpotent. Furthermore, an appropriately modified version of Lemma~\ref{lem:ideal:span} would correctly imply that $\g_1$ is a nonzero proper ideal of $L_\alpha^3$. Finally, the analogue of Proposition~\ref{prop:weak:simplicity:cond:graph:based} would imply that $L_\alpha^3$ is not simple, as $G(V,E)$ contains no closed directed walk that induces a subgraph $G(\Tilde{V},\Tilde{E})$ with $V=\Tilde{V}$, which is also correct.

        Note that we do not claim that these results hold in general for coarse $\mathcal{M}$-magma-gradings. Nonetheless, the example suggests that suitably modified analogues of the results considered above may remain valid in this generalized setting. Establishing their precise formulations and validity requires further investigation. 
    \end{example}
\end{tcolorbox}

\begin{remark*}
    We emphasize that this $\mathcal{M}$-magma-graded construction is not used in the proofs of the main results of this work. Rather, it should be understood as a possible extension of the present framework beyond graph-admissible Lie algebras. The precise formulation, the modified graph-construction algorithms, and graph-theoretic classification results are left for future investigation.
\end{remark*}

\subsection{Potential applications to quantum physics}\label{sec:discussion:application:to:quantum:physics}

Finally, we aim to highlight potential applications of the present framework in analyzing the dynamics of a quantum mechanical system and fostering a deeper intuition for the underlying calculations. The following approach can be especially beneficial as a first approach to better understand the subject, because associated graphs can provide an initial tool for visualizing dynamical Lie algebras~\cite{Allcock:2026,Wang:2016,Albertini:2020,DAlessandro:2020}, analyzing their solvability, and organizing the similarity transformations that arise in Lie-algebraic factorization methods~\cite{Wei:Norman:1963,Wei:Norman:1964,Charzynski:2013}.

To proceed, we consider a complex Hilbert space $\mathcal{H}$ that contains all physical states of the system. Let the initial state at time $t_0=0$ be denoted by $\ket{\psi_0}\equiv\ket{\psi(0)}\in\mathcal{H}$. The central objective in quantum dynamics is generally to determine the state $\ket{\psi(t)}$ that describes the system at a later time $t\in\R$. This is achieved by solving the Schrödinger equation
\begin{align*}
    \frac{\dd}{\dd t}\ket{\psi(t)}=-i\hat{H}(t)\ket{\psi(t)},\quad\text{with}\quad\ket{\psi(0)}=\ket{\psi_0},
\end{align*}
where $\hat{H}(t)$ denotes the Hamiltonian governing the system's evolution \cite{Cohen:Tannoudji:2020}.
This equation can be reformulated by introducing the unitary time-evolution operator $\hat{U}(t)$, with the initial time fixed at $t_0=0$, defined by $\hat{U}(t)\ket{\psi_0}=\ket{\psi(t)}$. The operator $\hat{U}(t)$ is a unitary propagator satisfying
\begin{align}
    \der{t}\hat{U}(t)=-i\hat{H}(t)\hat{U}(t)\quad\text{with}\quad \hat{U}(0)=\hat{\mathds{1}},\label{eqn:schroedinger:one:parameter:group}
\end{align} 
where $\hat{\mathds{1}}$ denotes the identity operator.
In most settings, the Hamiltonian $\hat{H}(t)$ admits a decomposition of the form
\begin{align*}
    \hat{H}(t)=-i\sum_{j\in\mathcal{J}}u_j(t)\hat{g}_j,
\end{align*}
where $u_j(t)$ are real-valued (potentially) time-dependent functions and $\hat{g}_j$ are time-independent skew-hermitian generators belonging to a particular Lie algebra. This representation  therefore emphasizes the deep connection between quantum dynamics and Lie-algebraic structures. 
The differential equation \eqref{eqn:schroedinger:one:parameter:group} can, in principle, be solved using formal expansions such as the Magnus series \cite{Blanes:2009} or the time-ordered Dyson series \cite{Dyson:1949}. However, these solutions are often not conducive to obtaining an exact closed-form expression for $\ket{\psi(t)}=\hat{U}(t)\ket{\psi_0}$, due to the lack of analytical control. Consequently, one often resorts to truncating these expansions to obtain approximate solutions \cite{Blanes:2009,Kirchhoff:2025}. Depending on the Hamiltonian and the relevant operator assumptions, the convergence and practical error control of such expansions may be restricted to finite time intervals~\cite{Mariani:1961,Wei:1963}.
An alternative approach, pioneered by Wei and Norman \cite{Wei:Norman:1963,Wei:Norman:1964}, offers a constructive method for solving the time-dependent Schrödinger equation by means of factorizing the time-evolution operator into exponentials of elements of the dynamical Lie algebra. This technique provides a local Wei-Norman factorization whenever the Hamiltonian Lie algebra $\g:=\lie{\{\hat{g}_j\}_{j\in\mathcal{J}}}$ is finite dimensional, and, importantly, yields, after choosing a suitable basis and under an appropriate  ordering of the corresponding basis elements, a global solution when the Lie algebra $\g$ is solvable \cite{Wei:Norman:1964}. The core idea is to select a basis $\mathcal{B}=\{\hat{g}_k\}_{k\in\mathcal{K}}$ of $\g$ and seek a factorization of the form
\begin{align*}
    \hat{U}(t)=\prod_{k\in\mathcal{K}}\hat{U}_k(t)\qquad\text{with}\qquad\hat{U}_k(t)=\exp(-f_k(t)\hat{g}_k),
\end{align*}
where the real-valued time-dependent functions $f_k(t)$ encode the dynamics. To determine these functions, one differentiates the factorized form of $\hat{U}(t)$, uses the unitarity of the factors, and compares it with the differential equation \eqref{eqn:schroedinger:one:parameter:group}, obtaining:
\begin{align*}
    \left(-\frac{\dd}{\dd t} \hat{U}(t)\right)\hat{U}^\dagger(t)&=\sum_{k=1}^m\left(\frac{\dd}{\dd t}f_k(t)\right)\left(\prod_{j=1}^{k-1}\hat{U}_j(t)\right)\hat{g}_k\left(\prod_{j=1}^{k-1}\hat{U}_{k-j}^\dagger(t)\right)=i\hat{H}(t)=\sum_{j=1}^n u_j(t)\hat{g}_j,
\end{align*}
where, without loss of generality, we have set
$\mathcal{K}=\{1,\ldots,m\}$ and
$\mathcal{J}=\{1,\ldots,n\}$, and have removed any linear redundancies
among the generators $\{\hat{g}_j\}_{j\in\mathcal{J}}$, so that
$m=\dim(\g)\geq n$. Thus, the functions $f_k(t)$ are determined by a system of coupled differential equations, whose structure depends on the similarity transformations $\hat{U}_j(t)\hat{g}_k\hat{U}_j^\dagger(t)$.

To illustrate this approach, consider the bosonic Hamiltonian of the displaced harmonic oscillator
\begin{align*}
    \hat{H}(t)=\omega \hat{a}^\dagger\hat{a}+\alpha(t)\hat{a}+\alpha^*(t)\hat{a}^\dagger=\omega\hat{a}^\dagger \hat{a}+\re{\{\alpha(t)\}}(\hat{a}+\hat{a}^\dagger)+i\im{\{\alpha(t)\}}(\hat{a}-\hat{a}^\dagger),
\end{align*}
obtained via the displacement operation introduced in \cite{Cahill:1969}. Here, $\hat{a}^\dagger$ and $\hat{a}$ denote the bosonic creation and annihilation operators respectively, which satisfy the canonical commutation relation $[\hat{a},\hat{a}^\dagger]=1$. For this system, we can choose the generating set of operators as $\mathcal{G}:=\{i\hat{a}^\dagger \hat{a},i(\hat{a}+\hat{a}^\dagger),\hat{a}-\hat{a}^\dagger\}$ and a basis of $\g:=\lie{\mathcal{G}}$ is clearly given by $\mathcal{B}=\{i\hat{a}^\dagger \hat{a},i(\hat{a}+\hat{a}^\dagger),\hat{a}-\hat{a}^\dagger,i\}$. This Lie algebra is known as the Wigner-Heisenberg algebra $\mathfrak{wh}_2$ \cite{A1:project}. Moreover, the basis $\mathcal{B}$ satisfies the structural relations \eqref{eqn:desired:basis}, confirming that the Lie algebra $\mathfrak{wh}_2$ is minimal-graph-admissible.

The framework developed in this work offers three direct applications in this setting.
\begin{enumerate}[label=\textbf{\arabic*.}]
    \item \textbf{Visualization of the Lie algebra.} Constructing the associated graph provides immediate insights into the structure of the Lie algebra and offers an intuitive starting point for understanding its properties. For example, for the Wigner-Heisenberg algebra $\mathfrak{wh}_2$ with the basis $\mathcal{B}$, the corresponding graph clearly encodes the commutation relations among these generators. Such a visualization is particularly useful for identifying subalgebras, connectivity patterns, and structural features. The corresponding graph is illustrated in Figure~\ref{fig:wigner:heisenberg:just}.
    \begin{figure}[H]
        \centering
        \includegraphics[width=0.35\linewidth]{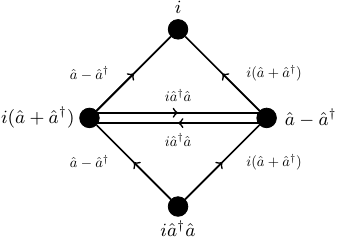}
        \caption{Depiction of the graph associated with the Wigner-Heisenberg algebra $\mathfrak{wh}_2$, constructed using the basis $\mathcal{B}=\{i\hat{a}^\dagger \hat{a},i(\hat{a}+\hat{a}^\dagger),\hat{a}-\hat{a}^\dagger,i\}$ and Algorithm~\ref{alg:creating:graph}. Here, we used, for visual simplicity, the convention to depict vertices by black circles instead of writing the name of the vertex inside a circle that represents it, as would be the correct convention for this graph.}
        \label{fig:wigner:heisenberg:just}
    \end{figure}
    \item \textbf{Solvability analysis.} If the dynamical Lie algebra $\lie{\{\hat{g}_j\}_{j\in\mathcal{J}}}$ is graph-admissible, one can construct an associated graph and employ Theorem~\ref{thm:non:solvability:condition:strong} to decide whether the algebra is solvable or not. This, in turn, allows us to infer whether a global Wei-Norman factorization of the time-evolution operator can be obtained after a suitable choice and ordering of basis elements. Returning to the example of the displaced harmonic oscillator and the Lie algebra $\mathfrak{wh}_2$, we observe that the graph depicted in Figure~\ref{fig:wigner:heisenberg:just} contains no closed directed walks that induce self-contained subgraphs. Thus, the Lie algebra $\mathfrak{wh}_2$ is solvable and a global Wei-Norman factorization of the time evolution can be obtained after a suitable choice and ordering of the basis elements.
    \item \textbf{Similarity transformations.} When calculating the time-evolution operator $\hat{U}(t)$ using the Wei-Norman method, one needs to determine the real time-dependent function $f_k(t)$, whose determination requires computing multiple similarity relations. As explained in the literature \cite{Adesso:Ragy:2014}, the following identity holds:
    \begin{align*}
        \exp(-s_j\hat{g}_j)\hat{g}_k\exp(s_j\hat{g}_j)&=\sum_{n=0}^\infty\frac{(-s_j)^n}{n!}\operatorname{ad}_{\hat{g}_j}^n(\hat{g}_k)=\exp(-s_j\,\operatorname{ad}_{\hat{g}_j})(\hat{g}_k),
    \end{align*}
    where $\operatorname{ad}_{x}^{n+1}(y)=[x,\operatorname{ad}_x^n(y)]$, and one adopts the convention $\operatorname{ad}_{\hat{g}_j}^0(\hat{g}_k)\equiv\hat{g}_k$. Note that, for the identity above, we implicitly assume that the dynamical Lie algebra $\g$ is realized in an associative operator algebra or a suitable representation in which the relevant products and exponentials are well defined. For Weyl-algebra realizations (cf.~\cite{Dixmier:1968}) and other unbounded operator realizations, this identity may instead be understood in a suitable completion, on a common invariant domain, or formally when only the algebraic commutation relations are used.
    
    The associated graph provides a direct intuitive interpretation of this relation: the similarity transformation $\exp(-\hat{g}_j)\hat{g}_k\exp(\hat{g}_j)$ can be visualized as a possibly infinite walk through the graph. Starting at the vertex corresponding to $\hat{g}_k$, one iteratively traverses edges labeled by $\hat{g}_j$, reaching each time the resulting vertex corresponding to the next nested commutator. The contribution of each step in the walk is weighted by the constants $\alpha_{jk}$. To retain this information, we equip the associated graph $G(V,E)$ with the weight function $\omega:E\to \mathbb{F}$ with $\mathbb{F}\in\{\R,\C\}$ defined by the concrete action
    \begin{align*}
        \omega:e\mapsto\omega(e):=\alpha_{jk},\qquad\text{where}\qquad[\varpi_\mathrm{s}(e),\varpi_\mathrm{l}(e)]=\alpha_{jk}\varpi_\mathrm{e}(e).
    \end{align*}
    The resulting weighted walk encodes both the successive output directions and the coefficients of the nested commutators and therefore allows one to compute the similarity transformation $\exp(-sv_j)v_k\exp(sv_j)$. Let us formalize this with the following lemma.

    \begin{lemma}\label{lem:walk:hadamard:relation}
        Let $\g$ be a finite-dimensional graph-admissible Lie algebra over $\mathbb{F}\in\{\R,\C\}$, let $\mathcal{B}=\{v_j\}_{j\in\mathcal{J}}$ be an admissible, possibly overcomplete, spanning set, and let $G(V,E)$ be the corresponding associated labeled directed graph equipped with the weight function defined above. Fix two vertices $v_j,v_k\in V$, and set $w_0:=v_k$. Starting from $w_0$, recursively proceed as follows. If, for some $p\in\Np{}$, there exists an edge $e\in E$ satisfying both $\varpi_\mathrm{s}(e)=w_p$ and $\varpi_\mathrm{l}(e)=v_j$, let $e_p\equiv e$ denote this unique edge and define  $w_{p+1}:=\varpi_\mathrm{e}(e_p)$. If, at some first index $L\in\Np{}$, no such edge exists, the recursion terminates at $w_L$, and we set $\mathcal{L}:=\{0,1,\ldots,L\}$. If such an edge exists at every step, we set $\mathcal{L}:=\Np{}$. For every $\ell\in\mathcal{L}$, define the corresponding directed walk $W_\ell^{(j,k)}:=(w_0,e_0,w_1,e_1,\ldots,e_{\ell-1},w_\ell)$, where $W_0^{(j,k)}:=(w_0)$.
    
        Then, for every $\ell\in\mathcal{L}$, one has
        \begin{align}\label{eqn:iterated:adjoint:action}
            \operatorname{ad}_{v_j}^{\ell}(v_k)=(-1)^\ell\left(\prod_{q=0}^{\ell-1}\omega(e_q)\right)w_\ell,
        \end{align}
        where the empty product is understood as one. Consequently, for every
        $s\in\mathbb{F}$, one has
        \begin{align*}
            \exp\left(-s\,\operatorname{ad}_{v_j}\right)(v_k)=\sum_{\ell\in\mathcal{L}}\frac{s^\ell}{\ell!}\left(\prod_{q=0}^{\ell-1}\omega(e_q)\right)w_\ell.
        \end{align*}
    \end{lemma}

    \begin{proof}
        First observe that by Proposition~\ref{prop:conditions:for:edges}, the walks $W_\ell^{(j,k)}$ are uniquely defined because the construction above allows only for unique edges $e_p$. Now, let us start to prove the formula for the iterated adjoint action~\eqref{eqn:iterated:adjoint:action} by induction over $\ell$. The base case $\ell=0$ follows from definition because $\operatorname{ad}_{v_j}^0(v_k)=v_k=w_0=(-1)^0 w_0$. Thus, we proceed with the induction step and assume that the formula~\eqref{eqn:iterated:adjoint:action} holds for some $\ell \in\Np{}$ and that the edge $e_\ell$ exists. By definition of the weight function, one has $[w_\ell, v_j]=\omega(e_\ell)w_{\ell+1}$. The antisymmetry of the Lie bracket therefore implies that $[v_j,w_\ell]=-\omega(e_\ell) w_{\ell+1}$. Hence,
        \begin{align*}
            \operatorname{ad}_{v_j}^{\ell+1}(v_k)&=\left[v_j,\operatorname{ad}_{v_j}^\ell(v_k)\right]=(-1)^\ell\left(\prod_{q=0}^{\ell-1}\omega(e_q)\right)[v_j,w_\ell]=(-1)^{\ell+1}\left(\prod_{q=0}^{\ell}\omega(e_q)\right)w_{\ell+1},
        \end{align*}
        which completes the induction step and therefore proves the validity of the induction hypothesis for all applicable $\ell\in\Np{}$. 

        Since $\g$ is assumed to be finite-dimensional, $\operatorname{ad}_{v_j}$ is a finite-dimensional endomorphism of $\g$~\cite{Hall:Lie:groups:16}. Its exponential is therefore defined for every $s\in\mathbb{F}$ by the convergent series
        \begin{align*}
            \exp\left(-s\operatorname{ad}_{v_j}\right)(v_k)=\sum_{\ell=0}^\infty \frac{(-s)^\ell}{\ell!}\operatorname{ad}_{v_j}^\ell(v_k).
        \end{align*}
        Substituting the formula~\eqref{eqn:iterated:adjoint:action} for non-terminating sequences of non-terminating directed walks $W_\ell^{(j,k)}$ obtained above yields
        \begin{align*}
            \exp\left(-s\operatorname{ad}_{v_j}\right)(v_k)=\sum_{\ell=0}^\infty\frac{s^\ell}{\ell!}\left(\prod_{q=0}^{\ell-1}\omega(e_q)\right) w_\ell.
        \end{align*}
        However, if a sequence of walks $\{W_\ell^{(j,k)}\}_{\ell\in \mathcal{L}}$ terminates with a walk that ends at vertex $w_L$, this implies by the definition of the associated graph that $[v_j,w_L]=0$. Consequently $\operatorname{ad}_{v_j}^{\ell}(v_k)=0$ for all $\ell>L$. Hence,
        \begin{align*}
            \exp\left(-s\operatorname{ad}_{v_j}\right)(v_k)=\sum_{\ell=0}^L \frac{(-s)^\ell}{\ell!}\operatorname{ad}_{v_j}^\ell(v_k)=\sum_{\ell=0}^L\frac{s^\ell}{\ell!}\left(\prod_{q=0}^{\ell-1}\omega(e_q)\right) w_\ell,
        \end{align*}
        as claimed.
    \end{proof}

    \begin{corollary}\label{cor:qm:example:loop:or:termination}
        Under the assumptions of Lemma~\ref{lem:walk:hadamard:relation}, fix
        $v_j,v_k\in V$ and consider the recursively defined walk $W^{(j,k)}:=(w_0,e_0,w_1,e_1,\ldots)$ defined there. Then exactly one of the following alternatives occurs:
        \begin{enumerate}[label=(\roman*)]
            \item There exists some $L\in\Np{}$ such that no edge starting at $w_L$
            and labeled by $v_j$ exists. In this case, the recursion terminates
            at $w_L$.
    
            \item The recursion does not terminate and defines an infinite
            directed walk $W_{\infty}^{(j,k)}:=(w_0,e_0,w_1,e_1,\ldots)$ that is eventually periodic. 
        \end{enumerate}
        If $\g$ is nilpotent, then only the first alternative can occur.
    \end{corollary}

    \begin{proof}
        For a fixed starting vertex $w_\ell$ and a fixed labeling vertex $v_j$, the corresponding end vertex $\varpi_\mathrm{e}(e_\ell)$ is, by construction, uniquely determined whenever the bracket $[v_j,w_\ell]$ does not vanish. Thus, the continuation of the sequence of walks $W_{\ell}^{(j,k)}$ is uniquely determined by the last vertex in the respective walk $W_{\ell}^{(j,k)}$ and the labeling vertex $v_j$. Since $V$ is finite, every sufficiently long directed walk $W_\ell^{(j,k)}$ must visit at least one vertex more than once. Hence, there must exist indices $p,q$ with $p<q$ such that $w_p=w_q$. Since the continuation from a fixed vertex under the labeling vertex $v_j$ is unique, one has $e_p=e_q$, $w_{p+1}=w_{q+1}$, and the same argument applies recursively to every subsequent step. The walk is therefore periodic from $w_p$ onward.

        Suppose now that $\g$ is nilpotent. By Theorem~\ref{thm:nilpotency:criteria:strong}, the associated graph $G(V,E)$ contains no closed directed walk. If the recursion did not terminate, the preceding argument would yield indices $p<q$ with $w_p=w_q$. The finite segment
        \begin{align*}
            (w_p,e_p,w_{p+1},e_{p+1},\ldots,e_{q-1},w_q=w_p)
        \end{align*}
        would then be a closed directed walk in $G(V,E)$, contradicting Theorem~\ref{thm:nilpotency:criteria:strong}. Therefore, the recursion terminates after finitely many steps.
    \end{proof}

    The preceding lemma and corollary are exact finite-dimensional statements, which show that the graphical description of dynamical Lie algebras does not merely provide a qualitative analogy for understanding similarity transformations: the weighted walk constructed above reproduces the relevant nested commutators and their coefficients exactly.  Moreover, the graphical interpretation can also aid in bounding the error of a truncated adjoint expansion. Let
    \begin{align*}
        S_T(s):=\sum_{\ell=0}^T\frac{s^\ell}{\ell!}\left(\prod_{q=0}^{\ell-1}\omega(e_q)\right)w_\ell,
    \end{align*}
    for any truncation order $T\in\mathcal{L}$ for which the vertices $w_\ell$ and edges $e_q$ are defined, and let $\|\cdot\|:\g\to\R_{\geq0}$ be an arbitrary norm on $\g$ with induced operator norm $\|\cdot\|_\mathrm{op}$ on $\operatorname{End}(\g)$. Then
    \begin{align*}
        \left\|\exp\left(-s\operatorname{ad}_{v_j}\right)(v_k)-S_T(s)\right\|&\leq \sum_{\ell=T+1}^\infty\frac{\abs{s}^\ell}{\ell!}\left\|\operatorname{ad}_{v_j}^\ell(v_k)\right\|\leq \|v_k\|\sum_{\ell=T+1}^\infty\frac{\left(\abs{s}\left\|\operatorname{ad}_{v_j}\right\|_\mathrm{op}\right)^\ell}{\ell!}\\
        &\leq \|v_k\|e^{\abs{s}\left\|\operatorname{ad}_{v_j}\right\|_\mathrm{op}}\frac{\left(\abs{s}\left\|\operatorname{ad}_{v_j}\right\|_\mathrm{op}\right)^{T+1}}{(T+1)!},
    \end{align*}
    by the definition of the operator norm~\cite{Werner:2018} and Suzuki's theorem~\cite{Higham:2008}. Thus, in the finite-dimensional setting, the weighted-walk expansion may be truncated with a controlled remainder. Therefore, this graphical interpretation goes beyond the development of an intuition for such calculations by providing an intuitive framework for understanding the structure of these approximations.

    By contrast, the extension to infinite-dimensional Lie algebras and infinite associated graphs should be understood as an interpretative observation rather than as a consequence of the results above. In an infinite-dimensional operator setting, additional assumptions concerning, for instance, the existence of the relevant operator exponentials and the convergence of the adjoint series are required. Nevertheless, the graphical picture already highlights these subtleties: in this setting, a non-repeating and non-terminating walk illustrates the challenges that may arise when analyzing the convergence behavior of the corresponding series expansion.
\end{enumerate}

\begin{tcolorbox}[breakable, colback=Cerulean!3!white,colframe=Cerulean!85!black,title=\textbf{Example}: Obtaining similarity relations using walks in the associated graphs]
    \begin{example}\label{example:wigner:heisenberg:walks}
        We now demonstrate the results above in the case of the Wigner-Heisenberg algebra $\mathfrak{wh}_2$. We use again the basis $\mathcal{B}=\{i\hat{a}^\dagger\hat{a},i(\hat{a}+\hat{a}^\dagger),\hat{a}-\hat{a}^\dagger,i\}$, and for completeness, display the corresponding edge weights as secondary labels. Thus, $\varpi_\mathrm{l}(e)$ continues to denote the labeling vertex of an edge $e$, whereas $\omega(e)$ denotes its weight. Accordingly, in the notation used below, an edge displayed as $(x,(y,\lambda),z)$ represents the edge $e=(x,y,z)$ with $\omega(e)=\lambda$. This results in the graph depicted in Figure~\ref{fig:wigner:heisenberg:wlaks:similarity}.
        
        Following the approach described above, we define the walk $W=(v_0,e_0,v_1,e_1,v_2,e_2,\ldots)$, defined by:
        \begin{align*}
            v_{2\ell}:=\hat{a}-\hat{a}^\dagger,\quad e_{2\ell}:=(\hat{a}-\hat{a}^\dagger,(i\hat{a}^\dagger \hat{a},1),i(\hat{a}+\hat{a}^\dagger)),\quad v_{2\ell+1}:= i(\hat{a}+\hat{a}^\dagger),\quad e_{2\ell+1}:= (i(\hat{a}+\hat{a}^\dagger),(i\hat{a}^\dagger \hat{a},-1),\hat{a}-\hat{a}^\dagger),
        \end{align*}
        for all $\ell\in\N_{\geq0}$. This walk is associated with the similarity transformation $\exp(-i\hat{a}^\dagger\hat{a})(\hat{a}-\hat{a}^\dagger)\exp(i\hat{a}^\dagger\hat{a})$. An illustration of the corresponding walk can be found in Figure~\ref{fig:wigner:heisenberg:wlaks:similarity}. For the edge weights, we note that $\omega(e_{2\ell})=1=-\omega(e_{2\ell+1})$ for all $\ell\in\N_{\geq0}$, and consequently
        \begin{align}
            \prod_{j=0}^{n-1}\omega(e_j)=(-1)^{\lfloor\frac{n}{2}\rfloor}.\label{eqn:wigner:heeisenberg:weight:formula}
        \end{align}
        This allows us to compute the following similarity transformation, where, for simplicity, we set $s=1$.
        \begin{align*}
            \exp(-i\hat{a}^\dagger\hat{a})(\hat{a}-\hat{a}^\dagger)\exp(i\hat{a}^\dagger\hat{a})&=\sum_{n=0}^\infty\frac{1}{n!}\left(\prod_{j=0}^{n-1}\omega(e_j)\right)v_n.
        \end{align*}
        Splitting the series into terms of even and odd orders yields together with the application of \eqref{eqn:wigner:heeisenberg:weight:formula}:
        \begin{align*}
            \exp(-i\hat{a}^\dagger\hat{a})(\hat{a}-\hat{a}^\dagger)\exp(i\hat{a}^\dagger\hat{a})&=\sum_{n=0}^\infty\frac{1}{(2n)!}(-1)^{\lfloor\frac{2n}{2}\rfloor}v_{2n}+\sum_{n=0}^\infty\frac{1}{(2n+1)!}(-1)^{\lfloor\frac{2n+1}{2}\rfloor}v_{2n+1}\\
            &=\sum_{n=0}^\infty\frac{(-1)^n}{(2n)!}v_0+\sum_{n=0}^\infty\frac{(-1)^{n}}{(2n+1)!}v_1.
        \end{align*}
        Substituting $v_0=\hat{a}-\hat{a}^\dagger$ and $v_1=i(\hat{a}+\hat{a}^\dagger)$, we obtain:
        \begin{align*}
            \exp(-i\hat{a}^\dagger\hat{a})(\hat{a}-\hat{a}^\dagger)\exp(i\hat{a}^\dagger\hat{a})&=\cos(1)\left(\hat{a}-\hat{a}^\dagger\right)+i\sin(1)\left(\hat{a}+\hat{a}^\dagger\right)=e^{i}\hat{a}-e^{-i}\hat{a}^\dagger.
        \end{align*}
        This matches the expected result, as it is in agreement with the literature \cite{A1:project}. This example illustrates a walk that loops the same sequence indefinitely, as predicted by Corollary~\ref{cor:qm:example:loop:or:termination}, which can be treated by splitting the series into the terms that have a common vertex label. 
        
        To contrast the previous example, we consider the similarity transformation $\exp(-i(\hat{a}+\hat{a}^\dagger))(i\hat{a}^\dagger\hat{a})\exp(i(\hat{a}+\hat{a}^\dagger))$. Here, the corresponding walk $W'$ starts at the vertex $v_0=i\hat{a}^\dagger\hat{a}$. Following the edge labeled with $i(\hat{a}+\hat{a}^\dagger)$, one reaches the vertex labeled by $\hat{a}-\hat{a}^\dagger$. Continuing along the next edge, also labeled by $i(\hat{a}+\hat{a}^\dagger)$, one arrives at the vertex labeled with $i$, where the walk terminates because no further edge originates from this vertex, especially no edge labeled by $i(\hat{a}+\hat{a}^\dagger)$. This yields by Lemma~\ref{lem:walk:hadamard:relation} for $s=1$:
        \begin{align*}
            \exp(-i(\hat{a}+\hat{a}^\dagger))(i\hat{a}^\dagger\hat{a})\exp(i(\hat{a}+\hat{a}^\dagger))&=i\hat{a}^\dagger\hat{a}+\omega(e_0) \left(\hat{a}-\hat{a}^\dagger\right)+\frac{1}{2}\omega(e_0)\omega(e_1)i,
        \end{align*}
        where $\omega(e_0)=1$ and $\omega(e_1)=2$. Substituting these values yields:
        \begin{align*}
            \exp(-i(\hat{a}+\hat{a}^\dagger))(i\hat{a}^\dagger\hat{a})\exp(i(\hat{a}+\hat{a}^\dagger))&=i\hat{a}^\dagger\hat{a}+\left(\hat{a}-\hat{a}^\dagger\right)+i.
        \end{align*}
        The corresponding walk is also depicted in Figure~\ref{fig:wigner:heisenberg:wlaks:similarity}.
        \begin{figure}[H]
            \centering
            \includegraphics[width=0.45\linewidth]{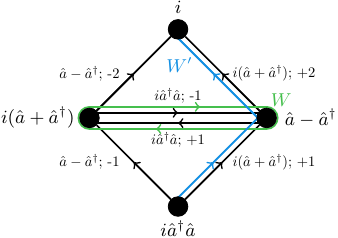}
            \caption{Depiction of the graph associated with the Wigner-Heisenberg algebra $\mathfrak{wh}_2$ from Figure~\ref{fig:wigner:heisenberg:just}, now including the edge weights as secondary labels. The two walks $W$ (shown in green) and $W'$ (shown in blue) from Example~\ref{example:wigner:heisenberg:walks} are also illustrated.}
            \label{fig:wigner:heisenberg:wlaks:similarity}
        \end{figure}
    \end{example}
\end{tcolorbox}

\subsection{Further remarks}

Finally, we want to devote this last section to a broader reflection and prospective developments stemming from our study. Beyond the specific results established in the preceding sections, the framework of associating labeled directed graphs with Lie algebras opens several avenues for further exploration. These include both theoretical questions and potential applications, particularly in the classification of Lie algebras and their structural invariants. We now present a selection of open problems and conjectures that naturally arise from our findings.

\subsubsection{Open questions and conjectures}

We begin by briefly discussing several open questions and potential directions for future research. These are presented in a logical sequence that mirrors the progression of the main text.

We started by defining and investigating the notion of graph-admissibility with particular emphasis on the subclass of minimal-graph-admissible Lie algebras. Among the general results established, which include the existence Theorems~\ref{thm:existncae:of:non:minimal:graph:admissible:} and~\ref{thm:existence:of:non:graph:admissible:lie:algebras}, we showed that every finite-dimensional nilpotent Lie algebra is graph-admissible. This naturally raises the following questions:
\begin{quote}
    \textbf{Q:} \emph{Is every nilpotent Lie algebra minimal-graph-admissible? If not, what are the structural differences between those that are minimal-graph-admissible and those that are not?}
\end{quote}
To approach this question, let us consider the subclass of two-step nilpotent Lie algebras, i.e., Lie algebras $\g$ that satisfy $[\g,\g]\neq\{0\}$ and $[\g,[\g,\g]]=\{0\}$. Engel's Theorem \cite{Humphreys:1972} implies that $\operatorname{ad}_x^2=0$ for all $x\in\g$. Generally, one has even $\operatorname{ad}_x\circ\operatorname{ad}_y=\operatorname{ad}_y\circ\operatorname{ad}_{x}=0$ for all $x,y\in\g$. This is because $\operatorname{ad}_x(z)\in[\g,\g]\subseteq\mathcal{Z}(\g)$ for all $x,z\in\g$, and consequently $(\operatorname{ad}_x\circ\operatorname{ad}_y)(z)=0$ for all $x,y,z\in\g$. Henceforth, the adjoint actions of any pair of basis elements commute.

Motivated by the ability to construct a similarity transformation that brings a set of commuting normal matrices into their diagonalized form, one can conclude the following:
if it were possible to construct a simultaneous similarity transformation that brings all adjoint operations into their respective Jordan normal forms, it would follow that every two-step nilpotent Lie algebra is minimal-graph-admissible. In fact, for certain pairs of commuting nilpotent endomorphisms, the existence of such a similarity transformation has already been established \cite{Hua:2023}. However, in the case of three-step nilpotent Lie algebras, it is generally no longer true that all adjoint actions commute. This observation motivates the following conjecture:

\begin{conjecture}
    Let $\g$ be a finite-dimensional $k$-step nilpotent Lie algebra. Then, the following holds: (a) If $k\leq 2$, then $\g$ is minimal-graph-admissible; (b) If $k>2$, then $\g$ is not necessarily minimal-graph-admissible. In fact, there exist $k$-step nilpotent Lie algebras that fail to be minimal-graph-admissible.
\end{conjecture}

The focus on nilpotent Lie algebras further connects this work to ongoing efforts aimed at investigating such algebras through alternative graph-based approaches, as mentioned in the introduction (see for example, \cite{Molina:2023,Dani:Mainkar:2005,Grantcharov:2017,Chakrabarti:2019}). A key difference between our framework and that in \cite{Molina:2023} lies in the interpretation of the Lie bracket. In \cite{Molina:2023}, the bracket $[x_1,x_2]=x_3$ is represented by an edge originating from $x_1$, pointing towards $x_2$ and being labeled by $x_3$. This encoding captures the sign of the bracket through the edge's direction, since by the antisymmetry of the Lie bracket $[x_2,x_1]=-x_3$. However, this approach enforces an even more restrictive condition on the basis elements, compared to the relations~\eqref{eqn:desired:basis} or~\eqref{eqn:desired:basis:overcomplete}. Specifically, the Lie bracket of two basis elements can only yield one of three outcomes: either a single basis element, its additive inverse, or zero. This limitation excludes certain Lie algebras, for instance real Lie algebras that inherently involve irrational coefficients. One notable example is the Lie algebra $\mathfrak{N}(q)$, first introduced in \cite{Scheuneman:1967}. This real Lie algebra is spanned by the six basis elements $\{x_1,x_2,x_3,x_4,y_1,y_2\}$, which satisfy the following nontrivial Lie bracket relations:
\begin{align*}
    [x_1,x_2]&=y_1,\;&\;[x_1,x_3]&=y_2,\;&\;[x_2,x_4]&=y_2,\;&\;[x_3,x_4]&=qy_1,
\end{align*}
where $q\in\R\setminus\Q$. All other brackets either vanish or follow from the antisymmetric property of the Lie bracket. Using the current framework, the associated graph for $\mathfrak{N}(q)$, according to our prescription, is given by:
\begin{figure}[H]
    \centering
    \includegraphics[width=0.4\linewidth]{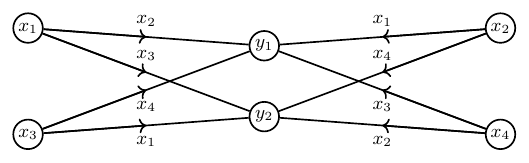}
\end{figure}
These two approaches, the one presented in this work and the one discussed in \cite{Molina:2023}, share a notable feature: a single Lie algebra can, in general, be represented by multiple distinct graphs, which is demonstrated by Example 2.3 in \cite{Molina:2023}. To further highlight the similarities and differences between the different graph-based approaches, we compare the present framework with the nilpotent constructions discussed in~\cite{Molina:2023,Chakrabarti:2019} and the related solvable construction from~\cite{Grantcharov:2017} in Table~\ref{tab:comparison:related:graph:constructions} of Appendix~\ref{app:comapring:graph:Lie:algebra:frameworks}. In particular, the table compares the direction of construction, the classes of Lie algebras covered, the interpretation of the vertices and edges, the quantitative relations between the graph and the Lie algebra, and the corresponding uniqueness properties.

Another interesting research direction emerges when considering Theorem~\ref{thm:non:solvability:condition:strong}, which provides sufficient and necessary criteria for solvability by analyzing closed directed walks within the associated graph. The only additional requirement imposed is that the subgraph induced by such a walk has to be self-contained. It might therefore be worthwhile to improve the efficiency of this criterion, which is especially relevant for computer-based implementations, by further restricting the set of admissible walks. For instance, one could aim to establish conditions on the maximum length of such walks or determine how frequently a vertex must be revisited. These refinements would significantly reduce the number of walks that need to be considered, making the criterion computationally more tractable. Inspiration for such an investigation can certainly be drawn from the field of computer science, where graph-based methods are widely adopted \cite{Deo:2016}. One idea could be to incorporate Dijkstra's algorithm to determine the shortest path between two vertices \cite{Dijkstra:1959}. Such techniques might help identify minimal or optimal walks relevant to solvability.

Currently, the number of possible walks can be infinite, since a closed directed walk can be repeated arbitrarily many times, which poses practical challenges. Restricting this to a finite set would greatly simplify verification, as infinite loops do not contribute meaningfully to the analysis. This motivates the following question:
\begin{quote}
    \textbf{Q:} \emph{Is Theorem~\ref{thm:non:solvability:condition:strong} still valid if one replaces \enquote{closed directed walk} with \enquote{closed directed trail}? If not, what are the structural differences in the associated Lie algebra?}
\end{quote}

As emphasized earlier, the current framework does not guarantee uniqueness of the graphs associated with a given Lie algebra, and vice versa. This observation opens several intriguing questions. For instance, what is the relationship between Lie algebras that share an associated graph, particularly when considering minimal graphs? Conversely, what is the relationship between different graphs associated with the same Lie algebra? Furthermore, one could explore connections between Lie algebras of different dimensions that share the same graph, i.e., in cases where at least one graph is not minimal. Similarly, it would be interesting to investigate the interplay between minimal and redundant graphs associated with the same Lie algebra, and to examine how such relationships extend to other Lie algebras which admit these redundant graphs as minimal ones. These, as well as related questions, point toward a broader structural theory that classifies Lie algebras not only by their intrinsic properties but also by the equivalence classes of graphs they admit. Such an approach might offer valuable insight into the investigation of Lie algebras.

A related and equally promising direction is to ask whether this graph-based approach can be used to explore additional structural properties beyond those discussed in the preceding sections. In particular, one might consider the following question:
\begin{quote}
    \textbf{Q:} \emph{Does a relation exist between the nullity of a Lie algebra and the associated graphs?}
\end{quote}
Recall that the nullity of a Lie algebra is defined as the minimal number of elements necessary to generate a given Lie algebra under the Lie closure \cite{A1:project}. Exploring this connection could uncover specific connectivity conditions for the associated graphs, among other results. Such findings might include existence theorems for particular graph configurations and upper bounds on the nullity, thereby linking this framework to other research on the nullity of Lie algebras \cite{Knebelman:1935,Kuranishi:1951,Sato:1974,Alburquerque:1992}.

Another potential direction relates directly to the question of whether the proposed framework can be useful for classifying Lie algebras, a topic we address in the following part.

\subsubsection{Contributing to the Lie algebra classification effort}
When working with algebraic structures, a natural question to raise is whether a classification of such structures can be achieved. A prominent example of such an endeavor is the classification of all finite simple groups \cite{Aschbacher:2004}. For Lie algebras, however, a full classification is not yet known. Nevertheless, the classification of every simple and therefore every semisimple Lie algebra over an algebraically closed field is well established \cite{Cartan:1894,Killing:erster:1888,Killing:zweiter:1888,Killing:dritter:1889,Killing:vierter:1890}. Building on these foundational results, significant progress has been made \cite{Bianchi:1903,Levi:1905,Matltsev:1942}, and current approaches are informed by the Levi-Mal'tsev theorem, which states that any complex Lie algebra can be decomposed into the semidirect sum of a semisimple component and its radical \cite{Kuzmin:1977}. Despite this, major challenges remain, particularly in classifying all solvable Lie algebras and understanding the structure of the semidirect sum \cite{Popovych:2003}. Contemporary research primarily focuses on low-dimensional solvable Lie algebras \cite{Qi:2019,DelBarco:2025,DeGraaf:2007,Ancocha:2015,Choriyeva:2024}. 

The framework introduced in this work can be connected to these efforts by attempting to classify all graph-admissible Lie algebras together with the set of graphs that represent them. This is especially interesting for minimal-graph-admissible Lie algebras and their associated minimal graphs, as these present a faithful representation and are, so far, better understood. Such a classification effort has the potential to complement and aid the broader goal of classifying all Lie algebras. Consequently, developing a systematic understanding of graph-admissible Lie algebras could yield valuable insights. One promising avenue would be to classify all labeled directed graphs that can serve as minimal graphs for minimal-graph-admissible Lie algebras. This approach could reveal structural invariants and provide a combinatorial perspective on the classification problem. 

In a first step, one could attempt to determine the number of inequivalent and valid graphs with $n$ vertices that can be associated with a minimal-graph-admissible Lie algebra of dimension $n$ as a minimal graph. Let us denote this number by $S(n)$. Calculating $S(n)$ can be refined by introducing the number $S(n,m)$, which counts such graphs when the number of edges is fixed to $m$. Clearly, $S(n)=\sum_{m=0}^\infty S(n,m)$. Moreover, by Propositions~\ref{prop:conditions:for:edges} and~\ref{prop:tight:bounds}, one~has:
\begin{align*}
    S(n)=\sum_{m=0}^{n(n-1)/2}S(n,2m),
\end{align*}
since the number of edges must be even, due to the antisymmetry of the Lie bracket, and is bounded from above by $n(n-1)$. Thus, $S(n,2k+1)=0$ for all $k\in\N_{\geq0}$. Additionally, one has $S(n,0)=1$ for all $n\in\N_{\geq1}$, as there exists only one abelian Lie algebra of dimension $n$. It is straightforward to verify that $S(1)=1=S(1,0)$, while $S(2)=2$ with $S(2,0)=1$, and $S(2,2)=1$. Furthermore, as shown in Theorem~\ref{thm:neccersary:conditions:proper:graph} and Corollaries~\ref{cor:proper:choice:dependent} and~\ref{cor:3:graphs:Bianchi:identification}, one has: $S(3,0)=1$, $S(3,2)=2$, $S(3,4)=4$, and $S(3,6)=3$, while $S(3,m)=0$ for all other $m\in\N_{\geq0}\setminus\{0,2,4,6\}$. Consequently, the number of inequivalent graphs that can be associated with a minimal-graph-admissible three-dimensional Lie algebra as a minimal graph is: $S(3)=10$. 

Conversely, one can also ask:
\begin{quote}
    \textbf{Q:} \emph{What is the number of graphs that can be associated with a particular Lie algebra, when restricted to minimal-graph-admissible Lie algebras and minimal graphs?}
\end{quote}
For example, each abelian Lie algebra can only be associated with one minimal graph, reflecting its trivial bracket structure. In contrast, as we have seen in Example~\ref{exa:second:example:graphs:assocaited:with:algebras}, the special linear Lie algebra $\sl{2}{\R}$ can be associated with two distinct minimal graphs.


\subsubsection{Generating Lie algebras for a provided graph}

Another potential research direction involves reversing the perspective adopted in this work. Instead of starting with a Lie algebra and associating a labeled directed graph to it, one could ask:  
\begin{quote}
    \textbf{Q:} \emph{Can we generate new and interesting Lie algebras by first constructing a valid graph and then associating a Lie algebra with it?}
\end{quote}
Pursuing this approach would require further development of tools to identify valid graphs, particularly valid redundant graphs, building on the initial effort in Section~\ref{sec:valid:graphs}. Such a methodology would not only relate the current framework to similar approaches in the literature, where Lie algebras are also constructed from graphs \cite{Dani:Mainkar:2005,Ray:2015,DeCoste:2022}, but could also yield valuable insight into the role of graph symmetries. Recall that in Section~\ref{sec:prominent:examples:linear:quantum} we observed that a subgraph depicted in Figure~\ref{fig:om:example} possesses a $\Z_3$-symmetry, yet we were unable to trace this symmetry to its origin in the Lie algebra. It is conceivable that such symmetries do not correspond to any meaningful invariant of the Lie algebra, resulting in what might be termed a \emph{apparent symmetry}. However, if these symmetries do reflect structural properties, their study could be highly informative. In particular, redundant graphs may offer greater potential for visualizing such symmetries. This line of inquiry also relates to the following question:
\begin{quote}
    \textbf{Q:} \emph{If $\g$ is a graph-admissible but not minimal-graph-admissible Lie algebra, what is the smallest overcomplete basis that satisfies \eqref{eqn:desired:basis:overcomplete}?}
\end{quote}
Addressing this question would require investigating how and which elements can be added to an overcomplete basis to maintain validity. Furthermore, one could explore the connection between proper graphs and choice-dependent graphs, as well as the existence of proper but not proper-minimal graphs. For instance, in Theorem~\ref{thm:neccersary:conditions:proper:graph} each of the obtained proper graphs with three vertices can also be associated with a Lie algebra as a minimal graph, raising deeper questions about the interplay between minimality and redundancy.

\section{Conclusion}\label{sec:outlook}

The primary aspiration of this work was to develop a visualization framework that facilitates intuitive understanding of the structural properties of Lie algebras. It is in this spirit that we created a tool that allows us to immediately identify fundamental algebraic properties, such as solvability, nilpotency, and the presence of certain ideals, by providing a graphical representation of the algebras. A strong motivation for our work lies in the programme recently initiated to study Hamiltonian Lie algebras that are key to quantum dynamics \cite{Bruschi:Xuereb:2024}, together with the ongoing efforts to determine finite-dimensional Lie subalgebras of the skew-hermitian Weyl algebra $\hat{A}_n$ generated by elements of significance to the dynamics of physical processes \cite{A1:project}. In this context, a visualization tool is expected to be helpful for understanding key aspects of the Lie algebras of interest.

We started by restricting our attention to finite-dimensional Lie algebras that admit a suitable (possibly overcomplete) basis whose elements obey a specific sparsity constraint under the adjoint action. Concretely, the Lie bracket of any two basis elements must be proportional to another basis element, or vanish. Lie algebras fulfilling this requirement were termed \enquote{graph-admissible}, whereas those that did not were classified as \enquote{non-graph-admissible}. This condition is motivated by the idea of associating a labeled directed graph with a given Lie algebra and encoding the Lie bracket structure within the vertex-edge configuration. Graph-admissible Lie algebras admit a further subclassification: those Lie algebras that can be represented by a graph in which the number of nodes coincides with the dimension of the Lie algebra and those that cannot. In the former case, the graphs provide a \enquote{faithful} representation and the corresponding Lie algebras were called \enquote{minimal-graph-admissible}. Conversely, those that cannot are referred to as \enquote{redundant-graph-admissible but not minimal-graph-admissible}, since their associated graphs exhibit redundancies: in particular, the vertex set forms only a spanning set rather than a basis. 
We demonstrated that the classification into non-graph-admissible, minimal-graph-admissible and redundant-graph-admissible (but not minimal-graph-admissible) is well founded by providing explicit examples of Lie algebras in each category.  This was achieved  by analyzing the family of solvable three-dimensional Lie algebras $L_\alpha^3$ (cf. \cite{DeGraaf:2004}). Furthermore, we established that a broad class of mathematically and physically significant Lie algebras, including all complex semisimple Lie algebras, all nilpotent Lie algebras, and all finite-dimensional Lie subalgebras of the complex single-mode Weyl algebra $A_1$, are graph-admissible.

The graphs that can be associated with graph-admissible Lie algebras must satisfy specific general structural properties reflecting the antisymmetry of the Lie bracket, as well as the constraints imposed by the Jacobi identity. In this work we provided relevant preliminary results that restrict the number of admissible graphs. However, we were unable to determine the exact number of graphs that, for a given number of $n$ vertices, can be associated with an $n$-dimensional Lie algebra, except for the case $n\leq 3$. The conditions of interest were mostly derived for minimal-graph-admissible Lie algebras associated with minimal graphs only, but not redundant graphs associated with redundant-graph-admissible Lie algebras. Moreover, they only provide a weak criterion  for proper graphs, i.e., graphs that can be associated with a Lie algebra, dictated by the Jacobi identity. Further work is required to refine and generalize these criteria.

We then provided a methodology to identify key structural properties of a given Lie algebra solely via its associated graph. Specifically, we formulated graph-theoretic criteria for identifying central elements, solvability, nilpotency, the existence of ideals, as well as simplicity, semisimplicity, and reductiveness. The criteria for solvability and nilpotency were inspired by classical algebraic constructions: we mirrored the formation of the derived series and lower central series with corresponding sequences of graphs, and provided accessible pruning algorithms to compute them. These constructions enabled us to establish strong necessary and sufficient conditions for solvability and nilpotency based exclusively on the initial graph, without requiring iterative graph reductions. These results present a significant step toward bridging algebraic and combinatorial perspectives as they allow structural classifications using purely graphical features. In order to illustrate the practical perspectives of these methods, we applied them to three important Lie algebras that are also key to physics: the Schrödinger algebra $\sl{2}{\R}\ltimes\mathfrak{h}_m$, the Lorentz algebra $\mathfrak{so}(3;1)_\C$, and a finite-dimensional subalgebra of the skew-hermitian Weyl algebra $\hat{A}_n$. These examples illustrate the effectiveness of our framework in capturing deep structural properties through visualizations alone.

We concluded our analysis by examining the strengths of the proposed approach and outlining potential generalizations of the framework. One particularly promising direction lies in exploring the relationship between graph-admissibility and magma-gradings of Lie algebras, which could lead to deeper structural insights, new classification techniques, and a broadening of the applicability of the current framework. In addition, we discussed several avenues for future research, presenting open questions and conjectures and highlighting their connections to broader mathematical efforts. Finally, we emphasized the practical significance of these results by considering the role of solvability in determining the time evolution of certain physical systems. In this context, our approach can assist in better understanding concepts such as similarity transformations, since it allows one to map computations of sequences of nested commutators within algebras of interest to walks on the associated graphs. This perspective helps clarify some of the difficulties that can arise when extending such factorization methods to infinite-dimensional settings, offering a conceptual bridge between abstract algebraic theory and practical applications in physics.

\clearpage

	\section*{Data availability statement}
	
	Data sharing is not applicable to this work as no datasets are used to support the presented results.

    \section*{Acknowledgments}
    We thank Robert Zeier for helpful comments and suggestions. We also wish to acknowledge the independent use,  found in \cite{Kazi:2025}, of the particular labelling of algebraic graphs considered here. Moreover, we thank Lidia Westphal for helpful comments and her advice regarding the runtime discussion of the main-text algorithms in Appendix~\ref{app:runtime_discussion}.

	\section*{Conflict of interest}
	
	The authors have no relevant financial or non-financial interest to disclose.

	\section*{Author contributions}
	TCH and DEB conceived the idea of a graphic representation of algebras with the chosen labelling approach. TCH proved all results. DEB supervised the project, provided constructive feedback for all claims, and proposed a subset of the claims. TCH generated all figures. All authors contributed to writing the manuscript.

	\section*{Funding}
    TCH acknowledges support from the joint project No. 13N15685 ``German Quantum Computer based on Superconducting
	Qubits (GeQCoS)'' sponsored by the German Federal Ministry of Education and Research (BMBF) under the framework programme “Quantum technologies -- from basic research to the market”.

\bibliography{main}

\onecolumngrid
\newpage
\appendix

\section{Table comparing different frameworks for associating graphs with Lie algebras}\label{app:comapring:graph:Lie:algebra:frameworks}

As discussed in the introduction and, in particular, in Section~\ref{sec:discussion}, several different frameworks have been developed for associating graphs with Lie algebras. For completeness, Table~\ref{tab:comparison:related:graph:constructions} compares the present framework with the nilpotent constructions from~\cite{Chakrabarti:2019,Molina:2023} and the related solvable construction from~\cite{Grantcharov:2017}. In particular, we summarize the classes of Lie algebras and graphs considered, the encoding of the Lie bracket, the quantitative relations between the graph and the Lie algebra, the corresponding uniqueness properties, and the direction of each construction.

\begin{table}[htpb]
    \centering
    \caption{Comparison of the present framework with closely related graph-based constructions of Lie algebras. The table contrasts different aspects such as the direction of construction, the encoded bracket data, the classes of Lie algebras covered, quantitative graph data, and uniqueness properties.}
    \label{tab:comparison:related:graph:constructions}
    \scriptsize  
    \begin{tblr}{width=\linewidth,colspec={X[1.15] X[1.55] X[1.65] X[1.65] X[2.05]},cells={halign=c,valign=t},hspan=minimal,row{1}={font=\bfseries},hlines,rowsep=2pt,colsep=2pt}
        Comparison feature & Nilpotent construction~\cite{Chakrabarti:2019,Molina:2023} & Solvable construction~\cite{Grantcharov:2017} & Present framework: minimal setting & Present framework: redundant setting \\
        
        Class of Lie algebras &  Specific two-step nilpotent Lie algebras $\mathfrak{n}$. & Specific three-step solvable Lie algebras $\mathfrak{s}$ & Finite-dimensional minimal-graph-admissible Lie algebras $\g$, including nilpotent, solvable, and non-solvable examples & Finite-dimensional redundant-graph-admissible Lie algebras $\mathfrak{h}$, including minimal-graph-admissible, nilpotent, solvable, and non-solvable examples \\
        
        Class of graphs & Finite labeled directed simple graphs & Finite simple undirected graphs $G(V,E)$ & \SetCell[c=2]{c} Finite proper labeled directed graphs & \\
        
        Role of vertices & Subset of basis of $\mathfrak{n}/\mathcal{C}^1\mathfrak{n}$ & Subset of basis of $\mathcal{D}\mathfrak{s}/\mathcal{D}^2\mathfrak{s}$ & Basis of $\g$ & Overcomplete spanning set of $\mathfrak{h}$ consisting of pairwise non-proportional elements \\
        
        \SetCell[r=4]{c} Encoding of Lie brackets & \SetCell[c=4]{c} Starting vertex: first input of the Lie bracket & \\
        
        & Edge label:  Resulting central basis element & Edge data: Resulting basis element of $\mathcal{D}^2\mathfrak{s}$ via the identification of the Lie bracket between two connected vertices $[v_j,v_k]$ with antisymmetric product $v_j\wedge v_k$ & \SetCell[c=2]{c} Edge label: second input of the Lie bracket & \\
        
        & \SetCell[c=2]{c} Target vertex: second input of the Lie bracket & & \SetCell[c=2]{c} Target vertex: result of the Lie bracket up to a non-zero scalar coefficient & \\
        
        & Edge direction: order of the inputs and hence the sign prefactor of the result & --- & \SetCell[c=2]{c} Edge direction: order of the inputs & \\
        
        Number of vertices & $|V|=\dim(\mathfrak{n})-\dim(\mathcal{C}^1\mathfrak{n})$. & $|V|\leq\dim(\mathfrak{s})-\dim(\mathcal{D}^2\mathfrak{s})$ & $|V|=\dim(\g)$ &  $|V|\geq \dim(\mathfrak{h})$ with $|V|<\infty$ \\
        
        Number of graph edges & \SetCell[c=2]{c} $0\leq |E|\leq |V|(|V|-1)/2$ & & \SetCell[c=2]{c} $0\leq |E|\leq |V|(|V|-1)$  \\

        Number of edge labels $N_\mathrm{L}$ & $ N_\mathrm{L}= \dim(\mathcal{C}^1\mathfrak{n})$ & $0\leq N_\mathrm{L}\leq \dim(\mathcal{D}^2\mathfrak{s})$ & \SetCell[c=2]{c} $0\leq N_\mathrm{L}\leq |V|$ \\
        
        Uniqueness & If the graph is connected and every edge has a distinct label, reversing any subset of edge directions produces an isomorphic Lie algebra; cf. Theorem~3.9 in~\cite{Molina:2023} & Two Lie algebras constructed by this method are isomorphic if and only if their defining graphs are isomorphic; cf. Theorem~3.4 in~\cite{Grantcharov:2017} & \SetCell[c=2]{c} A single Lie algebra may admit several non-isomorphic associated proper graphs, while the same unweighted proper graph may be associated with several non-isomorphic Lie algebras; cf. Example~\ref{exa:second:example:graphs:assocaited:with:algebras} & \\
        
        Direction of construction & Starts from a labeled directed simple graph and constructs an associated two-step nilpotent Lie algebra & Starts from a finite simple undirected graph, uses its vertices and edges, supplemented by sets of three-cliques $v_j\odot v_k\odot v_\ell$, and constructs an associated solvable Lie algebra & Starts from a finite-dimensional Lie algebra together with a suitable basis and constructs an associated labeled directed graph & Starts from a finite-dimensional Lie algebra together with a suitable finite overcomplete spanning set and constructs an associated labeled directed graph \\
    \end{tblr}
\end{table}

\section{Supplementary proofs}\label{app:supplementary:proofs}

This appendix collects detailed proofs of several graph-admissibility results presented in the main text. Since these proofs are largely technical and provide little additional conceptual insight, their inclusion in the main text would distract from the central results. Throughout this appendix, we use the definitions and conventions introduced in the main text without restating them.

\begin{theorem*}\label{thm:supp:existncae:of:non:minimal:graph:admissible:}
    There exist finite-dimensional Lie algebras that are minimal-graph-admissible and ones that are not minimal-graph-admissible.
\end{theorem*}

\begin{proof}
    The fact that there exist finite-dimensional Lie algebras that are not minimal-graph-admissible follows directly from the following Lemma~\ref{lem:supp:L:alpha:3:alpha:leq:non:mga}. The existence of minimal-graph-admissible Lie algebras can easily be shown by providing a simple example. The simplest possible such example is the one-dimensional abelian Lie algebra $\mathbb{F}$, which is admissible with any basis. 
\end{proof}

\begin{lemma}\label{lem:supp:L:alpha:3:alpha:leq:non:mga}
    The family of solvable real Lie algebras $L_\alpha^3$ spanned by the three basis elements $x_1,x_2,x_3$ and determined by the Lie bracket relations
    \begin{align}
        [x_1,x_2]&=0,\;&\;[x_1,x_3]&=-x_2,\;&\;[x_2,x_3]&=-\alpha x_1-x_2,
    \end{align}
    for $\alpha\in \R$ is not minimal-graph-admissible if $\alpha<-1/4$.
\end{lemma}

\begin{proof}
    Consider the family of solvable Lie algebras $L_\alpha^3=\spn\{x_1,x_2,x_3\}$,
    where $\alpha$ is a real parameter \cite{DeGraaf:2004}. We choose now $\alpha<-1/4$ and suppose that $L_\alpha^3$ is minimal-graph-admissible. In other words, $L_\alpha^3$ admits a basis $\{\Tilde{x}_1,\Tilde{x}_2,\Tilde{x}_3\}$ that satisfies the relations \eqref{eqn:desired:basis}. Each $\Tilde{x}_j$ can be expressed as a linear combination of the original basis:
    \begin{align*}
        \Tilde{x}_j&=c_1^{(j)}x_1+c_2^{(j)}x_2+c_3^{(j)}x_3\qquad\text{for}\qquad j\in\{1,2,3\},
    \end{align*}
    where the $c_k^{(j)}$ are appropriate real coefficients. Next, we compute the Lie bracket $[\Tilde{x}_p,\Tilde{x}_q]$ using the bilinearity of the bracket and the original structure constants:
    \begin{align*}
        [\Tilde{x}_p,\Tilde{x}_q]&=\left(-c_1^{(p)}c_3^{(q)}+c_3^{(p)}c_1^{(q)}\right)x_2+\left(c_2^{(p)}c_3^{(q)}-c_3^{(p)}c_2^{(q)}\right)(-\alpha x_1-x_2)\\
        &=-\alpha\left(c_2^{(p)}c_3^{(q)}-c_3^{(p)}c_2^{(q)}\right)x_1-\left(c_1^{(p)}c_3^{(q)}-c_3^{(p)}c_1^{(q)}+c_2^{(p)}c_3^{(q)}-c_3^{(p)}c_2^{(q)}\right)x_2.
    \end{align*}
    Since $L_\alpha^3$ is non-abelian and $\{\Tilde{x}_1,\Tilde{x}_2,\Tilde{x}_3\}$ is, by assumption, a basis satisfying the relations \eqref{eqn:desired:basis}, one must have $[\Tilde{x}_p,\Tilde{x}_q]=\Tilde{\alpha}_{pq}\Tilde{x}_{\Tilde{\delta}(p,q)}\neq0$ for at least one pair $(p,q)\in\{1,2,3\}^{\times 2}$ with $p\neq q$. Hence, one must have $c_{3}^{(\ell)}=0$ for at least one $\ell\in\{1,2,3\}$ and to ensure linear independence also $c_3^{(q)}\neq0$ for some $q\in\{1,2,3\}$. Without loss of generality, let us assume that $c_3^{(1)}=0$ and $ c_3^{(2)}\neq 0$.

    We can now consider the Lie bracket $[\Tilde{x}_1,\Tilde{x}_2]$. Here, one has:
    \begin{align*}
        [\Tilde{x}_1,\Tilde{x}_2]&=-c_3^{(2)}\left(\alpha c_2^{(1)}x_1+\left(c_1^{(1)}+c_2^{(1)}\right)x_2\right).
    \end{align*}
    Suppose this Lie bracket vanishes. This would imply that $c_2^{(1)}=0=c_1^{(1)}+c_2^{(1)}$ and consequently $c_1^{(1)}=0$. However, this in turn would imply $\Tilde{x}_1=0$, which is a contradiction to the requirement that $\{\Tilde{x}_1,\Tilde{x}_2,\Tilde{x}_3\}$ spans the three-dimensional Lie algebra $L_\alpha^3$. Therefore, the commutator $[\Tilde{x}_1,\Tilde{x}_2]$ must be proportional to either $\Tilde{x}_1$ or $\Tilde{x}_3$, since $[\Tilde{x}_1,\Tilde{x}_2]=\Tilde{\alpha}_{12}\Tilde{x}_2$ with $\Tilde{\alpha}_{12}\neq 0$ is clearly prohibited by the condition $c_3^{(2)}\neq 0$. We consider the remaining two cases:
    \begin{enumerate}[label = (\roman*)]
        \item $[\Tilde{x}_1,\Tilde{x}_2]=\Tilde{\alpha}_{12}\Tilde{x}_1$ with $\Tilde{\alpha}_{12}\neq 0$. Here, we need to solve the following equations:
        \begin{align*}
            c_1^{(1)}&=-\frac{c_3^{(2)}}{\Tilde{\alpha}_{12}}\alpha\, c_2^{(1)},\;&\;c_2^{(1)}&=-\frac{c_3^{(2)}}{\Tilde{\alpha}_{12}}\left(c_1^{(1)}+c_2^{(1)}\right).
        \end{align*}
        Solving for $c_2^{(1)}$, we find:
        \begin{align*}
            c_2^{(1)}\left(1+\frac{c_3^{(2)}}{\Tilde{\alpha}_{12}}\left(1-\alpha\frac{c_3^{(2)}}{\Tilde{\alpha}_{12}}\right)\right)=0.
        \end{align*}
        If $c_2^{(1)}=0$, then $\Tilde{x}_1=0$, which is not allowed. Thus, the bracket $[\Tilde{x}_1,\Tilde{x}_2]$ can only be proportional to $\Tilde{x}_1$, if the following quadratic equation has a real solution:
        \begin{align*}
            0&=\left(\frac{c_3^{(2)}}{\Tilde{\alpha}_{12}}\right)^2-\frac{1}{\alpha}\frac{c_3^{(2)}}{\Tilde{\alpha}_{12}}-\frac{1}{\alpha}=\left(\frac{c_3^{(2)}}{\Tilde{\alpha}_{12}}-\frac{1}{2\alpha}\right)^2-\frac{1+4\alpha}{4\alpha^2}
        \end{align*}
        However, for $\alpha<-1/4$, the discriminant is strictly negative, so there is no real solution \cite{Cox:2025}. One can therefore not have $[\Tilde{x}_1,\Tilde{x}_2]=\Tilde{\alpha}_{12}\Tilde{x}_1$ with $\Tilde{\alpha}_{12}\neq0$.
        \item $[\Tilde{x}_1,\Tilde{x}_2]=\Tilde{\alpha}_{12}\Tilde{x}_3$ with $\Tilde{\alpha}_{12}\neq0$. This is only possible if
        \begin{align*}
            c_1^{(3)}&=-\frac{c_3^{(2)}}{\Tilde{\alpha}_{12}}\alpha c_2^{(1)},\;&\; c_2^{(3)}&= -\frac{c_3^{(2)}}{\Tilde{\alpha}_{12}}\left(c_1^{(1)}+c_2^{(1)}\right),\;&\;c_3^{(3)}&=0.
        \end{align*}
    \end{enumerate}
    We must therefore have case (ii) and consequently $[\Tilde{x}_1,\Tilde{x}_3]=0$. Suppose now $[\Tilde{x}_2,\Tilde{x}_3]=0$. In this case, the structure constants would force $L_\alpha^3$ to be isomorphic to the real Lie algebra $L_0^4$, which is not possible \cite{DeGraaf:2004}. Thus, one needs to consider the three remaining possibilities for the bracket $[\Tilde{x}_2,\Tilde{x}_3]$:
    \begin{enumerate}[label=(\alph*)]
        \item $[\Tilde{x}_2,\Tilde{x}_3]=\Tilde{\alpha}_{23}\Tilde{x}_1$ with $\Tilde{\alpha}_{23}\neq 0$. This implies
        \begin{align*}
            [\Tilde{x}_2,\Tilde{x}_1]&=-\Tilde{\alpha}_{12}\Tilde{x}_3,\;&\;[\Tilde{x}_2,-\Tilde{\alpha}_{12}\Tilde{x}_3]&=-\Tilde{\alpha}_{12}\Tilde{\alpha}_{23}\Tilde{x}_1,\;&\;[\Tilde{x}_1,-\Tilde{\alpha}_{12}\Tilde{x}_3]&=0.
        \end{align*}
        One can therefore construct a Lie-algebra isomorphism mapping: $\phi(\Tilde{x}_1)=e_1$, $\phi(\Tilde{x}_2)=e_3$, and $\phi(\Tilde{x}_3)=-e_2/\Tilde{\alpha}_{12}$. Thus $L_\alpha^3$ would be isomorphic to $L_{-\Tilde{\alpha}_{12}\Tilde{\alpha}_{23}}^4$, which is not possible, since $L_\alpha^3$ and $L_\beta^4$ are never isomorphic for any $\alpha,\beta$, see \cite{DeGraaf:2004}.
        \item $[\Tilde{x}_2,\Tilde{x}_3]=\Tilde{\alpha}_{23}\Tilde{x}_2$ with $\Tilde{\alpha}_{23}\neq 0$. This is not possible, since $c_3^{(2)}\neq0$, but $[\Tilde{x}_p,\Tilde{x}_q]\in\spn\{x_1,x_2\}$ for all $p,q\in\{1,2,3\}$.
        \item $[\Tilde{x}_2,\Tilde{x}_3]=\Tilde{\alpha}_{23}\Tilde{x}_3$ with $\Tilde{\alpha}_{23}\neq 0$. This implies:
        \begin{align*}
            \left[\frac{1}{\Tilde{\alpha}_{23}}\Tilde{x}_2,\Tilde{\alpha}_{23}\Tilde{x}_1\right]&=-\Tilde{\alpha}_{12}\Tilde{x}_3,\;&\;\left[\frac{1}{\Tilde{\alpha}_{23}}\Tilde{x}_2,-\Tilde{\alpha}_{12}\Tilde{x}_3\right]&=-\Tilde{\alpha}_{12}\Tilde{x}_3,\;&\;\left[-\Tilde{\alpha}_{12}\Tilde{x}_3,\Tilde{\alpha}_{23}\Tilde{x}_1\right]&=0.
        \end{align*}
        Thus, one would be able to find an isomorphism from $L_\alpha^3$ to $L_0^3$ with the identification $\Tilde{x}_2/\Tilde{\alpha}_{23}\leftrightarrow e_3$, $\Tilde{\alpha}_{23}\Tilde{x}_1\leftrightarrow e_1$, and $-\Tilde{\alpha}_{12}\Tilde{x}_3\leftrightarrow e_2$. However, it is known that $L_\alpha^3\cong L_\beta^3$ if and only if $\alpha=\beta$, see \cite{DeGraaf:2004}. Thus, we find again a contradiction, since $\alpha<-1/4$.
    \end{enumerate}
    The results obtained above prove that $L_\alpha^3$ does not admit a basis that satisfies the relations \eqref{eqn:desired:basis} for $\alpha<-1/4$.
\end{proof}

\begin{corollary*}\label{cor:supp:existance:of:redundant:but:not:mimimak:graph:admissibel:algebras}
    There exist Lie algebras that are graph-admissible but not minimal-graph-admissible.
\end{corollary*}

\begin{proof}
    This statement is a direct consequence of the analysis in the proof of Lemma~\ref{lem:supp:L:alpha:3:alpha:leq:non:mga}. Recall that for certain values of the parameter, specifically for $\alpha<-1/4$, the real Lie algebra $L_\alpha^3$ does not admit a minimal-graph-admissible basis. However, in some cases, it is still possible to find an overcomplete basis for such an algebra that satisfies the weaker condition of redundant-graph-admissibility given by the relations \eqref{eqn:desired:basis:overcomplete}. For example,  consider $\alpha=-1$. In this case, the Lie algebra $L_{-1}^3$ is redundant-graph-admissible with the overcomplete basis $\{x_1,x_2,x_3,x_1-x_2\}$ defined by the basis $\{x_1,x_2,x_3\}$ given in the proof of Lemma~\ref{lem:supp:L:alpha:3:alpha:leq:non:mga}. This basis allows all nontrivial Lie brackets of these basis elements to be written as proportional to one of these elements, which can be verified by a straightforward calculation. By the proof of Lemma~\ref{lem:supp:L:alpha:3:alpha:leq:non:mga}, it is furthermore clear that no minimal basis exists with this property. Thus, the existence of such an overcomplete basis demonstrates that $L_{-1}^3$ is graph-admissible but not minimal-graph-admissible.
\end{proof}

\begin{theorem*}\label{thm:supp:existence:of:non:graph:admissible:lie:algebras}
    There exist finite-dimensional Lie algebras that are non-graph-admissible.
\end{theorem*}

\begin{proof}
    This claim follows directly from the following Lemma~\ref{lem:L:alpha:three:graph:admissible:if:and:only:if}, which provides a concrete example of finite-dimensional Lie algebras that fail to satisfy the conditions for graph-admissibility. 
\end{proof}

\begin{lemma}\label{lem:L:alpha:three:graph:admissible:if:and:only:if}
    The real Lie algebra $L_\alpha^3$ (see \cite{DeGraaf:2004}) is graph-admissible if and only if $\alpha> - 1/4$ or if $\arctan(\sqrt{|1+4\alpha|}/(1+2\sqrt{|\alpha|}))/\pi\in\mathbb{Q}^*$.
\end{lemma}

\begin{proof}
    The Lie algebra $L_\alpha^3$ is defined as the real three-dimensional vector space spanned by the elements $x_1,x_2,x_3$, with the Lie bracket given by:
    \begin{align}\label{eqn:structure:constants:L:alpha:3}
        [x_1,x_2]&=0,\;&\;[x_1,x_3]&=-x_2,\;&\;[x_2,x_3]&=-\alpha x_1-x_2,
    \end{align}
    where $\alpha\in\R$ is a real parameter \cite{DeGraaf:2004}.
    Suppose now there exists a finite set $\mathcal{B}=\{\Tilde{x}_j\}_{j=1}^m$ that spans $L_\alpha^3$ and satisfies the graph-admissibility conditions given by Definition~\ref{def:graph:admissible} and the relations in \eqref{eqn:desired:basis:overcomplete}. Each $\Tilde{x}_j\in\mathcal{B}$ can be expressed via the original basis elements:
    \begin{align*}
        \Tilde{x}_j=c_1^{(j)}x_1+c_2^{(j)}x_2+c_3^{(j)}x_3\quad\text{for all }j\in\mathcal{M}:=\{1,\ldots,m\}\supseteq\{1,2,3\},
    \end{align*}
    where the $c_k^{(j)}\in\R$ are real coefficients. Since $\mathcal{B}$ spans $L_\alpha^3$ there must exist at least one element $\Tilde{x}_j\in\mathcal{B}$ such that $c_3^{(j)}\neq 0$, 
    otherwise the algebra would be two-dimensional and abelian since each new element would be expressed as a linear combination of $x_1$ and $x_2$, and $[x_1,x_2]=0$. Without loss of generality, we assume $c_3^{(2)}=1$. The Lie bracket between two elements $\Tilde{x}_p,\Tilde{x}_q\in\mathcal{B}$ can be computed using bilinearity and the structure constants determined by \eqref{eqn:structure:constants:L:alpha:3}:
    \begin{align*}
        [\Tilde{x}_p,\Tilde{x}_q]&=-\alpha\left(c_2^{(p)}c_3^{(q)}-c_3^{(p)}c_2^{(q)}\right)x_1-\left(c_1^{(p)}c_3^{(q)}-c_3^{(p)}c_1^{(q)}+c_2^{(p)}c_3^{(q)}-c_3^{(p)}c_2^{(q)}\right)x_2,
    \end{align*}
    Since $L_\alpha^3$ is non-abelian and the graph-admissibility condition requires that $[\Tilde{x}_p,\Tilde{x}_q]=\Tilde{\alpha}_{pq}\Tilde{x}_{\delta(p,q)}$ with $\Tilde{\alpha}_{pq}\neq 0$ for at least one pair $(p,q)\in\mathcal{M}\times\mathcal{M}$, it follows that there must exist at least one $\Tilde{x}_k\in\mathcal{B}$ such that $c_3^{(k)}=0$. Assume, without loss of generality, that $c_3^{(1)}=0$. For convenience, let us write $\Tilde{x}_1=Ax_1+Bx_2$, where $A,B\in\R$. The Lie bracket between $\Tilde{x}_1$ and $\Tilde{x}_2$ then becomes: $[\Tilde{x}_1,\Tilde{x}_2]=-\alpha Bx_1-(A+B) x_2$. Here, one observes that $\Tilde{x}_2$ can never be proportional to $[\Tilde{x}_1,\Tilde{x}_2]$, since $c_3^{(2)}=1\neq0$ and the bracket lies entirely in the span of $x_1$ and $x_2$. It is, furthermore, easy to show that $\Tilde{x}_1$ can be proportional (with a non-zero proportionality constant) to $[\Tilde{x}_1,\Tilde{x}_2]$ if and only if $\alpha\geq- 1/4$. Let us prove this claim: To determine whether $\Tilde{x}_1\propto[\Tilde{x}_1,\Tilde{x}_2]$, one must solve the system
    \begin{align}
        \kappa A&=-\alpha B,\;&\;\kappa B&=-(A+B),\label{eqn:help:existence:of:non:graph:admissible:Lie:algebras}
    \end{align}
    where $\kappa\in\R^*$ is the proportionality constant. Note that at least one of the coefficients $A$ or $B$ must be non-zero; otherwise $\Tilde{x}_1$ would vanish, which contradicts the assumption that $\Tilde{x}_1$ is a non-zero element belonging to $\mathcal{B}$. Observe that if $B=0$, the first equation of \eqref{eqn:help:existence:of:non:graph:admissible:Lie:algebras} implies $A=B=0$ if $B=0$, which is prohibited. Hence, one has $B\neq0$.  Simple manipulations of equations \eqref{eqn:help:existence:of:non:graph:admissible:Lie:algebras} yield $0=B(\kappa^2+\kappa-\alpha)$ and consequently $\kappa^2+\kappa-\alpha=0$. This equation admits real solutions if and only if $\alpha\geq -1/4$, since the discriminant of the quadratic polynomial is $\Delta=1+4\alpha$, which is non-negative if and only if $\alpha\geq -1/4$ \cite{Cox:2025}. Therefore, the proportionality conditions $\Tilde{x}_1\propto [\Tilde{x}_1,\Tilde{x}_2]$ can only be satisfied if $\alpha\geq -1/4$. For $\alpha<-1/4$, the system has no real solution, and the proportionality fails.

    We consider, therefore, now the case $\alpha> -1/4$. In this regime, the proportionality condition derived earlier can be formulated as an eigenvalue problem. Specifically, we can examine the linear transformation defined by the adjoint action $-\operatorname{ad}_{\Tilde{x}_2}(\Tilde{x}_1)=[\Tilde{x}_1,\Tilde{x}_2]$. Restricted to the space $V:=\spn\{x_1,x_2\}$, we can represent the problem $-\operatorname{ad}_{\Tilde{x}_2}(\Tilde{x}_1)=\kappa\Tilde{x}_1$ with $\kappa\in\R^*$ using the matrix formalism:
    \begin{align*}
        \begin{pmatrix}
            0&-\alpha\\
            -1&-1
        \end{pmatrix}\begin{pmatrix}
            A\\
            B
        \end{pmatrix}&=\kappa\begin{pmatrix}
            A\\
            B
        \end{pmatrix}.
    \end{align*}
    The eigenvalues are given by $\kappa_\pm=(-1\pm\sqrt{1+4\alpha})/2$. Using these eigenvalues, we define the new elements $\Tilde{x}_1:=\kappa_-x_1+x_2$, $\Tilde{x}_2:=x_3$, and $\Tilde{x}_3:=\kappa_+ x_1+x_2$, which are clearly linearly independent if $\alpha>-1/4$, and hence form a valid basis of $L_\alpha^3$. Moreover, they satisfy the following Lie bracket relations: $[\Tilde{x}_1,\Tilde{x}_2]=\kappa_+ \Tilde{x}_1$, $[\Tilde{x}_1,\Tilde{x}_3]=0$, and $[\Tilde{x}_3,\Tilde{x}_2]=\kappa_-\Tilde{x}_3$. This confirms that $L_\alpha^3$ is minimal-graph-admissible and consequently also graph-admissible for $\alpha>-1/4$.

    Next, we turn to the case $\alpha<-1/4$. As previously established, the Lie bracket $[\Tilde{x}_1,\Tilde{x}_2]$ is, independent of the choice of $c_1^{(1)}$, $c_2^{(1)}$, $c_1^{(2)}$, and $c_2^{(2)}$, non-zero, lies in the subspace $V:=\spn\{x_1,x_2\}$, and is neither proportional to $\Tilde{x}_1$, nor to $\Tilde{x}_2$. Therefore, it must be included in $\mathcal{B}$, up to a non-zero scalar multiple $\kappa$. Without loss of generality, we may assume this proportionality constant is $\kappa=1$, since scalar multiples are irrelevant for satisfying the graph-admissibility condition \eqref{eqn:desired:basis:overcomplete}.

    To proceed, we introduce the iterated adjoint action: $\operatorname{ad}_{\Tilde{x}_2}^\ell(\Tilde{x}_1):=[\Tilde{x}_2,[\Tilde{x}_2,[\ldots,[\Tilde{x}_2,\Tilde{x}_1]\ldots]]]$, where the Lie bracket is nested $\ell$ times. This allows us to make the following useful observations:
    
    \begin{itemize}
        \item The element $\operatorname{ad}_{\Tilde{x}_2}^\ell(\Tilde{x}_1)$ is non-zero and lies entirely in $V$ for all $\ell\in\N_{\geq0}$. 
        
        This can be inductively shown: For the base case $\ell=0$, one has by convention $\operatorname{ad}_{\Tilde{x}_2}^0(\Tilde{x}_1):=A_0 x_1+B_0x_2$, where $(A_0,B_0)\in\R^2\setminus\{0\}$, since $\Tilde{x}_1\neq0$ by assumption. Then $\operatorname{ad}_{\Tilde{x}_2}^1(\Tilde{x}_1)=A_1x_1+B_1 x_2$, where the coefficients are given by the recurrence relations: $A_1=\alpha B_0$ and $B_1=A_0+B_0$. It is straightforward to verify that from $(A_0,B_0)\neq 0$ it follows that $(A_1,B_1)\neq 0$. Hence, $\operatorname{ad}_{\Tilde{x}_2}^1(\Tilde{x}_1)\neq 0$. With the base case shown, we can move on to the induction step. Assume the induction hypothesis holds for some $\ell\geq 1$, i.e., assume that $\operatorname{ad}_{\Tilde{x}_2}^\ell(\Tilde{x}_1)=A_\ell x_1+B_\ell x_2\neq 0$ holds for some $\ell\geq 1$. Then, $\operatorname{ad}_{\Tilde{x}_2}^{\ell+1}(\Tilde{x}_1)=A_{\ell+1} x_1+B_{\ell+1}x_2$, with the recurrence $A_{\ell+1}=\alpha B_{\ell}$, and $B_{\ell+1}=A_\ell+B_\ell$. By the same argument as before, we have that $(A_{\ell+1},B_{\ell+1})\neq 0$ since $(A_\ell,B_\ell)\neq 0$, which shows that $\operatorname{ad}_{\Tilde{x}_2}^\ell(\Tilde{x}_1)\in V\setminus\{0\}$ for all $\ell\in\N_{\geq0}$.
        \item The linear map $\operatorname{ad}_{\Tilde{x}_2}^\ell|_V$ is invertible for all $\ell\in\N_{\geq0}$. This follows from examining the matrix representation of $\operatorname{ad}_{\Tilde{x}_2}|_V$. 
        \item For all $p,q\in\N_{\geq0}$, one has $[\operatorname{ad}_{\Tilde{x}_2}^p(\Tilde{x}_1),\operatorname{ad}_{\Tilde{x}_2}^q(\Tilde{x}_1)]=0$ . This is a direct consequence of the observation that $\operatorname{ad}_{\Tilde{x}_2}^\ell(\Tilde{x}_1)\in V$ for all $\ell\in\N_{\geq0}$ and that the subspace is abelian, i.e., $[V,V]=\{0\}$, since $[x_1,x_2]=0$.
    \end{itemize}
    
    Since $\mathcal{B}$ is assumed to be a finite (possibly overcomplete) basis of $L_\alpha^3$, all elements $\operatorname{ad}_{\Tilde{x}_2}^\ell(\Tilde{x}_1)$ must belong to $\mathcal{B}$, up to an irrelevant non-zero scalar multiple, for all $\ell$ smaller than some particular $\ell_*\in\N_{\geq2}$. This follows from the fact that $\operatorname{ad}_{\Tilde{x}_2}(\Tilde{x}_1)\in\mathcal{B}$, which implies that its Lie bracket with any other element in $\mathcal{B}$ must either vanish or be proportional to another element in $\mathcal{B}$, particularly when considering the Lie bracket with $\Tilde{x}_1$ and $\Tilde{x}_2$. From the previous observation that $[\operatorname{ad}_{\Tilde{x}_2}^p(\Tilde{x}_1),\operatorname{ad}_{\Tilde{x}_2}^q(\Tilde{x}_1)]=0$ for all $p,q\in\N_{\geq0}$, and that $\operatorname{ad}_{\Tilde{x}_2}^2(\Tilde{x}_1)=[\Tilde{x}_2,\operatorname{ad}_{\Tilde{x}_2}(\Tilde{x}_1)]\in V\setminus\{0\}$, it follows that either: (a) $\operatorname{ad}_{\Tilde{x}_2}^2(\Tilde{x}_1)$ is proportional to $\Tilde{x}_1$ or $\operatorname{ad}_{\Tilde{x}_2}(\Tilde{x}_1)$; or (b) $\operatorname{ad}_{\Tilde{x}_2}^2(\Tilde{x}_1)$ belongs to $\mathcal{B}$. In case (b), the same logic must be applied recursively. However, since $\mathcal{B}$ is finite, this process must eventually terminate, requiring the existence of the aforementioned minimal $\ell_*\in\N_{\geq2}$, such that the iterated adjoint becomes proportional to a previous one. Thus, there must exist some $\operatorname{ad}_{\Tilde{x}_2}^{p}(\Tilde{x}_1)$, with $p\in \N_{\geq2}$, and some $q\in\N_{\geq0}$ with $q<p$, such that $\operatorname{ad}_{\Tilde{x}_2}^p(\Tilde{x}_1)\propto \operatorname{ad}_{\Tilde{x}_2}^q(\Tilde{x}_1)$. Since, by the observations listed above, $\operatorname{ad}_{\Tilde{x}_2}^p(\Tilde{x}_1)\in V$ and $\operatorname{ad}_{\Tilde{x}_2}^\ell$ is invertible on $V$ for all $\ell\in\N_{\geq0}$, one has that $\operatorname{ad}_{\Tilde{x}_2}^{p-q}(\Tilde{x}_1)=[\Tilde{x}_2,\operatorname{ad}_{\Tilde{x}_2}^{p-q-1}(\Tilde{x}_1)]\propto \operatorname{ad}_{\Tilde{x}_2}^0(\Tilde{x}_1)=\Tilde{x}_1$. One may therefore assume, without loss of generality, that the set $\mathcal{B}':=\{\Tilde{x}_2,\operatorname{ad}_{\Tilde{x}_2}^0(\Tilde{x}_1),\operatorname{ad}_{\Tilde{x}_2}(\Tilde{x}_1),\ldots,\operatorname{ad}_{\Tilde{x}_2}^{p-q-1}(\Tilde{x}_1)\}$ is a non-empty subset of $\mathcal{B}$, where $p-q$ is the smallest integer such that $\operatorname{ad}_{\Tilde{x}_2}^{p-q}(\Tilde{x}_1)\propto\Tilde{x}_1$, as these satisfy the relations \eqref{eqn:desired:basis:overcomplete}. By the properties established above: (i) all brackets $[\operatorname{ad}_{\Tilde{x}_2}^{j}(\Tilde{x}_1),\operatorname{ad}_{\Tilde{x}_2}^k(\Tilde{x}_1)]$ vanish, and $[\Tilde{x}_2,\operatorname{ad}_{\Tilde{x}_2}^\ell(\Tilde{x}_1)]=\operatorname{ad}_{\Tilde{x}_2}^{\ell+1}(\Tilde{x}_1)\in\mathcal{B}'$ for all $\ell\in\{0,\ldots,p-q-2\}$, while $[\Tilde{x}_2,\operatorname{ad}_{\Tilde{x}_2}^{p-q-1}(\Tilde{x}_1)]\propto\Tilde{x}_1\in\mathcal{B}'$; (ii) $0\notin\mathcal{B}'$; and (iii) For any $x,y\in\mathcal{B}'$, $x\propto y$  if and only if $x=y$. It is furthermore clear that $\mathcal{B}'$ spans $L_\alpha^3$, since $\Tilde{x}_2$ has non-vanishing support in $L_\alpha^3\setminus V$ and $\Tilde{x}_1$ and $\operatorname{ad}_{\Tilde{x}_2}(\Tilde{x}_1)$ are linearly independent elements within the two-dimensional subspace $V$.

    One is now left with determining the smallest $\ell$ such that $\operatorname{ad}_{\Tilde{x}_2}^\ell(\Tilde{x}_1)\propto\Tilde{x}_1$. This is a simple eigenvalue problem, and we start, therefore, by computing the spectrum of the linear map $\operatorname{ad}_{\Tilde{x}_2}|_V$. The eigenvalues of this map are: $\operatorname{Specc}(\operatorname{ad}_{\Tilde{x}_2}|_V)=\{(1\pm i\sqrt{|1+4\alpha|})/2\}$. Thus, the eigenvalues of $\operatorname{ad}_{\Tilde{x}_2}^\ell|_V$ are given by
    \begin{align*}
        \operatorname{Specc}(\operatorname{ad}_{\Tilde{x}_2}^\ell|_V)=\left\{\sqrt{|\alpha|^\ell}\exp\left(\pm 2i\arctan\left(\frac{\sqrt{|1+4\alpha|}}{(1+2\sqrt{|\alpha|})}\right)\ell\right)\right\}.
    \end{align*}
    These eigenvalues are real for some $\ell\in\N_{>0}$ if and only if $\arctan(\sqrt{|1+4\alpha|}/(1+2\sqrt{|\alpha|}))/\pi\in\mathbb{Q}^*$. Conversely, if this rationality condition holds, choose $0\neq v\in V$, and let $d\geq2$ be minimal such that $\operatorname{ad}_{x_3}^d|_V(v)\propto v$. Then $\{x_3,v,\operatorname{ad}_{x_3}|_V(v),\ldots,\operatorname{ad}_{x_3}^{d-1}|_V(v)\}$ is a finite overcomplete spanning set satisfying relation~\eqref{eqn:desired:basis:overcomplete}.

    This leaves us with the final case $\alpha=-1/4$. In this scenario, the linear map $\operatorname{ad}_{\Tilde{x}_2}$ has only a one-dimensional eigenspace, spanned by $\Tilde{x}_1:=x_1-2x_2$, corresponding to the doubly degenerate eigenvalue $1/2$. In order for $\mathcal{B}$ to span $L_{-1/4}^3$, it must therefore also contain an element $\Tilde{x}_3=c_1^{(3)}x_1+c_2^{(3)}x_2+c_3^{(3)}x_3$, which does not belong to $\spn\{\Tilde{x}_1,\Tilde{x}_2\}$. We can now, without loss of generality, consider the following two cases:
    \begin{itemize}
        \item \textbf{Case 1:} $\boldsymbol{c_3^{(3)}=1}$. Here, we compute: $[\Tilde{x}_3,\Tilde{x}_2]=-(c_2^{(2)}-c_2^{(3)})x_1/4+(c_1^{(2)}+c_2^{(2)}-c_1^{(3)}-c_2^{(3)})x_2$. It is straightforward to verify that $[\Tilde{x}_3,\Tilde{x}_2]\propto \Tilde{x}_1$ if and only if the set $\{\Tilde{x}_1,\Tilde{x}_2,\Tilde{x}_3\}$ is linearly dependent and that $[\Tilde{x}_3,\Tilde{x}_2]=0$ if and only if $\Tilde{x}_3\propto\Tilde{x}_2$. The latter case is prohibited by Definition~\ref{def:graph:admissible}, while the former case requires $\mathcal{B}$ to contain a fourth element outside the span of $\{\Tilde{x}_1,\Tilde{x}_2,\Tilde{x}_3\}$. We may therefore assume that $\Tilde{x}_3$ is such that $[\Tilde{x}_2,\Tilde{x}_3]$ neither vanishes nor is proportional to $\Tilde{x}_1$. $\mathcal{B}$ must contain an element proportional to $\operatorname{ad}_{\Tilde{x}_2}(\Tilde{x}_3)$. After rescaling and relabelling this element, we obtain an element of $\mathcal{B}$ belonging to $V\setminus\spn\{\Tilde{x}_1\}$, which is covered by the next case.
        
        \item \textbf{Case 2:} $\boldsymbol{c_3^{(3)}=0}$. Here, we have that $\Tilde{x}_3\in V\setminus\spn\{\Tilde{x}_1\}$, and we compute: $[\Tilde{x}_2,\Tilde{x}_3]=-c_2^{(3)}x_1/4+(c_1^{(3)}+c_2^{(3)})x_2$. This bracket is non-zero, since its vanishing would imply $c_2^{(3)}=c_1^{(3)}=0$, contradicting $\Tilde{x}_3\neq0$. Moreover, this bracket is proportional to $\Tilde{x}_1$ if and only if $\Tilde{x}_3\propto\Tilde{x}_1$, which is excluded by assumption. Thus, the graph-admissibility condition implies that $\mathcal{B}$ contains an element proportional to $\operatorname{ad}_{\Tilde{x}_2}(\Tilde{x}_3)=[\Tilde{x}_2,\Tilde{x}_3]$. We now need to show that $\mathcal{B}$ contains an element proportional to $\operatorname{ad}_{\Tilde{x}_2}^\ell(\Tilde{x}_3)$ for every $\ell\in\Np{}$ and that these iterates are pairwise non-proportional.

        Let us first show the latter. Suppose that there exist distinct $p,q\in\Np{}$ such that $\operatorname{ad}_{\Tilde{x}_2}^p|_V(\Tilde{x}_3)=\lambda\operatorname{ad}_{\Tilde{x}_2}^q|_V(\Tilde{x}_3)$ for some $\lambda\in\R^*$. Without loss of generality, we may assume $p>q$. Let us start by defining the map $N:=\operatorname{ad}_{\Tilde{x}_2}|_V-\frac{1}{2}\mathds{1}_V$. A straightforward computation shows that $N^2=0$ and $\ker(N)=\spn\{\Tilde{x}_1\}$. Next, we show by induction that
        \begin{align*}
            \operatorname{ad}_{\Tilde{x}_2}^\ell|_V=2^{-\ell}(\mathds{1}_V+2\ell N)=2^{-\ell}\left(2\ell\operatorname{ad}_{\Tilde{x}_2}|_V-(\ell-1)\mathds{1}_V\right)
        \end{align*}
        holds for all $\ell\in\Np{}$. The base case $\ell=0$ is trivial. Let us therefore proceed to the induction step and assume that the induction hypothesis holds for some $\ell\in\Np{}$. Here, it is helpful to observe that $N^2=0$ implies that $\operatorname{ad}_{\Tilde{x}_2}^2|_V=\operatorname{ad}_{\Tilde{x}_2}|_V-\frac{1}{4}\mathds{1}_V$. This allows us to compute:
        \begin{align*}
            \operatorname{ad}_{\Tilde{x}_2}^{\ell+1}|_V&=2^{-\ell}\left(2\ell\operatorname{ad}_{\Tilde{x}_2}^{2}|_V-(\ell-1)\operatorname{ad}_{\Tilde{x}_2}|_V\right)=2^{-\ell}\left((\ell+1)\operatorname{ad}_{\Tilde{x}_2}|_V-\frac{1}{2}\ell\mathds{1}_V\right)\\
            &=2^{-(\ell+1)}\left(2(\ell+1)\operatorname{ad}_{\Tilde{x}_2}|_V-\ell\mathds{1}_V\right),
        \end{align*}
        which completes the induction step and therefore proves the claimed identity. Next, we observe that this implies
        \begin{align*}
            N\operatorname{ad}_{\Tilde{x}_2}^p|_V(\Tilde{x}_3)=2^{-p}N\Tilde{x}_3=2^{-q}\lambda N\Tilde{x}_3.
        \end{align*}
        Since $\Tilde{x}_3\notin\ker(N)$, we have $\lambda=2^{q-p}$. Inserting this into the assumption that $\operatorname{ad}_{\Tilde{x}_2}^p|_V(\Tilde{x}_3)=\lambda\operatorname{ad}_{\Tilde{x}_2}^q|_V(\Tilde{x}_3)$, and using the fact that $\operatorname{ad}_{\Tilde{x}_2}|_V$ is invertible, we find that $2^{p-q}\operatorname{ad}_{\Tilde{x}_2}^{p-q}|_V(\Tilde{x}_3)=\Tilde{x}_3$. Comparing this with the induction hypothesis applied to $\Tilde{x}_3$ for $\ell=p-q$ shows that $2(p-q)N\Tilde{x}_3=0$. Since $\Tilde{x}_3\notin\ker(N)$, this implies $p=q$, contradicting $p>q$. Consequently, the elements $\operatorname{ad}_{\Tilde{x}_2}^\ell(\Tilde{x}_3)$ are pairwise non-proportional.
        
        It remains to show that $\mathcal{B}$ contains a proportional representative of every such iterate. We prove this by induction over $\ell\in\Np{}$. For $\ell=0$, set $b_0:=\Tilde{x}_3\in\mathcal{B}$. Suppose that there exists $b_\ell\in\mathcal{B}$ such that $b_\ell\propto\operatorname{ad}_{\Tilde{x}_2}^\ell(\Tilde{x}_3)$. Since $\operatorname{ad}_{\Tilde{x}_2}|_V$ is invertible, one has $[\Tilde{x}_2,b_\ell]\propto\operatorname{ad}_{\Tilde{x}_2}^{\ell+1}(\Tilde{x}_3)\neq0$. Graph-admissibility therefore implies that $\mathcal{B}$ contains, for all $\ell\in\Np{}$, an element $b_{\ell+1}$ proportional to $\operatorname{ad}_{\Tilde{x}_2}^{\ell+1}(\Tilde{x}_3)$. Hence, $\mathcal{B}$ contains a proportional representative of every iterate. Since these iterates are pairwise non-proportional, the representatives $b_\ell$ are pairwise distinct. Thus, $\mathcal{B}$ would be infinite, contradicting the assumption that it is finite. This shows that $L_{-1/4}^3$ is non-graph-admissible. This completes the proof, since
        \begin{align*}
            \left.\frac{1}{\pi}\arctan\left(\frac{\sqrt{|1+4\alpha|}}{1+2\sqrt{|\alpha|}}\right)\right|_{\alpha=-1/4}=0\notin\Q^*.
        \end{align*}
    \end{itemize}
\end{proof}

\begin{tcolorbox}[breakable, colback=Cerulean!3!white,colframe=Cerulean!85!black,title=\textbf{Example}: Graph-admissible Lie algebras from the family $L_\alpha^3$] 
    \begin{example}\label{exa:different:L:alpha:3:graphs}
        We examine the family of real Lie algebras $L_\alpha^3$ for specific values of the parameter $\alpha\in\R$, to illustrate the distinction between minimal-graph-admissible and redundant-graph-admissible:
        \begin{enumerate}[label = (\roman*)]
            \item \textbf{Case} $\boldsymbol{\alpha=0}$: Here, the bracket relations of the basis elements $x_1$, $x_2$, $x_3$ are: $[x_1,x_2]=0$, $[x_1,x_3]=-x_2$, and $[x_2,x_3]=-x_2$. This basis satisfies the minimal-graph-admissible condition \eqref{eqn:desired:basis}. The corresponding graph is depicted in Figure~\ref{fig:L:alplha:examples} (a).
            \item \textbf{Case} $\boldsymbol{\alpha=1}$: Here, one has $[x_1,x_2]=0$, $[x_1,x_3]=-x_2$, and $[x_2,x_3]=-(x_1+x_2)$. This implies $[x_3,Ax_1+Bx_2]=Bx_1+(A+B)x_2$, revealing the emergence of the Fibonacci sequence: We can choose $A=F_0=0$ and $B=F_1=1$ and find $\operatorname{ad}_{x_3}^k(x_2)=F_kx_1+F_{k+1}x_2$ for all $k\in\N_{\geq0}$, where $F_k$ is the Fibonacci Sequence defined by $F_{n+1}=F_n+F_{n-1}$ and the initial values $F_0=0$ and $F_1=1$ \textup{\cite{Lucas:1891}}. We can now define the \emph{golden ratio} $\varphi$ as the largest root of the polynomial $X^2-X-1$ and $\vartheta$ as the smaller root \textup{\cite{Schneider:2016}}, and we have $\varphi=(1+\sqrt{5})/2$ and $\vartheta=(1-\sqrt{5})/2$. A basis of $L_{1}^3$ satisfying the relations \eqref{eqn:desired:basis} is consequently $\{x_3,x_1+\varphi x_2,x_1+\vartheta x_2\}$, since they are linearly independent and
            \begin{align*}
                [x_3,x_1+\varphi x_2]&=\varphi(x_1+\varphi x_2),\;&\;[x_3,x_1+\vartheta x_2]&=\vartheta(x_1+\vartheta x_2),\;&\;[x_1+\varphi x_2,x_1+\vartheta x_2]&=0.
            \end{align*}
            The graph associated to $L_{1}^3$ is depicted in Figure~\ref{fig:L:alplha:examples} (b).
            \item \textbf{Case} $\boldsymbol{\alpha=-1}$: The real Lie algebra $L_{-1}^3$ is, by the proof of Theorem~\ref{thm:existncae:of:non:minimal:graph:admissible:}, not minimal-graph-admissible. It is, however, by Lemma~\ref{lem:L:alpha:three:graph:admissible:if:and:only:if} graph-admissible. This can be seen by considering the overcomplete basis $\{x_1,x_2,x_3,x':=x_1-x_2\}$. This satisfies the relations \eqref{eqn:desired:basis:overcomplete}, since:
            \begin{align*}
                [x_1,x_2]&=0,\;&\;[x_1,x_3]&=-x_2,\;&\;[x_2,x_3]&=x',\\
                [x_1,x']&=0,\;&\;[x_2,x']&=0,\;&\;[x_3,x']&=x_1.
            \end{align*}
            The corresponding graph is depicted in Figure~\ref{fig:L:alplha:examples} (c). 
        \end{enumerate}
    \end{example}

    \begin{figure}[H]
        \centering
        \includegraphics[width=0.85\linewidth]{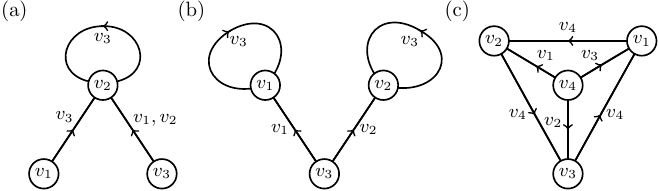}
        \caption{Depiction of the graphs from Example~\ref{exa:different:L:alpha:3:graphs}. Graph (a) is associated with the Lie algebra $L_0^3$ with $v_1=x_1$, $v_2=x_2$, and $v_3=x_3$. Graph (b) is associated with the Lie algebra $L_{1}^3$ with $v_1=x_1+\varphi x_2$, $v_2=x_1+\vartheta x_2$, and $v_3=x_3$. Graph (c) is associated with the Lie algebra $L_{-1}^3$ with $v_1=x_1$, $v_2=x_2$, $v_3=x_1-x_2$, and $v_4=x_3$.}
        \label{fig:L:alplha:examples}
    \end{figure}
\end{tcolorbox}

\section{Further remarks on proper graphs }\label{app:field:dependence:choice:dependent:graphs}
In Section~\ref{sec:valid:graphs} we investigated the rules for distinguishing valid proper graphs, i.e., labeled directed graphs that can be associated with a given graph-admissible Lie algebra, from improper graphs, i.e., graphs that cannot. We demonstrated that certain graphs can only be associated with a given Lie algebra if edges of the form $e=(v_\mathrm{s},v_\mathrm{l},v_\mathrm{e})$ are translated into Lie brackets $[v_\mathrm{s},v_\mathrm{l}]=\kappa v_\mathrm{e}$ with $\kappa\in\mathbb{F}^*$ under the conditions that the coefficients $\kappa$ satisfy specific constraints. This makes the corresponding graphs choice-dependent. Importantly, this dependence on the proportionality constants is affected by the characteristic of the underlying field $\mathbb{F}$, since the Jacobi identity
\begin{align*}
    0=[x,[y,z]]+[y,[z,x]]+[z,[x,y]]
\end{align*}
holds also when the right-hand side coincides with $pz$, where $p=\operatorname{char}(\mathbb{F})$ and $z$ is any element in $\g$.
Consequently, any criterion that differentiates between choice-independent and choice-dependent graphs must account for this subtlety. To illustrate this point consider Example~\ref{exa:derived:algebra:of:solvable:algebra:not:nilpotent} below, which demonstrates how the validity of the Jacobi identity can hinge on the choice of coefficients and the field's characteristic. 

\begin{tcolorbox}[breakable, colback=Cerulean!3!white,colframe=Cerulean!85!black,title=\textbf{Example}: The derived algebra of a solvable Lie algebra]
    Let $\g$ be a solvable Lie algebra over a field $\mathbb{F}$. If $\operatorname{char}(\mathbb{F})=0$, then its derived algebra $\mathcal{D}\g$ is nilpotent \cite{Jacobson:1966}. However, when $\operatorname{char}(\mathbb{F})\neq0$, this property generally does not hold \cite{Jacobson:1966}. The following example illustrates a solvable yet graph-admissible Lie algebra whose first derived algebra fails to be nilpotent for $\operatorname{char}(\mathbb{F})>0$. This example underscores the interplay between choice-dependent graphs and the selection of coefficients, demonstrating that these choices are inherently linked to the characteristic of the underlying field $\mathbb{F}$.
    \begin{example}\label{exa:derived:algebra:of:solvable:algebra:not:nilpotent}
        Let $\mathbb{F}$ be a field of characteristic $p$, and let $\g$ denote the $(p+2)$-dimensional Lie algebra spanned by the basis elements $\{x,y,e_j\}_{j=1}^p$, which satisfy the following bracket relations, see \cite{Jacobson:1966} for reference:
        \begin{align}\label{eqn:lie:algebra:characteristic:dependent}
            [x,y]&=x,\;&\;[x,e_j]&=-e_{j+1}\text{ for }j\leq p-1,\;&\;[x,e_p]&=-e_1,\;&\;[y,e_j]&=-(j-1)e_j,
        \end{align}
        while all remaining brackets between basis elements vanish. Then, the Jacobi identity yields:
        \begin{subequations}
        \begin{align}
            [e_j,[x,y]]+[x,[y,e_j]]+[y,[e_j,x]]&=e_{j+1}+(j-1)e_{j+1}-je_{j+1}=0,&&\text{for }j\leq p-1\\ 
            [e_p,[x,y]]+[x,[y,e_p]]+[y,[e_p,x]]&=e_1+(p-1)e_1=pe_1=0,&&\text{for }j= p.\label{eqn:jacobi:identity:depending:on:char}
        \end{align}
        \end{subequations}
        All other instances of the Jacobi identity clearly vanish identically. Thus, one can construct a graph for this Lie algebra using the basis $\mathcal{B}=\{x,y,e_j\}_{j=1}^p$. The corresponding graph is shown in Figure~\ref{fig:field:charactersitc:dependent}. However, if we change the field $\mathbb{F}$ to another field $\mathbb{F}'$ with $\operatorname{char}(\mathbb{F}')=q\neq p$, the Jacobi identity computed in \eqref{eqn:jacobi:identity:depending:on:char} would be violated. Consequently, the graph $G(V,E)$ cannot be associated with a Lie algebra spanned by elements in $\mathcal{B}$ which satisfy the bracket relations \eqref{eqn:lie:algebra:characteristic:dependent}. This failure directly affects the parameter choices for translating edges $e\in E$ into valid Lie bracket relations. 
        
        Such a situation can never occur over a field of characteristic zero, since no  nonzero multiple of the unit element ever vanishes. This observation emphasizes the choice-dependent nature of the graph representation, which is inherently tied to the characteristic of the underlying field. 
    \end{example}

    \begin{figure}[H]
        \centering
        \includegraphics[width=0.95\linewidth]{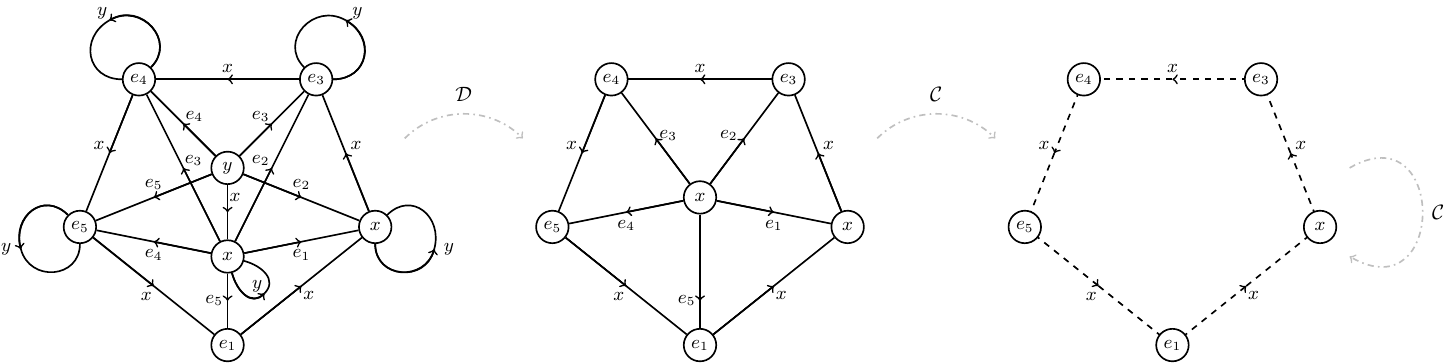}
        \caption{Illustration of the graph constructed in Example~\ref{exa:derived:algebra:of:solvable:algebra:not:nilpotent} along with the first derived graph and its lower central series shown for the case $p=5$.}
        \label{fig:field:charactersitc:dependent}
    \end{figure}
\end{tcolorbox}

This example motivates the following statement:

\begin{corollary}\label{cor:solbaleb:but:derived:not:nilpotent}
    Let $\g$ be a finite-dimensional graph-admissible Lie algebra over a field $\mathbb{F}$ associated with the labeled directed graph $G(V,E)$. If $G(V,E)$ contains no self-contained subgraph induced by a closed directed walk, but the derived graph $\mathcal{D}G(V,E)$ does contain a closed directed walk, then $\operatorname{char}(\mathbb{F})\neq 0$.
\end{corollary}

From the preceding example, it is furthermore clear that a graph described by the conditions in Corollary~\ref{cor:solbaleb:but:derived:not:nilpotent} is generally choice-dependent.

\begin{proof}
    This claim is a direct consequence of Theorems~\ref{thm:non:solvability:condition:strong} and~\ref{thm:nilpotency:criteria:strong}, since Theorem~\ref{thm:non:solvability:condition:strong} implies that $\g$ is solvable if $G(V,E)$ contains no self-contained subgraph induced by a closed directed walk. Furthermore, Theorem~\ref{thm:nilpotency:criteria:strong} implies that the derived algebra $\mathcal{D}\g$ is not nilpotent, since the derived graph $\mathcal{D}G(V,E)$ contains a closed directed walk. However, such a situation is prohibited if $\g$ is defined over a field $\mathbb{F}$ of characteristic $\operatorname{char}(\mathbb{F})=0$ \cite{Jacobson:1966}.
\end{proof}

\subsection{Summary of the criteria for proper graphs}\label{app:subsec:tables:summariting:proper:graph:conditions}

This subsection summarizes the current status of the criteria obtained in Section~\ref{sec:valid:graphs} for determining whether a labeled directed graph is proper. Table~\ref{tab:summary:criteria:proper:graphs} collects the individual necessary and sufficient conditions, together with their respective scopes and limitations. Table~\ref{tab:status:proper:graphs:minimal:redundant} organizes these results according to the Lie-algebra axioms and distinguishes between minimal and redundant graphs, while also indicating the conditions that remain to be established.

\begin{table}[htpb]
    \centering
    \caption{Summary of the available criteria derived in Section~\ref{sec:valid:graphs} for determining whether a labeled directed graph is proper. Here, \enquote{necessary} and \enquote{sufficient} refer to the stated scope of the respective criterion.}
    \label{tab:summary:criteria:proper:graphs}
    \scriptsize
    \begin{tblr}{width=\linewidth,colspec={X[2.0,l] X[0.55,c] X[0.65,c] X[1.45,l] X[2.3,l]},rows={valign=t},row{1}={font=\bfseries},hlines,rowsep=2pt,colsep=3pt}
        Criterion or result & Necessary & Sufficient & Scope & Remaining limitation \\
        
        Local edge conditions from Proposition~\ref{prop:conditions:for:edges} & Yes & No & Minimal and redundant graphs & These conditions solely encode consequences of the antisymmetry of the Lie bracket. \\
        
        Bounds on number of edges and vertices from Propositions~\ref{prop:criteria:number:of:vertices:and:edges} and~\ref{prop:tight:bounds} & Yes & No & Minimal graphs & Conditions~(a)-(d) can straightforwardly be lifted to redundant graphs by replacing $n$ with $m$.\\

        Three-vertex-subgraph condition from Theorem~\ref{thm:neccersary:conditions:proper:graph} & Yes & No & Minimal graphs with at least three vertices & Local admissibility of every self-contained three-vertex subgraph does not imply global properness. \\
        
        Complete classification from Corollary~\ref{cor:proper:choice:dependent} & Yes & Yes & Minimal graphs with exactly three vertices & Graphs of Type~X are choice-dependent, and thus properness additionally requires compatible non-zero proportionality constants. \\
        
        Special classification from Corollary~\ref{cor:criteria:three:vertices:vanish:under:twO.brackets} & Yes & No & Three-vertex graphs for which the relevant nested brackets vanish & The result applies only to the special case in which nested brackets of three distinct basis elements vanish. \\
        
        Trail criterion from Proposition~\ref{prop:allowed:trails} & Yes & No & Minimal graphs & The criterion identifies possible cancellation patterns enforced by the Jacobi identity, but does not determine whether compatible non-zero proportionality constants exist. \\
        
        Claim~(i) from Proposition~\ref{prop:conditions:proper:choice:independent:proper:minimal} & No & Yes & Minimal and redundant graphs satisfying a specific set of local edge conditions & This is a strong sufficient condition requiring every relevant nested bracket to vanish individually and is therefore not necessary. \\
        
        Claim~(ii) from Proposition~\ref{prop:conditions:proper:choice:independent:proper:minimal} & Yes & No & Minimal graphs & The condition prevents non-zero nested brackets of three distinct basis elements from being forced into several linearly independent directions, but does not guarantee cancellation of their coefficients. \\
    \end{tblr}
\end{table}

\begin{table}[htpb]
    \centering
    \caption{Summary of the extent to which the Lie-algebra axioms can be verified from a labeled directed graph in the minimal and redundant settings using the results from Section~\ref{sec:valid:graphs}, making such graphs proper.}
    \label{tab:status:proper:graphs:minimal:redundant}
    \scriptsize
    \begin{tblr}{width=\linewidth,colspec={X[1.30] X[2.25] X[2.25]},cells={halign=c,valign=t},hspan=minimal,row{1}={font=\bfseries},hlines,rowsep=3pt,colsep=4pt}
    
        Lie-algebra axiom & Minimal graphs & Redundant graphs \\
        
        {\emph{Bilinearity} \\ (Definition~\ref{def:assocaited:graph})} & \SetCell[c=2]{wd=0.79\linewidth,c} { Presupposed by construction and therefore not subject to an additional graphical condition.}& \\
        
        {\emph{Antisymmetry} \\(addressed by Proposition~\ref{prop:conditions:for:edges}; \\ partially addressed by Theorem~\ref{thm:neccersary:conditions:proper:graph}, Corollary~\ref{cor:proper:choice:dependent}, and Proposition~\ref{prop:conditions:proper:choice:independent:proper:minimal})} & \SetCell[c=2]{wd=0.79\linewidth,c} {Proposition~\ref{prop:conditions:for:edges} provides the necessary local edge conditions for both minimal and redundant graphs: the starting and labeling vertices of an edge must be distinct, every edge must possess its antisymmetric counterpart, and its end vertex must be uniquely determined by its starting and labeling vertices. At the graphical level, these conditions are also sufficient when checking antisymmetry alone, without additionally imposing the Jacobi identity.} & \\
        
        \SetCell[r=5]{c} {\emph{Jacobi identity}\\ (partially addressed by Theorem~\ref{thm:neccersary:conditions:proper:graph}, Corollary~\ref{cor:proper:choice:dependent}, Proposition~\ref{prop:allowed:trails}, and Proposition~\ref{prop:conditions:proper:choice:independent:proper:minimal})} & Theorem~\ref{thm:neccersary:conditions:proper:graph} provides a necessary condition by restricting the possible self-contained three-vertex subgraphs of every proper-minimal graph. & No analogous condition is presented in this work. Such a condition could in principle be investigated by similar means, but it would generally be less restrictive because linear dependencies among the vertex elements permit additional cancellations. \\
        
        & Corollary~\ref{cor:proper:choice:dependent} provides a complete necessary-and-sufficient classification for all three-vertex graphs into proper, improper, choice-independent, and choice-dependent graph types. & A corresponding necessary-and-sufficient classification of redundant three-vertex graphs has not been established. \\
        
        & Proposition~\ref{prop:allowed:trails} provides a necessary condition for arbitrary proper-minimal graphs by identifying the possible directed-trail configurations corresponding to cancellations in the Jacobi identity. & An analogous trail criterion has not been established because linear dependencies among the vertices permit additional cancellation patterns that are absent in the minimal setting, and would therefore lead to a less restrictive criterion. \\
        
        & Proposition~\ref{prop:conditions:proper:choice:independent:proper:minimal}~(i) provides a sufficient condition in the case that $\delta(j,\delta(k,\ell))$ vanishes for all relevant index triples. Indeed, the Jacobi identity is automatically satisfied when all nested brackets of three distinct basis elements appearing in it vanish. & The sufficient condition $\delta(j,\delta(k,\ell))=0$ from Proposition~\ref{prop:conditions:proper:choice:independent:proper:minimal}~(i) remains applicable because it forces every nested bracket appearing in the Jacobi identity to vanish individually for all elements of the algebra. \\
        & Proposition~\ref{prop:conditions:proper:choice:independent:proper:minimal}~(ii) provides a necessary condition by restricting the possible non-zero output directions of nested brackets. & No directly analogous condition is obtained in this work because nested brackets terminating at distinct but linearly dependent vertex elements may cancel.\\
        
        \SetCell[r=3]{c} \textbf{Missing conditions} & --- & {A general necessary graph-based criterion for the Jacobi identity, similar to the one established by Proposition~\ref{prop:allowed:trails} for minimal graphs, is still missing. \\ \emph{Principal difficulty}: Linear dependencies among the vertex elements permit additional cancellation patterns that are not encoded by the graph itself.}\\
        
        & {A general sufficient criterion for the Jacobi identity on graphs with more than three vertices is still missing. \\\emph{Principal difficulty}: An associated graph records only the proportionality directions of the non-zero brackets, whereas the Jacobi identity depends, for choice-dependent graphs, on their actual non-zero proportionality constants.} &  {A general sufficient criterion for the Jacobi identity beyond the special condition in Proposition~\ref{prop:conditions:proper:choice:independent:proper:minimal}~(i) is still missing. \\\emph{Principal difficulty}: The graph records neither the non-zero proportionality constants nor the linear dependencies among the vertex elements relevant for verifying the Jacobi identity.} \\
        & \SetCell[c=2]{wd=0.79\linewidth,c} { A general purely graph-theoretic necessary-and-sufficient criterion incorporating both antisymmetry and the Jacobi identity is still missing.} \\
    \end{tblr}
\end{table}

\section{Runtime discussion of the main-text algorithms}\label{app:runtime_discussion}
In this appendix, we want to derive rough runtime estimates for the main-text algorithms, namely Algorithms~\ref{alg:creating:graph}, \ref{alg:generating:generating:the:graph:derived:altered}, and~\ref{alg:generating:lower:central:series}.  Note that the purpose of this analysis is not to determine optimal or asymptotically tight bounds. Instead, we derive simple and conservative worst-case upper bounds for direct implementations of the pseudocode  provided in the main text above. In particular, no attempt is made to optimize the algorithms in concrete implementations.

Let us begin by recalling the asymptotic notation employed for this purpose, and let
$f,g:\Np{}\to\R$ be two functions such that $g(n)\neq0$ for sufficiently large $n$. Then, $f(n)\in\mathcal{O}(g(n))$ if and only if there exists a positive real constant $C>0$ and a non-negative integer $n_0\in\Np{}$ such that $|f(n)|\leq C|g(n)|$ for all $n\geq n_0$~\cite{Cormen:2009}. Similarly, we write $f(n)\in\Omega(g(n))$ if and only if there exists a positive real constant $c>0$ and a non-negative integer $n_0\in\Np{}$ such that $|f(n)|\geq c|g(n)|$ for all $n\geq n_0$~\cite{Cormen:2009}. Thus, $\mathcal{O}(g(n))$ and $\Omega(g(n))$ provide asymptotic upper and lower bounds for $f$, respectively. If both bounds hold, we write $f(n)\in \Theta(g(n))$, or equivalently $\Theta(g(n))=\Omega(g(n))\cap \mathcal{O}(g(n))$~\cite{Cormen:2009}.

For the subsequent runtime discussion, we employ a unit-cost computational model in which elementary operations, such as checking whether a specific vertex belongs to a given vertex set, the removal of an edge or vertex, access to $\varpi_\mathrm{s}(e), \varpi_\mathrm{l}(e), \varpi_\mathrm{e}(e)$ for a given edge, or checking whether a coefficient $\alpha_{jk}$ vanishes, require $\mathcal{O}(1)$ time. Moreover, we assume that an admissible spanning set $\mathcal{B}=\{x_j\}_{j\in\mathcal{M}}$, satisfying the corresponding admissibility condition from Definition~\ref{def:graph:admissible}, is supplied along with the matrix $\alpha_{jk}$ and the map $\delta$, as required by the bracket structure $[x_j,x_k]=\alpha_{jk}x_{\delta(j,k)}$. Thus, the cost of finding a suitable admissible spanning set or of verifying graph-admissibility is not included in the following estimates.

Recall that we employ the convention $m:=|\mathcal{M}|$, where by the requirements from Definition~\ref{def:graph:admissible}, one has $|\mathcal{M}|=|\mathcal{B}|$. Similarly, we define for convenience $e:=|E|$ as the number of labeled directed edges in the associated graph. By the same logic employed to show Proposition~\ref{prop:criteria:number:of:vertices:and:edges}, one has $e\leq m(m-1)$ and hence $\mathcal{O}(e)\subseteq\mathcal{O}(m^2)$.

Before discussing Algorithms~\ref{alg:generating:generating:the:graph:derived:altered} and~\ref{alg:generating:lower:central:series}, we need to supplement the respective algorithms with a suitable stopping rule. As currently stated, the pseudocode contains no explicit stopping condition and may therefore continue even after reaching the trivial graph or a fixed point. Thus, at the end of every pruning round, the newly obtained graph must be compared with the graph from the preceding round and with the trivial graph. The respective algorithms shall now terminate if the new graph coincides with the preceding graph or is the trivial graph. For this purpose, it is sufficient to store the graphs from the current and preceding pruning rounds. The cost associated with copying and comparing the relevant graphs during the $\ell$-th pruning round is bounded by $\mathcal{O}(|\mathcal{D}^\ell V|+|\mathcal{D}^\ell E|)$ and $\mathcal{O}(|\mathcal{C}^\ell V|+|E_\mathrm{r}|+|\mathcal{C}^\ell E|)$ for the Algorithms~\ref{alg:generating:generating:the:graph:derived:altered} and~\ref{alg:generating:lower:central:series}, respectively.

The pruning algorithms never add vertices to the sets $\mathcal{D}^\ell V$ and $\mathcal{C}^\ell V$ or edges to the ordinary edge sets $\mathcal{D}^\ell E$ and $\mathcal{C}^\ell E$. In Algorithm~\ref{alg:generating:lower:central:series}, however, it can occur that an edge may be transferred from $\mathcal{C}^{\ell-1}E$ to the supplementary edge set $E_\mathrm{r}^{(\ell)}$. By construction, every removal or transfer of an edge occurs only as a consequence of the removal of at least one vertex. Thus, a graph can differ from its successively pruned graph only if at least one vertex is removed. Conversely, if no vertex is removed during a pruning round, the newly obtained graph coincides with the preceding graph, and a fixed point has been reached. Since the initial graph contains $m$ vertices, there are at most $m$ strictly changing rounds. Including the final round in which the trivial graph or a fixed point is detected, the algorithms therefore perform at most $m+1$ pruning rounds.

We first consider Algorithm~\ref{alg:creating:graph} for constructing an associated labeled directed graph.

\begin{proposition}\label{prop:runtime:alg:graph:creation}
    The runtime $\mathrm{Rt}(\text{Alg. 1}, \mathcal{B}, \boldsymbol{\alpha},\delta)$ of Algorithm~\ref{alg:creating:graph} satisfies  $\mathrm{Rt}(\text{Alg. 1}, \mathcal{B}, \boldsymbol{\alpha},\delta)\in\mathcal{O}(m^2)$.
\end{proposition}

\begin{proof}
    Algorithm~\ref{alg:creating:graph} first initializes the vertex set $V$ (cf. line 1) together with an empty edge set $E$ (cf. line 2), which requires at most $\mathcal{O}(m)$ elementary operations. It then considers every tuple of distinct indices $(j,k)\in\mathcal{M}\times\mathcal{M}$ with $j<k$ (cf. line 3). There are $m(m-1)/2\in\mathcal{O}(m^2)$ such tuples. For each tuple, the algorithm performs the constant-time test whether $\alpha_{jk}\neq0$ (cf. line 4). If $\alpha_{jk}$ does not vanish, exactly two directed edges are added to $E$ (cf. lines 5-6). Under the assumed unit-cost model, each such insertion requires $\mathcal{O}(1)$ time. If $\alpha_{jk}=0$, the algorithm immediately proceeds with the next tuple. Therefore, the runtime satisfies $\mathrm{Rt}(\text{Alg. 1}, \mathcal{B}, \boldsymbol{\alpha},\delta)\in\mathcal{O}(m+m^2+e)$. The earlier observation that $e\in\mathcal{O}(m^2)$ implies consequently $\mathrm{Rt}(\text{Alg. 1}, \mathcal{B}, \boldsymbol{\alpha},\delta)\in\mathcal{O}(m^2)$.
\end{proof}

We proceed with Algorithm~\ref{alg:generating:generating:the:graph:derived:altered} for constructing the graphs associated with the derived series. 

\begin{proposition}
    The runtime $\mathrm{Rt}(\text{Alg. 2}, \mathcal{B}, \boldsymbol{\alpha},\delta)$ of Algorithm~\ref{alg:generating:generating:the:graph:derived:altered} satisfies  $\mathrm{Rt}(\text{Alg. 2}, \mathcal{B}, \boldsymbol{\alpha},\delta)\in \mathcal{O}(m^4)$.
\end{proposition}

\begin{proof}
    We begin by defining $\mathcal{D}^\ell m:=|\mathcal{D}^\ell V|$ and $\mathcal{D}^\ell e:=|\mathcal{D}^\ell E|$, i.e., the number of vertices and edges contained in the graph obtained after the $\ell$-th pruning round. As discussed before, Algorithm~\ref{alg:generating:generating:the:graph:derived:altered} performs at most $m+1$ pruning rounds. Moreover, the sequences $\mathcal{D}^\ell m$ and $\mathcal{D}^\ell e$ are monotonically decreasing, as, at no stage of the algorithm, are vertices or edges added to the graph during the pruning procedure, only ever removed.

    For the following estimate, we observe that each pruning round is explicitly implemented synchronously. That is, the vertices and edges to be removed in the $\ell$-th pruning round are determined from the graph $\mathcal{D}^{\ell-1} G(V,E)=G(\mathcal{D}^{\ell-1}V,\mathcal{D}^{\ell-1}E)$, and the corresponding deletions are applied only after all vertices have been inspected. This ensures that the result of a pruning round does not depend on the order in which the vertices and edges are considered.

    The algorithm first constructs the initial associated graph (cf. lines 1-6), as described by Algorithm~\ref{alg:creating:graph}. By Proposition~\ref{prop:runtime:alg:graph:creation}, this initialization requires at most $\mathcal{O}(m^2)$ elementary operations.
    
    After that the pruning procedure begins, where, at each round, indexed by $\ell$, one considers every vertex $v\in \mathcal{D}^{\ell-1} V$ (cf. line 11). For each such vertex, the algorithm then determines whether there exists an edge $a\in \mathcal{D}^{\ell-1} E$ such that $\varpi_\mathrm{e}(a)=v$ (cf. line 12). This test is performed by scanning the complete edge set $\mathcal{D}^{\ell-1}E$ and therefore requires at most $\mathcal{O}(\mathcal{D}^{\ell-1} e)$ elementary operations for each vertex. If no such edge exists, the corresponding vertex $v$ is marked for removal from the set $\mathcal{D}^\ell V$ (cf. line 13). Marking or removing the vertex requires $\mathcal{O}(1)$ time. Moreover, all edges $b\in \mathcal{D}^{\ell-1} E$  for which at least one of $\varpi_\mathrm{s}(b)=v$, $\varpi_\mathrm{l}(b)=v$, or $\varpi_\mathrm{e}(b)=v$ holds for a vertex $v$ that is marked for removal must be removed from the edge set $\mathcal{D}^{\ell-1} E$ (cf. line 14-16,18). For a fixed vertex, this requires another complete scan of the edge set and therefore contributes at most $\mathcal{O}(\mathcal{D}^{\ell-1}e)$ elementary operations. Finally, copying the resulting graph $\mathcal{D}^\ell G(V,E)$ and comparing it with the graph $\mathcal{D}^{\ell-1}G(V,E)$ from the pruning round requires at most $\mathcal{O}(\mathcal{D}^{\ell-1} m+\mathcal{D}^{\ell-1}e)$ additional operations. Thus, combining these contributions yields a total runtime estimate of
    \begin{align*}
        \mathrm{Rt}(\text{Alg. 2}, \mathcal{B}, \boldsymbol{\alpha},\delta)\in\mathcal{O}\left(m^2+\sum_{\ell=1}^{m+1}\left(\mathcal{D}^{\ell-1} m\left(\mathcal{D}^{\ell-1}e+\mathcal{D}^{\ell-1} e\right)+\mathcal{D}^{\ell-1}m+\mathcal{D}^{\ell-1}e\right)\right).
    \end{align*}
    As discussed earlier, since $\mathcal{D}^\ell m\leq m$ and $\mathcal{D}^\ell e\leq e\in \mathcal{O}(m^2)$, one obtains
    \begin{align*}
        \mathrm{Rt}(\text{Alg. 2}, \mathcal{B}, \boldsymbol{\alpha},\delta)\in\mathcal{O}\left(m^2+(m+1)\left( m\left(m^2+m^2\right)+m+m^2\right)\right)=\mathcal{O}\left(m^4\right),
    \end{align*}
    as claimed.
\end{proof}

Finally, let us consider Algorithm~\ref{alg:generating:lower:central:series} for constructing graphs associated with the lower central series. 

\begin{proposition}
    The runtime $\mathrm{Rt}(\text{Alg. 3}, \mathcal{B}, \boldsymbol{\alpha},\delta)$ of Algorithm~\ref{alg:generating:lower:central:series} satisfies  $\mathrm{Rt}(\text{Alg. 3}, \mathcal{B}, \boldsymbol{\alpha},\delta)\in\mathcal{O}(m^4)$.
\end{proposition}

\begin{proof}
    Similar to before, we define $\mathcal{C}^\ell m:=|\mathcal{C}^\ell V|$, $\mathcal{C}^\ell e=|\mathcal{C}^\ell E|$, and $\mathcal{C}^\ell r:=|E_\mathrm{r}^{(\ell)}|$. Thus, $\mathcal{C}^\ell m$, $\mathcal{C}^\ell e$, and $\mathcal{C}^\ell r$ denote the numbers of vertices from $\mathcal{C}^\ell V$, ordinary edges from $\mathcal{C}^\ell E$, and supplementary edges from $E_\mathrm{r}^{(\ell)}$, respectively. The runtime analysis of Algorithm~\ref{alg:generating:lower:central:series} is similar to the one of Algorithm~\ref{alg:generating:generating:the:graph:derived:altered}, except that, here, the pruning steps are different. 

    The algorithm first constructs the initial associated graph in the same manner as described in Algorithm~\ref{alg:creating:graph} (cf. lines 1-7). Hence, by Proposition~\ref{prop:runtime:alg:graph:creation}, this initialization step requires at most $\mathcal{O}(m^2)$ elementary operations. This leaves us with providing a runtime estimate of the subsequent pruning round. As discussed above, the algorithm performs at most $m+1$ such pruning rounds.
    
    In the $\ell$-th pruning round, the algorithm first loops over all vertices $v\in \mathcal{C}^{\ell-1}V$ and determines whether there exists an edge $a\in \mathcal{C}^{\ell-1}E\cup E_\mathrm{r}^{(\ell-1)}$ such that $\varpi_\mathrm{e}(a)=v$ (cf. lines 14-15). This requires a complete scan of the two edge sets $\mathcal{C}^{\ell-1}E$ and $E_\mathrm{r}^{(\ell-1)}$ in the worst case. Consequently, the test costs at most $\mathcal{O}(\mathcal{C}^{\ell-1}e+\mathcal{C}^{\ell-1}r)$. If the deletion condition is satisfied, the corresponding vertex is marked for removal from the set $\mathcal{C}^{\ell-1}V$ (cf. line 16). This, together with the later removal (cf. line 17), contributes at most $\mathcal{O}(1)$ elementary operations per vertex $v\in\mathcal{C}^{\ell-1}V$. Hence, the complete loop over the vertices contributes at most $\mathcal{O}(\mathcal{C}^{\ell-1}m(\mathcal{C}^{\ell-1}e+\mathcal{C}^{\ell-1}r))$.
    
    In the next step, the algorithm performs a loop over every supplementary edge $a\in E_\mathrm{r}^{(\ell-1)}$ and checks whether both vertices $\varpi_\mathrm{s}(a)$ and $\varpi_\mathrm{e}(a)$ belong to $\mathcal{C}^{\ell}V$ (cf. lines 21-22). A single membership test constitutes in the unit-cost model $\mathcal{O}(1)$ elementary operations. Thus, the two membership tests contribute at most $\mathcal{O}(2)=\mathcal{O}(1)$ elementary operations. If either condition fails, the investigated edge is removed from $E_\mathrm{r}^{(\ell)}$ (cf. line 23,31), adding per edge at most $\mathcal{O}(1)$ elementary operations. Thus, this loop adds a cost of $\mathcal{O}(\mathcal{C}^{\ell-1}r)$ during the $\ell$-th pruning round. 
    
    Finally, the algorithm loops in each pruning round over all ordinary edges $a\in \mathcal{C}^{\ell-1}E$ (cf. line 24). For each such edge, it first checks the membership conditions involving the three vertices $\varpi_\mathrm{s}(a)$, $\varpi_\mathrm{l}(a)$, and $\varpi_\mathrm{e}(a)$ in the newly obtained vertex set $\mathcal{C}^{\ell}V$ (cf. line 25). Each of these membership tests has a cost bounded by $\mathcal{O}(1)$ in the employed unit-cost model. If not all three vertices belong to $\mathcal{C}^\ell V$, the edge is marked for removal from the ordinary edge set $\mathcal{C}^{\ell-1}E$ (cf. line 26,30). The algorithm then checks whether both $\varpi_\mathrm{s}(a)$ and $\varpi_\mathrm{e}(a)$ remain in the vertex set $\mathcal{C}^{\ell} V$ (cf. line 27). If this is the case, the corresponding edge $a$ is marked for transfer to the supplementary edge set $E_\mathrm{r}^{(\ell)}$ (cf. line 28,31). For every ordinary edge, this procedure involves only a constant number of membership tests, at most five, together with at most one deletion and one insertion. It therefore contributes at most $\mathcal{O}(\mathcal{C}^{\ell-1}e)$ elementary operations during the $\ell$-th pruning round.
    
    Finally, copying the resulting graph data and comparing the newly obtained graph data with the data from the previous pruning round requires at most $\mathcal{O}(\mathcal{C}^{\ell-1} m+\mathcal{C}^{\ell-1}e+\mathcal{C}^{\ell-1} r)$ additional elementary operations during the $\ell$-th pruning round. Combining all contributions, we obtain
    \begin{align*}
        \mathrm{Rt}(\text{Alg. 3}, \mathcal{B}, \boldsymbol{\alpha},\delta)\in\mathcal{O}\left(m^2+\sum_{\ell=1}^{m+1}\left(\mathcal{C}^{\ell-1}m(\mathcal{C}^{\ell-1} e+\mathcal{C}^{\ell-1} r)+\mathcal{C}^{\ell-1}r+\mathcal{C}^{\ell-1} e+\mathcal{C}^{\ell-1} m+\mathcal{C}^{\ell-1}e+\mathcal{C}^{\ell-1} r\right)\right).
    \end{align*}
    From construction, it is clear that $\mathcal{C}^\ell m\leq m$. Moreover, the ordinary and supplementary edge sets are disjoint subsets of the original edge set. Hence, $\mathcal{C}^\ell e,\mathcal{C}^\ell r,\mathcal{C}^\ell e+\mathcal{C}^\ell r\leq e\in\mathcal{O}(m^2)$. It follows that the runtime of Algorithm~\ref{alg:generating:lower:central:series} is bounded by
    \begin{align*}
        \mathrm{Rt}(\text{Alg. 3}, \mathcal{B}, \boldsymbol{\alpha},\delta)\in\mathcal{O}\left(m^2+(m+1)\left(m\left(m^2\right)+m^2+m+m^2\right)\right)=\mathcal{O}\left(m^4\right),
    \end{align*}
    as claimed.
\end{proof}

We want to emphasize that these estimates are deliberately conservative worst-case upper bounds for a direct implementation involving repeated scans of the edge sets. Their purpose is only to demonstrate that the algorithms terminate after a polynomial number of elementary operations in the size of the supplied admissible spanning set; no claim of asymptotic optimality is made.

\section{Remaining algorithms}\label{app:remaining:algorithms}
This section is dedicated to the remaining algorithms mentioned in the main text. These algorithms complement the constructions introduced earlier and provide explicit procedures for generating graphs associated with the derived series of Lie algebras. Their inclusion ensures completeness and facilitates practical implementation of the theoretical framework developed in the main text.

\subsection{Generating labeled directed graphs associated with the derived series of a given Lie algebra}
In Section~\ref{sec:structural:properties:derived:series} we introduced an approach for associating a labeled directed graph with the $\ell$-th derived algebra of a given Lie algebra $\g$. The procedure starts with the construction of a graph $G(V,E)$ using Algorithm~\ref{alg:creating:graph} followed by iterative pruning rules that yield a graph $\mathcal{D}^\ell G(V,E)$ corresponding to $\mathcal{D}^\ell\g$. These rules remove vertices and edges that fail to satisfy specific structural conditions.

This construction is informed by the recursively defined index sets
\begin{align*}
    \mathcal{M}_{\mathrm{D}}^{(\ell+1)}:=\delta\left(\left\{(j,k)\in\mathcal{M}_{\mathrm{D}}^{(\ell)}\times\mathcal{M}_{\mathrm{D}}^{(\ell)}\mid \;\Tilde{\alpha}_{jk}\neq0\right\}\right)\quad\text{for all }\ell\in\N_{\geq0},\quad\text{where}\quad\mathcal{M}_{\mathrm{D}}^{(0)}:=\mathcal{M}.
\end{align*}
Consequently, one can directly employ these index sets to obtain a graph associated with any derived algebra $\mathcal{D}^\ell\g$, provided that $\g$ is graph-admissible and a suitable (possibly overcomplete) basis $\mathcal{B}$ satisfying the relations \eqref{eqn:desired:basis} (or \eqref{eqn:desired:basis:overcomplete}) is known. A graph associated with $\mathcal{D}^\ell\g$ can then be constructed by applying Algorithm~\ref{alg:creating:graph} to the basis $\mathcal{B}$ when restricted to the elements with indices from $\mathcal{M}_\mathrm{D}^{(\ell)}$. For completeness, this procedure is formalized in Algorithm~\ref{alg:generating:generating:the:graph:derived:unaltered}.
\begin{algorithm}[htpb]
        \DontPrintSemicolon
        \KwData{A (possibly overcomplete) basis $\mathcal{B}=\{\Tilde{x}_j\}_{j\in\mathcal{M}}$ of the graph-admissible Lie algebra $\g$, satisfying the Lie bracket relations \eqref{eqn:desired:basis:overcomplete} along with the determining matrix $\boldsymbol{\alpha}$, the index function $\delta$, and a non-negative integer $\ell$.}
        \KwResult{A labeled directed graph associated with the derived algebra $\mathcal{D}^\ell\g$.}
        \SetKwData{Left}{left}\SetKwData{This}{this}\SetKwData{Up}{up}
        \SetKwFunction{Union}{Union}\SetKwFunction{FindCompress}{FindCompress}
        \SetKwInOut{Input}{input}\SetKwInOut{Output}{output}
    
        $V$ $\leftarrow$ $\emptyset$\tcc*[h]{Initialize an empty vertex set}\;
        $E$ $\leftarrow$ $\emptyset$\tcc*[h]{Initialize an empty edge set}\;
        $\mathcal{M}_\mathrm{D}$ $\leftarrow$ $\mathcal{M}$\tcc*[h]{Initialize the index set $\mathcal{M}_\mathrm{D}^{(0)}$}\;
        $r$ $\leftarrow 0$\tcc*[h]{Initialize a counting index}\;
        \While(\tcc*[h]{Calculate the set $\mathcal{M}_\mathrm{D}^{(\ell)}$}){$r<\ell$}{
            $\mathcal{M}_\mathrm{D}$ $\leftarrow$ $\delta(\mathcal{M}_\mathrm{D},\mathcal{M}_\mathrm{D})\setminus\{0\}$\tcc*[h]{Compute $\mathcal{M}_\mathrm{D}^{(r+1)}$ using the previous set $\mathcal{M}_\mathrm{D}^{(r)}$}\;
            $r$ $\leftarrow$ $r+1$
        }
        \ForEach{$j\in\mathcal{M}_\mathrm{D}$}{
            $V$ $\leftarrow$ $V\cup\{\Tilde{x}_j\}$
        }
        \BlankLine
        \ForEach(\tcc*[h]{Add all relevant edges to edge set $E$}){$j,k\in \mathcal{M}_{\mathrm{D}}$ with $j<k$}{
            \If{$[\Tilde{x}_j,\Tilde{x}_k]=\Tilde{\alpha}_{jk} \Tilde{x}_{\delta(j,k)}$ with $\tilde{\alpha}_{jk}\neq 0$}{
                $E$ $\leftarrow$ $E\cup\{(\Tilde{x}_j,\Tilde{x}_k,\Tilde{x}_{\delta(j,k)})\}\cup\{(\Tilde{x}_k,\Tilde{x}_j,\Tilde{x}_{\delta(j,k)})\}$
                \tcc*[h]{Add edge to edge set $E$; each edge is an ordered triple: (start vertex, edge-label, end vertex)}\;
            }
        }
        \Return $\mathcal{D}^\ell G(V,E)= G(V^{(\ell)},E^{(\ell)})\equiv G(V,E)$\tcc*[h]{Return labeled directed graph $\mathcal{D}^\ell G(V,E)$ associated with $\mathcal{D}^\ell\g$}\;
\caption{Algorithm for generating a labeled directed graph that represents the $\ell$-th derived algebra $\mathcal{D}^\ell\g$ for a given finite-dimensional graph-admissible Lie algebra}\label{alg:generating:generating:the:graph:derived:unaltered}
\end{algorithm}

It is straightforward to verify that the graphs produced by Algorithm~\ref{alg:generating:generating:the:graph:derived:unaltered} coincide with the graphs obtained via Algorithm~\ref{alg:generating:generating:the:graph:derived:altered}, provided the former is applied iteratively for all $\ell\in\N_{\geq0}$. This equivalence follows from the observation that
\begin{align*}
    \mathcal{M}_\mathrm{D}^{(\ell+1)}&=\delta\left(\left\{(j,k)\in\mathcal{M}_\mathrm{D}^{(\ell)}\times\mathcal{M}_\mathrm{D}^{(\ell)}\mid\Tilde{\alpha}_{jk}\neq0\right\}\right)=\delta\left(\mathcal{M}_\mathrm{D}^{(\ell)}\times\mathcal{M}_\mathrm{D}^{(\ell)}\right)\setminus\{0\},
\end{align*}
for all $\ell\in\N_{\geq0}$, since $\alpha_{jk}=0$ if and only if $\delta(j,k)=0$.

\subsection{Improved algorithm for generating derived graphs from redundant initial graphs}
The application of Algorithm~\ref{alg:generating:generating:the:graph:derived:altered}, when working with redundant initial graphs, typically yields redundant graphs associated with the derived algebras. However, in certain cases Algorithm~\ref{alg:generating:generating:the:graph:derived:altered} may inadvertently remove redundant elements, resulting in a minimal graph. If the goal is to preserve the redundant nature of the representation, an improved procedure is required.

Algorithm~\ref{alg:generating:generating:the:graph:derived:altered:redundant:included} addresses this issue by modifying the pruning strategy of Algorithm~\ref{alg:generating:generating:the:graph:derived:altered}. Specifically, edges are removed at a later stage, after reintroducing vertices that belong to $\mathcal{D}^\ell\g$ but were previously eliminated. This ensures that the resulting graph remains a \enquote{faithful} representation of the derived algebra when redundancy is maintained.
Similar to Algorithm~\ref{alg:generating:generating:the:graph:derived:altered}, the improved algorithm does not discard edges associated with re-added vertices, as doing so would compromise the structural integrity of the representation. As discussed in Section~\ref{sec:structural:properties:derived:series} (see remarks following Algorithm~\ref{alg:generating:generating:the:graph:derived:altered}), this approach guarantees that the resulting graph can be consistently associated with $\mathcal{D}^\ell\g$. Nevertheless, it should be remarked that the pruned graphs obtained by Algorithm~\ref{alg:generating:generating:the:graph:derived:altered:redundant:included} are not necessarily redundant.

\begin{algorithm}[htpb]
        \DontPrintSemicolon
        \KwData{An (overcomplete) basis $\mathcal{B}=\{\Tilde{x}_j\}_{j\in\mathcal{M}}$ of the graph-admissible Lie algebra $\g$, satisfying the Lie bracket relations \eqref{eqn:desired:basis:overcomplete}}
        \KwResult{A sequence of labeled directed graphs $\mathcal{D}^\ell G(V,E)$ associated with the derived algebras $\mathcal{D}^\ell\g$ of the Lie algebra $\g$ for all $\ell\geq0$, constructed with respect to the given basis $\mathcal{B}$.}
        \SetKwData{Left}{left}\SetKwData{This}{this}\SetKwData{Up}{up}
        \SetKwFunction{Union}{Union}\SetKwFunction{FindCompress}{FindCompress}
        \SetKwInOut{Input}{input}\SetKwInOut{Output}{output}
    
        \BlankLine
        $\mathcal{D}$ $\leftarrow$ $\emptyset$
        \tcc*[h]{Initialize the set containing all graphs $\mathcal{D}^\ell G(V,E)$ associated with each derived algebra $\mathcal{D}^\ell\g$}\;
        \BlankLine
        $V$ $\leftarrow$ $\mathcal{B}$\tcc*[h]{Initialize the vertex set for $G(V,E)$ with the basis $\mathcal{B}$}\;
        $E$ $\leftarrow$ $\emptyset$\tcc*[h]{Initialize the edge set for $G(V,E)$}\;
        \BlankLine
        \ForEach(\tcc*[h]{Add all relevant edges to edge set $E$}){$(j,k)\in \mathcal{M}\times\mathcal{M}$ with $j<k$}{
            \If{$[\Tilde{x}_j,\Tilde{x}_k]=\Tilde{\alpha}_{jk} \tilde{x}_{\delta(j,k)}$ with $\tilde{\alpha}_{jk}\neq 0$}{
                $E$ $\leftarrow$ $E\cup \{(\tilde{x}_j,\tilde{x}_k,\tilde{x}_{\delta(j,k)})\}\cup \{(\tilde{x}_k,\tilde{x}_j,\tilde{x}_{\delta(j,k)})\}$
                \tcc*[h]{Add edge to edge set $E$; each edge is an ordered triple: (start vertex, edge-label, end vertex)}\;
            }
        }
        $\mathcal{D}$ $\leftarrow$ $\{G(V,E)\}$\tcc*[h]{Add initial graph associated with $\g=\mathcal{D}^0\g$ to $\mathcal{D}$}\;
        \BlankLine
        \ForEach{$\ell\in\N_{\geq1}$}{
            $V_\mathrm{old}$ $\leftarrow$ $V$\;
            $E_\mathrm{old}$ $\leftarrow$ $E$\;       
            $V_\mathrm{r}$ $\leftarrow$ $\emptyset$ \tcc*[h]{Initialize candidate vertices to be removed}\;
            \ForEach{vertex $v\in V_{\mathrm{old}}$}{
                \If{there exists no edge $e\in E_{\mathrm{old}}$ such that $\varpi_\mathrm{e}(e)=v$}{
                    $V_\mathrm{r}$ $\leftarrow$ $V_{\mathrm{r}}\cup \{v\}$\tcc*[h]{Mark $v$ as a candidate for removal}\;
                }
            }
            $V_\mathrm{D}$ $\leftarrow$ $V_\mathrm{old}\setminus V_\mathrm{r}$\tcc*[h]{Vertices occurring as endpoints of edges}\;
            $V$ $\leftarrow$ $V_\mathrm{D}$\;  
            \ForEach{vertex $v\in V_\mathrm{r}$}{
                \If{$v\in \spn\{ V\}$}{
                    $V$ $\leftarrow$ $V\cup\{v\}$\tcc*[h]{Retain redundant vertices belonging to the derived algebra}\;
                }
            }
            \ForEach{vertex $v\in  V_\mathrm{old}\setminus V$}{
                \ForEach{edge $e\in E_{\mathrm{old}}$}{
                    \If(\tcc*[h]{edge is of the form $(\cdot,\cdot,v)$ or $(\cdot,v,\cdot)$ or $(v,\cdot,\cdot)$}){$v\in e$}{
                        $E$ $\leftarrow$ $E\setminus\{e\}$\tcc*[h]{Remove $e$ from $E$}
                    }
                }
            }
            $V^{(\ell)}\equiv\mathcal{D}^\ell V$ $\leftarrow$ $V$\;
            $E^{(\ell)}\equiv\mathcal{D}^\ell E$ $\leftarrow$ $E$\;
            $\mathcal{D}$ $\leftarrow$ $\mathcal{D}\cup\{\mathcal{D}^\ell G(V,E)=G(V^{(\ell}),E^{(\ell)})\}$\;
        }
        \Return $\mathcal{D}$\;
\caption{Algorithm for generating a directed graph associated with the $\ell$-th derived algebra $\mathcal{D}^\ell\g$ for a given finite-dimensional graph-admissible Lie algebra $\g$}\label{alg:generating:generating:the:graph:derived:altered:redundant:included}
\end{algorithm}

\section{Realizing the Schrödinger algebra $\sl{2}{\C}\ltimes\mathfrak{h}_m$ within the Weyl algebra $A_m$}\label{app:realization:of:schroedinger:algebra}
It has been shown that the complex Schrödinger algebra $\sl{2}{\C}\ltimes\mathfrak{h}_1$ admits a faithful realization within the complex Weyl algebra $A_1$, see \cite{TST:2006}. Furthermore, it has been demonstrated that the real Schrödinger algebra $\sl{2}{\R}\ltimes\mathfrak{h}_1$ can be faithfully realized within the skew-hermitian Weyl algebra $\hat{A}_1$, see \cite{A1:project}. In this section, we aim to prove that the real Schrödinger algebra $\sl{2}{\R}\ltimes\mathfrak{h}_m$ can be faithfully realized within the skew-hermitian Weyl algebra $\hat{A}_m$, and, consequently, that the complex Schrödinger algebra $\sl{2}{\C}\ltimes\mathfrak{h}_m$ admits a faithful realization within the complex Weyl algebra $A_m$. Finally, we comment on the physical significance of the real Schrödinger algebra $\sl{2}{\R}\ltimes\mathfrak{h}_m$ in the context of quantum mechanics.

Before constructing the realization, let us recall the most essential definitions. The bosonic creation and annihilation operators, denoted by $\hat{a}_j^\dagger$ and $\hat{a}_j$, have the following \emph{canonical commutation relations}\footnote{Note that these commutation relations $[A,B]:=AB-BA$ define a Lie bracket.}:
\begin{align*}
    [\hat{a}_j,\hat{a}_k]&=0,\;&\;[\hat{a}_j,\hat{a}_k^\dagger]&=\delta_{jk}\hat{\mathds{1}},\;&\;[\hat{a}_j^\dagger,\hat{a}_k^\dagger]&=0,
\end{align*}
where $\hat{\mathds{1}}$ is a central element, i.e., an element that commutes with every creation and annihilation operator. These operators span the Heisenberg algebra $\mathfrak{h}_m=\lie{\{\hat{a}_j,\hat{a}_{j}^\dagger,\hat{\mathds{1}}\}_{j=1}^m}$, which is typically considered over either the real or complex numbers. The complex Weyl algebra $A_m$ can be regarded as the universal enveloping algebra $U(\mathfrak{h}_m)$ of the complex Heisenberg algebra $\mathfrak{h}_m$, under the restriction that $\hat{\mathds{1}}$ acts as the multiplicative identity. More precisely, if the abstract complex Heisenberg algebra is presented with central generator $z$, then the complex Weyl algebra is isomorphic to the quotient $A_m\cong U(\mathfrak{h}_{m,\C})/\langle z-\mathds{1}\rangle$~\cite{Gaddis:2025}. This algebra is equipped with the hermitian conjugation $(\cdot)^\dagger:A_m\to A_m$ with concrete action $(\cdot)^\dag:p\mapsto p^\dagger$, which is a conjugate-linear involutive anti-automorphism satisfying:
\begin{align*}
    (\hat{\mathds{1}})^\dagger&=\hat{\mathds{1}},\;&\; (i\hat{\mathds{1}})^\dagger&=-i\hat{\mathds{1}},\;&\; (\hat{a}_j^\dagger)^\dagger&=\hat{a}_j,\;&\; (\hat{a}_j)^\dagger=\hat{a}_j^\dagger,
\end{align*}
and reversing the order of multiplication, i.e., $(\hat{p}_1\hat{p}_2)^\dagger=\hat{p}_2^\dagger\hat{p}_1^\dagger$.

The Poincaré-Birkhoff-Witt theorem guarantees that the normal-ordered monomials
\begin{align*}
    \left(\prod_{j=1}^m(\hat{a}_j^\dagger)^{\alpha_j}\right)\left(\prod_{j=1}^m\hat{a}_j^{\beta_j}\right)=\prod_{j=1}^m(\hat{a}_j^\dagger)^{\alpha_j}\hat{a}_j^{\beta_j},\qquad\text{with}\qquad\alpha_j,\beta_j\in\N_{\geq0}
\end{align*}
form a basis of the Weyl algebra $A_m$, see \cite{PBW:Thm}. However, since bosonic dynamics are governed by hermitian Hamiltonians $\hat{H}(t)$, where the time-evolution operator $\hat{U}(t)$ satisfies the differential equation
\begin{align*}
    \der{t}\hat{U}(t)=-i\hat{H}(t)\hat{U}(t),\qquad\text{with}\qquad\hat{U}(0)=\hat{\mathds{1}},
\end{align*}
it is natural to restrict attention to the skew-hermitian Weyl algebra $\hat{A}_m$. This algebra consists of all skew-hermitian polynomials in $A_m$, i.e., polynomials $\hat{p}\in A_m$ satisfying $\hat{p}^\dagger=-\hat{p}$, and forms a real Lie algebra under the commutator bracket. A basis of $\hat{A}_m$ can be described by the following elements:
\begin{align*}
    g_+=i\left(\left(\prod_{j=1}^m(\hat{a}_j^\dagger)^{\alpha_j}\hat{a}_j^{\beta_j}\right)^\dagger+\prod_{j=1}^m(\hat{a}_j^\dagger)^{\alpha_j}\hat{a}_j^{\beta_j}\right)\qquad\text{and}\qquad g_-=\left(\prod_{j=1}^m(\hat{a}_j^\dagger)^{\alpha_j}\hat{a}_j^{\beta_j}\right)^\dagger-\prod_{j=1}^m(\hat{a}_j^\dagger)^{\alpha_j}\hat{a}_j^{\beta_j},
\end{align*}
where the integers $\alpha_j,\beta_j\in\N_{\geq0}$ satisfy the following conditions:

\begin{itemize}
    \item $g_+$ term: If there exists an index $k\in\{1,\ldots,m\}$ such that $\alpha_k<\beta_k$, then there exists an index $j\in\{1,\ldots,k-1\}$ such that $\alpha_j>\beta_j$ and $\alpha_\ell\geq\beta_\ell$ for all $\ell\in\{1,\ldots,j-1\}$.
    \item $g_-$ term: $\alpha_j\neq\beta_j$ for at least one index $j\in\{1,\ldots,m\}$, and if there exists an index $k\in\{1,\ldots,m\}$ such that $\alpha_k<\beta_k$, then there exists an index $j\in\{1,\ldots,k-1\}$ such that $\alpha_j>\beta_j$ and $\alpha_\ell\geq\beta_\ell$ for all $\ell\in\{1,\ldots,j-1\}$.
\end{itemize}
This concludes the repetition of the most essential definitions. More details can be found in the relevant literature~\cite{Bruschi:Xuereb:2024}.
Let us proceed with showing the following claim:

\begin{proposition}\label{prop:faithful:realization:of:schroedinger:algebra}
    The real Schrödinger algebra $\sl{2}{\R}\ltimes\mathfrak{h}_m$ can be faithfully realized within the skew-hermitian Weyl algebra $\hat{A}_m$.
\end{proposition}

\begin{proof}
    It has been shown that the following six elements of $\hat{A}_m$
    \begin{align*}
        \hat{h}_j&=\frac{1}{2}\left(\hat{a}_j^2-(\hat{a}_j^\dagger)^2\right),\;&\;\hat{x}_j&=-\frac{1}{4}\left(2i\left(\hat{a}_j^\dagger\hat{a}_j+\frac{1}{2}\right)+i\left(\hat{a}_j^2+(\hat{a}_j^\dagger)^2\right)\right),\;&\;\hat{y}_j&=\frac{1}{4}\left(2i\left(\hat{a}_j^\dagger\hat{a}_j+\frac{1}{2}\right)-i\left(\hat{a}_j^2+(\hat{a}_j^\dagger)^2\right)\right),\\
        \hat{p}_j&=\hat{a}_j-\hat{a}_j^\dagger,\;&\;\hat{q}_j&=i\left(\hat{a}_j+\hat{a}_j^\dagger\right),\;&\;\hat{z}&=-2i\hat{\mathds{1}},
    \end{align*}
    satisfy the commutation relations:
    \begin{align*}
        [\hat{h}_j,\hat{x}_j]&=2\hat{x}_j,\;&\;[\hat{h}_j,\hat{y}_j]&=-2\hat{y}_j,\;&\;[\hat{x}_j,\hat{y}_j]&=\hat{h}_j,\;&\;[\hat{h}_j,\hat{q}_j]&=\hat{q}_j,\;&\;[\hat{h}_j,\hat{p}_j]&=-\hat{p}_j,\\
        [\hat{q}_j,\hat{p}_j]&=\hat{z},\;&\;[\hat{x}_j,\hat{q}_j]&=0,\;&\;[\hat{x}_j,\hat{p}_j]&=\hat{q}_j,\;&\;[\hat{y}_j,\hat{q}_j]&=\hat{p}_j,\;&\;[\hat{y}_j,\hat{p}_j]&=0,
    \end{align*}
    while $\hat{z}$ commutes with every other element \cite{A1:project}. This extends naturally to the multi-mode case:
    \begin{align*}
        [\hat{h}_j,\hat{x}_k]&=2\delta_{jk}\hat{x}_j,\;&\;[\hat{h}_j,\hat{y}_k]&=-2\delta_{jk}\hat{y}_j,\;&\;[\hat{x}_j,\hat{y}_k]&=\delta_{jk}\hat{h}_j,\;&\;[\hat{h}_j,\hat{q}_k]&=\delta_{jk}\hat{q}_j,\;&\;[\hat{h}_j,\hat{p}_k]&=-\delta_{jk}\hat{p}_j,\\
        [\hat{q}_j,\hat{p}_k]&=\delta_{jk}\hat{z},\;&\;[\hat{x}_j,\hat{q}_k]&=0,\;&\;[\hat{x}_j,\hat{p}_k]&=\delta_{jk}\hat{q}_j,\;&\;[\hat{y}_j,\hat{q}_k]&=\delta_{jk}\hat{p}_j,\;&\;[\hat{y}_j,\hat{p}_k]&=0,
    \end{align*}
    with $\hat{z}$ being a central element. We can now define the three elements:
    \begin{align*}
        \hat{h}&:=\sum_{j=1}^m\hat{h}_j,\;&\;\hat{x}&:=\sum_{j=1}^m\hat{x}_j,\;&\;\hat{y}&:=\sum_{j=1}^m\hat{y}_j.
    \end{align*}
    Then the commutation relations become:
    \begin{align*}
        [\hat{h},\hat{x}]&=2\hat{x},\;&\;[\hat{h},\hat{y}]&=-2\hat{y},\;&\;[\hat{x},\hat{y}]&=\hat{h},\;&\;[\hat{h},\hat{q}_j]&=\hat{q}_j,\;&\;[\hat{h},\hat{p}_j]&=-\hat{p}_j,\\
        [\hat{q}_j,\hat{p}_k]&=\delta_{jk}\hat{z},\;&\;[\hat{x},\hat{q}_j]&=0,\;&\;[\hat{x},\hat{p}_j]&=\hat{q}_j,\;&\;[\hat{y},\hat{q}_j]&=\hat{p}_j,\;&\;[\hat{y},\hat{p}_k]&=0,
    \end{align*}
    where $\hat{z}$ belongs to the center. The PBW basis shows that the elements $\hat{h},\hat{x},\hat{y},\hat{p}_1,\hat{q}_1,\ldots,\hat{p}_m,\hat{q}_m,\hat{z}$ are linearly independent over $\R$: the first three have linearly independent quadratic components, the elements $\hat{p}_j,\hat{q}_j$ have linearly independent linear components, and $\hat{z}$ is a non-zero constant element. Consequently, the Lie homomorphism determined by the preceding commutation relations is injective. Thus, comparing these bracket relations to those found in \cite{Tao:2022}, one concludes $$\lie{\{\hat{h},\hat{x},\hat{y},\hat{p}_j,\hat{q}_j,\hat{z}\}_{j=1}^m}\cong\sl{2}{\R}\ltimes\mathfrak{h}_m,$$ which shows that the Schrödinger algebra can be faithfully realized within the skew-hermitian Weyl algebra $\hat{A}_m$.
\end{proof}

\begin{corollary}
    The complex Schrödinger algebra $\sl{2}{\C}\ltimes\mathfrak{h}_m$ can be faithfully realized within the Weyl algebra $A_m$.
\end{corollary}

\begin{proof}
    The skew-hermitian Weyl algebra $\hat{A}_m$ is a real form of $A_m$. Indeed, every $\hat{p}\in A_m$ admits the unique decomposition $\hat{p}=\frac{1}{2}(\hat{p}-\hat{p}^\dagger)+i(-\frac{i}{2}(\hat{p}+\hat{p}^\dagger))$, where both summands before multiplication by $i$ belong to $\hat{A}_m$. Thus, $A_m=\hat{A}_m\oplus i\hat{A}_m$ as real vector spaces. Complexifying the injective realization from Proposition~\ref{prop:faithful:realization:of:schroedinger:algebra} therefore yields an injective realization
    \begin{align*}
        \left(\sl{2}{\R}\ltimes\mathfrak{h}_m\right)\otimes\C\cong\sl{2}{\C}\ltimes\mathfrak{h}_m
    \end{align*}
    within $A_m$.
\end{proof}

Let us continue by discussing the physical relevance of this result. Bosonic systems, such as modes of the electromagnetic field \cite{Sakurai:2020}, phononic excitations in crystals \cite{Ashcroft:Mermin:1976}, or Bose-Einstein condensates in many body physics \cite{Dalfovo:Giorgini:1999}, are typically modeled by a collection of bosonic harmonic oscillators. When analyzing the time evolution of such systems, one often performs a basis change that transforms the Hamiltonian into one describing a set of uncoupled harmonic oscillators \cite{Bruschi:2021}. A corresponding time-dependent Hamiltonian $\hat{H}(t)$ for an $m$-mode bosonic system takes the form:
\begin{align}
    \hat{H}(t)&=\sum_{j=1}^m\left(\omega_j\hat{a}_j^\dagger\hat{a}_j+s_j(t)\hat{a}_j^2+s_j^*(t)(\hat{a}_j^\dagger)^2\right),\label{eqn:collection:single:mode:squeezed:oscillators}
\end{align}
where $\omega_j>0$ are the oscillator frequencies, and $s_j(t)$ complex-valued functions known as single-mode squeezing parameters \cite{Heib:2025:RWA,Ferraro:2005}. This Hamiltonian structure naturally involves quadratic combinations of creation and annihilation operators, i.e., precisely the same building blocks used in the realization of the Schrödinger algebra within the Weyl algebra presented in the proof of Proposition~\ref{prop:faithful:realization:of:schroedinger:algebra}. The terms $\omega_j\hat{a}_j^\dagger\hat{a}_j$ model a single free uncoupled harmonic oscillator, while the terms $s_j(t)\hat{a}_j^2+s_j^*(t)(\hat{a}_j^\dagger)^2$ correspond to active operations known as single-mode squeezing operators \cite{Ferraro:2005}. These operations are fundamental in continuous-variable quantum information theory and quantum optics \cite{Adesso:Ragy:2014,Heib:2025:RWA}. The fact that this Hamiltonian models uncoupled oscillators is evident from the absence of cross-mode interaction terms, i.e., there are no products involving creation and annihilation operators from different modes.

The Hamiltonian \eqref{eqn:collection:single:mode:squeezed:oscillators}, multiplied by $i$, can be expressed in terms of the operators $\hat{h}_j,\hat{x}_j,\hat{y}_j$, up to an additive term proportional to $i\hat{\mathds{1}}$ (equivalently $\hat{z}$), which corresponds to a physically irrelevant global phase factor. This phase does not affect observable quantities and can therefore be omitted in the analysis. The resulting expression is:
\begin{align*}
    i\hat{H}(t)&=-\sum_{j=1}^m\left(2\im{s_j(t)}\hat{h}_j+\left(2\re{s_j(t)}+\omega_j\right)\hat{x}_j+\left(2\re{s_j(t)}-\omega_j\right)\hat{y}_j\right).
\end{align*}
For simplicity of exposition, we impose the assumptions $s_j(t)=s_k(t)\equiv s(t)$ and $\omega_j=\omega_k\equiv \omega$ for all $t\in\R$ and $j,k\in\{1,\ldots,m\}$. These conditions correspond to a system of identical modes, which significantly reduces the complexity of the problem by introducing the symmetry among the oscillators. However, to broaden the applicability of the analysis, we incorporate displacement operators on each mode, as described in \cite{Ferraro:2005}, which can all be controlled independently. These operations introduce different mode-dependent shifts and are parameterized by complex displacement coefficients $\alpha_j$ for each mode. This displacement operator acts as follows: $\hat{\mathds{1}}\mapsto\hat{\mathds{1}}$, $\hat{a}_j\mapsto\hat{a}_j+\alpha_j\hat{\mathds{1}}$, and $\hat{a}_j^\dagger\mapsto\hat{a}_j^\dagger+\alpha_j^*\hat{\mathds{1}}$. The resulting displaced Hamiltonian acquires additional linear terms in the operators $\hat{p}_j$ and $\hat{q}_j$ (resp. $\hat{a}_j$ and $\hat{a}_j^\dagger$) and consequently reads:
\begin{align*}
    i\hat{H}_\mathrm{f}(t)&=-2\im{s(t)}\hat{h}-\left(2\re{s(t)}+\omega\right)\hat{x}-\left(2\re{s(t)}-\omega\right)\hat{y}+\sum_{j=1}^m\left(\lambda_j(t)\hat{p}_j+\kappa_j(t)\hat{q}_j\right),
\end{align*}
modulo an irrelevant global phase factor proportional to $\hat{z}$. The real-valued functions $\lambda_j(t),\kappa_j(t)$ depend on the displacement parameters $\alpha_j$, the squeezing parameter $s(t)$, as well as the frequency $\omega$. Their explicit form can be computed straightforwardly but is omitted here as it is not essential for the present discussion. Following the discussion in Sections~\ref{sec:prominent:examples:linear:quantum} and~\ref{sec:discussion:application:to:quantum:physics}, the Lie algebra required to compute the time-evolution of this system via a factorization ansatz is generated by $\mathcal{G}=\{\hat{h},\hat{x},\hat{y},\hat{p}_j,\hat{q}_j\}_{j=1}^m$, whose Lie closure satisfies, by the proof of Proposition~\ref{prop:faithful:realization:of:schroedinger:algebra}, $\lie{\mathcal{G}}=\sl{2}{\R}\ltimes\mathfrak{h}_m$, underscoring the importance of the Schrödinger algebra in describing the dynamics of multi-mode bosonic systems. 
Moreover, it is clear that the simplifying assumptions $s_j(t)=s_k(t)\equiv s(t)$ and $\omega_j=\omega_k\equiv \omega$  for every $t\in\R$ and $j,k\in\{1,\ldots,m\}$ dramatically reduce the dimensionality of the problem by collapsing  the system into a symmetric configuration. Without these assumptions, one would need to consider the full Lie algebra $\lie{\{\hat{h}_j,\hat{x}_j,\hat{y}_j,\hat{p}_j,\hat{q}_j\}_{j=1}^m}$, which is substantially larger (i.e., $\dim(\lie{\{\hat{h}_j,\hat{x}_j,\hat{y}_j,\hat{p}_j,\hat{q}_j\}_{j=1}^m})=5m+1$ against $\dim(\lie{\{\hat{h},\hat{x},\hat{y},\hat{p}_j,\hat{q}_j\}_{j=1}^m})=2m+4$) and more complex to handle computationally. However, using the graph-theoretic approach developed in this work it becomes evident that this Lie algebra is isomorphic to
\begin{align}
    \lie{\{\hat{h}_j,\hat{x}_j,\hat{y}_j,\hat{p}_j,\hat{q}_j,\hat{z}\}_{j=1}^m}\cong\left(\bigoplus_{j=1}^m\sl{2}{\R}\right)\ltimes\mathfrak{h}_m.\label{eqn:decomposition:generalized:schroedinger:algebra}
\end{align}
Let us explain this in more detail: First, one constructs the graph associated with $\lie{\{\hat{h}_j,\hat{x}_j,\hat{y}_j,\hat{p}_j,\hat{q}_j,\hat{z}\}_{j=1}^m}$ using the basis $\{\hat{h}_j,\hat{x}_j,\hat{y}_j,\hat{p}_j,\hat{q}_j,\hat{z}\}_{j=1}^m$. This graph is minimal because the basis is not overcomplete. Next, one recognizes that, according to Lemma~\ref{lem:ideal:span}, Lemma~\ref{lem:lower:central:series:of:graphs} and Algorithm~\ref{alg:generating:lower:central:series}, the vertices labeled by the elements $\{\hat{p}_j,\hat{q}_j,\hat{z}\}_{j=1}^m$ span a two-step nilpotent ideal, which is precisely the Heisenberg algebra $\mathfrak{h}_m$. Furthermore, Corollary~\ref{cor:simple:subalgebra} implies that the triples $\{\hat{h}_j,\hat{x}_j,\hat{y}_j\}$ span, for every $j\in\{1,\ldots,m\}$, a simple three-dimensional subalgebra. Employing then the Bianchi classification \cite{Bianchi:1903}, we know that there exist only two such algebras in the real case, namely $\sl{2}{\R}$ and $\mathfrak{su}(2)$. However, Theorem~4.4 from \cite{Joseph:1972} excludes the case $\mathfrak{su}(2)\cong\lie{\{\hat{h}_j,\hat{x}_j,\hat{y}_j\}}$, since these algebras can be realized within $\hat{A}_1$. Moreover, the commutator between two such algebras vanishes, i.e., $[\lie{\{\hat{h}_j,\hat{x}_j,\hat{y}_j\}},\lie{\{\hat{h}_k,\hat{x}_k,\hat{y}_k\}}]=\{0\}$ if $j\neq k$, which shows that the subalgebra generated by all such triples is a direct sum of $m$ copies of $\sl{2}{\R}$, i.e., $$\lie{\{\hat{h}_j,\hat{x}_j,\hat{y}_j\}_{j=1}^m}\cong\bigoplus_{j=1}^m\sl{2}{\R}.$$ Since the graph is minimal and each algebra $\sl{2}{\R}$ is a simple and non-solvable Lie algebra, the Heisenberg algebra $\mathfrak{h}_m$ forms the radical of the full Lie algebra $\lie{\{\hat{h}_j,\hat{x}_j,\hat{y}_j,\hat{p}_j,\hat{q}_j,\hat{z}\}_{j=1}^m}$. This leads to the Levi-Mal'tsev decomposition stated in \eqref{eqn:decomposition:generalized:schroedinger:algebra}, as the radical does not form the center of the Lie algebra.

\section{The Poincaré algebra}\label{app:poincare}
The \emph{Poincaré group}, also referred to as the \emph{inhomogeneous Lorentz group} \cite{Weinberg:1995}, comprises all coordinate transformations of Minkowski space that preserve the form of all physical laws \cite{Carroll:2019}. This group is a ten-dimensional and non-abelian Lie group that is generated by the following hermitian operators:
\begin{itemize}
    \item The Hamiltonian operator $H$,
    \item The momentum operators $P_1,P_2,P_3$,
    \item The angular momentum operators $J_1,J_2,J_3$,
    \item The three boost operators $K_1,K_2,K_3$.
\end{itemize}
These ten elements span the complexified Poincaré algebra $\mathfrak{p}(1,3)$ \cite{Fushchich:1987}, which is also denoted by $\mathfrak{iso}(1,3)$ \cite{Goze:2007}. The commutation relations, and therefore the Lie bracket relations, among these basis elements are
\begin{align*}
    [J_j,J_k]&=i\sum_{\ell=1}^3\varepsilon_{jk\ell}J_\ell,\;&\;[J_j,K_k]&=i\sum_{\ell=1}^3\varepsilon_{jk\ell}K_\ell,\;&\;[K_j,K_k]&=-i\sum_{\ell=1}^3\varepsilon_{jk\ell}J_\ell,\\
    [J_j,P_k]&=i\sum_{\ell=1}^3\varepsilon_{jk\ell}P_\ell,\;&\;[K_j,P_k]&=-iH\delta_{jk},\;&\;[K_j,H]&=-iP_j,
\end{align*}
while all other commutators vanish \cite{Weinberg:1995}. The corresponding graph $G(V,E)$ is depicted in Figure~\ref{fig:poincare:algebra}. 
\begin{figure}[htpb]
    \centering
    \includegraphics[width=0.75\linewidth]{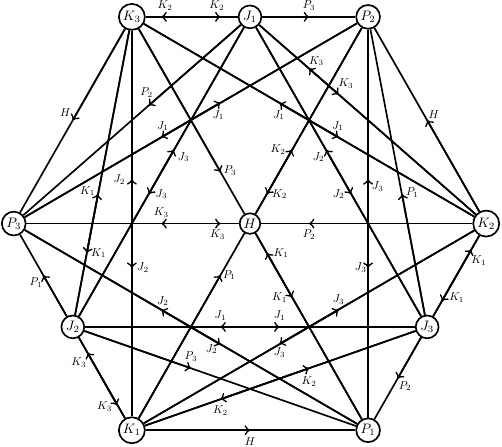}
    \caption{Depiction of the graph $G(V,E)$ associated with the Poincaré algebra $\mathfrak{p}(1,3)$ utilizing the basis $\{H,P_j,J_j,K_j\}_{j=1}^3$.}
    \label{fig:poincare:algebra}
\end{figure}

When considered over the complex numbers, the subalgebra spanned by the three boost and three angular momentum operators is the Lorentz algebra $\mathfrak{so}(3;1)_\C\cong\sl{2}{\C}\oplus\sl{2}{\C}$ (cf. Section~\ref{sec:prominent:examples:subsec:Lorentz:algebra}). The application of Lemma~\ref{lem:ideal:span} to the graph $G(V,E)$  implies that the vertices labeled with the three momentum operators $P_1,P_2,P_3$ and the energy operator $H$ span together an ideal of the Poincaré algebra $\mathfrak{p}(1,3)$. Furthermore, since the subgraph containing only vertices  and edges labeled by these four operators contains no edges, it follows from Lemma~\ref{lem:abelian:criterion} that the Lie algebra $\lie{\{P_1,P_2,P_3,H\}}$ is abelian. Moreover, because the graph $G(V,E)$ is minimal and the algebra $\mathfrak{so}(3;1)$ is semisimple, and therefore non-solvable, one has $\lie{\{P_1,P_2,P_3,H\}}=\operatorname{rad}(\mathfrak{p}(1,3))$. Thus, we can conclude that
\begin{align*}
    \mathfrak{p}(1,3)=(\sl{2}{\C}\oplus\sl{2}{\C})\ltimes\left(\bigoplus_{j=1}^4\mathfrak{a}\right)\cong(\sl{2}{\C}\oplus\sl{2}{\C})\ltimes\C^4,
\end{align*}
where $\mathfrak{a}$ denotes the one-dimensional abelian Lie algebra isomorphic to the complex numbers, i.e., $\mathfrak{a}\cong\C$. 

The Poincaré algebra plays a fundamental role in encoding the Lorentz invariance of inertial systems within the relativistic framework. However, in the non-relativistic limit (i.e., for velocities $v$ that satisfy $(v/c)^2\ll1$, where $c$ is the speed of light) this algebra reduces to the Galilei algebra, thus capturing the Galilei invariance of classical mechanical systems. To understand this demand, one needs to  reintroduce the relevant physical constants into the bracket relations, rather than working in natural units, where the speed of light is normally set to $c\equiv1$ \cite{Jackson:2014}. The correspondingly modified relations read then:
\begin{align*}
    [J_j',J_k']&=i\sum_{\ell=1}^3\varepsilon_{jk\ell}J_\ell',\;&\;[J_j',K_k']&=i\sum_{\ell=1}^3\varepsilon_{jk\ell}K_\ell',\;&\;[K_j',K_k']&=-\frac{1}{c^2}i\sum_{\ell=1}^3\varepsilon_{jk\ell}J_\ell',\\
    [J_j',P_k']&=i\sum_{\ell=1}^3\varepsilon_{jk\ell}P_\ell',\;&\;[K_j',P_k']&=-\frac{1}{c^2}iH'\delta_{jk},\;&\;[K_j',H']&=-iP_j',
\end{align*}
as discussed in \cite{Heranz:2002}, where the real form of this algebra is considered. This real form can be obtained, by multiplying each operator with $-i$. As shown in \cite{Heranz:2002}, taking the classical limit $c\to\infty$ yields the Galilei algebra, which obeys the same bracket relations as above, except that $[K_j',K_k']=0=[K_j',P_k']$. In this case, the operators $P_j'$ represent spatial translations, $J_j'$ represent rotations, $K_j'$ correspond to Galilean boosts, and $H'$ represents time translations. Note that this Galilei algebra is also known as the reduced classical Galilei algebra $A\Bar{G}_1(3)$ \cite{Nesterenko:2016}. The corresponding graph can be found in Figure~\ref{fig:galileo:algebra}.
\begin{figure}[htpb]
    \centering
    \includegraphics[width=0.75\linewidth]{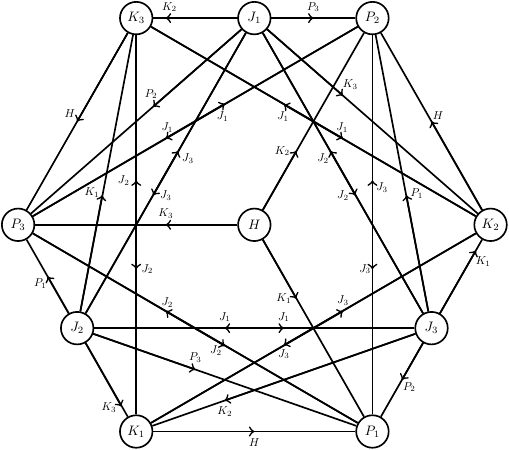}
    \caption{Depiction of the graph $G(V,E)$ associated with the Galilei algebra utilizing the basis $\{H',P_j',J_j',K_j'\}_{j=1}^3$ in the non-relativistic limit $c\to\infty$.}
    \label{fig:galileo:algebra}
\end{figure}
This graph can again be analyzed using the graph-theoretic framework developed in this work. Corollary~\ref{cor:simple:subalgebra}, applied for the real form, guarantees after complexification that this Galilei algebra contains a simple complex three-dimensional Lie subalgebra spanned by the three rotation operators $\{J_j'\}_{j=1}^3$, since $\sl{2}{\R}\otimes\C\cong\sl{2}{\C}\cong\mathfrak{su}(2)\otimes\C$~\cite{A1:project}, i.e., $\lie{\{J_j'\}_{j=1}^3}\cong\sl{2}{\C}$. Furthermore, the radical of this Galilei algebra, which we denote by $\mathfrak{ga}:=\lie{\{H',P_j',J_j',K_j'\}_{j=1}^3}$, is given by $\operatorname{rad}(\mathfrak{ga})=\lie{\{H',P_j',K_j'\}_{j=1}^3}$, since the corresponding vertices span by Lemma~\ref{lem:ideal:span} an ideal, which is by Theorem~\ref{thm:nilpotency:criteria:strong} nilpotent and therefore solvable. This ideal is furthermore the maximal solvable ideal of $\mathfrak{ga}$, since the graph is minimal and the three remaining vertices span a simple Lie algebra. Moreover, by computing the lower central series of the associated graphs, we find that the radical is a two-step nilpotent Lie algebra, where the first derived algebra is three-dimensional. The edge structure of the graph associated with the radical allows also a direct identification of $\operatorname{rad}(\mathfrak{ga})=\mathfrak{tn}_{3,7_A}$, where $\mathfrak{tn}_{3,7_A}$ denotes the indecomposable two-step nilpotent Lie algebra labeled by $3,7_A$ in \cite{Seeley:1993}. This identification follows because every edge terminating at a sinkhole either originates from the vertex labeled by $H'$ or is itself labeled by $H'$. The Levi-Mal'tsev decomposition theorem allows therefore the conclusion that:
\begin{align*}
    \mathfrak{ga}\cong\sl{2}{\C}\ltimes\mathfrak{tn}_{3,7_A}.
\end{align*}

\section{Further remarks on symmetries}\label{app:symmetries}

In the main text, particularly in Section~\ref{sec:prominent:examples:linear:quantum}, we recognized that some proper labeled directed graphs admit certain symmetries. This appendix shall be dedicated to a more detailed examination of these symmetries and their interplay with the graphical representations of Lie algebras. For clarity and simplicity, we restrict our attention to minimal-graph-admissible Lie algebras and their rotational symmetries.

\begin{definition}\label{def:symmetry:admission}
    Let $G(V,E)$ be a minimal graph associated with an $n$-dimensional minimal-graph-admissible Lie algebra. We say that $G(V,E)$ \emph{admits a $\Z_n$-symmetry} if the following holds: There exists an equidistant distribution of the vertices $v\in V$ on a circle in the Euclidean plane, such that the graphical representation of the graph $G(V,E)$ is invariant under rotation by integer multiples of $360^\circ/n$, combined with appropriate relabeling of vertices and edges.
\end{definition}

Our work leads us to believe that rotational symmetries of minimal graphs reflect underlying algebraic invariances. To illustrate this concept, we provide an example below.
\begin{tcolorbox}[breakable, colback=Cerulean!3!white,colframe=Cerulean!85!black,title=\textbf{Example}: Rotational symmetries of the graph associated with $\mathfrak{su}(2)$]
    Let us examine how this definition works in practice:
    \begin{example}\label{exa:symmetry:discussion}
        Consider the real Lie algebra $\mathfrak{su}(2)$ discussed in Example~\ref{exa:first:example:graphs:assocaited:with:algebras}. Here, we can choose the basis $\{e_1,e_2,e_3\}$ satisfying the bracket relations:
        \begin{align*}
            [e_1,e_2]&=e_3,\;&\;[e_2,e_3]&=e_1,\;&\;[e_3,e_1]&=e_2.
        \end{align*}
        The corresponding graph consists of three vertices and six edges. A possible representation of this graph in the Euclidean plane, with the vertices equidistantly distributed on a circle, is depicted in Figure~\ref{fig:rotations:of:su2}. To verify that this graph admits a $\Z_3$-symmetry, observe that rotating the Euclidean plane counterclockwise by $120^\circ$ and $240^\circ$ respectively, combined with appropriate permutations of the vertex and edge labels, leaves the illustration invariant. Specifically:
        \begin{itemize}
            \item A rotation by $120^\circ$ paired with the permutation $\sigma:e_1\mapsto e_3\mapsto e_2\mapsto e_1$ preserves the graph representation.
            \item A rotation by $240^\circ$ paired with the permutation $\sigma:e_1\mapsto e_2\mapsto e_3\mapsto e_1$ preserves the graph representation.
        \end{itemize}
        \begin{figure}[H]
            \centering
            \includegraphics[width=0.95\linewidth]{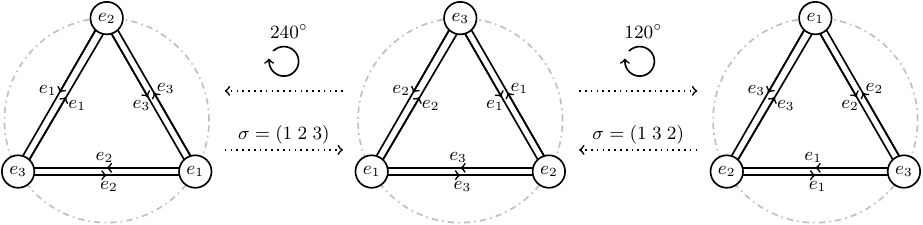}
            \caption{Illustration of the graph associated with $\mathfrak{su}(2)$, which is represented in the Euclidean plane with vertices placed equidistantly on a circle (indicated by the gray dash-dotted line). This configuration highlights the $\Z_3$-symmetry of the graph.}
            \label{fig:rotations:of:su2}
        \end{figure}
    \end{example}
\end{tcolorbox}

The following propositions formalize some immediate consequences for the symmetry concept introduced above. They highlight the relationship between rotational symmetries of minimal graphs and structural properties of the corresponding Lie algebras.

We begin by noting the following fundamental observations:
\begin{proposition}
    Let $\g$ be an $n$-dimensional abelian Lie algebra. Then every minimal graph $G(V,E)$ associated with $\g$ admits a $\Z_n$-symmetry.
\end{proposition}

\begin{proof}
    This is immediate, since any minimal graph $G(V,E)$ associated with an abelian Lie algebra contains by Lemma~\ref{lem:abelian:criterion} no edges.
\end{proof}

\begin{proposition}\label{prop:Zn:symmetry:only:if:semisimple:or:ablelian}
    Let $\g$ be a complex, non-abelian, $n$-dimensional minimal-graph-admissible Lie algebra associated with a minimal graph $G(V,E)$ that admits a $\Z_n$-symmetry. Then $\g$ is perfect with nilpotent radical.
\end{proposition}

\begin{proof}
   Let now $\g$ be a non-abelian, $n$-dimensional, minimal-graph-admissible Lie algebra associated with the minimal graph $G(V,E)$ that admits a $\Z_n$-symmetry. According to Lemma~\ref{lem:abelian:criterion} this implies that $E\neq\emptyset$ and there exists at least one edge that targets a vertex $v_j$. This edge $e$ cannot satisfy $\varpi_\mathrm{s}(e)=\varpi_\mathrm{l}(e)=\varpi_\mathrm{e}(e)=v_j$, due to Proposition~\ref{prop:conditions:for:edges}. Therefore, by the antisymmetry requirement formalized by Proposition~\ref{prop:conditions:for:edges}, there exists at least one edge $e'$ such that $\varpi_\mathrm{s}(e')\neq \varpi_\mathrm{e}(e')$. By the symmetry requirement, this property must be true for all vertices $v\in V$. Consequently, there exists a closed directed walk $W_v$ for every vertex $v\in V$ such that $W_v$ visits $v$ at least once. A closer inspection leads us to conclude that these walks must induce self-contained subgraphs. We explain this conclusion in detail below:
   \begin{enumerate}[label = (\roman*)]
       \item \textbf{Claim:} \emph{There exists a closed directed walk $W_v$ for every vertex $v\in V$ such that $W_v$ visits $v$ at least once.} We have seen that for every vertex $v\in V$ there exists at least one edge $e_v\in E$ such that $\varpi_\mathrm{e}(e_v)=v$. Now assume that there exists a vertex $w\in V$ such that every edge $e\in E$ that satisfies $\varpi_\mathrm{s}(e)=w$ also obeys $\varpi_\mathrm{e}(e)=w$. Since $G(V,E)$ is assumed to admit a $\Z_n$-symmetry, we can therefore rotate the graph by $360^\circ/n$ and must, modulo a permutation of the labels, find the same graph. This implies that every vertex must be of the same form and consequently every edge $e\in E$ satisfies $\varpi_\mathrm{s}(e)=\varpi_\mathrm{e}(e)$. This contradicts Proposition~\ref{prop:conditions:for:edges}, as any edge $e''\in E$ must satisfy $\varpi_\mathrm{s}(e'')\neq\varpi_\mathrm{l}(e'')$, and since an edge $e''$ belongs to $E$ if and only if the edge $e'''=(\varpi_\mathrm{l}(e''),\varpi_\mathrm{s}(e''),\varpi_\mathrm{e}(e''))$ also belongs to $E$. This edge $e'''$ satisfies $\varpi_\mathrm{s}(e''')=\varpi_\mathrm{l}(e'')\neq \varpi_\mathrm{s}(e'')=\varpi_\mathrm{e}(e''')$, contrary to the assumption that every edge has identical starting and ending vertices. We can conclude that for every vertex $v\in V$, there exists an edge $e\in E$ such that $\varpi_\mathrm{s}(e)=v\neq\varpi_\mathrm{e}(e)$. It is now straightforward to construct for every vertex $v\in V$ a closed directed walk $W_v$ that visits $v$ at least once: Start at the vertex $v$ and follow one of the edges that satisfies $\varpi_\mathrm{s}(e)=v\neq \varpi_\mathrm{e}(e)$. The vertices of the graph $G(V,E)$ are equidistantly distributed on a circle. Thus, one can count the number of vertices between $v$ and $\varpi_\mathrm{e}(e)$ in a counterclockwise direction. Let us denote this number by $N\equiv N(v,\varpi_\mathrm{e}(e))$, where $N(v,w)=0$ means that $v$ and $w$ are next neighbors and define for convenience $N(v,v):=-\infty$ for all $v\in V$. Note furthermore that $N(v,w)=n-2-N(w,v)$ if $v\neq w$, as then $N(v,w)+N(w,v)=n-2$. The rotation-invariance of the graph under rotations by integer multiples of $360^\circ/n$ reveals now that for the vertex $\varpi_\mathrm{e}(e)$ there exists an edge $e'\in E$ such that $\varpi_\mathrm{s}(e')=\varpi_\mathrm{e}(e)\neq \varpi_\mathrm{e}(e')$ and $N(\varpi_\mathrm{e}(e),\varpi_\mathrm{e}(e'))=N(v,\varpi_\mathrm{e}(e))$. Thus one can append the corresponding edge $e'$ and vertex $\varpi_\mathrm{e}(e')$ to the walk $W_v=(v,e,\varpi_\mathrm{e}(e))$ and continue this process. Here it is important to recognize that $N(v,\varpi_\mathrm{e}(e))>-\infty$ such that this process must after at most $n$ steps yield a closed directed walk that visits $v$ at least once. This is exemplarily visualized in Figure~\ref{fig:rotatio:symmetry:closed:walk:requirement}. Note that by Definition~\ref{def:symmetry:admission} one cannot assume that $N(v,\varpi_\mathrm{e}(e))=1$ holds, as one is only guaranteed that one equidistant distribution of the vertices on a circle exists that satisfies the symmetry constraint, but not that all distributions satisfy the desired symmetry condition.
       \begin{figure}[htpb]
           \centering
           \includegraphics[width=0.85\linewidth]{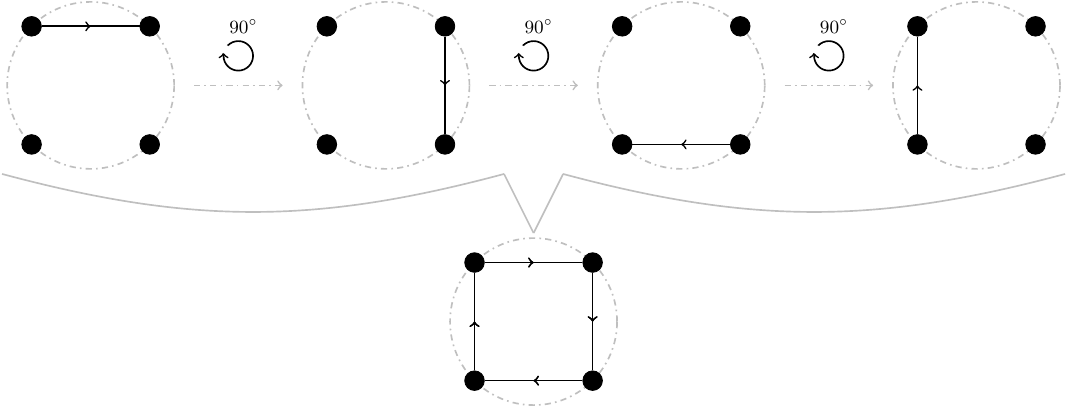}
           \caption{Illustration of a directed graph with four vertices. The graph is assumed to admit a $\Z_4$-symmetry. For clarity, one instance is shown with only one edge, while the remaining edges are omitted. The graph is assumed to contain an edge that connects two distinct vertices, as illustrated in the first instance. The symmetry condition requires that the graph remains invariant under successive rotations by $90^\circ$. Thus, the complete graph must include the four edges that are obtained via the depicted rotations. Note that the labels are not displayed, as they are irrelevant for this argument.}
           \label{fig:rotatio:symmetry:closed:walk:requirement}
       \end{figure}   
       \item \textbf{Claim:} \emph{There exists a closed directed walk $W_v$ for every vertex $v\in V$ such that $W_v$ visits $v$ at least once and induces a self-contained subgraph.} We have in the previous argument already shown that, for every vertex $v\in V$, there exists an edge $e\in E$ satisfying $\varpi_\mathrm{s}(e)=v$ and $\varpi_\mathrm{s}(e)\neq\varpi_\mathrm{e}(e)$. We may now write $\varpi_\mathrm{s}(e):=v^{(0)}$, $\varpi_\mathrm{l}(e):=v^{(q)}$, and $\varpi_\mathrm{e}(e):=v^{(p)}$, where $v^{(r)}$ denotes the vertex obtained from $v^{(0)}$ by $r$ successive rotations by $360^\circ/n$. Note that we assume that all superscripts are understood modulo $n$. Since the starting and end vertices are distinct, one has $p\not\equiv0\pmod n$. The invariance of the graph under the corresponding rotations now implies that, for every $r\in\N_{\geq0}$, the graph contains an edge $e_r:=(v^{(rp)},v^{(q+rp)},v^{((r+1)p)})$. Since the vertex set is finite, there exists a smallest positive integer $m$ such that $mp\equiv0\pmod n$. Consequently,
        \begin{align*}
            W_v:=\left(v^{(0)},e_0,v^{(p)},e_1,v^{(2p)},\ldots,e_{m-1},v^{(mp)}=v^{(0)}\right)
        \end{align*}
        is a closed directed walk within $G(V,E)$ that visits the initial vertex $v=v^{(0)}$.
        
        We now distinguish two cases. If $p=q$, i.e., if $\varpi_\mathrm{l}(e)=\varpi_\mathrm{e}(e)$, then the labeling vertex of every edge $e_r$ coincides with its end vertex, since
        \begin{align*}
            \varpi_\mathrm{l}(e_r)=v^{(p+rp)}=v^{((r+1)p)}=\varpi_\mathrm{e}(e_r).
        \end{align*}
        Hence, every labeling vertex of an edge occurring in $W_v$ is already visited by $W_v$. Therefore, the graph induced by the closed directed walk $W_v$ is self-contained, as desired.
        
        Suppose now that $p\neq q$. In this case, the closed directed walk constructed above need not yet induce a self-contained subgraph, since some labeling vertices $\varpi_\mathrm{l}(e_r)=v^{(q+rp)}$ may not be visited by $W_v$. Assume that this occurs, and choose an edge $e'$ from the previously constructed walk $W_v$ whose labeling vertex $\varpi_\mathrm{l}(e')$ is not visited by $W_v$. Since the end vertex $\varpi_\mathrm{e}(e')$ is visited by $W_v$, one necessarily has $\varpi_\mathrm{l}(e')\neq \varpi_\mathrm{e}(e')$. By Proposition~\ref{prop:conditions:for:edges}, the antisymmetric counterpart $e''=(\varpi_\mathrm{l}(e'),\varpi_\mathrm{s}(e'),\varpi_\mathrm{e}(e'))$ also belongs to $E$. 
        
        Applying the preceding rotational construction to the edge $e''$, we obtain a closed directed walk $W_{\varpi_\mathrm{l}(e')}$ that visits the labeling vertex $\varpi_\mathrm{l}(e')$. Moreover, because the first edge used in this construction is $e''$, the walk $W_{\varpi_\mathrm{l}(e')}$ also visits its end vertex $\varpi_\mathrm{e}(e')$. The vertex $\varpi_\mathrm{e}(e')$ already belongs to $W_v$. Since $W_{\varpi_\mathrm{l}(e')}$ is closed and visits $\varpi_\mathrm{e}(e')$, we may cyclically reorder its sequence so that it starts and ends at $\varpi_\mathrm{e}(e')$. We can therefore insert $W_{\varpi_\mathrm{l}(e')}$ into $W_v$ at an occurrence of $\varpi_\mathrm{e}(e')$. This produces a new closed directed walk $W_v^{(1)}$ that visits every vertex visited by $W_v$, together with the previously missing labeling vertex $\varpi_\mathrm{l}(e')$. The new walk $W_v^{(1)}$ may itself contain edges whose labeling vertices are not yet visited. Whenever this happens, we repeat the same construction: choose such an edge, use its antisymmetric counterpart to construct a closed directed walk that visits the missing labeling vertex and passes through the end vertex of the chosen edge, and insert this new closed walk into the previously constructed walk at their common end vertex.  See for a visualization of this procedure Figure~\ref{fig:closed:directed:walk:under:Z:n:symmetry}.
        
        At every iteration, at least one vertex that was not previously visited is added to the resulting closed directed walk. Since $V$ is finite, this procedure terminates after finitely many steps. Let $\Bar{W}_v$ denote the resulting closed directed walk. By construction, every labeling vertex of every edge occurring in $\Bar{W}_v$ is also visited by $\Bar{W}_v$. Hence, the graph induced by $\Bar{W}_v$ is self-contained. Moreover, the original vertex $v$ remains part of the walk throughout the construction. Thus, for every vertex $v\in V$, there exists a closed directed walk that visits $v$ at least once and induces a self-contained subgraph.
       
       \begin{figure}[htpb]
           \centering
           \includegraphics[width=0.4\linewidth]{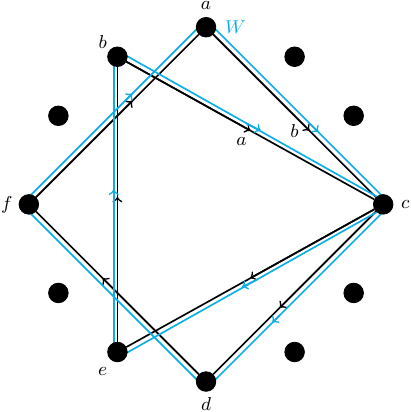}
           \caption{Illustration of a graph with 12 vertices. The graph contains two closed directed walks: $W_1$, which visits the vertices $c,d,f,a$, and $W_2$, which visits the vertices $c,e,b$. However, each walk induces a graph that is not self-contained, because the edges targeting vertex $c$ are labeled by vertices not belonging to the respective walk. By joining these two walks into the blue walk $W$, which visits the vertices $c,e,b,c,d,f,a,c$ in this sequence, the resulting induced graph contains the same edges but also includes the vertices labeling these edges.}
           \label{fig:closed:directed:walk:under:Z:n:symmetry}
       \end{figure}
   \end{enumerate}
   It follows immediately that $\g$ is perfect \cite{Burde:2024}, i.e., $\mathcal{D}\g=\g$, which can be verified using Algorithm~\ref{alg:generating:generating:the:graph:derived:altered} and Lemma~\ref{lem:derived:series:graph:alg:valid}. Thus, by the Levi-Mal'tsev decomposition theorem $\g=\mathfrak{s}\ltimes\mathfrak{r}$, where $\mathfrak{s}$ is semisimple and $\mathfrak{r}$ is solvable \cite{Kuzmin:1977}. Furthermore, by \cite{Souris:2024}, since $\g$ is perfect, $\mathfrak{r}$ is nilpotent. 
\end{proof}

Building upon the previous propositions, we now strengthen the results by showing that under certain dimensional constraints, the presence of rotational symmetry in the associated minimal graph implies simplicity of the Lie algebra.

\begin{theorem}
    Let $\g$ be a non-abelian complex $n$-dimensional minimal-graph-admissible Lie algebra associated with a minimal graph $G(V,E)$ that admits a $\Z_n$ symmetry, where $n\leq 3$. Then $\g$ is simple.
\end{theorem}

\begin{proof}
    We denote the graph of $\g$ with $G(V,E)$ and proceed by considering the cases, where $\dim(\g)\in\{1,2,3\}$:
    
    \begin{itemize}
        \item \textbf{One-dimensional case.} There does not exist any non-abelian one-dimensional Lie algebra.
        \item  \textbf{Two-dimensional case.} The only complex non-abelian two-dimensional Lie algebra is $\mathfrak{aff}(\C)$. For the associated graph to admit a $\Z_2$ symmetry, there would need to exist an edge $e\in E$ that satisfies $\varpi_\mathrm{s}(e)=v_1$ and $\varpi_\mathrm{e}(e)=v_2$, since this algebra is non-abelian and due to Proposition~\ref{prop:conditions:for:edges}. By the symmetry requirement, there must also exist an edge $e'\in E$ satisfying $\varpi_\mathrm{s}(e')=v_2$ and $\varpi_\mathrm{e}(e')=v_1$. This configuration violates the antisymmetry of the Lie bracket. Therefore, no graph associated with $\mathfrak{aff}(\C)$ can admit a $\Z_2$ symmetry.
        \item \textbf{Three-dimensional case.} In Figure~\ref{fig:all:3:vertex:subgraphs}, we classified all minimal graphs with three vertices. A close inspection of these graphs reveals that the only graph admitting a $\Z_3$-symmetry is of Type IX, which is always associated with a simple Lie algebra, as the corresponding Lie algebra is by Theorem~\ref{thm:non:solvability:condition:strong} non-solvable, and the only non-solvable three-dimensional Lie algebras are by \cite{Bianchi:1903} simple.
    \end{itemize}
\end{proof}

We present here some final remarks. One might conjecture that there exists no non-abelian five-dimensional Lie algebra that can be associated with a minimal graph that admits a $\Z_5$-symmetry. Furthermore, one might conjecture that the only non-abelian six-dimensional Lie algebra that can be associated with a graph that admits a $\Z_6$-symmetry might, in the complex case, be the algebra $\sl{2}{\C}\oplus\sl{2}{\C}$. While these conjectures are consistent with the classification of semisimple Lie algebras, according to which there exists no semisimple Lie algebra of dimension five, and the only semisimple Lie algebra of dimension six is $\sl{2}{\C}\oplus\sl{2}{\C}$~\cite{helgason2024differential}, they do not follow from this classification alone, since perfect non-semisimple Lie algebras with nilpotent radicals already exist in dimensions five and six~\cite{Burde:2024}.

Moreover, Definition~\ref{def:symmetry:admission} can clearly be generalized to arbitrary symmetry groups. For instance, the graph depicted in Figure~\ref{fig:rotations:of:su2} is also invariant under three reflections along the bisectors of the triangle formed by the three vertices, combined with appropriate permutations of the vertex and edge labels. Hence, the corresponding graph $G(V,E)$ admits a $D_3$-symmetry, where $D_3$ is the dihedral group generated by the aforementioned rotations and reflections~\cite{Seiler:2008}.

\subsection{On the Killing form}

In this context, it could be helpful to investigate the role of the Killing form \cite{Kirillov:2008}. This is a symmetric bilinear form that arises naturally in the structure theory of Lie algebras and Lie groups \cite{Knapp:1996}. It is commonly defined as $B(x,y):=\operatorname{Tr}(\operatorname{ad}_x\circ\operatorname{ad}_y)$ for every pair of elements $x,y\in \g$ \cite{Knapp:1996}. One important property is that $B$ has an invariance property in the sense that $B(\operatorname{ad}_x(y),z)=-B(y,\operatorname{ad}_x(z))$ or equivalently $B([x,y],z)=B(x,[y,z])$ for all $x,y,z\in\g$ \cite{Knapp:1996}. The Killing form is of particular interest for Lie algebras over fields of characteristic zero, as in this case, Cartan's criterion states that $B$ is nondegenerate if and only if the Lie algebra $\g$ is semisimple \cite{Dieudonne:1953}. 

Consider now a basis $\mathcal{B}$ of a finite-dimensional minimal-graph-admissible Lie algebra $\g$ satisfying the relation \eqref{eqn:desired:basis}, i.e., $\operatorname{ad}_{x_k}(x_\ell)=[x_k,x_\ell]=\alpha_{k\ell}x_{\delta(k,\ell)}$. Applying now $\operatorname{ad}_{x_j}$ to this element yields: $$\operatorname{ad}_{x_j}(\operatorname{ad}_{x_k}(x_\ell))=[x_j,[x_k,x_\ell]]=\alpha_{k\ell}\alpha_{j,\delta(k,\ell)}x_{\delta(j,\delta(k,\ell))}.$$ The entries of the Killing matrix are given by $B_{jk}=\operatorname{Tr}(\operatorname{ad}_{x_j}\circ\operatorname{ad}_{x_k})$. Since the trace of a linear map is the sum of the coefficients of $x_\ell$'s in the image of the $x_\ell$'s, only terms with $\delta(j,\delta(k,\ell))=\ell$ contribute. Hence:
$$B_{jk}=\sum_{\ell\in\mathcal{N}} \alpha_{k\ell}\alpha_{j,\delta(k,\ell)}\delta_{\ell,\delta(j,\delta(k,\ell))}, $$ where $\delta_{\ell,\delta(j,\delta(k,\ell))}$ denotes the Kronecker delta that evaluates to one if $\delta(j,\delta(k,\ell))=\ell$, and zero otherwise. As we stated before, the Lie algebra $\g$ is semisimple if and only if the Killing form is non-degenerate \cite{Dieudonne:1953}. This motivates the following observation for the entries $B_{jk}$ of the Killing matrix:
\begin{lemma}\label{lem:killing:form:helper}
    Let $\g$ be a finite-dimensional minimal-graph-admissible Lie algebra associated with the minimal graph $G(V,E)$.  Fix $j,k\in\mathcal{N}$ and let $v_j,v_k\in V$ be the corresponding vertices. Then $B_{jk}\neq 0$ only if there exist $v_\ell,v_{\ell'}\in V$ and two edges $e_1,e_2\in E$ such that $\varpi_\mathrm{s}(e_1)=v_\ell=\varpi_\mathrm{e}(e_2)$, $\varpi_\mathrm{l}(e_1)=v_k$ and $\varpi_\mathrm{l}(e_2)=v_j$.
\end{lemma}
Employing Proposition~\ref{prop:conditions:for:edges},  this trace condition can be illustrated graphically by the subgraph substructure in Figure~\ref{fig:Killing:form:condition}. This structure witnesses a nonzero trace contribution: the trace contribution is nonzero if and only if there exists a vertex $v_\ell$ that returns to itself by traversing first an edge labeled by $v_k$ and then an edge labeled by $v_j$, i.e., along the two-step pattern visualized in Figure~\ref{fig:Killing:form:condition}.

\begin{figure}
    \centering
    \includegraphics[width=0.25\linewidth]{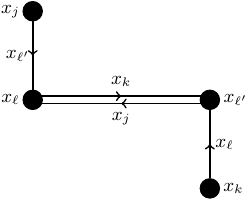}
    \caption{Illustration of the subgraph that appears in every graph associated with every semisimple Lie algebra for a given $\ell\in\mathcal{N}$, provided the associated graph is minimal. Note that multiple vertices may coincide, as for example $k=j$ may occur.}
    \label{fig:Killing:form:condition}
\end{figure}

\begin{proof}
    Let the Lie algebra $\g$ and the associated graph $G(V,E)$ be as stated in the claim above. By the computation above, the $B_{jk}$ components are determined by the coefficients of the elements $\sum_{\ell\in\mathcal{N}}[x_j,[x_k,x_\ell]]=\sum_{\ell\in\mathcal{N}}\alpha_{k\ell}\alpha_{j,\delta(k,\ell)}x_{\delta(j,\delta(k,\ell))}$. Thus, by the definition of the trace, if $B_{jk}\neq 0$, then there exists an index $\ell\in\mathcal{N}$ such that $\delta(j,\delta(k,\ell))=\ell$, i.e., a vertex $v_\ell$ such that $[x_k,x_\ell]\propto x_{\ell'}$ and $[x_j,x_{\ell'}]\propto x_\ell$. Under the graph construction rules, these relations are exactly the existence of edges $e_1=(v_\ell,v_k,v_{\ell'})$ and $e_2=(v_{\ell'},v_j,v_\ell)$, i.e., the condition stated in the lemma. Note that the prefactor $\alpha_{k\ell}\alpha_{j,\delta(k,\ell)}$ is by the definition of the matrix $\boldsymbol{\alpha}$ and the index-function $\delta$, non-zero in this case, as $\ell\in\mathcal{N}$ and $\ell\neq 0$.  Conversely, if no such $\ell$ exists, then every Kronecker-delta factor vanishes and $B_{jk}=0$.
\end{proof}

We can now connect this result to Cartan's criterion for semisimplicity. Indeed, if for some  $j\in\mathcal{N}$ one had $B_{jk}=0$ for all $k\in\mathcal{N}$, then the $j$-th row of the Killing matrix would be zero, hence the Killing form would be degenerate. Therefore, if $\g$ is semisimple, for every $j\in\mathcal{N}$ there must exist $k,\ell\in\mathcal{N}$ such that $[x_j,[x_k,x_\ell]]\propto x_\ell$. This allows us to provide the following necessary subgraph condition for semisimplicity in minimal graphs:
\begin{proposition}\label{prop:semisimple:killing:form:condition}
    Let $\g$ be a finite-dimensional minimal-graph-admissible Lie algebra associated with the minimal graph $G(V,E)$. If $\g$ is semisimple, then, for every vertex $v_j\in V$, there exist vertices $v_k,v_\ell$ such that the Figure~\ref{fig:Killing:form:condition} structure occurs in $G(V,E)$. 
\end{proposition}

\begin{proof}
    By Cartan's criterion, $\g$ is semisimple if and only if the Killing form $B$ is nondegenerate. This implies that the Killing matrix is nondegenerate and thus contains no zero row. If one now fixes an index $j\in\mathcal{N}$, one cannot have that $B_{jk}=0$ for all $k\in\mathcal{N}$. Hence, there must exist an index $k\in\mathcal{N}$ such that $B_{jk}\neq 0$. By Lemma~\ref{lem:killing:form:helper}, the Figure~\ref{fig:Killing:form:condition} subgraph structure must occur for $(v_j,v_k)$, completing the proof.
\end{proof}

We can deduce the following practical rule following Lemma~\ref{lem:killing:form:helper}: Fix $j,k\in\mathcal{N}$ and define the set $$\kap_{jk}:=\{\ell\in\mathcal{N}\;\mid\,\delta(j,\delta(k,\ell))=\ell\}.$$ Then $\kap_{jk}=\emptyset$ (and hence $B_{jk}=0$) if there is no Figure~\ref{fig:Killing:form:condition} subgraph structure compatible with the vertices $v_j$ and $v_k$.

This subgraph structure provides a necessary local constraint on minimal graphs of semisimple Lie algebras and allows one to determine a possible support of the Killing matrix (i.e., the set of indices $(j,k)\in\mathcal{N}\times\mathcal{N}$ where $B_{jk}\neq 0$ is possible) efficiently. However, nondegeneracy is a global property: even if no row is identically zero, the determinant may vanish due to cancellations among different contributions.

To illustrate the comment on the support of the Killing matrix, consider
the graph depicted in Figure~\ref{fig:rotations:of:su2}. This corresponds to a strictly diagonal Killing matrix, as for $j\neq k$, the graph lacks any $x_\ell$ that returns to itself under $\operatorname{ad}_{x_{j}}\circ\operatorname{ad}_{x_k}$, hence $B_{jk}=0$ by Lemma~\ref{lem:killing:form:helper}. If, on the other hand, $j=k$,  the subgraph structure from Figure~\ref{fig:Killing:form:condition} exists within the graph, showing that the Killing form has possibly nonzero diagonal elements.

\end{document}